\documentclass{article}
\usepackage[hidelinks]{hyperref}
\usepackage{amssymb,ComplexSystems}
\usepackage{graphicx,booktabs,mathtools,longtable}

\newtheorem{algo}{Algorithm}

\newtheorem{remark}{Remark}

\newtheorem{background}[proposition]{Proposition}

\newcommand{\FF}{\mathbb{F}_2}
\newcommand{\ZZ}{\mathbb{Z}}

\newcommand{\ANF}{\textsc{anf}}
\newcommand{\Inc}{\operatorname{Inc}}
\newcommand{\conv}{\ast}

\newcommand{\Xor}{\bigoplus}
\newcommand{\xor}{\oplus}

\newcommand{\bin}[2]{\binom{#1}{#2}}
\newcommand{\Zet}{\zeta}

\providecommand{\mathscr}[1]{\mathcal{#1}}
\newcommand{\Zm}{\mathbb{Z}_N}

\newcommand{\lead}{\operatorname{lead}}

\newcommand{\src}[1]{}
\newcommand{\roadmap}[1]{}
\newcommand{\paperhead}[2]{}

\makeatletter
\renewcommand{\verbatim@font}{\normalfont\ttfamily\scriptsize}
\providecommand*{\toclevel@title}{0}
\makeatother

\renewcommand{\mycopyright}{\hbox{\textcopyright}}

\newcommand{\Needspace}[1]{\par\begingroup
\skip0=\lastskip\vskip-\skip0\penalty-100
\dimen0=#1\relax\advance\dimen0 by\skip0
\dimen2=\pagegoal\advance\dimen2 by-\pagetotal
\ifdim\dimen0>\dimen2
  \ifdim\dimen2>0pt\vfil\fi\break
\else\vskip\skip0\fi\endgroup}
\makeatletter
\let\UThreeSubsection\subsection
\renewcommand{\subsection}{\Needspace{9\baselineskip}\UThreeSubsection}
\let\UThreeSection\section
\renewcommand{\section}{\Needspace{10\baselineskip}\UThreeSection}
\makeatother

\newcommand{\EmbeddedFigure}[2]{%
  \resizebox{#1}{!}{%
    \hbox to \csname EFWidth#2\endcsname bp{%
      \vrule width0pt height\csname EFHeight#2\endcsname bp depth0pt
      \pdfliteral{q /EF#2 Do Q}\hss}}}
\begingroup
\catcode`\#=12 \catcode`\%=12 \catcode`\_=12
\catcode`\^=12 \catcode`\~=12 \catcode`\&=12 \catcode`\$=12
\pdfobj reserveobjnum
\expandafter\xdef\csname EF1O1\endcsname{\the\pdflastobj}
\pdfobj reserveobjnum
\expandafter\xdef\csname EF1O2\endcsname{\the\pdflastobj}
\pdfobj reserveobjnum
\expandafter\xdef\csname EF1O3\endcsname{\the\pdflastobj}
\pdfobj reserveobjnum
\expandafter\xdef\csname EF1O4\endcsname{\the\pdflastobj}
\pdfobj reserveobjnum
\expandafter\xdef\csname EF1O5\endcsname{\the\pdflastobj}
\pdfobj reserveobjnum
\expandafter\xdef\csname EF1O6\endcsname{\the\pdflastobj}
\pdfobj reserveobjnum
\expandafter\xdef\csname EF1O7\endcsname{\the\pdflastobj}
\pdfobj reserveobjnum
\expandafter\xdef\csname EF1O8\endcsname{\the\pdflastobj}
\pdfobj reserveobjnum
\expandafter\xdef\csname EF1O9\endcsname{\the\pdflastobj}
\pdfobj reserveobjnum
\expandafter\xdef\csname EF1O10\endcsname{\the\pdflastobj}
\pdfobj reserveobjnum
\expandafter\xdef\csname EF1O11\endcsname{\the\pdflastobj}
\pdfobj reserveobjnum
\expandafter\xdef\csname EF1O12\endcsname{\the\pdflastobj}
\pdfobj reserveobjnum
\expandafter\xdef\csname EF1O13\endcsname{\the\pdflastobj}
\pdfobj reserveobjnum
\expandafter\xdef\csname EF1O14\endcsname{\the\pdflastobj}
\pdfobj reserveobjnum
\expandafter\xdef\csname EF1O15\endcsname{\the\pdflastobj}
\pdfobj reserveobjnum
\expandafter\xdef\csname EF1O16\endcsname{\the\pdflastobj}
\pdfobj reserveobjnum
\expandafter\xdef\csname EF1O17\endcsname{\the\pdflastobj}
\pdfobj reserveobjnum
\expandafter\xdef\csname EF1O18\endcsname{\the\pdflastobj}
\pdfobj reserveobjnum
\expandafter\xdef\csname EF1O19\endcsname{\the\pdflastobj}
\pdfobj reserveobjnum
\expandafter\xdef\csname EF1O20\endcsname{\the\pdflastobj}
\pdfobj reserveobjnum
\expandafter\xdef\csname EF1O21\endcsname{\the\pdflastobj}
\pdfobj reserveobjnum
\expandafter\xdef\csname EF1O22\endcsname{\the\pdflastobj}
\pdfobj reserveobjnum
\expandafter\xdef\csname EF1O23\endcsname{\the\pdflastobj}
\pdfobj reserveobjnum
\expandafter\xdef\csname EF1O24\endcsname{\the\pdflastobj}
\pdfobj reserveobjnum
\expandafter\xdef\csname EF1O25\endcsname{\the\pdflastobj}
\pdfobj reserveobjnum
\expandafter\xdef\csname EF1O26\endcsname{\the\pdflastobj}
\immediate\pdfobj useobjnum \csname EF1O1\endcsname {<< /F2 \csname EF1O2\endcsname\space 0 R /F1 \csname EF1O9\endcsname\space 0 R >>}
\immediate\pdfobj useobjnum \csname EF1O2\endcsname {<< /Type /Font /Subtype /Type0 /BaseFont /GCWXDV+DejaVuSans-Oblique /Encoding /Identity-H /DescendantFonts [ \csname EF1O3\endcsname\space 0 R ] /ToUnicode \csname EF1O8\endcsname\space 0 R >>}
\immediate\pdfobj useobjnum \csname EF1O3\endcsname {<< /Type /Font /Subtype /CIDFontType2 /BaseFont /GCWXDV+DejaVuSans-Oblique /CIDSystemInfo << /Registry <41646f6265> /Ordering <4964656e74697479> /Supplement 0 >> /FontDescriptor \csname EF1O4\endcsname\space 0 R /W \csname EF1O6\endcsname\space 0 R /CIDToGIDMap \csname EF1O7\endcsname\space 0 R >>}
\immediate\pdfobj useobjnum \csname EF1O4\endcsname {<< /Type /FontDescriptor /FontName /GCWXDV+DejaVuSans-Oblique /Flags 96 /FontBBox [ -1016 -351 1660 1068 ] /Ascent 929 /Descent -236 /CapHeight 0 /XHeight 0 /ItalicAngle 0 /StemV 0 /FontFile2 \csname EF1O5\endcsname\space 0 R /MaxWidth 695 >>}
\immediate\pdfobj useobjnum \csname EF1O5\endcsname stream attr{/Length1 4524 /Filter [/ASCIIHexDecode /FlateDecode]}{
789cb5366d7054d775e7be73ef7bbb6f57abddd56a2581102bad5642111244429f50bc803e40924120212440a095562bf1a10ff42d5401092218db09
18c362546c888b5decb88e423a44318ec7499cda2d653c2d30198fa7c14e98e9a4833dee0c495a4d75d5f3566b179866a6f99177f7dc7bceb9e79eaf
77eedb030c001c3471b057969557401c9801583271dd95b59beaca866a06895e41b4a7b26eebdacccb4ffc9ce87aa2c39bea96e5edcdecfe80e8fb44
37b475057a99aefe23805246f495b6a101cf8b17a73288a67db63fd4dbd1f56f7bffe33d32d648fbcf7404fa7b41a301fc1ed1d68efda3a11fbcf3d3
62a2bf0010d73adb0341ade1674e00fd0aed177612c3fa9a207bfa2744a777760d8c78eeb37500166ed0fb7bda02ca87fc23a29710edea0a8cf462b7
4af1585619fe7707bada730bd7f511bd937cb2f7f6f40fc062f60440cc7eda5fd1dbd7de7bfa39258de84be483158cdc5861fe51884230457806e890
0c29c0ca2a6aea896bd887b9b9a8ecfcfe693803aec8fec1405fa01526027d5ddd30d1da17d803136d81ee7e9a3bdbfb681eeddb0f131ded3d8477f4
b5ef8389ce4037c974b6b712675fa03b0013fb033d1e631e200d5d81814e98e8de67707a3a025d30d137d84d9203a1ee0e9a3b0dfd8f7914f58bdbd8
291094db7c711e0a598ab1ce95e32f21a45096158b8a88dca2f0680c5f3db5a1f22071b6c088ea922e36a975b1df3c228351488e440eb0932816a139
f82359345106396968835e18809188670f517357e75e9e0bcf3d3f77ea51cbd1fcd744291139350f86976d5130f8bd51500906a240b50523513091c3
5eb60ca6e106c14fe175785f791f2e41118d2df00e3ba1e4d0ceab701c3e120afc1d5c6025ccc54a68f796ea5247c5317185f677d2d9f5a4e523789e
741a9aa6e1b032a6d44208de17376192464f84ff05bcc58ec21d78016e28ebe1f77014ebe1248d49e8e720ee301da4920d970d4b34005a0916618eb8
13195fc06118837ab8ac4eab2e762be2f5abece774d37e4b3edfc29d78804e5c802bec18f7f22b7c3d9c9cf7175be0a472944df296c818a3f73a0c17
780b7b5d75c12f0c5f89534b9e86e06d8261b8c956b26378823c1b333c1077e0a656c597cd7ba58d6301c503046fc255c8c1309d8fc4a286e0821222
5bbf274f6e62196491b67eba76f46580f86baae0a83058eab14f29be0dc129ffe646cf074da9394b1f233d76cd3305b55331a39ee9b9b9da46be5034
4d89e429f499a6b8cffbe91fdbfc346769756da367ea87e56551ade52d65c4ab6b24d4a0884dfcf2b29cf9db4bf5473587841551e9dca4284d549ddf
f1af62c21d0fc2ed8e77111227e2ddf1a5b1d61833dbee84edee0a578ccb1d1fb320ce196b33513826b7754182d92d9217b895f8858bec0feeddbeef
2829719438128cc959f27558b6eac1fddbf74b9c2525cbc919cd2e3ed3e837bf34a55d6d8c65acd91fdf18db646f723439f7c5eeb3ef738cc58ed9c7
1c7a33c3d4bcc2821519deb8149610e7c55c96cdd0cb52f3c83b55e121c6cfffe4f0f0a07f6bf1e13b83394307ef8edda0927effc37f7d5b2c9bfdfc
b5a7ba06665f8a1d9a929decf099e0ec0971e7fad4857bc63ddd35f71b5e45b5b018aafdc98e6afbc6981a1754b38dbc66d141d7668f9d3940b74142
9ac7fee03d23208ac37eefc17f3fb86f970f3e5fceaaa79cf58d3f01cfdcbba4cb43c35fec282a6e625a2ef3a6a99a8dc5bbdcf97985458479d3320a
56388b56b3fc3c37aff2ad3d5bdb78b1b061d7b9c99aef357b6b5f6bbafbc5f858fdafbfb57192175c5bbe7c75c9b67f1f3f772637f75a6aeadd0fdf
64bec1da2d5e7a5333b2815f237f1de083dffa375974c56a16e95e2e548da349a4a77bd3885cec5152528537ddeb8a53e29dc4f0193ba5a94eb4d6a4
604dfcc194cd6987d30f65a465f8cc90ea75a6832d31dde7cdb0dfae9e32d7574fe9f53baaa72cc66435a61863b21953ac31d98dc951bfa3f1c7e09d
7bb7b8e9de6dbff970e2eec427126313793365c969bc64672455f7ee51aa3e773cf2cea9240ccc58e60b601ea3b9298d5d6bc9eccdbc9889cd7ef3a5
cc539947320f67f266466ffee174ae662c92c77897329fd94c2a867cfccf23235bce9696941d19da76b6a0a462f6efabfe61ec8d77867aaadeea187e
7949ff34cb9956ae3e7fa1be76d327ddcf9edbb271cbafd89e40f0d31bafca5bc39beada5a67df54757a8f47e993dd2a5e822498f4af369b34bab1aa
8ec299e0544482338920512424384b9d688e633d09f170c83c6e712d484a4ca07be2b45a74b34a373c16626317d8df7b28a1b66856296fb146deecef
2644b2b4ecfe8355efddcf7b3449ff9b1aba1b8971c6ddb0242989711d89a3eed124d1ccb4548ce46325d322a5555814e72dca8fcb476c4d96cf0dab
ea935bdf5a3c75f9bbaa3a7c61d1f9d5d37fc15b4ece5e4acb8df1974e1ef9664251aad27292ad963fa35b3077c3b8f9bc052cec907f3dcf169c290c
b3158559b2291f7ab65967c662d1359366ce36993491ad6a2acf46ae696aa9ae71c64d7048283a65cbf8ba99aca05a8d62b2d647c37ea8665423f67f
a2c08c4fc2bdbc8468d4a6f96f41044c9f3d4e53124c9c41b3ffb60095695c0837737397f81a64b10c25137d2243cd326599b3c8cb52286685bc4014
a8455a89bed252051b58252f17e56aa5b641afb2344003db2a1ad4466d9b390421165242d82e426a8739a4872c4330c886f8a0185487f521cb093cce
bf2126d4e3dab7ccc7f5a72d6721cc5ec073fcb408abe7f457e015cb8ff0476a95999b35555fc0335836cfd232f44256cc8bb542bd826de01bb40abd
91ed619d5aa73ecc47f467f909ed983ec95ee4e7b5b03ead7da02d6e66cdeceaa918d66c56bcc8bc66961fc7f8cd59c956defa4c06ef2a99f2e35bbf
64c3bc65f6e3d95fb0b76599b25171ca7dec6cb4d7509ace27f53e65db1dbbea77b0d8146908fef990edee97eb1f3e99cd8fa9303d13f9da7fd5f3d0
7f58975c443dddb13f7cf25f3f8ea98020f5b80f3f666efc171aea5f2778869a899350c443b04bb8e666d00b47f9db46e57cf51c847f61136c56f938
e2931937423674529fa2809d7a1dea38d0a6c4447b420db61b9d10377aeae5911ec9c019b8899ac715b0b18a288e90419df53cce1fc20524b283515c
8574761ad6d1ff712f8c421fec810eb23e40dfe425d40965d19a07cb69e413d64a121e584b3203f4cf3c40d2ed10802e584adc0dd04df2b984ad81fd
343cd4037da9ab3f42b5d3da4e6786680e92a4feffb05af895d57ab23444b6f6d2996e9236fc08d0993fcd6219617be95c030c92441bc90622dada23
270291883ca4a59be65e926925bd7b48ce43e77bc87a20b2f7b89eba88967ed814953f40dcf63f22e3794caa21e2613fd13d11ab79e4673e143c72fa
cbb3398f9d55a823ff1dc121aa8affeb3147ea55a1375d0015a4651bec98568ef8e77e2071ca87dfcfc337c3f8b7367ca3df26dec8c3ef497cdd87af
d9f08a0fff268cafcee02b337859e25f97e2cb12bf9b87972ed6894b61bcf8e41a71b10e5fcac3175d78218c7fa5e3a4c4f34e7c611ccf5dc7b0c433
2471661c9f9778fab94a717a1c9fabc45327178a53124f2ec4ef48fcb6c467253e23f1e91329e269892752f0a93c3c2e71c28d47257e53e237241e91
7858e22189e3d53e311ec4bf9438e6c083a3d7c54189a323cd62f43a8e1ee123c33e31d28c237e3eecc3218983611c0862bf0dfb0ef8445f100ff43a
c5011ff63ab187dcea99c16eff9cc42e89fb25ee73e3de3da5626f10f7908d3da5d8b9d1223a13b12364131d7918b2617b1083742c18c63689ad01ab
689518b062cbee24d112c4ddbbec627712eeb263b38e3b77c4889d1277c4e0763ab13d8c4d8d36d1b4041b6db86d061bb65e170d12b7d6378badd771
eb115e5fe713f5cd58efe7753edc2271736daed82cb136173791139b5cb8d1824f92574faec11a5a6a2456573944b50fab1cb841e2fa4a87582fb1d2
811512cb2596495cb7765cac93b8761cd748f4cfe01333b87a065715ae15ab24aefc004b092badc312e9efc5e2712c22b290e788c2b558207185c4fc
52cc9bc1e5565c263147e25289d9b49dfd75fc9a1db3d02eb2bcb8240533336c2233881936f4315df8f230dd9a28d2c7d12b4a8557621a5169d73195
e45317a267b145786291baed77fd937cb10553cc98e2e78bec984ce2c9615c18c605493eb1208849894e91e4c3442726b87d22610dba7d182fd12531
6e069d8e24e194e820ad8e24b44b8c9568230db630c690c19871b45aacc29a88162bea923aba52610aa34ae2aa4441518852e444f11c443bfdd3ea42
4944a623f373484636cd82c7becdb2ffbc0ffc99f5ffa9cfa2ff018f05ecd0
>}
\immediate\pdfobj useobjnum \csname EF1O6\endcsname {[ 82 [ 695 ] 99 [ 550 ] 112 [ 635 ] 116 [ 392 ] 120 [ 592 ] ]}
\immediate\pdfobj useobjnum \csname EF1O7\endcsname stream attr{ /Filter [/ASCIIHexDecode /FlateDecode]}{
789c636018028085a00a569c326c509a1d4a73000003aa001f
>}
\immediate\pdfobj useobjnum \csname EF1O8\endcsname stream attr{ /Filter [/ASCIIHexDecode /FlateDecode]}{
789c5d914d6f84201086effc8a396e0f1b76addb7a30269bedc5433f52db93e9416134241508e2c17fdf01ac4d4a026f78665e98017eab9f6aad3cf0
376744831e06a5a5c3d92c4e20f4382acdce194825fcb68bab983acb38999b75f638d57a30ac2c81bf5370f66e85c3559a1eef1800f05727d1293dc2
e1f3d624d42cd67ee384dac3895515481ce8b8e7cebe7413028fe6632d29aefc7a24db5fc6c76a11b2b83fa792849138db4ea0ebf488ac3cd1a8a01c
68540cb5fc17bf24573fece9978cd293b449bf027eb88f38489b34e2c753c441daa409e709e71bce375c245c6cb8204c35fdde1eca0b6fb9f72e16e7
a8edf8e0b1dfd0a9d2b8ff893536b8c2fc015b86869b
>}
\immediate\pdfobj useobjnum \csname EF1O9\endcsname {<< /Type /Font /Subtype /Type0 /BaseFont /BMQQDV+DejaVuSans /Encoding /Identity-H /DescendantFonts [ \csname EF1O10\endcsname\space 0 R ] /ToUnicode \csname EF1O15\endcsname\space 0 R >>}
\immediate\pdfobj useobjnum \csname EF1O10\endcsname {<< /Type /Font /Subtype /CIDFontType2 /BaseFont /BMQQDV+DejaVuSans /CIDSystemInfo << /Registry <41646f6265> /Ordering <4964656e74697479> /Supplement 0 >> /FontDescriptor \csname EF1O11\endcsname\space 0 R /W \csname EF1O13\endcsname\space 0 R /CIDToGIDMap \csname EF1O14\endcsname\space 0 R >>}
\immediate\pdfobj useobjnum \csname EF1O11\endcsname {<< /Type /FontDescriptor /FontName /BMQQDV+DejaVuSans /Flags 32 /FontBBox [ -1021 -463 1794 1233 ] /Ascent 929 /Descent -236 /CapHeight 0 /XHeight 0 /ItalicAngle 0 /StemV 0 /FontFile2 \csname EF1O12\endcsname\space 0 R /MaxWidth 974 >>}
\immediate\pdfobj useobjnum \csname EF1O12\endcsname stream attr{/Length1 11088 /Filter [/ASCIIHexDecode /FlateDecode]}{
789cd57a7b5c55d5b6f0986bacb5f6fbc9de3c64c3deb0d9203e09c227e58e7ceb314a32b53450407c82a29692079f20a9c7276866ba2b345333328f
819a69526ae4cd8edaa95b1d2b2d7b90713a56f720cc7dc75a1b8ccebdf7fbf5fd717fdfef5b8bb1e6638d39e7788f39f7021800d8e9218267d8e021
43210c7a003002081f9675dfd8b8677da9d41e0c202c1f36f6c1cca49a41a701f02cbdaffec33dd9c3639a279ea4c13984f3c07d637ba7cee8bf3106
403a4cefc74d9d9d5b3cf49bd46a00d948ef574e5d38df03d363fa036836509b17144f9b3df7ce85330074d486fdd3724b8a414337e833a86d9c366b
51c17f7c98bc82dac301227616e6e7e66927bd590e10f735bdef53481da65d9a2d00f11e6a2714ce9efff8cb132389de78a217f6cc2a9a9abbb4b694
c6c637517be2ecdcc78bc5a7e4b9005eea03cf9cdcd9f949dfdf759cda8544cfa5e2a292f9adefe4bc0ae053d6df573c2fbf78a0e6ef54f5d17a5221
28b2224ed44ba0164277b54f013df4820c10060f1d9d0de659b9f3e74024c994ae6010e0764dc59e993f6f0e68dbc7317a27a8a516047656c1645f08
46308023b450f013f85fb8823504af115c545608be46f0c9af6b054f116c2798a9beafe93ceaf668d259d04db05a19c75f02e035a1914a8d5f547bce
119cfadfa0fff75d9d65471c5e0cb57faddd7eb7e177ce37f377a131d29e9d3cc9014e08872888815870831712e84d18bd61a4710cd90448ed6364b2
7b6d7b5d47d6c4c046e33b3045c20b61e8696e1d59a109cc60012b61292b75d8e026d8020ed50617e7cecb9d022b73e7cd9e032ba7cccb9d0e2ba7e6
ce29a16761fe3c7a2e9a370b564ecb2fa2fab479f933616561ee1cc229cc9f423d3373e7e4c2ca59b9451ee549b6bc7276eefc42583967a6d253342d
7736ac9cb7600e61ce2f98338d9e85cafc9decbdb32c14ba3e83af210b328ca0f954e91496876286743984a4d43b5f1d6d56d85e7ff07708dd4cf32e
fc1dda9918c255cadbf54e73fca67f62a7f6bc5feb423f805beba99e717b6857d247b8eac95d2053d59ba233857bb11d030987f0740be89944f5aeba
75f43e51e9d555ff06affb6dbc6eb7f1926fe31129a2996d502c474a939ea2666ca8c4bf42816027ea0c32a2561404654685027f0791590543f2a8e5
89b7cb0eee60db35b3d9d5769c8eb543e002456b0087a9c5d4b608dba88c078f1aaf3c444f0fe80303e05ed2eb03900db9900fd3a10816c2a278bb6a
031ea2b83ba4403f9247074e1ee1cc8279219ce055f2c48f831f06ff1abc106c089e0ebe193c123c1c7c25f872f0a5e0febf8575a2eb7fbe42f1f842
7bcbaaae1b0254b502446948b68a5495bc2613a410f421503c8e74497c00791310a5401c51465484d50e4a1c7ea01d9c04d9ed104e90db0e11047904
f9049104d3dba10bc12c8222826802c586141ba5b8098bda2181c46a0f01f8583ad44123dda7601fec607ba85540fd73a927201c8255b0807a4eb346
5629f4a4be3dd00c1709b3021a719f086c24a4512fc0479200375936e96f07ebcf1cacbf46a62032463c2c3e20d689d7c5f3d0572c11cf8b3962094b
c3e7a471d21e82fef816d9cf39a2ae8e5d8112388adf601a1e17078b66b882e7711f7c49ab28b26c84f55003a5448b83154199502a3c403d67a4f3b0
9dee227a7f9eed641789baa36c055c866d280ac36127bb4c7c35c2cfb002b3853252519a5040f49fa1b9ced3f8ed50424e7399e9810bdda98fa8a7b5
a6a8cf18ec295d56ef6628a395b3a146ae931d1a2fada2486c0f3bcd9ae4cd10808bf808cec58fd92ad12bee1587c3fa90043007d6d3dcdb95317201
5b44bc2b77a932bbf09898c3f6c137628e660acdfd96c29162fdc203c451011c27784cb6124f03d92aac244a95b731705e3352ec4de36906cd1255c7
45980e33a8560a07e110f4c46a584f33a9fcca7da59f69e40ef173e2793d5b27fc0ce77130596781788364ad981979f76b1a59125160d0c363ad157c
23f26afdf78ff79c9d10d7b3c7bf343d568da716b26a4d8b3c75c160d678315a9a502bb96ad1a7ad157ddecfffa7979ff7ec312a6bbca7b66dc8e0f6
5987e40ca6beb1e3a9aab4a89bfa870c56df298bd64a3efa1b9153eb995ae879d2faa477c093d6fc013d157f13945d0b79bb12256af98f42a96ca7cc
d4d76f91b7c156b349036897214c6fb67e32aa362c7b7c3de88327fb4d18556b51ebe0ef37e15a6a93ad7fff3bc0dada94c264c1e9b047781385f43b
ed7d85d2f2e52b5605aaabb66c95ed5ff1bbaf5fe703bffc8ebdfdd915d6d044ebd5d07a45ea7a6ebf45a3aca7a1c46b17c3b440eb65dcfc75deb0b4
70bbd32168bc7dece9770a353465557560d58a15b2bd89675cf98c0ff8ee4bf6d6f5ebec4d9af5238aaa1f4b226559afdf063ab6438392e01421422f
3bb546eb27ad19ad194d7740ef4b5436a4305b9c33cee6b5c5a5c7d9f0a0d0b3ede2deb68b424f496cbbb84fa9ec232f65b0336867a78153c48ef21b
7127ac905164511029137197de0dd1d837cd89deb0e68b35cb1ee007f849e627d95604af8aebc9de0d145fbcfe303960878071a37d6da4ce65894597
333a9246de246aacd76e36596fa4b078c166b5a7a5da6d562129156c56f0c62b4f61cd8e679ea1bf679eb9c574fc975bb7f82f4c2765f1f3fc5d82f3
2c8dee3b595a8097f0725ec14bd83ab6882d66eb94ccf03905cf8994a749597e6726064421202dd34040a775cb2e043733582fb5eb9329fa6c6a6855
c4d3947ab3e952530a59d5847876d882165198d437ce26a5fbd214917136923fc5f2df61235b6bf68925c3eb86b75cde471390078823896317ecf427
457589c648978d9461932431d3faac6d8b29e0d828522404ab5e607a578415e5186beba85a67f6a8daf0ec8747d53ab21f264a50b1b2864b4d274fda
ecfddba9b9a952a3b14adf6ba4ef59adcb6a8be84fb4f9531f14c749e3348bc5c5d2c2e88a280dc5c928b10b398c6b3e2c94177429899eef5a0ee551
cbbb2c8f5eeeda0b7ba36d9360928f9848ef037def66e977267ae3654dfadd2c2d55743a648d0c149c4fb58e2631a6e5fee185f2472f3ebef8d2f8af
9963c8c351fce6be7dfb1e631b07ccde3ae2b1eacc7bdfbd23f5eb371fd95d1cc3bf23ee7790be4b88fbae50ecef05ce307db9ce5dee090b384d01dd
66d915f06cf66e94d73a9f4f0e7785013aa25c891eab0b1d6e9d9cac08213cbb837f9dca3f09809c2042f582a66b37af3559bfba61556f924a0af3eb
f26273ddb99ebc381126b158e6748871f18949e9b1c4481fe2aa3b4b0f557ec31e0edaf83cbfc0bf9e7c6646f6d9d927ced4ef3e78a46ae7f3dbc69e
9857726ec257ccf827f4b91b367cfaa3cf77fa8ed4eaf52babf63c565c529a9078d8e379ffd013fb9598419952ac219b12c87797f96398094d8068ca
0434680212c3653a66d4834bd68a46357218883193ca985161ec52464353aa4dd1ebb54b194da9c48baa58f11c29f79ca2d26e06da380d870994821f
83274113ceba4322eb8e7dd818769ff13ed33856c016b0c5b88a9948953a168769b634a757f56394b9c0783abf7cf95cdb64c9d77a15cfb7a6ede501
96735af5e4ab621e511e0393fd5eb18bc6566e8de912d03802d64a93108065a6b59a9ad80817d3a30bf45639d6daca3aebc5da29fa59156f2115591b
6e280eac7830a9873784b413a6c41545e6e074c06fd4a268e3538c6a0bf418dfa38525f04bfc87c9a70b279e9cf9d23befbc74ffb3d9d2e57d7c93c5
c26f7cfb77fe93c7d37847ca911d3b8e242492b4d713f5d56a3c4980f1fe8430194ce5460884cb0157f86e6bc05819bfd1b5d6678cd7b9a262c35c18
e78ef651802123baa686986badd77e351fbf837629ecbc701ecf8b8d52a34c7c1f8a1526b1492c5e763ac243b432672fe68d17b08311af47094771a9
e142cdea5dbb561330dde8a7479fbd68197868e6e74ce2cd5ff0367e8365b1e8d14fe3c0a3cf3d7becd8b3cf1d1516d52524f21ff90f0f4de23f7cf7
15ff560d5053d8ee586597b786b85aa372e585de70afdf17493c25c981d89e01fbc6d8b549cfa7441a13bab99c092e8b8e6226054e4b5c748ab5b5a1
e9664393ca4e8787a8adfee41a9d58f0f5220f4f484b0d575c5b75126f7c42fa9d7dc23a10481fc29a0dbb776fd8b06737dfbd7c2304ff76856f5cb6
e979fecb2fbff05f6a866f5cb17cf3e6e52b360a6f6dafa8d8fe7479c5f6719e434b5fbd70e1d5a5873cf16fafffe8ebaf3f5aff36cb9dbf7cf97c02
d2d332e2a882388a54f5e4d5b8a358394405f4bbc5005486bb03d68de16b7d1a972b2e2c16e2e35d26554d447e4726f88affd4a1a5f086a837bb9c8c
3ee93a19f3666c835bb3cf7edcfe8d1d494f7d558bb287994943907e27a48574139fc83ad822197c3e7ac728d2ce8043b33ee3b798f50b86ccc65fe1
5f8edec1ee6ed7a09b74c34ccc3eee1166f9ee2b16ae26915dfce158616b87fe149f6fa664725af492be34e0f29be515e21e582190b38b10a9b5524e
4d55e2f44d35117a156f6cbe4817e7a297d3987314319ac8efb474def5908ebd5150a5d755d997b12afd4b6e9b412b8445b92530bbc2a528572f1db8
ec621cc92495f2ab1a01d524d0a428b87fca214b3c2333b5c585f8bc5df1c57556721cdbcc063fffcc33cff3e3acfb968d1bb77083205e6f59fa44d5
6ede7cabed6be15cdba7156bd6ae120af8dd45f3e616ef39f94ae5730e4fe3b6b3ff4e71a22478554a220d46411f7f17d3b3e683fa2a1b7b160e8a55
111b6d6bbb68a24c90e2b07651486c579a42dbcf37520e5ba2ddd10291a7784ebbb7f4e9eb34df6e844b4905d797078137332b83e5d70b667cbf92bf
c417b37236b6fc7b69cae54727f333fc43fe113f33f9d18bc387b35d6c1a2b64bb86910e1429ae53a5180603fc91acca0a55ba6576ab5e4b474529ca
34c8062e9de808098ec4464253b3b8df6071ba9d839c8f3a5f764a8ae8d26e4b8c5282d89dd9bc6c335fb77dfb3ade8f9dbdc5180fdee2ef48bddbde
db5451be69cfd58f3ffda26d6f687d89e213adee8257fce93aad06f5b20d45946ca28899744471a2e8acd239aa4ccb0ca224a38d34196e96f45151a2
6d9043ef328a312a710d0a75b6905a331412edfd95fb76822750b3fb217f2ca9da9fb9388c492031499051233ac1c91c423846883e3a6ff984444c92
133589da449d27b60feb230c654385426981b8407a2c6cb5bc5ab34ddea6714f5253644498177b31855f5b9c473196dba2c075f794de7dfea33746ae
79fc9377d85906ad2bda2af9a6aaaa4dc2f1f00d7fe485acac7a4a5ba574f9830fd71d15ee6bbb51b162c52a7289600be9e41b9289064692570855b0
4c647eca1e7e496bbd441157d5416a0aa5117df6f8d749757ef5a7422d65125bdf7e13fc61a07383955905b7c6aaf3eb8a75bb74ba49d8be2b95c51f
da6e34b6dda09cd07259eaaef86129c5969e74ead1830f8ed31ecb6d88d099e1c508b9de6cf394bb8fbaeabd75b6b5114688c048934e6b70a3d63124
91a4feeea5a6d4543554f66eb876b395ecf56d3576da14c1fbe7a4c4a4c4a6b8533c297129f18392fc31fe58bfdbeff1c7f9e3b362b262b3dc599eac
b8acf8aca4e2a4553115b115ee0a4f45dcaaf80d4981a4e6a4d88ea11d833a06e4c4e6b8733c3971c5b1c5ee624f71dcd2d8a5eea59ea571919df3cc
5dacafcd9bae84b1448acb69719d772ce1c2892b0796153d555f5737e8f8ea038d6db798f0c2d69c23d9f92726fea359482b289d52f2d1e1e4d16dcb
f615e49e7aeef593f6b235bd7aed4b4a6a556435976435517650aa76413f7f54977a303bea25ed5a731ddb4ae6035a6198cd6e1812a3eec253d50076
4dc9290d37528ee4c42e8d0dc4a2ea2dede4d1e61c8824d62973e0737575035e79a23108c1c6275e693bf3c2a64d7bf76edaf4021e1126ffb3696f5e
2e1bccb4740fcee5cec6ebd71b09dae92a231d3a209a768a0964cfba72ed6ac9f92293ea8dec5864bdbdceb8d615ed14b44e2d8c12ec96212e95c406
7537ac449b5022bf19ca11c983628a63023117629a63a44130880d12063907454b3d34bdb5bd753df44550c48a84226751b46ed25c2532c5a95b915f
8312e5768d92e7658d58d67ac878feb51967a64cbd3093dfe4675872eb174c5327ec5ebdbdde2c4c9e78e2cc9d771eecd683f5637a16c6eee59f366c
3d7c70a722eb53e40c8b48d64a7ee84e9e70427c058e537ed08a30f4767eb8d6aa0423c5ceb3743964eb148cc248ba4aba3855479798732b203bbe51
ce2dbfce17ff1a6c159816868ad6d0d124c56fb24a7e294bca918aa566490e4d4213c88e7f3629632961696248bef130d19f28db7591169063344e63
458c07eba28f47593560b368b572964d6bc972456abb0cf5aa61a995a2a67adec8c8b87653dda4da2348be6129095909c5091b120274bf9170252198
a02339aaf6e924ea43c9b67325cda9be1493879c5cfef289fa790bd6efa99ff7d8ba3df5f5836a172dde8f954f2cfce98bb647849dcfee3851d35621
ec7ceee9379e6fab10730e4e9bf2443b07621e711046f9270a7580662657986d75c6e37a2668618ce2d943d540af1c1f7b672866ab127b38c7f99e53
c93fff1d3979754f3c5175a0be3ef3d505a7de166a140276ed5408a085f3f37e68b7cc05b4ae72621d4927d67a3bd41beb9413abdd723fda9d43fee5
c4eaf70e8a2a8552b94c53a62dd395e9cb0ca5c6325399b9cc52662db395da0351cd51b6dfee297f73b02dd972607fd5e603073637333bbfd1fc77fe
03b3e195ebe7ce5dfffaec996f76f0b3bc897f4f66d89faccdc1fa8524832389421bd016563690260d5861aed31dd7e8650aab43edcac650b535f297
4bef2a0e72382b6c57982213d5c83b09240247ba47f4d8f1427dfd80a3abc27ab9f0b0ddd678a2ed1089a360aaa4fcfa5b14bc8a6768b524b8eecf30
1905b361ac3b56ab1334fab16e776ca6de10eba67c54ce2a4547b9b332b2de26d6fb28fc768dd51bdcd11a78205a6bd6681df143ba2a545d6aba46e4
a8594f159fb2cbfbe986a23535ef99959ca731b7673e485232df6c97de6570197b911bf730f4300ed40dd40f340c341a3ce0610942577d5743b7b0de
8edece6ee15d63bbba933dc9710949e5fa7243b9b1dca4fc1ec90441d6cb0634a209cd68412b4661178c469718a34bea9d3c28f9d1e4b2e4a5c91b92
03c9cdc9919424e786b6fc8abadcea9952f6763ebcf466ca76ba0fc90ed78cd93bb1b272ca96410dbb7ff970e2e959056fe72e5f9bbfdfbf7fdb67ff
5670581c74b06bd7ec6cff883873b7a72a771cf17a4fa4a74fb87f5496cf9250b57ce78158c5cafb52e0f851da49564ed1d92c692df822d8d8716d85
de4032a6f86cb59b152bcf68a0bfd4d02630e497b4c93ff4b233b40da454e9081fc89c4afea0489d66638fb152be6a54c9ebaf5f7eaea242dac9df5c
df16a81cb37dd75f849cf5ec6e255a1d243b1faffa970306fa5dbf7ad85a3d3beea833927f390c63c8d3863a1583ef1fb2a86ba9b7ddacc8795271b3
304a0e21c3be9d2512d941c5cd5eaaabbbf79505a7ceb2f7d851614f5beeae5d276a84d25b810305539b71afc2fd5de4e365620ec870cb9f84365112
051b1324a5404106da9b00c89902c21b922ca1c0241134ca6f357ae5f409cae95339962abf9d284753507f328808fd5ad26927a56d07e97be564fda7
e1c20ca1542813ca85a5c246a146d02a0be950a7eeaaba601731914edac9982c7ab4e990ce06e00031453b1486b21138421c2a0d97fdda71308e4dc0
096296b6000ad8749c2e4e930ae51ced0298cf4ab194f65d8be555b08a5562a5582995cbd550cdb60adb719bb84dda2aef955e906bb527b557b441ed
dd647061693a96c6bc779d6693d9e4d3fc911631a7351b0fdc0a2812e2e3540959d81affbd1aada0b381456f33e8012c669b052c269bd1044a6136e9
0d7aa3cd60d0679a0c3a2b18a40a7cdd6c386e359b8c7a9d8ca0b58816833524bd51b55a5562860e112af5d00f2eea9ed94a6a6e6adf17fd3762544b
e9fb8854459ecd32485a5987a6707d84c96af29ad24d23f4f7e9c79826ea26ea67e82b4c4b4d9b4d763d101106c968301b2c11cc295845ab14a17718
1cc62ee62e962448205ff6881e2959db55e7d32718128c49a66ee66e168fad2fe9205d481153a47efa3e863ec67ea6fee6fe9614db3de0677ec18f7e
d12ff965bfc6afcdd40dd10f338d308fb0f86dd9703fbb5f7810b3c42c699cfca0669cf621dd43fa070d0f1a27982758b26c05ac4028d44f374fb7e4
d84ab58f9b1fb754c293ba558655c64a53a5b9d2f294aeca5065dc6ede6ea931d418f79bf75b6a6defd9aed882b67cd2986466a19f960631a6284fd8
3c66cb139b678dce4e8be303436a2c3cbb78fbf0f26c714ceb169ca5d8fa47424fe963e1af94cd6d476087203210ad971a48d6743a517e2af948125a
b8f0d77d0a6e6d5b2b666936917e8079d3ad61f6b03440728d92d381a37b02bc7562595bebb7b8957d21a4306cfb072f6fbbd1fa7d681c2bd1542adf
8d94fd40ede9d39aca9f4bd419b98366ac5467742a137a059a309d95946d6e0c5cd810d0547edbba9b3fc41d3c9f0d6137e84491faedbf52926e2517
4cb387d96d4e615dd944de1ad87334a0101226d8d922c1ccdbda2ef0d8d6e9df52cc5da78eab54ff1780b24d0a66a9a4a8df013bd392e6a449d1ab4c
2aac3b1dd87021d0b8b94c21e67cdbc7dcce8fb1a759137b111f86f66f7cc284a7caa32eec78d492f113b8431f9bdfffa3f95a47f9cb07ada3cd1374
9fb47f170b5d344e339bc70098f92f1fb4dc6f9ef05fbeca2589e7d56f542090f88535945fdd504b50231f26bdad819d520454107c4e504db083208f
40e95f4fb086609970089aa525704eaa8612f12295ebe19c7835d8225e8752c90173c526982bf487530a681ae1a86487a3e29730174752d91d8ad00b
7da9efa0781cee5240599bc4571bba8575f4ec7cf583a9b00e5e84ef583c7b56700af384f73103778b16b197582ebe2fd9a4a7a58bf274f9cf1ab766
98265ff3b656d08ed53ea9ebaf5ba8abd38f30f4320c36e418661b2a0dbb0daf1ade6b976f0fcc86ee500846ca1056784a91a2e814c2a954be996960
a2f29d44d411728afa7d52a93308a756a82e80960d6daf63a77eb1535d824836a6bd2e838315c0bd5004c5b008e6c1749846abcf57bfbc4e85642a53
2185ee34aa4d210c0f6412ce7c28219807f9900bb3a107f58e803984df8b6af7c02cba3df0c0edb94ad4563e95f93466213df30853ff3b56ed737b55
e5cbef425a4bf9323687b0153a7269ccffdd8a83a93683c68d83058431957073d5d9f2d511b92a471e9a650e3d8b09670acd3b9df03c34be8856cf55
dffdeb3c63d5594a88223ae7c04cea55562d51bf65cf5179e945f24bffcda88e3142c898827f54ff2fe0bf5e3d54bb50fe83c34036a1fcbf46c7ff84
44402444d1e9cd0531740e4f5457e84b6b0c83e1301aee832cb89f64f020adf6108c8709f008fcc004864c6412939946b360cef4d4b4fe99ede5bded
e5e0f672487b393454de43d656272cf5076f716c71e03f7df81fa9f84b35fe6cc69f38dee4f80f1ffe68c6bf5763b30f7f78f21ee9078e37aaf1fb6a
6c6ac1ef5af05b8edf0cc0af33f13ac7af52f1cb6b63a52fabf11a215e1b8b57bfe82d5d6dc12f7ae3e71c3fe3782515ffe6c04fabf1138e1fdbf1df
97e047c7f0438e1f10fa074bf0f2a561d2e5257869185efc4bb47491e35fa2f17d8e1738bec7f1df389eafc6771b63a5773936c6e23ba9788ee3dbab
6cd2db2e7c2b1c1b389ee6f826c7531c4f727c83e3098eaf733ccef118c7a336ac2ff749f51ceb5e3b26d5717cedc824e9b563f8da52f1c89f7dd291
49fe201ef18b7ff6e1618eaf56e3218eaf70ace5f832c78379f892190fecf74907f270ff3ebbb4df87fbecf82211fd620beee5f802c73d1c77dbb186
e3f3cf99a5e753f139333e9b8701420954e32e8e3b9f31d25e119f31e28ea7a3a41d79f8f476abf474146eb7e2537adcc6716bb549dacab1da845534
a8aa1ab76c364b5bbae266336e6ac18d1b8e491b396e583f49da700c372c15d7ffc927ad9f84ebfde29f7cb88ee3da35bda4b51cd7f4c22789cd27ef
c1cad506a9d281abe900431d1579584e922af7e12a1baee4b862b94d5ac171b90d97715ccab18ca33ff8c7254ba43f725cb2049fc8c3d26ca754eac3
c51c17717cdc8c8f1971a11e17709cdf82252d38af05e7b66031c7228e7338ce8ac3991c67d832a51963713ac7c225388d1a051cf339e6719cca710a
c7dc0198d382938d3889e3c31c27729c305e2f4d68c1f17a7c283c4a7a2815c7717c90567e3013b39d389659a5b191f88003ef1f1926ddcf31cb80f7
711cf307ab3486e31fac389ae3287a338ae3c811566964188e88314923ac38dc84c3380eadc621d53898e3bdb441b8b705338fe13da3d0cf7110c7bb
efb24b773bf0ae0c8b74971d33069aa40c7fd082034d3880637f8efdfa3aa47e2dd8b78f55eaebc03ee906a98f15d30d78672ca69930f50e8394caf1
0e03a6f436482926ec6dc05e3d75522f2bf6d4618f54ecdecd2775cfc36ec976a99b0f93edd835c92775bd07937c98e833488916f4193081a39763bc
05e388cf383b7af2d0dd82b1c4426c1ec698d0451274718c6ec12e9918458d288e91791841928ae0184e83c2a3d0c9d1c1318ca39d10ec1c6dc4ab2d
13ad4bd09287668e2663b864e268246c63381a38eaada8e3a825342d478d03e53c14e9a54816e044ea454e7b11ab24f4446645e0c8ea58deaa75acfb
ff0f17fcbf26e0ff78c5fc2762150aa6
>}
\immediate\pdfobj useobjnum \csname EF1O13\endcsname {[ 32 [ 318 ] 40 [ 390 390 ] 43 [ 838 318 ] 48 [ 636 636 636 636 636 636 636 ] 56 [ 636 636 337 ] 66 [ 686 698 ] 80 [ 603 ] 82 [ 695 ] 84 [ 611 ] 97 [ 613 ] 100 [ 635 615 ] 105 [ 278 ] 108 [ 278 974 634 612 ] 114 [ 411 521 392 634 592 ] 121 [ 592 ] 8722 [ 838 ] ]}
\immediate\pdfobj useobjnum \csname EF1O14\endcsname stream attr{ /Filter [/ASCIIHexDecode /FlateDecode]}{
789cedcf370a03411004c00679efbd77ffffa28ee3a203a14089822a98661b760736f951a3d69b6915d94ea7eaddf4d2cf20c38c322efa24d3cc6a6f
e6597cdcbf2c6655ccfae38d4d99dbecaabe2ff390634e3997e74baeb9e59e47d99edfbe04000000000000000000000000000000007fe5f506ebf202
bd
>}
\immediate\pdfobj useobjnum \csname EF1O15\endcsname stream attr{ /Filter [/ASCIIHexDecode /FlateDecode]}{
789c5d533d6f833014dcf9151ed32122408046424855ba30f4434d3ba10c603f22a4622c4306fe7d6d9f215291e07477ef0bf308cfd56b25fb99859f
7ae4179a59d74ba1691aef9a136be9d6cb208a99e8f9ec997bf2a1514168922fcb34d350c96e0c8a82855fc69c66bdb0dd8b185b7a0a1863e18716a4
7b7963bb9ff305d2e5aed42f0d24677608ca9209ea4cb9b746bd3703b1d025ef2b61fc7e5ef626ed11f1bd2862b1e31146e2a3a049359c74236f1414
077395ace8cc550624c53f3fca90d6765b7c6ce30135f0eae467c8272f6f146e0bcabdbb51e72628996470571a016240023802d2350305d02d697c01
4f4fabea828e28744c10f4a0ce4dd134f5af95fad74a1165a10642c618166aa093338c6ca106424674967a79a3703166e64f2ef3479671c89d973d15
005a4d179b63cadc1fe04a716439dae5e91a831474cd7dd7dc758de3c8a6026ae0d5aec5fafded86d875ded68fdfb5369be776dead9c5db65ed2f65b
a851d92c7bff019aa3d7f0
>}
\immediate\pdfobj useobjnum \csname EF1O16\endcsname {<< /I1 \csname EF1O17\endcsname\space 0 R /I2 \csname EF1O18\endcsname\space 0 R /I3 \csname EF1O19\endcsname\space 0 R /I4 \csname EF1O20\endcsname\space 0 R /I5 \csname EF1O21\endcsname\space 0 R /I6 \csname EF1O22\endcsname\space 0 R >>}
\immediate\pdfobj useobjnum \csname EF1O17\endcsname stream attr{/Type /XObject /Subtype /Image /Width 160 /Height 85 /ColorSpace [ /Indexed /DeviceRGB 4 <ffffff66a9b8d5b66a98526924668b> ] /BitsPerComponent 4 /Filter [/ASCIIHexDecode /FlateDecode] /DecodeParms [null << /Predictor 10 /Colors 1 /Columns 160 /BitsPerComponent 4 >>]}{
789cedd74d0a8020108661175d20e804d10df20056deff4ce940423490fd4026efb7d2419f59cd628cc949e3435cd6d3bc000216034a0001014fc0a9
0d794905042c069cfb2df1e60001015570194364f4e2c13e9a4140c0624077acfbbb2a2060bda0f23b9595668080ff0095b5426eb1dc5d540101eb05
95b542c0541e76cd00013f04572c565dd0
>}
\immediate\pdfobj useobjnum \csname EF1O18\endcsname stream attr{/Type /XObject /Subtype /Image /Width 160 /Height 85 /ColorSpace [ /Indexed /DeviceRGB 3 <ffffff66a9b8d5b66a24668b> ] /BitsPerComponent 2 /Filter [/ASCIIHexDecode /FlateDecode] /DecodeParms [null << /Predictor 10 /Colors 1 /Columns 160 /BitsPerComponent 2 >>]}{
789cedd5b10d80300c0440470c9091588182af598d86c62cc15c0864625a20504411c5bb8aecebac8f452ed5982dd7ee4d119680d18cf087d06cd599
b0124c6101ba0f9921ac0aa36d22c0f09a19c212d0c73be0cffc66086bc234529d0087d9f3415802b6be12ff9d02903d1f8455a17703ceea4554c7a7
cc10168107e1758269
>}
\immediate\pdfobj useobjnum \csname EF1O19\endcsname stream attr{/Type /XObject /Subtype /Image /Width 160 /Height 85 /ColorSpace [ /Indexed /DeviceRGB 4 <ffffff66a9b8d5b66a98526924668b> ] /BitsPerComponent 4 /Filter [/ASCIIHexDecode /FlateDecode] /DecodeParms [null << /Predictor 10 /Colors 1 /Columns 160 /BitsPerComponent 4 >>]}{
789cedd5c10dc2300c05d0025920c000101800280304c8fe33916fa95f187188c4a185fe7f4a6ce7ddac745d4b42a9c94da36d1128706cb0b8081428
f0b3956a0e7dcd39d67ca70a14383218f01042c46d81c3062af8845e162870ee20e7ed753fe4ca8f853781027f03e49b071d2cc91af27ed89bb77181
02670506346dd958e26d89c30e277e33a1bc260b14383dd09a378cda6e2417f4ee1e648fabb4f5aa40817f077adaf6085317ccb3bc42e9c481e8c2f1
237a02054e017c02aabc1f68
>}
\immediate\pdfobj useobjnum \csname EF1O20\endcsname stream attr{/Type /XObject /Subtype /Image /Width 160 /Height 85 /ColorSpace [ /Indexed /DeviceRGB 4 <ffffff66a9b8d5b66a98526924668b> ] /BitsPerComponent 4 /Filter [/ASCIIHexDecode /FlateDecode] /DecodeParms [null << /Predictor 10 /Colors 1 /Columns 160 /BitsPerComponent 4 >>]}{
789cedd7410e40301484e11207e006e20a2e60e1fe67d292cc4ba3a815ea9f55f332f9d8bc34752e27cdec336555f30208f834286b5501010193d64a
549d4f1b8d0001bf08da5a0c3e7d38d4e1a40f4d80807f07d51f43c427468080df006f5e19d77540c0d240db25edd959aeeb80800f838d66e1f9d0c5
fd440eeaa602021607ea64af063d24e6bdaa51a2befd0420e00bc0056e0657b7
>}
\immediate\pdfobj useobjnum \csname EF1O21\endcsname stream attr{/Type /XObject /Subtype /Image /Width 160 /Height 85 /ColorSpace [ /Indexed /DeviceRGB 4 <ffffff66a9b8d5b66a98526924668b> ] /BitsPerComponent 4 /Filter [/ASCIIHexDecode /FlateDecode] /DecodeParms [null << /Predictor 10 /Colors 1 /Columns 160 /BitsPerComponent 4 >>]}{
789cedd7c10d83300c85610e59a0a80b14b1010c90b6d97f26624bb8b294a8913810c1ff4e89b1bf5b24330c2d09292736b5b60510f06c30b9000202
96ad47ce38e5bce4744c05043c190c3228cc2487b7584f578a80807707a5ffbbe4f8db6aafcedf0001bb0753b92d54de46a51d10f0c2a00e59a21fb2
fc6d0704ec06d48fba10e92e649b91a6f05b51698f8080d705e5f491fabc4feb90ee5b1a9bd6212bafaefd270002f6006e56896540
>}
\immediate\pdfobj useobjnum \csname EF1O22\endcsname stream attr{/Type /XObject /Subtype /Image /Width 160 /Height 85 /ColorSpace [ /Indexed /DeviceRGB 4 <ffffff66a9b8d5b66a98526924668b> ] /BitsPerComponent 4 /Filter [/ASCIIHexDecode /FlateDecode] /DecodeParms [null << /Predictor 10 /Colors 1 /Columns 160 /BitsPerComponent 4 >>]}{
789cedd6410e832010055013bd40939ea0f600463d800af73f533b341df2a3386c3484fcbfc271782b26d03439e9fc374b566b5e0812bc18f4561e10
b39d20c13ac1e9243a4aab58cfd02f39dd43906009e0eb38bdfc0c475f1683302d94c6c4468204eb0413ef2717c608be6618bdf8658720c1db410f49
3c88e29d62b61324582738fd33c000b55ac7c1b2db09122c01d49751e8f81df61dd84979d983d83eeaeb8b20c1ea40296db27ac325e17427de295a9e
a13d0a040996007e006497f7ff
>}
\immediate\pdfobj useobjnum \csname EF1O23\endcsname {<< /A1 << /Type /ExtGState /CA 0 /ca 1 >> /A2 << /Type /ExtGState /CA 1 /ca 1 >> >>}
\immediate\pdfobj useobjnum \csname EF1O24\endcsname {<<  >>}
\immediate\pdfobj useobjnum \csname EF1O25\endcsname {<<  >>}
\immediate\pdfobj useobjnum \csname EF1O26\endcsname stream attr{/Type /XObject /Subtype /Form /FormType 1 /BBox [0 0 306 338.4] /Resources << /Font \csname EF1O1\endcsname\space 0 R /XObject \csname EF1O16\endcsname\space 0 R /ExtGState \csname EF1O23\endcsname\space 0 R /Pattern \csname EF1O24\endcsname\space 0 R /Shading \csname EF1O25\endcsname\space 0 R /ProcSet [ /PDF /Text /ImageB /ImageC /ImageI ] >> /Filter [/ASCIIHexDecode /FlateDecode]}{
78daed5b4d93dcb6119d338ef905ac9ca44ac445e31bb95976ac8a4f5979ab72707c70ad64c9aa592b926c2bfef779e0904403437238c3dd4ab62a96
5733d324ba9b8dfe78783bfad05023f1879a6732fd7f7b27aebe7afddb4fb7af5fbe78de7cf92dff74fb4950f30e3f6fb0e01d7e3e63d90bfcbc1149
c59dd0d2e175dfbd6a1d5a83f7727cf756881fc5d517b8fd13ee7a2184562de94619d946e523410119df4a63c871f19e8bb5a4d63b135432d3ade792
cec487a6528ca5b121326da4a44346498d936dd02eca183ca9e6e3ebe61fcdcfcdd5172a398787c4cfe7468a174df9f81fa0c6b6aa8b9823bc99b274
7bd75cfd8d9aafde37d7e2baf9306895085bd20cc34937020789d0b6b5704285220e4c8ae00d6110cf11fbcfe283485bf62ced9996adb6dafaa09553
b857b58eac19dc10cf6fc4d5d7d4f8e6e6c76e876e5e89ef9a277ffcc32e3c6dbe6f6ebe117fbd11d7a2734884d8dae025c2c81d61d24547826fb593
96a23496d63ab293c76e90d1ad73466b5b2546162f3a821d6e8df6de2be5745cedc94440faad8da1d5c1c4a8bafcae44531e28d33aef8287a7d1b65e
056f284643e7c5a1b7eea9d5c66967b3f52c3a61dd8516098e8485b36ec1ba99b56e4caba5f788e3683d8b4e58c73e384b2604efad5a1bf9542ef750
bcb2f5e273f35da35a8bb52deef91e065f414e06eba3933ee083923a2ab22eb8f4c168d25e59587b59977d70b817d7ca6695a5bc05896f79d1abe61b
e83636557cbac54b1b8232e943e7058ca60fa1f34279dc581817307edc258ffb1e8c4eb6cebbb9d68905abbbef786f56b0a457a647ee7b1db14ec7d2
c4f943574161dad4c4fc61bdd44a0ea9820c810394520989921227a4c479b27bf9b4b97927d01f6db051e37297494f76bff672e9a526358af79d58b5
0afdc0ba2c7fddc9cd60a39736fdddd6a8e898dc0dca43c49665b99c91ffa5d7e3a2e5cecca99f931fd47795d1cf4d3e7b653f7b05053c87b341fb62
3b558c6d08c556f622be636cf1d1244d893c5e4ff5162fab4bc1876a7372a82e9aeca6abcad3f5defc83e154a46ef8cfa7a670284a2c48bda3ab5d6d
d544871069d0a8a031a1bab62f310b9d499ba023ea04e6b978cfc521a0f7582b7d27ce5ae6e45ccddbaadd4c630c822a6dac73d54c65e2e5991a744b
0123cc2aed681bcc5046b546a94825e0e1e2455fd22234706d6d4cad7a03d240a9016449f2aaaa9a2c5ef6249a5679a09e485e872d4883d72f871b14
d4847c72873c3a2a011ad0a10836000fee0c471fdc996508523bb30187149d8d8111eecc3222a99db91896dc671fc4f46f36e193ea54c6ba09cfe451
faa0f86466facccc15983f9a4b77537329419ad573adbc99e39469cd2b300a7041df914c8af9230029340346e8bf015eaebe5643306ed7a19803c64c
e7374de8edb604994c5ef001d88c16c989413912028568c431b57ae98c399b1410abf0cba4a90eb7e88c5b36fb93f08ab818aff8d81ae36df40dacb7
3a3a8f35fc4cc3c47b2e066870da934995c494cc88b992756025531f452a3031dada9809a72811d21a68c9cb71132ee4440a5f9878d997921559ebcb
095aa4aa902c5ff6a52246563b33cb8c102089f70e0937b2034cb4cc0e60c461ea498d7b6534973023040c9260d768baff7cc2aec3480c847717d021
04a4e122dab9ce56b3e88461e05bcc62525a47e9cea74336f52d71df7883b583a22e324dc25bf0033122657b38eefc2527323757ea25674ca1cca164
1d8bbad7302301f93932234a86830213b575ff93a0433f0e70c1d061b1b73d3c2cf675c0977cf738149d8417d5f1e0b27addc6934c010e334b946cf0
f0de9892027ae4e34a811a9898630cae654e7e3ef8609448d90298fcc4982db9926df883b12265d632f9b23b155db205823062a42aa12c3fe14cc998
6c8120bc9e390e61c4c009305211035b100977668425dc93056c52bb71194029da1b4329dc8965a852fb71315eb9d746f870444999c2993f7950e032
3785e6a64b264ba686d611da3963c655ec0a073033dad750262a219391327904e8253e266a04fb87f2a45082cc515ad0222493d520e3c88a7049458a
f42afc61569eff4509b19a13292d7508c5d694c8e5dedc172392d8f3101284e0679a2ce57c4840e7500173a6e03da6a55cc3b96408dbfc2c44db1af6
fe14118259ab637063e42fe441981b59b8e846c981ac74e304055254c2285d74a3a23fd6fa510483a9c344e7e7346a1df954e459db71a3c3803d54fd
dffb1e6595c210c1e8a0a161bc1fae04a4ad96c6c5e1caa7ee8a6ea3f4983ca9170e577eea9b8c44273752694cd9feca2ffd9568a403fe5061c59a79
0f7eeeaf78ab4daa4ecc88aacde960a4d66af42d35baf1a1ff9d5b1d2755306cc805e568e454b264995249b013888d4829132ea1903cfe56091b8fa6
b364d934409af13afd3a0bfb7f099564305a224999f9ab2c59369d90bf463b0952da783e8fb4a9d53fdcd76a785bc9dc129b580fc421f1ae7a34254b
02697afa56f7af1cd499671a972f695dc31b59fe8d1a02f24aeb43b4961ed5efaaeca3a393d866f7c8996df400bdd97e72943e05c9aaf3d34505bb8d
479a00696e9646badcbf8d2c9242c6aa34d00bb8c6e81f8eb698984333a664467c3e64635411ef0359ba0c534afa68136a632c114fd22c5df4a4628e
360037461015c5324a97fd2849a3adc04d295d9c51ff8fdcce476eac8571f8c6c8a6650c57714d1b801cf384a339e6c932a4ab3cd980ebf84460e08e
79b28cf02a4f2e8679f7393e1e8e7de37d2053720f89f7a627f6cc30cebcdbf174af11e24a2450f1730cf74deb5df50da5826e7bbca04f2e8241f42d
320ff1edea833bae75a5f63f75e2d0a27b793a81115756c2c737f33523ea33d26a9dd50111bd53e32538c3369fbadd1af67ee8e8cffb8007842f063e
5edef75794444d4663f3b01a324307af0ed9518e97a92baf862bd176c332efd50fa39d508db18ffdb834d13815a274a3b6df8734d11280806fd9b0c3
5e269a09bb1ccb410acca09577da798a65eac203008c327d3ff56967140641421a2bb40d29199cf687b43ce55b9750c386fceb3850cf08b3c8e5ed53
8d6da3e5d56b1071a9a20f965c53ddcd36fb9f4f06afc99143d388a62e3c237db4a9f8c62b7f1e2cd8f415d7944597af299ef4b6bf09b046eb54ee61
b889bbfc94955877fe1a9a5ff70f15534ea6de475da00863cc06731abcdd8c00a9cefa795075d75df1adc618964131f036b43e4746d922b75781aa5f
f2137214a07c82b7e9df4fc9544f76e9f7805fee5eeff6f8d3ec7edbfd80d75ff1b9d9ddeddeef5ee1fd1eaf4d6aa1191f8c47d9d533f470d8951af7
a017a6210514405175200b1f642dde1762d4ea28664a98f86d8242e341ea78c6a593aa8b873e9cadacc688c03ad43d6b6a9ee9f16c77b84c2db56bca
4bc11075304ea943a45400ccc786eb3252ba6ba8a5785f8879a49892732285f35a9087a9bd26525446ca7447f17cf8f6aa7f767b325045d6882e508b
da10276b5ab89a8abf88132a59d7e27d21e671624ace8993431f9587f1b1264eaac2fa898fe8d2203d563f8ca54acf9c93e582d45aa917b1f3385bc9
d4eacbd885d49221ae728c8979ec989273621709213f8c9935b1d365ec2c4f199d53223de3a1ae6cbc24762bf50aa1a441c57507601ebb04a46d2dde
17621624aee48cd8297208f90180ae895d71c6bc16e2fa3ff54cce73
>}
\expandafter\gdef\csname EFWidth1\endcsname{306}
\expandafter\gdef\csname EFHeight1\endcsname{338.4}
\pdfobj reserveobjnum
\expandafter\xdef\csname EF2O1\endcsname{\the\pdflastobj}
\pdfobj reserveobjnum
\expandafter\xdef\csname EF2O2\endcsname{\the\pdflastobj}
\pdfobj reserveobjnum
\expandafter\xdef\csname EF2O3\endcsname{\the\pdflastobj}
\pdfobj reserveobjnum
\expandafter\xdef\csname EF2O4\endcsname{\the\pdflastobj}
\pdfobj reserveobjnum
\expandafter\xdef\csname EF2O5\endcsname{\the\pdflastobj}
\pdfobj reserveobjnum
\expandafter\xdef\csname EF2O6\endcsname{\the\pdflastobj}
\pdfobj reserveobjnum
\expandafter\xdef\csname EF2O7\endcsname{\the\pdflastobj}
\pdfobj reserveobjnum
\expandafter\xdef\csname EF2O8\endcsname{\the\pdflastobj}
\pdfobj reserveobjnum
\expandafter\xdef\csname EF2O9\endcsname{\the\pdflastobj}
\pdfobj reserveobjnum
\expandafter\xdef\csname EF2O10\endcsname{\the\pdflastobj}
\pdfobj reserveobjnum
\expandafter\xdef\csname EF2O11\endcsname{\the\pdflastobj}
\pdfobj reserveobjnum
\expandafter\xdef\csname EF2O12\endcsname{\the\pdflastobj}
\pdfobj reserveobjnum
\expandafter\xdef\csname EF2O13\endcsname{\the\pdflastobj}
\pdfobj reserveobjnum
\expandafter\xdef\csname EF2O14\endcsname{\the\pdflastobj}
\pdfobj reserveobjnum
\expandafter\xdef\csname EF2O15\endcsname{\the\pdflastobj}
\pdfobj reserveobjnum
\expandafter\xdef\csname EF2O16\endcsname{\the\pdflastobj}
\pdfobj reserveobjnum
\expandafter\xdef\csname EF2O17\endcsname{\the\pdflastobj}
\pdfobj reserveobjnum
\expandafter\xdef\csname EF2O18\endcsname{\the\pdflastobj}
\pdfobj reserveobjnum
\expandafter\xdef\csname EF2O19\endcsname{\the\pdflastobj}
\pdfobj reserveobjnum
\expandafter\xdef\csname EF2O20\endcsname{\the\pdflastobj}
\pdfobj reserveobjnum
\expandafter\xdef\csname EF2O21\endcsname{\the\pdflastobj}
\pdfobj reserveobjnum
\expandafter\xdef\csname EF2O22\endcsname{\the\pdflastobj}
\pdfobj reserveobjnum
\expandafter\xdef\csname EF2O23\endcsname{\the\pdflastobj}
\pdfobj reserveobjnum
\expandafter\xdef\csname EF2O24\endcsname{\the\pdflastobj}
\pdfobj reserveobjnum
\expandafter\xdef\csname EF2O25\endcsname{\the\pdflastobj}
\pdfobj reserveobjnum
\expandafter\xdef\csname EF2O26\endcsname{\the\pdflastobj}
\pdfobj reserveobjnum
\expandafter\xdef\csname EF2O27\endcsname{\the\pdflastobj}
\pdfobj reserveobjnum
\expandafter\xdef\csname EF2O28\endcsname{\the\pdflastobj}
\pdfobj reserveobjnum
\expandafter\xdef\csname EF2O29\endcsname{\the\pdflastobj}
\pdfobj reserveobjnum
\expandafter\xdef\csname EF2O30\endcsname{\the\pdflastobj}
\immediate\pdfobj useobjnum \csname EF2O1\endcsname {<< /F2 \csname EF2O2\endcsname\space 0 R /F3 \csname EF2O9\endcsname\space 0 R /F1 \csname EF2O16\endcsname\space 0 R >>}
\immediate\pdfobj useobjnum \csname EF2O2\endcsname {<< /Type /Font /Subtype /Type0 /BaseFont /EVICAO+DejaVuSans-Bold /Encoding /Identity-H /DescendantFonts [ \csname EF2O3\endcsname\space 0 R ] /ToUnicode \csname EF2O8\endcsname\space 0 R >>}
\immediate\pdfobj useobjnum \csname EF2O3\endcsname {<< /Type /Font /Subtype /CIDFontType2 /BaseFont /EVICAO+DejaVuSans-Bold /CIDSystemInfo << /Registry <41646f6265> /Ordering <4964656e74697479> /Supplement 0 >> /FontDescriptor \csname EF2O4\endcsname\space 0 R /W \csname EF2O6\endcsname\space 0 R /CIDToGIDMap \csname EF2O7\endcsname\space 0 R >>}
\immediate\pdfobj useobjnum \csname EF2O4\endcsname {<< /Type /FontDescriptor /FontName /EVICAO+DejaVuSans-Bold /Flags 32 /FontBBox [ -1070 -416 1976 1175 ] /Ascent 929 /Descent -236 /CapHeight 0 /XHeight 0 /ItalicAngle 0 /StemV 0 /FontFile2 \csname EF2O5\endcsname\space 0 R /MaxWidth 830 >>}
\immediate\pdfobj useobjnum \csname EF2O5\endcsname stream attr{/Length1 4824 /Filter [/ASCIIHexDecode /FlateDecode]}{
789cdd567b7854d5b55ffbfcce9e09936432934c8684249319920958120809414341060c504284288f0634cae44578e46112903480498d56b1888a46
0c542ca5d17211d3dc3e80440d865e408d0a1a69eb935aaed6462fb5ca63cadde93a8748b5dfedf7ddfbc7fdfa7dddfb9cb5d77badbdcedefb6c1244
e464a093634edeacd91443562291c85cf79cc2050b532bc6bdcbf424a637cd59b878e6983dd3fb98fe19d36d0b164ec85a3de7b6a3445a01d34b4aab
82b5d488fb98dec37473e9ba06efb5072786112195798b2a6a5754dda4ad3e4b24d7b37cdf8a607d2d47e378169653c48a358d154d17bffb67a63359
3fadb23c5836a2fa855b89c28a583eb99219f667ac6c1bf600d3a995550deb93d2b4db9836f2495c53531aa453e2db4cbfc6b4a32ab8bed65260c960
fa7da6bdd5c1aaf2f48819b544239814aab6a6be61e830751385334d99b575e5b51b82bb3f649ae35b9f24a3361174b9694c812c26cf786d349ea692
9637bb6011d9d7041baa298e6bc86d6888e80a666aaf2eafaba6b0613bc132cd1cc3781c6d6aa67307f7e136b4e74bcc8ca99334e34a3307ed4afc87
e8617299f1bf13ac0b96506bb0aeaa9a5a4bea822ba9b534585dcfb0b2bc8e6163dd1a6a5d515ec3f88abaf2d5d45a19ac669dcaf212e6ac0e5607a9
754db0c66b409e476b55b0a1925aab571b9c9a15c12a6aad5b5bcd9a0d15d52b18561afebf32d7bf35a13fa2f5709e24b3e5634c7a2e8f3845155a34
97303c0cb0e89aa60fcfe14a2bac9855465ecaa3c516977289766b95f8e06b3a187e13cd99131530254c5a27fe565ccb30b33279b490169b590d6343
bd435d433ffd8aa7cbdf73d257fce60dbf46560b875f9e032d1e7e2da242648b0aeaa13f704ad368378590ccd134aa60ae31ee138b6890e525acd9a2
df2516f158a5ef218de577e8fdc6ea11d95442bc52c9afef113d7488ceb0758bd822bf259719dac36b4ca373b24f9c95b95a2e15e955fa34bd4b6fd1
bb5863ad5ea1b75027c35ced84be536fd25fd59ba8c8c84c1418af9107b58b7c9142ed5abbc813f1224feba7c39c7385982edac537e54bf2251aa001
51c89afbe876cd268e8acfc4045124bad8ea1c9d13c94ce56839e253f12167bc9d4ea048daa89db68a68a67aa89ff33e439f51bdce5e69ab1cd0c6c9
01eaa3d3748af944ab84f11592902107b89fa5a7681557e6b4d0e480c565f5e915da051a14776a1dda05912234eed12299ab790bfaf5e5fa51fd5e96
727584866c246326c39b0d0d3920da398bd3960ad1c87a466fe238835a9f7680e7f81cbdc3f3e2e8dacd5a93d64eef88fde210674c7497d8af2fb796
e809d46e69d78be853a3367442ebe77a149af5b88feeb34ca473ba85cea2402cd79f322a467e795890f059f32dd1d426f2ad77f24c08575313b9587a
9c37e3e1cb9db5c22c49d4a68fc1e39cbba66dfcb26ea291fab55c94d04eb36f1307681b1da07a621748fba5d522756882d2bd8e4ecd3fb7ac337043
91f7d8525f46fadf915e87d5db49859d918dde03434385457a825cda29133be10febd4fd29a7ff91f07446fabcc222ef01317656deb0db59cbf398b9
b08851836236f367e59932236aa7f4f3337779a7b7b4d2bbd9b13965ca6647f9940c7365f2fee27d05c67a784be4f2570fa391343e104bbd23b6885e
7798e6b691cc708c27f708c4392e0d5e1a74e6729b4813063f1fccec2a8c17a25838b3b326e74c4a4b113eaf3ed21d1debd2ac16fd2aa14ddb3c78f1
c227fffd8568e3cd73fded2b2b2a56ae579ddc57e95d976efbf8bd773f1229c1867275e1c99fa8f3e50d41e36cd836f481cce03536869e0c7c83f795
2bce690bf724ebb1a24f8fed8bdbefd4f7fb773ab78ff5d8c29313ac94106f7759e3478f75bc3d78e48dc123cee8cbd9397efff9a0439d719c718ecc
651e3f998145b55789dca45c4f6e72ae37df939f9cef2db21527ddecb925f916ef32dfaac49aa41a4f4d72a5b7c65bed6b086f886888dc90bcc1bbc1
d716fe68c40e4f7bf24e6fbbaf23bc23a223726fd25ecfdee4bddebdbeb1c5c6f4c78b94d19658973b597844ac4bf78d4e1de374eb978b32418c1739
93527d59bafe52d3c7959b5b97aedd73f1a47a4bbd7ebffaddd6ad227cc3a6bb6fbae791f75f155e616f12baec5047aebea6a070ea7571beac977bce
ff69728e985570fda2f9b30b3cbecc935def9df573957ab84a7dfcb59c941d18e90c0fe3c357f6dab78fa02dd1616e5bae7013a2cdaf659422ebf341
2e43e62f0231b531bb62208a637c4e97db48cfe9733ac8e725678f6869bc7773a368512fa843aa5bbd20072eb51cdcb5eb20eeb8d4a2fad48b62329f
9197cf6e6de963c7eff79ebc356aea17941c661e712736d9dffa723c5f77697f54ff8872f35ff8b73f1e59ab5412ffbe7de7eb2e6645f57ffd4fc1cd
ca276a85795ceee5f738afc63dd423fb689be50e636d5e69797c178831adad984fe3683d9ffd1a39286044939f683e1e8d33df4acb8cbf89ce970291
69fe130c5c909ba9cb38af7e317b180779c5fc615cff0a2e294e540fe316f288bbe93aaa21be0f511dada41554490dfc871b4ba574158f59fcd7caa4
6cc64a58c34b3359a781cf8606d62ea72055f18dc04b73a99af5c7333683d670f7d28d577cd59b54398fe56cb38e61196bdafe1751275f89ba8823ad
e358abd8a69ab58d3c826cf37f8b98c7d82ab65b426b59a3947583a6b772d32268cec8cb5eaa19d6b24e09fb5dc97a5eb6afe1e84153f6f77e169a5e
ea39a31ad62ffb0752ef15f91233ab7af6556346cae2dcb229e76b765f5a655cb122f3d660dc6436f12af89f9ad55c931a7fd9315cb96b66b46acd22
8d1441f8794f41a40a1f4d21295229c4540af17e12a38779a34d3d0387f09af264be6d427838324492294da4788609e46138cae4c49b30ce84234de8
3661ac70919dbdc69a948183d7b781479b304ad86923cba34ccac0212245047d9f7991262f927a491711229c9632cf90806133f3c2858dd2986748c0
30c03c830331c2b40c33a1957790010d0b4bd7a3e3e58c186131e7254da89b5a3067a4991c61420a0c6dc4d0b5500a97fe922e2f29fc251d21858b17
e6c88b1b71610ece87704ee10b85cf15fedc8dcf14fea47056e1bf3cf854e193419bfc4461d086c180fec78f6df28f59f8d8863f84f0d1836ef991c2
8721fc6708679838a3f07b850f147ea7705ae17d85f714de0de19db7e3e43b65783b0e6f3de1916f95e1b7bff1cbdf86f01b3f7e7dc22f7f1dc2a937
5df2941b6f0e38e49b2e0c38f0c6ebe1f20d2f5e0fc749d63819c209f67fc28fd71e8e90afa5e0d5575cf2d534bcd21f2d5f71a13f1a2fb3f8e524bc
e4c28bc7bbe58b0ac78f15cbe3dd38deac1f0b0c1df5cb63c53816d08ffaf11f0abf2ac391071cf288425f225e5038acd0fbfc14d91bc2f34f27c8e7
a7e0b96747c9e7b2f06c8f533e3b0a3ddd51b2c789ee4311b23b0a87227090831d5438a0f0cb58fc221a3f57f899c2bf2b748dc44fe3d1e9c633ece7
9910f6f3b03f84a759ffe904ece361df46fc9bc2de34fc44e1298527153a147e6cc31e851fedb6cb1f29ecb6637740ff2117ea87213cc1264f78b08b
875d213cce937f3c113f50d8b9a35bee54d8d15e2c77746347b3debed52fdb8bd11ed01f53d8ceab63bbc2a3e3d1c6866d9ec0101e61d347bc783802
db98b56d1e1ee2e1218507b90e0fbaf180035bfdb85f618bc2f715ee53d8ac70afc23ddff3cb7b14bee7c7dd0a7729b466e1ce367c57a145a1391e77
d8b04961a3c20685a610be1342a3c2edeb3ae4ed0aeb3ab0b62141ae0da12101f521d46dc46d0ab535e9b2261dd5215485b02684d50aab14562a5496
46c8ca2cac50a8c84279994d962b94d95016d04b4b6cb2340225360497c7ca601b960ba75c1e8b5b6db845a158f12dd5296f56b8695982bc49611953
cb12b054a128846f2b2c613a30b44461b1c2220f16ba70e30df1f2c6106e60c10df1285c102f0b435830df2917c463be13d77b5030cf250b62312fdf
29e7b9903fd72ef39d986bc7b7429833db25e7c462b60bb342c8bbce2ef3a2709d1d3367f8e5cc1066b0cf197e04a647c980c2f46bed727a14aeb563
dad44839cd8da991f86619a628e4ba708dc2d531989c334a4ef62367924be68c424eaf3ec9162927b930a9992f261132db85ec809e158189991d72a2
4226fbcfecc084088c8f4146fa149911427aac5fa64fc1b8327ca30c57298c8dc598914e39c683342ffc1ea4a67001c6a57a90e2c4688a94a343f045
c117d0bd2e24dbe0f12029315e26f99118152313e3917880cf8c07f584488c8a9f27476d443c078d9f873885914eb8399a3b8458e6c5fae12a438c13
d10a4ea69d0a8e3244d91d322a0651bdbadd017bb31ec992c81022b210ce530b7723bc59b745c216d04728842958152cd2262d0ad20619d0f51050c6
9773a7d4149f5e9152384191100744d95d5bc4b87f8d46ffec04fe1f5b12fd1513dd5a75
>}
\immediate\pdfobj useobjnum \csname EF2O6\endcsname {[ 68 [ 830 ] 83 [ 720 ] 85 [ 812 ] ]}
\immediate\pdfobj useobjnum \csname EF2O7\endcsname stream attr{ /Filter [/ASCIIHexDecode /FlateDecode]}{
789c636018248005af2c2b10b3010001570010
>}
\immediate\pdfobj useobjnum \csname EF2O8\endcsname stream attr{ /Filter [/ASCIIHexDecode /FlateDecode]}{
789c5d503d6fc32010ddf91537a64344bebc5996aa74f1d0368ad329ea80e1b0906a40180ffef7392075a59e044ff7eede7df173fbd65a13815f8293
1d46d0c6aa80939b8344e8713096ed0fa08c8c4f2fff72149e711277cb14716cad76acae815f2938c5b0c0e655b91e5f1800f0cfa030183bc0e6ebdc
15aa9bbdffc1116d841d6b1a50a8a9dcbbf01f6244e059bc6d15c54d5cb624fbcbb82d1ee190fd7d19493a8593171283b003b27a47d640adc91a8656
fd8b1f8baad76bfae944e905ee05bf135d1d339de05eb0d055a1ab275d114d5d7eeba586e93aeb36720e8116c927cc1ba4d98dc5f5cadef9a44aef01
04487cf4
>}
\immediate\pdfobj useobjnum \csname EF2O9\endcsname {<< /Type /Font /Subtype /Type0 /BaseFont /GCWXDV+DejaVuSans-Oblique /Encoding /Identity-H /DescendantFonts [ \csname EF2O10\endcsname\space 0 R ] /ToUnicode \csname EF2O15\endcsname\space 0 R >>}
\immediate\pdfobj useobjnum \csname EF2O10\endcsname {<< /Type /Font /Subtype /CIDFontType2 /BaseFont /GCWXDV+DejaVuSans-Oblique /CIDSystemInfo << /Registry <41646f6265> /Ordering <4964656e74697479> /Supplement 0 >> /FontDescriptor \csname EF2O11\endcsname\space 0 R /W \csname EF2O13\endcsname\space 0 R /CIDToGIDMap \csname EF2O14\endcsname\space 0 R >>}
\immediate\pdfobj useobjnum \csname EF2O11\endcsname {<< /Type /FontDescriptor /FontName /GCWXDV+DejaVuSans-Oblique /Flags 96 /FontBBox [ -1016 -351 1660 1068 ] /Ascent 929 /Descent -236 /CapHeight 0 /XHeight 0 /ItalicAngle 0 /StemV 0 /FontFile2 \csname EF2O12\endcsname\space 0 R /MaxWidth 695 >>}
\immediate\pdfobj useobjnum \csname EF2O12\endcsname stream attr{/Length1 4392 /Filter [/ASCIIHexDecode /FlateDecode]}{
789cb5166b7054d5f93bf7bbe7ee33fbcee64112966c964721011248c84af1f20eef40022440201b360f48b2094920243140258c02555e3195888296
5aa44a574a9d60b4838a830ed24e4718eb54079de2d4ce446b6710c78c39e9776fc2cba933f587e7dc73eef77de77b9f2730007052278363fe9cb9f3
c00d2600964454effcfc6505692f8fff23e153088fcc2f58396bccc907df22fc24e15dcb0a26666e1e13b908204d277cd5c6da503d1bc38957ea24bc
73e3b626dfd3c7a3a301702cc9d454d457d6fe73f37f885fd6c6f757861aebc14015b88f706b654d4bc5a3fbdfa8237c02f1dcaa2a0f850dabde7401
98d3683cbb8a08d617f83b84af233cadaab669fbc873928df05d1a5e53b7318446791fe19a7f9edad0f67a8c28148fb987705f24545b9e913dbb81f0
bf918f8efabac6260830f2ddfa398d4fa96f28af3f7c484a0588a1f8652b68b9b1c25091084330ea34ad992103a6833467dee242b0d5849a22104f39
a43238087007d2b9abcb1b222409c3b232e9d2fe46fa7b74ce4448215cd2ed30aa86610e838eddb677183ac1a3db6b0d3584caa023d4501b818eb286
d026e8d8188a34525f55de407d4b430d745496d7115cd9505e0d1d55a108f154959711a53a140941474da8cea7f5e477476da8a90a3a22d51aa5ae32
540b1d0d5b23c4d95411a9a4be4ad37f4f6c770b936dec20709aaf2c7e14b2598af61f9c8b1f408544b326591444942d923c1cc39d925f31374c9415
d0a4788487751b6ad93feee3c1e196a4470e504818d37119a6e9342d7b1269a8860668d2fdd2e02d1a3c7876b07bf089c1cefb6c0ecda33a8cc93aff
50d3e0eae146b1900e209d000ab5a6e146b3c7fc6c22f4c0656a6fc069b8245d821390437505fc89ed95d269e47978043ee4129c83632c9779582e8d
beaf789416be879fa2f175249b475a3e8423a453d3d4033ba536291f2ae012bf02dd54eb74fa57f02adb0dd7e049b82ce5c12dd88d8570806a3734ca
c0af313308693c9cd42c510528a3968ce9fc9a5ebf829dd046393ba9f4281ef6beeef5f3ec2dd607ff229fdfc775b885248ec129b647f6cba7e43c38
30e42f96c2016937eb964bf5da4673d80cc7e452765af1c0db9aaf44c9274f2be0356acd70853dc0f6e05ef2ac4df3805f832b8685f2c421af0ced38
95e2016a67e02ca46317c9ebb12815704caa205bb7c8932b3807c691b646dab274aa40ec2b0a97516230c1e7884a8105e1a8babcc8f74ef1a8f409df
437d0e832f0af9d198165fcfe0607e913c82174779521403c6a81cf07ffa43839fa64f58945fe48bfe61ee9c61ad734be710ada088400d2332d1e7ce
491fdaf9b4d66885214139b45cae5094465a898fabd319f7c602f77a633d04b879ac373668b7c698d81a17acf1cef3c478bcb131896e97dd66a4708c
5e6b629cc9cb9312bd52ec8864c7cd1b57fb9cb9b9ce5c679cd6b97227c3c4e937fbaef6e5ba72732791330607ffc240dfd0af38f56c919db11235b6
c85eec287616bbaaedd58e6a679bbdcdd1e63497301c95993d75ca68bf3b85c5b9fd98c1c633f4b35199e49d22c9154c3efafacee6adeaca693baf6d
4ddfd67abded322de94b7ff9f8353e71e0cb171ead6d1a78c6be2d2aaad8cecef0c05e7ead377aec86b627f328e6cf68fecdf0a99aa77083d1c0b9cc
8d46038e93b92c8d6312e1260e66238d8159e22899c1cc682b9983462e21ee807526c564342b5c62326db78039cd42d1bf477153a45ad47d37e2ee46
ec307e61300e056efce20ec5a167c0286b1998c125af942d154885bcd058c5371977494724b395cefea01234a8d6222ce2ab0dab8d9b788db2d9d082
2d4a8b618ff438ee55f619bad15302256c51d45558f43a58062f508416aaea3467ceb46296852ccb4479a1cc35f744a5e67f0fa892745da46ea69df4
dd56f6c1c0e7033d927fe0635a0bebc51af912ed1137a4c24beae224963c82bb9cf109525c2277ba9c12931088e0e62e97d341408c4db2dab591a01d
92e31663f2626b2b2e77ed74ecf0bb3d0e2778ec89ae80d3635352fd8e9bba7fe7e90c52a715f75d3deb70b2921bfa5aa13a191c376edcb8d9e7f8d2
7977a15032efae150db943a39e32579a46993b77226d579a5412b0317feae8a9535c39335856a6372e278b22f678b3325ddafa49550c3686ab16bedb
7afa8de6babcdeca874ea4f9ce8b6be745f0a1ad676eed6c5edd99139c7d2a14feecdde758c6f6a585651b4f7dfb1daee87c8a2df9f2e0d1e54b577c
34b46af01c65c7094faa31163328b2d360729a8c692ec7cd8be7ea4da526a9e4bd73f5ce7ca754e2d2b6c06450ad8ac5004e5c6c6b35ed304ca408af
bea747180de617a989066e3419655a500a371963ac92cd427f3337998cf7447d27760a79995b5b2c6e1554a64aaa4db5ab0ed5a9baf3dda692b3c7dd
ac8465302d58f4bbb30c4389c0434979e38e7c94b6a5e7ef3dfe05971b5382c9f88add71bd77e06db9f44ca8caa0ddffbb29b232fe0c2440b73a838c
d369a59891bbe25c128f7325508be77171aea00b4d6e5617170b3b4ced164f62427c1c9d112eabc56c52e874b383dd9ee8b8b8286a2a5c143517ae5d
14b5699da5702d4dbb7df0c2b462c7056d7fd0f6e8bb39fd625fa6f3fe436168ae87e6365e0fd49220c5bb2be35bbc2d09bc841946a18dc57abc0f30
833ed7d9396e7f4e963b0bb12c491c6a5694252b5f1d193df9aca2341f4b3e3aa3e7e772e9818113a919316ab07bd7c37139a3a4d2036c86781386ef
65a9f8e84bf17f1ebbc13efd6b1869d42fd1bfeeb05dbffdffe69381ac9879c6fdfa6979e77d407740ad48a6f7d49e6f3ef9f67ccc3c08d36eb9b798
64ed2ed1d49fa6b69f2ee00390233f0779dc03eb7121e4a15fcbf79dd2ca14f68424d13da2f964c2a5301eaae86e97c041773bbd24d126c59007da9d
6e8035daab41d6deb393f4f7840633f01236044b6063f3866184d1ac701896ef8139c4b3d661588134761866d37d560f2df446d8049564bd097c3016
36d21de6834c9844358ba032e2f0c12ce269a29bad89b8cb2104b53081a80b2042fc1904cd841aaa3e7a43dcd6d5a863e5f42f27996dd48789d3fc7f
58cdbe63b5902c6d235b9b492642dc9a1f2192f97116e710b499e456c156e2d848bc215d5bb92e11d223f2919608f5f5c453467a37119f8fe4ebc87a
481ffbbe9e025d4b232c1be6df42d4f21fe0f17d8f6b95ee6123e175bad54cf2330ba6de277d5b36fd7bb2f4a21efc9ada0e5a15ffab98f4f52ad14c
4f85a594c195b0ba47daa50ebe2c301ac0df67e2992e7cc9862f36daf88b99f83b81a703f8820d4f05f0b75df87c3ffea61f4f0afc75109f13f86c26
9e385ec04f74e1f12533f9f1027c26139ff6e0b12e7cca8cdd028fbaf0c976fc552f7609ec248ece763c22f0f0a1f9fc703b1e9a8f070f8ce007051e
18818f0b7c4ce02f05ee17b86f6f0adf27706f0a3e9a898f08ecf0e26e810f0bfc85c05d02770adc21b07d5180b787f121816d4e6c6de9e5ad025bb6
97f0965e6cd9256f6f0ef0ed25b85d959b03b84de0d62e6c0a63a30d1bb604784318b7d4bbf89600d6bbb08edcaaebc7883a28b056608dc06a2f6ede
14e49bc3b8896c6c0a62d5520baf8ac7ca0a1bafccc40a1b9687314c62e12edc28b02c64e5650243562cdd90c04bc3b861bd836f48c0f50e2c31e3ba
b5317c9dc0b531b88624d6746171918d178fc5221baeeec7552b7bf92a812b0b4bf8ca5e5cb94b2e2c08f0c2122c54e58200ae10b83c3f832f17989f
81cbc889651e5c6ac125e4d59299b8987e8b052e5ae8e48b02b8d0890b04e6cd77f23c81f39d384fe05c817304ce9ed5ce670b9cd58e3305aafdf860
3fcee8c7e9d9b3f874810fbc8341828205982bd47a9cd68e398466cbe93c7b164e15384560561033fb719215270a4c173841e0781a1e3f197fe6c071
e8e0e3fc383605c78cb6f131611c6dc30033f34026a659e3795a3bfa7990fb05a61296da8ba3887fd408f48db4709f1de9b57a41ed96475a30c58429
aa9cecc024624feac2115d989810e089614c8877f18400c6bb30ce1be07133d11bc058811e81ee7e743913b84ba093b43a13d021d02ed0461a6c5d18
430663dad16ab1726b3c5aac681668a42163172ac4ae08e414050fa24c989c8ee8a09bcacca578646664aa0c49c87a5878cf636cfc4f5be027d6ff63
4bf27f0146bda70c
>}
\immediate\pdfobj useobjnum \csname EF2O13\endcsname {[ 82 [ 695 ] 107 [ 579 ] 113 [ 635 411 ] 116 [ 392 ] ]}
\immediate\pdfobj useobjnum \csname EF2O14\endcsname stream attr{ /Filter [/ASCIIHexDecode /FlateDecode]}{
789c636018028085641dac486c36067620c9010002b2001f
>}
\immediate\pdfobj useobjnum \csname EF2O15\endcsname stream attr{ /Filter [/ASCIIHexDecode /FlateDecode]}{
789c5d513b6fc32010def91537a643446cb9cd62598ad2c5431faadb29ea60c36121d580301efcef7b80eb484582d37d0fb83bf8b57d6e8d0ec0dfbd
151d0650da488fb35dbc401870d4861525482dc296a5534cbd639cccdd3a079c5aa32cab6be01f44cec1af70b8483be0030300fee6257a6d46387c5d
bb0c758b733f38a10970624d0312155df7d2bbd77e42e0c97c6c25f13aac47b2dd159fab4328535ee492849538bb5ea0efcd88ac3ed16aa056b41a86
46fee3abec1ad42e7f2c499ec32dc7ef083f0d098ee1966382cf4582cf9bfa9e66b6ca69b5b115c154c3df6bb19c38bbbd57b1784f6da601a7fe6267
dae0fe07cebae88afb17fe0683b1
>}
\immediate\pdfobj useobjnum \csname EF2O16\endcsname {<< /Type /Font /Subtype /Type0 /BaseFont /BMQQDV+DejaVuSans /Encoding /Identity-H /DescendantFonts [ \csname EF2O17\endcsname\space 0 R ] /ToUnicode \csname EF2O22\endcsname\space 0 R >>}
\immediate\pdfobj useobjnum \csname EF2O17\endcsname {<< /Type /Font /Subtype /CIDFontType2 /BaseFont /BMQQDV+DejaVuSans /CIDSystemInfo << /Registry <41646f6265> /Ordering <4964656e74697479> /Supplement 0 >> /FontDescriptor \csname EF2O18\endcsname\space 0 R /W \csname EF2O20\endcsname\space 0 R /CIDToGIDMap \csname EF2O21\endcsname\space 0 R >>}
\immediate\pdfobj useobjnum \csname EF2O18\endcsname {<< /Type /FontDescriptor /FontName /BMQQDV+DejaVuSans /Flags 32 /FontBBox [ -1021 -463 1794 1233 ] /Ascent 929 /Descent -236 /CapHeight 0 /XHeight 0 /ItalicAngle 0 /StemV 0 /FontFile2 \csname EF2O19\endcsname\space 0 R /MaxWidth 974 >>}
\immediate\pdfobj useobjnum \csname EF2O19\endcsname stream attr{/Length1 15328 /Filter [/ASCIIHexDecode /FlateDecode]}{
789cd57b09781455bae87feaafaaaede977467ed249d7416c296d821ec4a1bd9418cb208289a40089bb2055c881a9621310a4310088a081101d96422
7231c188205144644607f00e17bca8a0b8446466709c09c9c9fb4f7507027aef9db9dff7bef7bdae9c3aa7aacef2efcba90a300070d14906dfa0fe03
06c2482801609de9ae7750de5d23bd99a9afd1f560eab067d0c8d1b955d3b6c70128cdf4fcca9db78f1aecdf9b3f1f40bd4c7de6dd35323330bdd7ca
59005a213d1f33e9e182d9d21fe745d375bd9863d223f37c302dbe178049a66b5e347bcac373ba3d321dc042d7b0734a41f16c30d001961abab64c79
e8f1a2c5597f394b4d1a9ff2f7a9930b0ab509ef9601dcb2979e779f4a37ac1b0d04df2d8d749d32f5e1798fa55a52ee070838687ee9a159930aec55
311e806c373dffcdc3058fcd9677ab73e8fa79baf6cd2c787872fa0fb7d2dcd9b5d4ffe4ec59c5f3eeccfdac1b40ce327ade387beee4d97d0c7f164b
3d46384f05412b0b847e121d087de85e5f08525b3c334157ba92fa0f1c3e0a6c0f15cc9b09d1203083d656806b2dd193cd983c772668a2a517996610
b506126e153df123d9024e88839b7eadfd6fbea3ffd2ae3d3f4ae5ecaff6f9bff6bb713dfe75ebba5fde6f6b876a8145eb6bad6fb61ebad677069513
7a6b73bb519bdbd7ad95ff6b08fbb75a5b135b9fbe0e037f8dd39c3c34ef211deecdfc042759e247f5ab43ffdbb5fe65d84214a96c9dd1fae63fdfff
bfe27278b67633b5b54358fe736bfcea2fed5f93abf07a4f12ed6f1ad77aa2f5ace07518d6b3e2fadab3a3a1736bc1ff9edb37ac75b60d92f650ddd8
839f15d2f72ffd18e9a60651a4df31104b5aea857848063fa44347e8045de016084036f58aa2fb4cb7143228a00adb16fe69600cb74cfad9ac5b166b
f89e0dece0d0ed41677dfcaffd503f87e615338b194de199acfa0c620e2048ede082086ab9c10391d7a086b0cdda00bbc0addbac0505730b264265c1
dc876742e5c4b905d3a07252c1cc623a4f9d3c97ce8fcf7d082aa74c9e45ed297327cf80caa90533a9cfd4c913e9ce8c82990550f950c12c9f3893ed
237b3b6f2a54ce9c21eecc9a52f03054ce9d3f937ace2b9a3985ce53c5fcff857dd469fcd0b4290537d848194236924177bd56085337d13e193a900d
46a27f07e8a7d769701b51ad03dc4ae7346193e91ea3c31fc6f973f806f2a0af050c9f89a5a4c5b47c3ed9f853a1a545bbfdafed9a4d0db747ff1322
62a3791ff99ffbc1f8505f515f6bb79be386fbe3db5dcfbdde967a025c5d41edbed7867620de47eb7e290106c24c5d8a84ac082a4486680d42867a8a
be46f2e1d083c674302ed7a9db01c05805d0ae5f9f6bfd7a5febd7eb17fd249a9bfa29426223c5534540c0641baba4b541c9565e20c8134235fe3b14
492e1a64561135599242fc05e265f8975734a090ae7cc985aa9bbbd93ac3c3ec7cb84f1b5ca1e285901eada72ba65fcbb050af1d74c7023ec820dfdc
9be4a13f0c27ee4d83d9f0082c482ed425cf479076822ca2de1df4f41e28801930171e134f5bcf937d3ad3fa1fada75b4fb61e6ffda0f5edd6fad6b7
5af77ff6493b287eed178a137685afacfa2aa1823a7780200ad1ad1395ae54542a593a7d43b8085ede4a85e217820b4078fe282ac3c32586ca3d1092
45a1cd0554a65149a022ecd96c2ae954849c0839ec4285e21858402500905c182a90c372a0168ed1710876c07ab695ae8ae8fe1cba532ded81a5309f
ee1c66c75885d485ee6d85cb70827a96c331dc21031b4a96ee18f53f4dbcbec246c15e9aa31773b35e0695547684bc57be47ae952fcac7a1875c2c1f
97f3e562968d9b9431ca562abdf03d9281a39008b5ec1c14c37efc16b3b15eee2fdbe01c1ec71df015ad22e8740c56c0668a4e6bc1cd6641a95422dd
43778e28c7611d1db3e8f971b6819d20e8f6b325700a9e47591a0c1bd829c2eb18fc0d96e028a994c89f2d1511fc4768aee3347e1d14935a9c6226e0
5227ba47d0d35a13f5733c76514ee9c76528a59547c166b556751bfcb48aa0d856769835aaaba01a4ee0fd3807cfb0a5b25fde260f8615210a603eac
a0b9d789316a117b9c70174789985d7a54ce673be05b39df3091e67e4f60446bee95ee218c8aa09ecaa3aa8370eac3966205412a9ec6c371c3503993
c6d30c8627096b80599803d3a95502bb610f74c12a584133e9f8aa3d94bfd1c8f5f21784f30ab65cfa1b1cc7fe247945f225a235194f20ed7dd3a02a
324a0c3afb1c3552ea90c29ae0dd637d1f8c4bead2f9a64b9fc3e0ab81bc1aebe3bedad6d6bcb1729c32ae46f1d660aa5623a7fabff8af1e7ed1a5f3
b0bcb1be9a9601fdc3b30ec8ef4ff7468ea5a6b8a2db747f407ffd9958b44649a5bf21f935be49537dcf389ef1f77ec631b97797908da12899f45668
7a0dff8b54a2ba48bf7a04edeaf3b0d6663500ba548830d91c6787d5448c1a5b07a6d6833dc70dabb1eb6d08f61c7721d0e8ecd5eb16703437663155
f2b85d51fe3429a79bab875452b678c9d2eaaa35abd7aaaeaff96d172ff23e5f7dcfdefffc1c6b68a4f536d37ab3f4f51283768358cfc0c0ec922334
a0f5fa5eb93e6f4476a4cbe3960cfeeeae9c6ed2669a724d55f5d2254b545723ef7bee73defbfbafd87b172fb27769d6d36419cf2832796e7fd00946
b6de808ae49121caa47a348be36c73dfe6be8db740e649aa1bb29833c993e4f43b9372929cb85bead272625bcb09a98b22b79cd8211a3b484b196c68
75b1c3c0c9eac6042db80196a828b318885609b8931f8560ec91ed417fc4e5139b17ddc377f1832c48e30ad939a9545a42d475ee83f5924c7ed77156
f43f4938d1c285525ccb57d292cdc20a9fa1d36e5a83fabe094b2431bd4cf312a4fadcfe33274e704efc2a6f3d2faf201d3293fdf20723d46a17545b
56ba96451bbdf604f47ae2a269d4151ae7b870a5d171298b254b4e872b3be0723aa4f400381de04f1667e9d9f52fbd447f2fbd749519f9cf57aff29f
9951c9e3c7f947548eb36c3abab1ec6a5ecccb78392f66cbd9e36c015b2e60fd828ced78f2ef2400414f2e56cb52b5b2c800d5462d51f5222432b3e3
645846989091c6861022812b8d843949eab864b6d78e76599ad023c9a9e4a4660b36703694bfc0267fc886366fde21170fae1ddc746a074d405a250f
258cbdb021981e131b87d15e2731d8a92872aee365e76a6bb57ba54cd6151c268999bc510e54e31dcdc36a3ca386d5448eba6f588d7bd47d04090ac9
6d38d978f0a0d3d52b0ccd151d1a8343f9c1a0fcc06abc0e67542f822d18182d8f51c61816c80b9447e2ca630c647b63e4585242ef3c78449d1f5b1c
37cfbb18ca6216c72e8e5becdd06dbe29c1360422a2191d31d7adcc672baa5f9935543ce6d2c3b207bdcaa410532f8879a871319b30bee7cb5ecc113
8f2d3839f61be61e705f0cbfb263c78e47d9cade0faf1df26855ee1d1fdd12f8e6ddfbb7cc8ee7df13f6eb89dfc5847d07981dec0a9e08539931b1cc
1751edb1561b57a9de6adf2aff4a7599e7958c486f04a03bc69be67378d19d68543304112247b5e16fd4f12702906245e99ad578e1ca8546c7d7971c
fa4154c96241636142416281af304986092c8179dc7252725a7a4e0221d29db0eac472428d1bd0c37e2b5fe11ff36f1e38327dd4070f1f3852b765f7
be351b5e797ee481b9c547c77dcd2cbfc5d4c486cacffe929a7af89640d58adfacd9fae8ece29294b4bd3edf277b9ed829ec10794f7933c99444f660
51309e59d10a88d65c40b3a15a61b8c8c82c26f0aa9a6cd1ad919910b3ea8859046227fb3634069c82af174ef66d0c102e3a63e5a3c4dca382a51dcd
944e0c8671e4d61f8567c010c93a411aeb84ddd9087697e52eeb1856c4e6b305b894598995469684d94e523bdd36a0ca25c673f8a953475b1e50529b
cfe3f1e6ec6dbc9ae51fd6adc379b990208f8707827e39d6e02c73c4c7561bdcd58e0aab540d8baccb0c9b13a2bccc845e3039d40447336bcf17473b
8bea10da422c72345c120a2c3498d8c31b42dc1126c329680e1e37dcc016c18dcf30a6a5baf3d8ce4d2c859fe43f3e7078eaf883335efbf0c3d7ee7e
7994726a077fce6ee797befb33ffc9e73b764bd6bef5ebf7a5a411b55710f455ba3d4981b1c1940815ac6516a88e54abbd915b1cd5968ae495de65a9
9664a3372621c28b498971a9646048882ee826e642f385ebe2137453e4c38e4bc7f1b87c4c39a612de7b12a4096c024b563deec810acccd395f99325
6c43c4ef13e6282910296d7e7ae3c6a7a930e3f017877f70c2de67cf8c2f98c22f7fc95bf82596c7e286bf887df66f7af9adb75edeb45f7abc36258d
ff85ff78ef04fee3f75ff3ef740335916d4910166a1b49d354e2890a9382d18a5342099d32d90b85f8810a3232c5aac1d1fc5183ee6332dbd9012a24
2e824163dfa64450987103188847ce1e3dc7055d6325a662acd24b19ac4cc11aa8510d242dc418e66749dbf060cb9727186fc9564e8d695aa47412f1
e9b344df6775fafa2113ee08a6461375d3d5ea842ed5ae9509cbd25fc98ab6a474f47a52bc7623596f32e1f6a4b82c477343e39586469db06dbaaa5f
f522256d47ccd4ae646b52b20391c2c8e8eaea4f4ec9e9d63da2ad034986f46ce5962d95955bb7f02d8b5742eb7f9ee32b173df70afff9e79ff9cf9b
07af5cb278d5aac54b564aefad2b2f5ff76259f9ba31be3d0bdff8f8e33716eef125bfbfe2f437df9c5ef13e2b98b778f13c2a24318b08a372c2285a
9718bf21318695414cb5698b5c0d159189d58e9591cb520d5e6f52440224277badbac010f86d3ee96bfe539bbc4436c4bc1b7b30eea0f760fcbb090d
89861dae7ad7b72e2489e9a1cbb62bc246b20239dd203b2425c969ac0d2da2c117c3d70f2339e9bde7a1cff955e6f8922173f2d7f957c3d7b3dbc2b2
944852c2accc35e67e66fffe6b16a9bbb38dfcbe04696d9b2409eb23f65c0fcb7ee29701bc419bba44de4a4e98cc8e0cd19a83228680f0185742ae58
d885cb274e08872cfbb91e733ca09ca1112adc19ec28ad4719d97a4ad4442531455560bdaae42a3245850a6e535f35300952643fc95f6340f0f64a63
54c85cc9247e5ab8c8248646c9c3729872e6eadf65ad892b125ee62bf9aa7dec93adec1301f56996af9cc14d61a8ad1462a8eb654d6132f805d00d01
3dd408071be2a0b442ccd3c471d3b6cbc2b716b79e57d2899731d03d186b7dd9b6dbb4c6c95e86ddf29aa895ce65b186182b64b91db102d230fb8418
feed52d65e7b5c629c449a2db439acc1dd7b786cd72e2295f4a28b8b5b815f660e068b2f164dffe137fc35be8095b191653f28134f3df8003fc2ffc4
4ff3230f3c7862f060b6914d6153d9c64184d7518ae16a487b357042d7a007d61817b1350e4d72984089b106c06b945d7a9443baa1ebafe0cc9efc08
46f038c392929aa4d7198cadba42444ce45ff0633c9756d9c3aaf8549ec70b94ccab8fb268d6957566515bf95abe903fc5ab8826b4ba3c9d56572035
6891d6c0229979b137c822da1361cc85c6aca031cb906758880b65994d88d0a397a31f4affdefca0724ac42c210cc8ea4204452daf07738c9a014daa
930443719258e45232e741d9b3c6e85e635d649615159d46f046da14534c8cecece736792d72bc203a990172d6ce10967d45f8e4ea258ef6e64a8f59
f6041308fb60ee8208a680c214494583ec010f734b911825a7422a4b95d2305d4d33a46969465f4277d65d1ac8064a5395f9f27ce5d188a7d5a70dcf
abcf1b1227e88e3f2ac28f5d592726fc9f4f5817226c48f370f9ed25b71d3ffdced0671f3bfb21fb8041f392960afedc9a35cf49f591954ff1a9acb4
6a624b8572ead33f2ddf2fddd572a97cc992a5c2268bc87513c95a3a3c15ec6bb54836b3949098a0192583494a4c4cc835991312650f03cfcbeed5d1
6b9cf21a589d4a42d821c1644e8c3340725c8cad8b21c69ddcc171b6a1912c0a59109d2e0e3d8c21a174bcef0c2b92c326a812ae883876128309fb12
333233eecac090ccea0624f157429c4cd66653e5c1c51f3db8e58d47b72ef8f2dff967fce2f41f179634ce7dadbe7c5dc9971fb2a89fa6fd87b2f9bd
1edd173e326972624ca7d3fb4e7f9e95f9f180814f3f35f389c4e82e0777be7f218df06e6d2279fa9664c10043c9b284042a48b14050d11c27c97fea
fe2790453ec7247c8ea6fb1c0db4369f1301c6447030879468701883c6d9c68d46e3040ce72daafc63cba5632d97c8c3379d121e87c11ed2e90c5acf
49d179a42639cda0acb12d33c22297e635f5645eb8dd25224387083840041f8dcdba6e07c2fa94c5f62546ac88d8188142b4438ec4992d748a346acf
b1dd87dfdd7d8c9fe35ff3aff839e554f37c328597f1d9e6fbf959fe29ebc85284652a215e77a1dcdc04a9504f517ba239ca6883ed516a9dcde92b4b
dcefadf3d73a974559200aa3ad46cd9c889a7b401a49fc47271b038190536eb870a599accdfbba0f740aa10fcecc8acf4ac84accf265256525f74b0f
c607138289415f3029989c179f97909798e7cb4bca4bce4b9f9dbe34be3ca13cb1dc579eb434b932bd3afd727a42dbd0b6416d03f213f213f37df949
b3136627cef6cd4e5a98b03071a16f615274fbc8e556d6c3e9cf11ee288dfc6b7652fb1838523a706ed7a2592fd4d5d6f6ab7f7ad7b196ab4c7a756d
febe51930f8cffeb6529bba86462f1e9bd19c35b16ed282a38b4e9ed83aed267bb76dd919ede2c68b59f68b559755384e0859ec118acb3d88d75d19e
65f6dab8b531e0720d8ab6a85aec40dd1e04aee889d405111bbc7f296b5f7ec2c284ea042438db7c2281caf4d08af23d82355db82bfcead5e79e7b55
9496dff67ebde423686dfda8e4f5de757552e6b18b178f5191ee292ce0f5fcef74d417146e236818cc693d8f17898731d02f180765ec69d956667dda
54e794eba26a8573705961b07b0039870b6dcec1c1af5c72fc74292b68b6c739e216c655c655c729ac9dc265879d4472d849e0c5112fe5bdf1fefb6f
e4bd34e2ce2d135a4882ba3075f426396757a74ee78f1f3fdfa9d38e941442c8c65cacb79fa84550c9e3093e47885ab1756073d729da325b2d5b4b86
0e346990d3651e10affb8740e01ab51a6ea016c9738899941c033190b58b9770536d6defd79f38d60aadc79e78bde508d16ddb36a21dee931ef847e3
b6c202d69f6974f42fe09e30f9c2709512b5dc1047995a0a595e6399f6b4e2d9ce943a0b7b2bbace556b59e68df3489a47836192cb3ec0ab83d8a067
a38278a140fa4a2832cae8173f3bbe3afee3f8cbf14a3fe8c7fa49fd3cfde294ce864c2dd3d8d9340b66b159d22ccfac38e3843982c049baf9baee80
49000c3ad10d7269f31ecbf137a71f9938e9e319fc0a3fc2329abf64865a69cbd3ebea6cd203e30f1ce9d66d77c7ceac2733b1087607ffac61eddedd
1b8425c924f1fc3bd13a02c605bd8a8359b4ed2a2b87b536b5de244550586c5434abdd3cdc2dec8949d813b3486086d5d8f4b6b02d7d1b280e6970e9
0a7d21d04c769a72336162829e3c4fb547186302329e85cc993f275b2897f4f79a4977b24cfe495d4dcdeeb755f70b795327ad68cec44f568c786ba7
a0351f238f275a9b292b1e1af4c758e28daeb288c83a3bd6a5f96bd3eb8d75f6b763e3d36240b30c525d2edf800c3d9e0e8943c3859040f05382d2bd
482a3a2eec58ddf1261d8a7248d7bdc2ad2c2c2a2e1295a89c6cdcb465cdea2d5b56afd952cb7953c1aebbefde70cfbfededb5e789df3737fffe893d
bd6aa55b3f387bf6832367cf7ecfbfe4dfc627bcd1b9e3dbefdc376922ebcd44eed17be2a41d82be147aca853a7dbb91de1b016d4c2db7396b2d6b4d
4cd26084b08c03dd42ef75b5ef2b5202a78b7cdc9e7c8f1eebf89d2190a991ad6786917261ed134facd9555797fbc6fc43ef4b9b5bee97366cdc7060
734bb9ea6ed930b9f04761710ed1e28fd3ba226eec443ee980fc3ad453b4abc930f05ab47ba199e21cb3f03879c67cf23a8a1ee9e8c1efa15afac9f9
57ab55f7b76d78aca7f94c64eb33342745bc06a7aa2a7aca25e9fb341a05af580f6b8d2a45a72a690f0c340bbc4e36e8397bdfc60b6d9e5bfea12dae
09c5c07b1c1611d83c284b262d524a9732944eda18a9489aa2154b8f2a8ba50ae5b7da2aa94a59abbd22b98c8a5195cc683274c074b983d249ed6408
5aa662bea50297ca15ca727585611dae35ecc057957d86f70c9f1a7ec6cbf8b37c598e9d3007047a2cdb48999cd3bfbf4e4afdbe65b734e372cb913a
d5dd3c8d9d6fb9d2b24bf2b77c46f85ea75ff29bb05612d85cdbab0b5a1d4a50c953f295d9ca65450d118d08a6baffd118a695219e643719c607d354
9731da0e6abcc163298ff7616d5c7d8cc3004ebba6a9794ecd9ee78d26c3efd703c1e6e6c6d0be55dfbe17aee89b1d4210821159297929b3532a53aa
e97827e55c4a6b8a91244397054f7bf9b82e289e90a0640c38b8f87707eae6ce5fb1b56eeea3cbb7d6d5f5ab797cc14eac78e2919fbe1462f3f27a21
36d2864d2fbef34a4bb99cbf7bcac427ae492d61104139c30d525bffeb527ba14d6af7e67bfee0916e965bcfff20b7b4b010db90859daf6b7d14697d
845ae7823a4badd8f974d9ef469767c04d3b9f417fbf981228514b0da55aa9b1d4546a2eb1945a4b6da5f65247a9b3c4551d7339c679e3dec40d1ba4
c5ab77ed5cb36ad7ae5597998b5fbafc67fe2373e2b98b478f5efce68323dfaee71ff046fe0399d35e6435ddaca7f0e46497361384c237dd168c6bf3
4db5b665ec6dac8f27bf3448f750ed7cb9e3c28536f7143486fcd3e7099454a45e234dd891dfe0e08bebeaaefb71a9679b77dfd6b25b35ed68e7c9d9
f76d0e2ac4371c4ad039212be856cd2467662cb7d51aeb0d2695c2cd812e612475cd27af74f223e186f6e65118283816f2dfd7d91585431387745eff
2ac1b17f6944572fee75398f1d68d943cc2a9aa48837a3b3287a3842aba5c3c570b43f321cec8fbc1eec53545121bbcb3c15d122aa48adbd1eeddf13
a7d90c9a3b79400701d5c91ba27db2dd3f8930c37563b4df16ec43ba30180f7b4d5eb3d7d2959c656773674b1f631f531f731f8bd9073e9622753075
30778cc874677a3a467648e89098e1cb484a492f339599cb2c655697782f2a49aa4935a305ad68433b3a300663310ebd72bc313d33a35fc68319a519
0b332a33aa332e674453d234e7e6b442f5ff32ade84eb4c367476c1b5f51317175bf862d3fff69fce1878ade2f58bc6cf2cee0cee73fff7dd15eb9df
ee0e1d468d0a0e49b2757ca162fd3ebfff404eceb8bb87e5a5da53d62cdeb04bdfdfea4166fc2fca06d2418a816c8a66c7ede064f55ab9c94c342619
73b86c420775f71b08e7c8a12d52f21ebf0b790fe173dd917d84074ecb11bed7c91e65257ce9b0e2b7df3eb5a9bc5cd9c0df5dd1525d3162ddc63f4a
f92bd86dc277ec262d1cab6bbf1bfa04bdd7f57f9989d5bb6b2da4fd6ef308b203033d421d7b85248a928b362330cb735018810867bbdc221c58b3dd
c208bc565b7bc7ebf30f7dc0fec0f64b5b5b0a366e3cb0592ab95abdab68d265dc26b0bf952c50a99c0f2a5c0da687bd0c931451a1a402e5aa006aae
84f00eb9229498228341bc91d04314088528ee51e20dc17d6d39504354e89d40bbccba6d2746df100cfe76b0345d2a914aa53269a1b452da2c696221
231af52c3b1663e5344863199821fbb41cc861bdb1b79ca50d84816c080e91072a83d5a03606c6b071384eced38aa0884dc369f21465aa9aafcd8779
ac044b280f5fa02e85a5ac022bc85b95a95550c5d64aebf079f97965adba4d7955add10e6ae7b456ed36b1eb48be2a9bf96f3dcc1e600f1ce6f737c9
f9cda370d7d56aa2501fa2d0e3442133bb23385071aa065576a26c1095223389a1539298d9493d4d4ea38989ca6c326806a353d30cb92683cc648da8
27855b84ac254440b1517d9f481ee9e4d4c9a7b6d153b443ef581a9c5121af1e88fa557afe1a7d9f37c9b22956f698d24cb7cab79846cbf71ac69a8a
4c8fb005f2238679a6e5f262d30bf24679ade13953a5692bdb2eff4ede6278c5546df29a505614a3c91c8b1ec5638c3567609a926aec68f6597bb35e
d843e966e86eec65ceb20ec181ca00e35073d03a4ef0411a87f72a63d4718631da18e338739e7596f531566a7d91ad36ec649b0d35d63f58cf595bad
99e26d80e4a7b0c0288203b990cf603b4ef3fd7cff69f6069f7b9a65b00c39bfe55ccb2156cb074b43a5483e87add0a594bc8190523b7b3678874193
8c4eb00b3203d86d4e3bd8ad4e8b154465b39acc268bd36c36e55acd46079895727cdb66ae77d8ac16935145d0ecb2ddec686380a693dddc8eece6d0
ab1d9deae13cde79c326517bd26be25d575440d0fcb20a8aa61ad11a698ab23aac7e6b8e7588e92ed308eb78e378d37453b975a17595d5650202c2ac
58cc36b33d8a792487ec50a24c6eb3db126b8bb5a7430ad9539fec5332b40ec654538a39c5926eed68eb68f7397b88ef21a42c394be969ea6eee6ee9
69ed65eb65cf72de0e411694821894834a500d1a825aae7180699075886d883de81c0577b3bba5d19827e7117f46137fee35de6b1a6d1e6d19671b67
cf7316b12269aa699a6d9a3ddf59a23d667bcc5e01cf18979a975a2aac15b60afb0bc635e6359675b675f6cde6cd969db69df61ae71f9ce79cadcec9
c44bc5c6422fb1fa31c1cf6c69d588d54fac7a68f8a8ec24de27a44a533f58b06e70d9287944f36a7c4870722c79ce33c44923bc148cd5426f134859
72b5ed508fdb150d19c84c35b5bd96b284f422f45e477f6baeeb484320fcb6a1f117af1b82b98a1429a54983a42106c5acd9cdd118a775d27ce6eed8
4bcb320b6a0dd0a97587762f8ed31e34e7b37ca908f3e57c65a2566a5e68fe9d392efc1e42bc83644973707acb70696ff393d2de96c972feb6e633ab
b661aabe672cde8fcbdb0824138c0a76fd1f76ab3f515f35ac37324d821483ec37ebb941208cc37fb7696d147028679a837830b4e1ac6f3aff9d7fcf
2fed61b377b059cd929e01d5b434639ee139926660fe1c47842b221b908c79f1e1eafd5bab79f3f8d296e6ef702dfb52ca62d8f2575ed672a9f987d0
38566ca810dfff88f8bae6f06143c5df8af519b99b66acd067f48809fdc42e670e2b2e5d75acfae3ca6a43c577cd5bf8bddccd27b301ec92948681ef
6e8624c7414e23db15e1727aa4e5a5e37973f5d6fdd5029008c9c51e976cbca5e5639ed03ced3b8a1296ebe32ac4d74722aecdc23c1d14fd9baef6b0
647b6852f48b49a5e587ab2b3fae3eb6aa540073bce50c77f1b7d88bac916dc7fbc279b7b283f28b14181c8c48d3d36c4b52b43541735a921ceee1a9
fadebef0ec0ef1114343c32d10741aadceed2e29b61ca2d7aa89ae7ab33db3efd78100ef7b2940297720eb8634fb7aaaad3b7f837820621365475bde
cd0d7aeabdbb66527a1afbc70d39785b1efe42870e532785f2f14cf1865e155f44c6c03d4167dc4088d222ed6e59d330d2a40e8fbd0e2f17df5e045d
1a852b8e725bf481c8d76d6b8d50af3001ed25aebf490c50dad91a571d5719b730cea16f1ffd1264829828ea978786207ded8d3a01f93feaea4456da
06e39bbf1340b33ddf86691a86b147d01e3590725d9345d31cb2cb363c52c017024f4067b59bb71be57258eb34d65b250118d7a1623ae96edeac208b
0f75fccdebfb1522d1534a6edab1a0d5d5665a3d03ba50de179d3930aa93d6d111e7d1623b1a2151d552128cc969c3bb5e275443409cf54f55825171
89feed294e560e5d0e747cdd016b230d29f531f14922150c0484bd773406e82fc4e530377b74ef718d546d3c6fb7c5a210f5c4368b60efbddef4114b
88d5774a7b0439c3bc47a226b13dc4e5d1de887441cc36e2b6a126b5d156c72e0d4604233b0cb46a8ec868b7e6308ad7f54971c644fff0f47698e988
e96210edf56d4f724ae596b4b51e4352bd3d362184d295bebfc4a77bf64de4bf69bf2824a7ed7911c6e31a0e3bdbf3e51a6fc49787d2b817e6be78f1
8f0fdafbfe04899afe91e1274fd92eb4d53f7fda3cdc36ce28beabd6ae7d9348e30c0ff378001bfff9d3a6bb6de37ef10d63503eae7ffb07c2d449cf
920f49841a2a9bd5bd709aae3748e7a190ea334a149453f9824a1595f5540aa96ca0b282ca362acf525924ed81cb62ac284a1514ab1970549e0f4795
1554dc502e9f6f6d52cec01ef92294d0f57eb908e6503d476e8439d2279029da8a0bf64bbde090bc29541b8ec17e714ffe4aefbb1f8752bb13cc423f
f4a0fbbbe57ab8954a1fbdce86b1626d327535371669b9a89525b48629540cf524159937d0a3274c82e5b01dde8113ccc386b363523f69bd74158b70
a7ec969f925f924fcb4dca00e537ca56354bdda55e30a419ea354d1ba96d341a8da38d75c69f4cf1a635a637cc6e739ef965f3654bbae5536b8475b4
f56bdb36fb72fb5f1dbd1d631c931ca58e171ddb1def39e39db39cbb5df1ae2774cedc8ea3a0134c050b69a3035e109c943d5224d5e27b48038c17df
c0c946626c96fe5da9683388a4ab505b028d0d0cb7b1dd7db95d5b81683622dc56c1cd8ae00e9805b3e171980bd3600aad3e4fff3e7612d9011f0420
8b8e6c6a4da41e3ec8a53ef3a098ca5c980c05f03074a6bb436026f5ef4aaddbe1213a7c70cfb5b98af5abc9544fa6318fd0b9907a9afe8955bb5f5b
7514adf408ad25be7a9c49bd051c0534e65f5bb13fb5a6d3b831309f7a4ca2be05fa6c93f511053a463e9a65269d67539f8934ef34eae7a3f1b368f5
02fdd9cdf38cd46729268866d13183ee8a558ba9ef2c7da600ad9d0d39378c6a1b13fecebff529fdbbee5ffe6ed77558fc3781f8c63ff43f024e3dab
6dff5d7fe8bf111269ad34482779ce21caf52049ee0f0360200c82c1c49fa1300c86c308b80bf2e06ea2ce48184d90dc4b51e338b80fee8709b05a64
5c4c660a5399816970100e31233319e6cf9c16c8ee951baeef08d7fdc3f580703d3054df9e25ea81b9596d7576b8ee562b2d0cb65ee5d8e4c67fa4e2
df03f87315fecd863f71bcc2f1afa9f8171bfeb90a2fa7e28fcfdcaefcc8f15215fe50858d4df87d137ec7f1dbdef84d2e5ee4f87500bfba3052f9aa
0a2f50c70b23f1fc9799caf926fc3213bfe0f839c77301fc4f377e568567399e71e17f3c89a7dfc23f71fc94ba7ffa249e3a394839f5249e1c8427fe
18a79ce0f8c738fc84e3c71cffc0f1f71c8f57e147c712948f381e4bc00f037894e3fb4b9dcafb5e7c2f121b381ee6f82ec7431c0f727c87e3018e6f
73ace7f816c7fd4eac2b4b55ea38d6bef99652cbf1cd7d139437dfc23717cafbfe2d55d93721d88afb82f2bfa5e25e8e6f54e11e8eaf73ace1f83b8e
bb0bf1351beeda99aaec2ac49d3b5ccace54dce1c2ed04f4f626dcc6f1558e5b396e71e1668eaf6cb229af0470930d5f2ec46aea525d851b396e78c9
a26ce0f89205d7bf18a3ac2fc417d739941763709d035f30e1f31cd7565995b51cabacb88606ada9c2d5ab6ccaea0eb8ca86cf35e1cacab794951c2b
574c502adfc2ca85f28adfa62a2b26e08aa0fcdb545cce71d9b35d95651c9fed8acf109acfdc8e154f9b950a373e6dc672ba515e886544a9b2545cea
c4df705cb2d8a92ce1b8d8898b382ee458ca31d8fad4934f2a4f717cf2497ca2104b46799492545cc0f1718e8fd9f0510b3e62c2f91ce735617113ce
6dc2394d389be32c8e33393e948433384e77e62ad347e2348e539fc4297451c47132c7428e93384ee458d01bf39bf0010b4ee0781fc7f11cc78d3529
e39a70ac09ef8d8c51ee0de0188ea369e5d1b938ca8323994319198df7b8f1eea111cadd1cf3cc7817c711773a94111cef74e0708ec3e8c9308e4387
3894a1113824deaa0c71e0602b0ee238b00a0754617f8e77485d943b9a30f72dbc7d180639f6e378dbad2ee53637dedad7aedceac2be7dac4adf60ab
1dfb58b137c75e1c7bf6702b3d9bb0477787d2c38ddd73cc4a7707e698b15b02665b31708b590970bcc58c59996625cb8a9966ecdac5a87475601723
760e60a78ea94aa742ec98e1523aa662860b3ba4a72a1d6ec7f4544c4b352b69764c35630a473fc7643b26119e492ef41562621326100a0985186f45
2f51d0cb31ae09637331862e623846176214512a8a63240d8a8c410f4737c7088e2eeae0e2e8245c9db9e87812ed8568e368b5442a568e16ea6d8944
334793038d1c35eaa67134b8512d44991eca24011ea4bbc8297771285217640e048eac96152e5dce3afdfff083ffd700fcb7bff8ff035b9b44a9
>}
\immediate\pdfobj useobjnum \csname EF2O20\endcsname {[ 32 [ 318 ] 40 [ 390 390 ] 43 [ 838 318 361 318 ] 48 [ 636 636 636 636 636 636 636 636 636 636 337 ] 60 [ 838 838 ] 67 [ 698 770 ] 76 [ 557 ] 82 [ 695 635 611 732 ] 97 [ 613 635 550 635 615 352 635 634 278 ] 107 [ 579 278 974 634 612 635 ] 114 [ 411 521 392 634 592 ] 120 [ 592 592 525 ] 8804 [ 838 ] ]}
\immediate\pdfobj useobjnum \csname EF2O21\endcsname stream attr{ /Filter [/ASCIIHexDecode /FlateDecode]}{
789cedcfc56e44410c45c1238599999939f9ff5fcb5b8e26ca6ab6558bd6b5255bee9ad0d4583dddccf0ce36d77c0b435a6ca9e5565a6dadf536da6c
abeda1bfd3eec8d45efb637b0e46f261471d77f2cf05a79d75de45975d75dd4db743efaefb1e7aeca9e7a17ae9b5b7defb18f2675f7d4ff05f000000
00000000000000000000000000e0af9f5fca330595
>}
\immediate\pdfobj useobjnum \csname EF2O22\endcsname stream attr{ /Filter [/ASCIIHexDecode /FlateDecode]}{
789c5d53cb6e833010bcf3153ea68788f0328d8490aaf4c2a10f35ed29ca01ec25422a061972e0ef6b3c36918a04a3d99d59afcd3a3c55af95ea6616
7eea419c69666da7a4a669b86b41aca15ba7822866b213b363f62bfa7a0c42633e2fd34c7da5da21280a167e99e434eb85ed5ee4d0d053c0180b3fb4
24dda91bdbfd9cce089defe3f84b3da9991d82b264925a53eead1edfeb9e5868cdfb4a9a7c372f7b637b28be9791586c798496c420691a6b41ba5637
0a8a83794a56b4e6290352f25f3e4a606bda4d1faf7ac00578b5e167848f2ebc51641b507259470540faa4d5265820a9a1f53402c480049002320007
e40074901c7d3514c7928974c5376ab329caa629b20f8a2cc42b5c80369ca1a52c43d85378b3d427ad966317dc9d92a7b07058382c1c9be2d814c7a6
f8b3f7a31cce3177ffc2518e3e398e9513a0f552ebccb164cee1f4141de4e820cfbc0616ac9dbbffe2e9d1478d288e6df3800bf0ba8e959f9f75c2d6
ebb08dafb86b6d26d7de193bb2ebb0768ab66b350ee3ea5adf3fc7b2e600
>}
\immediate\pdfobj useobjnum \csname EF2O23\endcsname {<< /I1 \csname EF2O24\endcsname\space 0 R /I2 \csname EF2O25\endcsname\space 0 R /I3 \csname EF2O26\endcsname\space 0 R >>}
\immediate\pdfobj useobjnum \csname EF2O24\endcsname stream attr{/Type /XObject /Subtype /Image /Width 100 /Height 100 /ColorSpace [ /Indexed /DeviceRGB 1 <ffffff26343e> ] /BitsPerComponent 1 /Filter [/ASCIIHexDecode /FlateDecode] /DecodeParms [null << /Predictor 10 /Colors 1 /Columns 100 /BitsPerComponent 1 >>]}{
789c63f8ff8101090c62deffff0d4390f781a1fec750e081c181a1c6fb0022edff0c7a1e384dd8ffffff6048f1c02904e807f97f839b87043e0c151e
009baf2d2c
>}
\immediate\pdfobj useobjnum \csname EF2O25\endcsname stream attr{/Type /XObject /Subtype /Image /Width 100 /Height 100 /ColorSpace [ /Indexed /DeviceRGB 1 <ffffff26343e> ] /BitsPerComponent 1 /Filter [/ASCIIHexDecode /FlateDecode] /DecodeParms [null << /Predictor 10 /Colors 1 /Columns 100 /BitsPerComponent 1 >>]}{
789c63f8ff8101090c6a5efd8fa1c80302fb3f43800774359027ff6f68f190c0a0e60da614399a76079e07007e8e6810
>}
\immediate\pdfobj useobjnum \csname EF2O26\endcsname stream attr{/Type /XObject /Subtype /Image /Width 100 /Height 100 /ColorSpace [ /Indexed /DeviceRGB 1 <ffffff26343e> ] /BitsPerComponent 1 /Filter [/ASCIIHexDecode /FlateDecode] /DecodeParms [null << /Predictor 10 /Colors 1 /Columns 100 /BitsPerComponent 1 >>]}{
789c63f8ff8101090c621e10fc187a3ca01fd8ffffff33e879607080f9ffff7f4389f701e4767ba01f1881cc41cc03a7097bb0db87100f9c420ea0a4
e441c98381074388070095388b81
>}
\immediate\pdfobj useobjnum \csname EF2O27\endcsname {<< /A1 << /Type /ExtGState /CA 0 /ca 1 >> /A2 << /Type /ExtGState /CA 1 /ca 1 >> >>}
\immediate\pdfobj useobjnum \csname EF2O28\endcsname {<<  >>}
\immediate\pdfobj useobjnum \csname EF2O29\endcsname {<<  >>}
\immediate\pdfobj useobjnum \csname EF2O30\endcsname stream attr{/Type /XObject /Subtype /Form /FormType 1 /BBox [0 0 306 309.6] /Resources << /Font \csname EF2O1\endcsname\space 0 R /XObject \csname EF2O23\endcsname\space 0 R /ExtGState \csname EF2O27\endcsname\space 0 R /Pattern \csname EF2O28\endcsname\space 0 R /Shading \csname EF2O29\endcsname\space 0 R /ProcSet [ /PDF /Text /ImageB /ImageC /ImageI ] >> /Filter [/ASCIIHexDecode /FlateDecode]}{
78dacd5c4d931b4772c519bf02e19318b69af5fdb1611f562bef86377c9144850f5e1f6888a2a825a51529ad6cff7abf2c0095af7ad0839e01358190
3833785395955f5d5999593d3fedeccee03fbbfbd4c8fffb77dbe79fbffafb9bfdab2ffff4d9ee0f5ff1a7fd87addd7d8f7faf31e17bfcfb15d3fe84
7fafb742e2ddd69b84ef6fdb776fea94f0b3e93f7db7dd7ebbfd69e7d264ecce663b1557f0c94f36ed6a9d4c76bbf7af76ffb1fb6167a66a730ad9c4
52f0a1381f5d0dd696ad99d2e91346838f29808bc5f1bbd9f8edd6e5c9bb28ebd6a99a0a0ede6d639e62a98eb0b78a853c59e7c0db5b9d4b1884fa6c
fbc56ea55845b8124eadf06c6a31be3a9b2016bec55492cb91c55a18bf9b8ddf6e93999c33a35815544a18c5ea1889d0e73e5aacdfcc5a15cbd999b5
acf393cf3373294842f4d9b7672fdb969f190c23269767165390a4d0e98f16cd54636b32b9047cf0879fbc0d102d1e65c0434ca22d8cdfcdc643b458
c0fcdc68a58289b9d13ac8a2f5e93768b58a75cdcc6aceb9c9a499d51464d1faf4db13cdb930a53ab39a8b61aa71663505796becd36f5034d9cbebdc
6a254f25ceadd64116ad4fbfbdfd5123529c2cc0219a758ca219dcd2c7304633c26e2d9a91581ab954ac8e91087deea3c5faeda3198945814be55290
84e8b36fcf5e148e48320d5c245907490a9dfe68d19e209ab1681ab848b40eb2687dfa0d5a4dc3918a46814b45539045ebd36f4f340a47249a062e12
ad83bc35f6e937289a8623124d031789d64116ad4fbfbdfd5123129ec9ecc768d6318d66b6ca0f6988668cdd5a3423b13472a958278c4538cd7dbc58
bf7d3423b12870a95c1d64214eb36fd05e148e48320d5c24d9096429faf41bb4198523124d031789760207d14ed36fd16a1a8e54340a5c2a5a0707d1
4ed36f50340a47249a062e12ed040e5be369fa2d8aa6e18844d3c045a29dc041b4d3f41b7cd63422c52999384433c5289a6537155bc66846d88d4533
164b23978ad53112a1cf7db458bf793463b12870a95c0a92107df6edd98bc21149a6818b24eb2049a1d31f2dda6f9f9b0da269e022d13ac8a2f5e937
68350d472a1a052e154d4116ad4fbf3dd1281c91681ab848b40ef2d6d8a7dfa0681a8e48340d5c245a0759b43efd9c680689dde9c9c007d79890d5e5
031e109fa1af2a7cdbddf7f8f7eb3063bb3ce3cb3fada7bdfd692bfdc24fa56118300bd25817a1683b455fabfcbc7fb7fdecc5f6f91fed0e416cf7e2
dbd61b7cf1cdf63f779f6cccb3dd7fed5efc79fbaf2fb65f30295fa798432ad0a38b167ba98d5529b95d99d1f96a814ef69309f9412cd9255278acbc
b3ab59fa7c818e3571aae1616a72f7d0f2e923e8c9fa3215ff3045f9455a7065f7584d4933f9d44466a27882e5e0f71006c31283d94c7935775f3395
f5cfdce2c8ed4779d69cc1aee71e66b0b8a00fac3825fb115cdbc1f2d13dcc486991569aa2c91f8129583bd887692a2fd14a750ae623688ab74a089a
b3bdc85359b15522daf9e0afdf2a57b254576c956b585ade2ae198fe412cd9c59032ec95d728ca8a677aff30aeec9acdf2c1ba5ada2cc54ddd03f5e6
2eee966bd8bbd5ddd2217f7ea0c9fc9aedf21af746be3981ccc3b80a6bf6cbabb84a05fbda037515d76c98d770258f4938a8ca627781d52f339556ec
98060121c747ee987830827f184f79c596b986a7355be66aa6ca9a2df31a4df196b99aabba66cbbc4a57b44daee4ca2d87171c506dfc08bae2bd6a35
5776cd5e758dae78af5acd955bb3575dc515ed55abb9f26bf6aa6bb8d2bd0ab426686e0553e1f25e65739a8ac9d7ee55eb798a97f7aa553cadd8abd6
339556ec5557698af6aaf55ce5157bd5c37575f978b79ec372e978b78abd9b3fdead57485db1655ee5deb465aee5ca2f8617de32afe34ab7ccf55cd9
155be6c3b97afe7b38f6871dbf3e608edeeecd94762148815f4af8b06fb4a7cf6f4f9fad4d5388b6bd5f80e1fab1bd582025522683f85c4f5f8f95df
e7bf7707068ee5512c3ebeddf013f6baf63604be31adfdbbddf37fb3bbcf7fdc7db1661d59c04cf5503b96ea7088a696a17f210f422cc716943c0823
235bf892a9c18b9ead990ce293f7f28285549073f59561b88a3526fadc40e782cd41c0907c2ab681b1646feb6ebf055ce0122e35b846834112c08c37
2509013b85e04db233b0c694c46db60cbb09a2d852774cd54d550a5579e0c04b3bb37a7be0a0738b27af5493d38ee5f253c1489b77ac830e8200f45e
4cb176181b2603755a3f50f513b60f93ea8c03fc9c6c7561e0d6e338974bf1835c04b20e06f8a42fa2aa9a250ed80ac4ad5a8ce452db9ef383fd79f7
680fc2574fe4a04144386c2b383b19ecd4451c54d6ac09fa673880db521db42c20480668096088b2c53730d61293a80670b131c27a02d78473360848
57477edec956102266d91958538b6a2040b09f702e0bf02aa6eaa79a8381ff310761f21ebbad3d70d0b90d53c621d31d7479922b4cd5048f278075d0
4110f0136489c90e63fd0497acf009a6ea2688576c9d712037f34c824f30b70ea13225ef07b908641d0cf0495f4455354b1cb015885bb518c9a5b63d
e707fbf3eef1a40e1a458416a92aceafc1a5b681028d258aea148de03517d1313007bd4145c0825c14b08225acdd9e5ba005e684e580d6ec444111ee
025161e20ab3a428fa19b09a0f11774b28a85727de4414e12a258adbd1da38e306697e1fd63e711961332bce4cd244682f8ad793dc1dc36c44d058c4
0d68a4348ea3780151b4f04d274e30ac6de19c557c80b844feeba39895a4518ce566f4a421a5a89ad4b559e7caa55a47a5512b9eb1f7feac173ca91b
2661df48dc2860055a68db2450a4791262144df829c9132c988b594209b0e0aa842d60c959d39e4ea0889512e180d612249000ab354b282c30474e12
4706ac56e3dbe64068c419d34bd4258a708f9a243cd3da093b474d6d6f222e13629b93a04fd2a4493c012e4e72770cb3b1c788d62c8fc4fe1083444b
a208bb542bc17258db4cd21284031097c820103f6156924631969bd1938694a26a52d7669d2b976a1d9546ad78c6defbb35ef0a46e087a265451440e
f23c1d8e93406372a2324525db8f51540b0c077e510e3098578c002c49ea252a019a8b157365517314e500c3fe2066cd88aea5887208436c36f67088
1ad068a3b80f51c431d36471335a1ba94a3687331c7199a7ec83b82e4993a78aa0061727b93b86d9509d8fe2003412e7286cc47000a588742715230e
c06b03f57299281097c0706e14b3aa3403d6e526b46b8828764dd2daac73e552ada3d2a815cfd87b7fd60b9ed40d0bd83f5480529413c6e1d0083446
2b0144d122ac060934c01c622494030c669190042c057f38ab00cdd948f04aa2e624a102188e4d12e48005d9c42c63301b4cd07685018d2e4930258a
450e93127469ed32f9e20e2735e2b24cd087047292a64c15ea868b93dc1d83392a560c120e150586d3928443a598cb84a3a784435e1ba83751c2a172
090ce72d09722acd8075b909ed1a228a5d93b436e99cb8ecd62169d48a67ecbd3feb054feb8660dfb5d26184a0251fce86406330a23245abb0ea45b5
7239b25a514e6db712c508c0520c87730ad08c0413e68aa2e62cca01866f6256b9d9658d288730980d01a9ed0a038a6d4ddc87288aba8db819ad8d3d
a3fac3f98cb8ac9893c475499a0ada465c9ce4ee18cc8115b11358cb238537230e401433a4cde200bcb614956a1007502e8119f823ccaad20c58979b
d0ae21a2d835496b93ce89cb6e1d92a65bf19cbdf767bde049ddb0ca731b252c045143399c0d8162bc041045e506a3b312688041c7122a80f9122424
014b291dce2940732812bc8228a74aa80056a393201784b695504198982d1ece48030a2d4b30558a51828f95a0ab6b0383971dce67ca25d08c43a2a4
355d1a6038764ac057b915c36c3800ac992c8f94005b251c12453800105bc7b5e100c8082439502e61562421302b49c398caada86aa853244df6b507
9d772ec93a5d1ab2e2197befcf7a01b9a1d9bd3efa5a39d62c8d144c7b01a8574d0972c7c2e9f6b3adddfd3af497dc04bdb6bdd7433e578652f07db7
35dbfa7482e8eb12b4bc2e425439de14bbbc6eb8bb2e29e9b42e41cbeb4a665b6a3b795f5e77a9ab93821c3677ede271eec57373b8ff2bfb07e63fff
2362e51485e2279bbf3edbbdf89ed93f16ac7b0d516e058fc839e69d97324e2bb61b3c70e690c83cc45c6d5d3d7d1e9755e0fe55b39f7c88f72f1a16
16d5a7ebb8a802f72f1a92bc5e31eb51dc6fa8f6d78b8494cd48eca29c59901d9299fc191bfdac3672bb3f1f3a12ed111b37e9b133316b3ce011bdbf
738101a611a74bd963730f43e1c4cd8dad738daeca3d6fa77cf94cc61f9dee93cd2f9bb79b579bdd266d0cbefee593cdd79bbf3c5bd97191d758aa27
cef1a494403d97c36795f538e14cd76524756ddf65a4d63a2f4e3b2f17d7fa08611947e409a78f3bcd17c1b1cfcdba2fd6d529f879fb45508cbcd37f
b15e6e82ce1b3082fa34efc0ccd0de7e18f0deab1828f7b6c6c007f54006ae7bc36490b0f756067d5023c63a24cf76de89111409fbac152328f6d03b
bd1879bd0859ecac19236828f36ecc80b23e46fca43da6ac9a663ed82ecc35599124248b9ff38ffd82df50287f12e7c5fe19cca143cb8d19c1ab6db5
62c22d8e1cc1b4aa3215fa05cdb955a087ae80c541c8fa56ada61682a0deb7ba36751b66686f4d0c78ef630c947bcb63e083fa2303d7bd993248d8fb
2e833ea84963a525545a0d79181d91d9b46af34039ca0d1851c28c0f9c0852ab620f5c07a97237f5b284840efa18f0ae3da24c9a263e06bb10d76445
92902c7ece3ff60b7ef3d4ce8b93966c3fb3a68dc0a5b6ca9ec23641e8d24a80da181010c94a9af56d6c4a7277bf0c8d1b01bd6d1548ed4b8c606f61
30dcbb1d4cb5b7459803eaa030b7bdd9c272f5ae0ceb801a383661874dadcac7637d7b3f7a68e108686cab1c8e1c38c95dead0c41150ea57431787c1
41070c777d2955d2ac72305841b9258ba95c64db337eb03fef1e4feda245daa8eda4ca0d1d814b6ee516856d913b01ad2ea34d0301b36b251cee2fd8
221dd356edd1568480708f5641ec3d8b11eced0d867b2784a9f696097340dd15e6b6376258aedeb1611d5073070471646ca5171e6ba5eddc8ab24415
5628ad9c337280c3b46b951fe61659be6b2522964bc141070c777d2955d2ac72305841b9258ba95c64db337eb03fef1e17b3793a39e8819cb0e5fc56
4e39b65e91d0d333a54b1376cfd278f44f2f353d2aa7979794915f0e493d63f72c2daf73db6bd37aec96b2275f93d7f7a488327bebec0c3b2b82e451
f9aae4feb838a7f7c7b5ef4ff069e94766f8a78d9672fce3caf767f9b4f283d3fcfb92f43bc9ec3c57c5d37729e1bd9caab7cbbdf97827f3b1b97a3d
e5ea5fadcfd59dcd93a32a8313f7cfcafaf1b38a7b9c7026571f495d9bab8fd45aaeee3557bfb8d64708c70e6e97e3dd5c5dfe2a824bf35c5dd09ce6
b9ba93e293bb9bab0b2e27b931571734d579ae3e437b6e3ae03d931d28f7ac77e08372e481eb9e510f12f6ec7bd007e5ea82bb30cfd5e5efebe0b03a
cbd5db5fdd29777375c18d9be7ea00a662e7b9fa80b23e46bc6b8f2893a6898fc12ec435599124248b9ff38ffd82df50a07e12e7f5fa9606e7eace43
9438bf4529680ef36b94ce3bb94f70e71ea5e021cf2f520a0a9f99dda49ca13d371df09ec90e947bd63bf04139f2c075cfa807097bf63de8837275c1
9d9bdfa87452daf2f32b958222c9b873a75270e424b34b9538394cc5cc6f550e28eb63c44fda63caaa69e683edc25c93154942b2f839ffd82ff8cd93
3baf0872b7d0846710cfe1bcd02428b4312b34e1c9955b08770a4d8287342f34090a9799159a66a81a8b71352d535637603ed869986b753196909c97
f431386f95d742668526e78bdc0c9f159a04f5e96ea149fed042adf34293a0b9ce0b4d033a3ccc03de1f7da24cdb04f1316c2ac435392f4948ce7bce
3ff60b7ef3d4ce1b4490963f70a1c9e1f8ea7cbbc9a8b080489725d669d9c24997b9b4db915ce21018e76989a25a0e11506eb05a2e348d602fb230dc
0b324cb5976e98032af330b7bd24c472f5e211eb800a4d023bd36e35f258388f6dd71f992af2a7d86e4a8e1c20f92fed5225710b309776fb92e42290
7530c0277d1155d52c71c056206ed5622417d9f68c1feccfbbc753bb28b253b97138ab85620e22c5ac162a60b6b35aa893fb38f94e2d54e01066b550
015398d54247905c54617251a54a2eaa1cb08310b7ea4c2497ba1de9805d1430d2b7b1168add0a6367b550017db8530b7511d3f2ac162a60ceb35a28
83830e18eefa52aaa459e560b082724b1653b9c8b667fc607fde3d9eda4593b4d85a01886ba1d87d100fda0d488505cca65d95d4ca9abc9c2807a7d9
fd7681a539395c701730f97653536b7b23d8eb800cf79a2153edd545e6802a91cc6daf5ab25cbdbec93aa05aa8c0b6b4db903c36b696e470d55d40ef
da0dcb9103c893da654ce6164785d46e6db25c0a0e3a60b8eb4ba992669583c10aca2d594ce522db9ef183fd79f7786a17cd10a2dc29d7bb9c100c66
e57a01539d95eb9ddcf18e77caf502073b2bd70b98dcac5c3f826a1e82d59444558d4e1cb08310b7ea4c2497ba1de9805d14b0cdb372bd93cbe66556
ae17d0db3be57a9721799c95eb05cc7156ae6770d001c35d5f4a9534ab1c0c56506ec9622a17d9f68c1feccfbbc753bb28644fb995bcf9dd0c07f5e1
c0d65e8be8b080a9b42b967affdfc12826b4db98fcaa80c0d225cffc7e8680f2b6bbe5173418a4371546f8f8520353ed6f3f3007f4a20473dbdfa960
b9facb17ac037a4f436030d35e22a1b1bef5c6931fa8fac99b763373e4008138b44b9cccadc3b476db93e5629074a030e94ba992669583c10aca2d59
4ce522db9ef183fd79f7b8d851a2fa96968d095beead48852eba2b3a4a7432d1a509bb67e9282abea2a344baa262b962f72c5da4b97d6d474972329c
6cafe928f5d23ddf15b571869d15c154a94e5dd3513a2e3e5c183dac7de1caa82efdc88ed269a3e55ba387952fdc1bd5953f6a47e94ecb65de51c1d3
77a92d73b9a32485bc940e7f32e4b10d257f6a287d3e6b28f13a88ae65e72af2ae147889c9d1226d15107a094220b9dbbcd87cd796709bb8499b4f37
effba27b7cfd61f3019f3fe0d31bfc2c8c387c7dff6cebe45d16e3532c7220f804433f60c837187c98fa162b7cc07fc0c1ef76995f3996d4eae8294a
c333e40fcfd097ed194af2fcd68808817dd71c06591df42f6d50968b8292a5dad6c4c2a0f6776df19b20a12e6553aafcd9f1c36fdcc26f86d57f3aae
6e6c8848c2d289f0b0fa3f2eac3e507adf06c96b9138d138b9b57e86d23f1d07798343780caee613b7bb47fcc69c9877d18684f4f824fb3f7cb386e3
55b2fff3a2e6cba2e67f77e418c7c0849819bdbd2ccbab23b590e42882c7a49c7ef3f2f49b80f35935cef775f6ed3772cc4ca6f8aaf26fbe3bce8935
d46c90aaafe06099da326f6f8fd4a042c420879ce9f26f96395872a1fb38f8dbc90ac6996824f05f43edc35107392379b029afb1cf0f8bbafef9c801
b20039f824f59de5751ee3216f1675fd7a513bdf3d82ebc758ee97c575963de4927deefe61a7f667f3714a2bfdc589c3436ce701e2ebcdefb08bff82
2dffcde6efd8c3df63277f893dbdc520b709a7e021ff7dd5c67ed8fc8c116f5a68903fda278515177d9427a685825787c919d369f2e76df237f8edeb
166d64a99720751c6ddb6b0f2de42d8710eb7110721723de7f9fd6fdf767f20e894526836d42b268895e2fc18484b8d718208cfc0c46deb4a886cfcf
b6f309ffb3f95b0b7187a8f8f33186ca77d878888d62a8634c7cb9221e5a272fdddd6f9d3f343defeeeaf9ed91e17758eafde6afc4d9b732a5196907
8e7fc4bfff135ddf25f223c9ffbff8f9474cdd35fdec416ac5e43dbebe85ebbcc3b08945fd62bbfde2ff01bd49bddd
>}
\expandafter\gdef\csname EFWidth2\endcsname{306}
\expandafter\gdef\csname EFHeight2\endcsname{309.6}
\pdfobj reserveobjnum
\expandafter\xdef\csname EF3O1\endcsname{\the\pdflastobj}
\pdfobj reserveobjnum
\expandafter\xdef\csname EF3O2\endcsname{\the\pdflastobj}
\pdfobj reserveobjnum
\expandafter\xdef\csname EF3O3\endcsname{\the\pdflastobj}
\pdfobj reserveobjnum
\expandafter\xdef\csname EF3O4\endcsname{\the\pdflastobj}
\pdfobj reserveobjnum
\expandafter\xdef\csname EF3O5\endcsname{\the\pdflastobj}
\pdfobj reserveobjnum
\expandafter\xdef\csname EF3O6\endcsname{\the\pdflastobj}
\pdfobj reserveobjnum
\expandafter\xdef\csname EF3O7\endcsname{\the\pdflastobj}
\pdfobj reserveobjnum
\expandafter\xdef\csname EF3O8\endcsname{\the\pdflastobj}
\pdfobj reserveobjnum
\expandafter\xdef\csname EF3O9\endcsname{\the\pdflastobj}
\pdfobj reserveobjnum
\expandafter\xdef\csname EF3O10\endcsname{\the\pdflastobj}
\pdfobj reserveobjnum
\expandafter\xdef\csname EF3O11\endcsname{\the\pdflastobj}
\pdfobj reserveobjnum
\expandafter\xdef\csname EF3O12\endcsname{\the\pdflastobj}
\pdfobj reserveobjnum
\expandafter\xdef\csname EF3O13\endcsname{\the\pdflastobj}
\pdfobj reserveobjnum
\expandafter\xdef\csname EF3O14\endcsname{\the\pdflastobj}
\pdfobj reserveobjnum
\expandafter\xdef\csname EF3O15\endcsname{\the\pdflastobj}
\pdfobj reserveobjnum
\expandafter\xdef\csname EF3O16\endcsname{\the\pdflastobj}
\pdfobj reserveobjnum
\expandafter\xdef\csname EF3O17\endcsname{\the\pdflastobj}
\immediate\pdfobj useobjnum \csname EF3O1\endcsname {<< /F1 \csname EF3O2\endcsname\space 0 R >>}
\immediate\pdfobj useobjnum \csname EF3O2\endcsname {<< /Type /Font /Subtype /Type0 /BaseFont /BMQQDV+DejaVuSans /Encoding /Identity-H /DescendantFonts [ \csname EF3O3\endcsname\space 0 R ] /ToUnicode \csname EF3O8\endcsname\space 0 R >>}
\immediate\pdfobj useobjnum \csname EF3O3\endcsname {<< /Type /Font /Subtype /CIDFontType2 /BaseFont /BMQQDV+DejaVuSans /CIDSystemInfo << /Registry <41646f6265> /Ordering <4964656e74697479> /Supplement 0 >> /FontDescriptor \csname EF3O4\endcsname\space 0 R /W \csname EF3O6\endcsname\space 0 R /CIDToGIDMap \csname EF3O7\endcsname\space 0 R >>}
\immediate\pdfobj useobjnum \csname EF3O4\endcsname {<< /Type /FontDescriptor /FontName /BMQQDV+DejaVuSans /Flags 32 /FontBBox [ -1021 -463 1794 1233 ] /Ascent 929 /Descent -236 /CapHeight 0 /XHeight 0 /ItalicAngle 0 /StemV 0 /FontFile2 \csname EF3O5\endcsname\space 0 R /MaxWidth 770 >>}
\immediate\pdfobj useobjnum \csname EF3O5\endcsname stream attr{/Length1 10204 /Filter [/ASCIIHexDecode /FlateDecode]}{
789cd579797c1455b6f0b975ea767755ef9dee2ca4d3e92c9d852db14358a3b411904d8c1290a068024908b224107654108684084c8248104488b2c8
36189127094446461419cca803f88ddfe87314c6e54d6498798c3890dcbc53d50982f3defce6fdf5fdbeae9cbafb3dfbb9e7a68001808b5e32f8ef1d
367c0404611000eb4dbdde7bf3ee1f9f74b8f74a6a8f2478f7def1137237bdb0f720003a68fcea7d77e78f548666fd8c16cbd4fefefef119c1197f2f
d7daedd49e386d765105f6b33d0bc0f369fd23d316cef7c38c38dadfd0486d515a317df6dc7e0b1f0730511b0e4c2faaac00233da03c426dcbf4594b
4a9fc9ac8da7761980e789b292a262d3945f5501f82ed178ff32eab0ee30ae0388f7523bb96cf6fcc56bae582f523b44eda766954f2bbaeff8c897a8
fd16b57367172dae90d71be602f8357afd738a6697a47e77670b5593899ef315e595f3fff0d6650920e1071a5f5a31afa46288f12f544d9c4f3c9481
262b0b847f123d083da82f16e2a9ae8da9d01772401a36626c3ed86615cd9f03d1a06182ce4e809b356d269b59326f0e98b49a0e32eda0952690984f
9bc9fab149a0801dfe855fe72e82a304e7083eedfc547bdf327a2e3ce327f3ebfec97e9ffef391ee711dd3d1dbc64fdedad3794ec37ddb8a33b7eea5
cdfe5ffe18c9c44c3ab082039c10011e88d175602149315d86649cbaa43981a16b9591e4ca689d439f23e933388d92a5d188427a0bef6903e8d2e376
38086e5d8f4b8be6154d85baa279b3e740ddd4794533a06e5ad19c4a7a9795cca3f79279b3a06e7a4939d5a7cf2b990975654573684e59c954ea9959
34a708ea661595fbb537d9c3cf6617cd2f83ba3933b59ef2e945b3a16ede823934737ee99ce9f42ed3f6ff1f6c46e77fd68ce945b7d98d0c61bb61d0
5f2f3971e2062f24421ad9259294dce0d34b17c411efda9844f5587a3b6905d3df1acf7f806f200f722c60fc4c432591dfcb8524c80b61d45afdd65f
779b9575d527fc0bea23094b0bff053d4f0ecfd5ca9bf55bf6b8ad7ff22ded793fd6a58100376aa99e7373691a71a9863709bf651babd32c8567f12d
d4f4854bfc3f502ab96807b301d1244b927c7345d72faf74783184c00f4b0c6ee1665b8db3d9c5dbe66017687149c3b88b5a4c6fcbb0864a1fd92292
bdf9a11fe96d200c23c93f0845300d4a6106cc8205b044d7bf1f327f325e42e333611e2ceeececbcd8799e7cecb79da73b5b3a8f77fea2f360e781ce
7d9d7b3b7776be743bbdffcd2f1cc79aba5a361d571834ba3309fa41d8fefa77814230b00b349e8675815913481768bb3ed8055682a22ed0f69fd605
5a4c2b212825d0ce91195d401227ce80b807f26c004d970b083c048b099610c4b06ca2f92c3d27613f6c637ba8a5ed33977a1aa4c3b09a5634c1dbec
2cab91fa50df1eb802e76866359cc5fd32b0d19045bd009f7009aeb27c38427b0c626e36c86820571a271f911f949be4afe556182057caad72a15cc9
b2f0653e91ef211884ef905d9c21bf6a629f43251cc36f310b5be461b20d3ec756dc0f7f242c9accce422de97c19d1e266e5b05c5a263d483da7792b
6ca5a79cc65bd976768ea83bc656c105781e6569246c671788afb3f03dacc27c6939a9224b2a25fa4fd35eadb47e2b5452e0bac05410522fea23ea09
d754fd1d877df805fdb902cb09733eec323419dcc624c2a2496c0f7b9bb5193642039cc347702efe9ead9693e4bdf248a80d4b000ba196f6deaaad31
94b225c4bbf62cd3769716c9856c3f7c2b171aa7d2deef681c11ce23d283c45129b4102c323888a7216c35d610a5da681cb41a47cb19b49e76303e49
5c039463363c4eb56570080e431fac875ada49e7d730807f4f2bb7c95f10cfb56cbdf43db4e230488752f932c99a0217d4031c351ab88c1283de7e47
a3141855dc187a6092ffbd82843ebd7fd2f43b8cfe46c86bb42ef1377576e64d9263794123f73662c0d4280792bef89f06bfe8d37b4cde247f63c7f0
615dbb0e2f1c467de32751556b5137f50f1fa68f69481b7980fe461536faa795f99f713c9334f81947c9e03e9a5f49da894e9e8c54abeebc28d79276
cc100549a10843830b1a2c1b5ceba215afdd875e4f6cb4a3bded6adb1de0b874b5cd713993254a4e872b2be8723aa4d420381d9094a8bda5b5db5e7c
91fe5e7cf10653c4b51b37c435a6f03cd12ade27686559f4f463590da25254896a51c9d6b3256c295bafc5a72fc8a5275344572114f2e462832c35f0
a78dd0a098e20d5e847866769c1fd368cf9fd44c9343030bda4eb51341196dc1ab6de7db3249060589ec881dedb234654082936707b29c099e04c146
8b2dace4d76c74fbaefd72e5c8a691d72feca70d485ff268e2d80bdb43a9313d6231daebe432383997731d2f399fb336b837c8e4b7e05025a67aa31c
688873b48f69f4e48f698ccc7f784ca33bff61a2043bdf1a5870ea7cdb5b6f395d83baa8b9aa536374f0ef8cfc3bd6e87538a306116da1e00479229f
685c2a2fe50b63ab638ce4d531720f52af773e2c342ce851193bdfbb12aa6256f65819bbd2bb17f6c63aa7c094003191dd1f06dcc5b2fba524251a8c
d977b1aca0ec711b8c06a05072b27d2c8931abe8be57aa1e3bb778e9f949df30f7f08763c4d5fdfbf72f621b06cfde3c6a517dee3defdf11fce6578f
ecae88137f22eeb791be2b89fb34a808f5054f845aa5c457f9231a3cd60665a3c1dbe0df98b4c1b0ceb3333dd21b01e88ef1a6f81d5e74c72b86744d
0891f9ddfc2b3aff2480ab6dc42549c0d176e9eaa536c757971dfa4352c96421a5d857145fe42f4e90610af3318f5b4e484c49cdf61123fd89ab5e2c
3b5cb98d3d1cba61a7f8507cf3e8e9c7f3df9b7de274f3ee436f6cdabef3f9f127e6559e29f88a597e8e81f853759ffd351078fb8e607dedcf36ed59
5451b92c39e588dfffd1e1270e68165e4c5ade45362551c47f3a14c7ac6805446b2ea0d9d8c0193ead308b0a5e8349b6d81c9f8e69341363569d318b
c6d8f99c536d41a7a6d74be773da82c48bae58f90c29f78ca6d29e66e80923a1800e8945f00c1823592f4861bdb03f1bc7eeb7dc6f9dc84ad902b614
57332ba952610998e5ccf22439939c09d968101213d9e2c285331d8ff240fb456c6dcfda2b1a58e1db446f2d69a85ef7c86498144a8e3080b5ca020d
9186066fe46e4783a5267183775dc092a878637c115e4c888f0d908b921a2ee94e7aa9fdd28f0a08b9e95462ad522bb6ca67f959837197efb04f9ac2
a6b04483c71d19963df3f4654989126ae2078f1b92fc9a43270423a55d6b76ec5843c094b12f8c7def9c7dc8e1995f302eae7c293ac46596c762c7be
80438ebdfcd2f1e32fbd7c4c5ad2949c22fe2afefcd014f1e73f7d25fe4377f1a96cb74f3bbdd712576b75ae922003ee0905a289a7544383af4f836b
836f5deacecc684b724faf27d96b5728ea50e8b127c4663ada4fb55d3dd5a6b3d36d637a6b1019d72d2c04fa928f2467052335e7d0cd2c293139bb5f
ff88ee09645fd2dabaddbbebeaf6ec16bb576e80ce7fff5c6c78fad99de2dab56be2daae911b56addcb871e5aa0dd23b5bababb7be5055bd75a2fff0
8ad73ffcf0f51587fd89efd67ef2cd379fd4becb8ae6af5c399f80f444b706de481666a25cb56fc8039b94a7d926874972a8c063ac41f02ab24b8f9d
44f9203d425078385c18c148faceacb08c03097a99ced8c6ab2c9bc58b2fc45991cb76b0c3ac5e94893c51c4336e2c62d1ac2febcda2f688cd628578
4ad487b1cbeb75ec113038144d9835025c0ed5448921e11fead408701301c176cd4109bf1e2d4366bb27de33d4f398e7550fd72909cb2790a0b99edc
8b3993d846b17eebd6f562207bef0663a2f386f835cfe8f8e0d9eaaa67f75cfcfd675f76ecede23e9ff047501c7d2d94ad988ca81a9c282377ca32e6
52e2e241d9b349716fb23e6d96b9019d0a78236d5c8d89919d43ddaad722c7e9c49dd2a873862594a391e81aa43d370329811e450f877c24b950eed2
08c681332e19d0287bc0c3dc522446c90108b0809482a98614638a2945f1fbfab3fed20836422ae30be4057c51c41ac31ae3f386e78df153f4501415
91847d99c6af33c1afd9cd4d51e0fabb97ddd5fac92f47af5dfce9afd97b0cda5775d48867376d7a566a89ac7b4a94b1e5f5533b6af8858f7fb7fe98
747fc7e5ea55ab566b316719d9781fca4d5408400b9d2df1e628c506fba20ccd36a7bf2afe98b739a9c9b92eca0251186d554ce67834b987a79014de
3fdf160cea069e71ead2d57632f177758b776a8208cdc98ccbf465c667fa3313321387a686e242be507cc81f4a0825e6c5e5f9f2e2f3fc7909798979
a915a9abe3aa7dd5f1d5feea84d58975a90da957527ddd4bbb17752f28f415c617fa0b132a7c15f115fe8a8415be15f12bfc2b12a26f8d0e77b201ce
a46c1b858714f2a6ac845b2375a474e2f3834f976f696e6a1adab2e6e0d98e1b4c7a6573e11bf9252726ffe71529ab74d9d4ca4f8ea48fed787a7f69
d1c997df7ccbb57c6ddfbefb5353dbb5337f6ee745fc9a6415034343b150c5d6c8b62aeb1ab5d929374791907a185d5618e91edec3d17e29d89d7d88
ab971d7fbbac9970ac2376456c5d6c43ac66c27af8ea2278804723b62b7ee1d7e35ecc7bfddd775fcf7b71dc7dbba774888f591f6698f0b29c7db057
af8badad177bf5da9f9cccee6236e662839348834495bc9ca872d34db422944cd6a55499d670cf3ec69b2dec7874b3abc9b2ce1beb914c1e138c915c
f6e15eddc54fe93980466438f85e0d47dff4a17115710d711fc65d89e34361281b2a0df50c8de5bd8d19a60ca5b75a0ee5ac5c2af794c72a53e66a8c
24e8e7a2ce831e17281e1b75e68cf2f2f6c396d6a38f9f9e3aedc399e2aa38cdd2dbbf64c62669f79aadcd36e9d1c9274ef7eb77a8676f3690a92c82
dd233e3bb5f9c8a1ed9aa433c8347f30b8c9530b425eee6016d33e03ab86cd36438b2a4518c1a87093d56e1eebd6ce77554bb2cc5a9235a6d1a6d7b5
d330e7547bcea9532edd402f513c715ca6139122f01b214f9ea7c183443a1119c7286121774acaced28c45faa171da7d2c437cd4dcd878e84d837b4b
5ed9b4daf60cfca876dcf10324eb93441cdd54e96c3042af90cd70427e0d5a24ce4c328c3039dae9d4d5b0b56b21cba184943ca550a95048df115984
838ed2934df4930b6f3418dcdf6a7c1ea388b88df653c9f7d24d4e6ee046a7c1c0d1297359d2b33b1397105b60b362e04c369898094698b51874fe94
7ed2e7b45d8aea3ee7bfeb8e3d5a9d828fc3a2059fc7644935454aa9523aef659a28954ad34d95d222be52aae13f376d94eaf966d34ec9a570c52099
5135a661aa9cc67b197a194396322cb4d4e06ab986af37d41ab7e266e37e7c85bf617cc7f8b1f11a5ec16bf215b9c794b9a0b1c7b214964481e958b3
14f853c72169e6958ed3cd0677fb0c76b1e36ac74129a9e333e2f747f9251e85cd92c68dec0827c69921ab8387781e2fe415fc0a378485460233b8ff
ded6252b1c4d76ee044a110c662338cd586d6b525a8caac104a6112eede0d5e54fb67dfe7dcd988fe445ec8890c84ac3dea6eb400f0351383a7e54ef
6daf34370f3eb63aa2af178fb89c674f741c960b0f954ee39cb09593af9f266ca9f07528c76a916ce6f1f13e932219d5f1f1f1be5cd5ec8ba7485ec5
6a647795a7265a8b01018a01693ed51c1f6b8407634d36a3c99d383c4da3ea7cdb2522473f2fc241e16f5a5070756bcea69d16465bd79901a99ada66
7b55afd96be94b2ed7dbdcdb324419a20e310fb198fde067c9529a9a66ee1991e1cef0f48c4cf3a5c5a7fbd3139253abd42a7395a5caaaddff992419
5483192d68451bdad18131d80363d12bc729a919e943d31f4b5f9ebe22bd2ebd21fd4a7a341d2f737f8c49f17ad66b48ba35ebcd605abad29f64876b
c7ed9d5c5333f5b9a1a7765ffbdde4b76795be5bb4725dc981d081e7fff09bd223f2d0436969f9f9a15109b69e5b6ab6bd919474223bbbe081317901
7bf2a695db0ffa345d0e20fffe2bdf4efe3d301463e3263bee03276b3155ab6692b10c2687cb36829c5b77e2605726124e6f07651e7ed5a3e7239ae7
ba2387687e9c92ad79b0932d62cbc4ea31956fbe79e1e5ea6abe5dfcaab6a3a166dcd61dbf950a6bd95dda797788a2e524d26b04c5cb21212f2a8036
66a8b639d7a9acc5dd64a133ce6d1e2791617ab42839286c51978294d13935851d29f7bce521938a706639c3891a55ba8e1b76a8e9892736fda2a9e9
9ed7169c7c8f7dc08e497b3a8a76ec38b14b5a76a3e160e9b42bb857e37e0879fd12b910ccec9ed008eea4ab914c5988512bb8cc24864e496266edbf
89aa535199569855a3c9a8384d2663ae6a94996c825f72a9ab26990c16edb2a9e66b174ebaec39b49753bf1818c261d1ad854543f8f6772a9c5951ac
8aba25633175417704f9b1addd1c9e576559ed217bd414f54ef90e7582fc9071925aaa2e644be585c6f9ea7a79a5ba45de216f363eabd6a97bd83ef9
5579b771a7daa07a5594395754730ff4708fd2c39c8e293ca0f434fbad83d9201cc0fb19fb2b83cc99d65138820f57469b43d60298c80aa4027c884f
341418279a262a05e63c6bb975315b6e7d813d673cc076191bad1f583fb7765a33b47b8a9444a147d102905c2c66b2fd9f8863e2d827ec7531ef1396
ced2e5c28ecf3b4eb22631521a2d458ab9ac56d3c19d62229d9a8560676b43f7184d92e204bb266600bbcd6907bbd569b18256d8acaa59b538cd6635
d76a561c60e6d5f8a6cddce2b0592daa624030d965bbd9d1ad00932e76f32d6237872f9dbad41d64486d5d39d27f2b7a93760b8f0a6a32bf62006e32
28688d54a3ac0e6b9235db3a4abd5f1d679dac4c561f57abad2bac1bad2e15880833b7986d667b14f3480ed9c1a354b7d96de961eb614f85648a167e
d9cfd34d694a404d36275b52ad3d6d3ded7ee700c866d952a69cc907aafdcdfd2d03ad836c83ec99cebb21c4425208437288870c2163c894ab0c57ef
b58eb28db2879cf9f0007b409a8079721ee96702e9e721e52175827982a5c05660cf7396b252a94c9d619b612f742e332db62db6d7c033ca6af36a4b
8db5c65663dfa26c326fb26cb56db5ef32efb21cb01db0373a3f707eeeec7496902eb98d85afd74399a6cf2c69e3b8e79ed8386b6c7e568218f2367b
943dfa76d97b4bb78eacca97c7b53f87b3ba7205be9fce946418198a48d153034b42b4d567725a121ceeb1012d8a04b538e2c8d13282537740c8a958
9dfb5c528f6a88de6c8877b598ed19395f058322e77290d28460e66da9c18fe9811e6a8cda801609f9feee5c4118f574e150e3b4d414f6f7dbf286ee
dc614b5a5ad9b4700e91a1fd2f87e8d5beb83c1872c68e802853a4dd2d9b4c18a91ac6f6f8915e91432763c865a2e0e8a8b6459f887ccdb6598116ce
346a2f0bfdc61ca454a39392ca3a4a2e1d7a6af98f2413c58c920f797498d25fbcdeac51fef7e6662d13e9a6f1e8ab1ad1ecf0b75d32eda27140c81e
3582f21bd562323964976d6ca4465f983c8d3acac2f62932a5664ea5c52a6984099d2aa68beea709167920348ba33fe658dae1ce97fd24cb22ec8676
c29e0e7d6072283a6344542f534f47acc7d4a3a702f10653b24f494c19dbf747419d0a6aef765d5c51b1f149fb929d942df639d1f335076c8e3426b7
c4c42564e4e45c0a0635ff73b405e92face52e6d0ee83fe0a6a8ba757e4b5ac8497a5a6aa8a9f7216feab855a4eafba4c39a38bb748f244d527b58cb
13bc11a99a30bb85dbcd9ad42d5b9dbb1418178a4c1b61353922a3dd2687a2fd63272156894f1a9b7a0b673a63ba19447bfdfb129c52b52565b3c798
d062efe10bb37435e71ff9a16bcfede2ff498e1bb6d35b75d1c5c74d1e0edcaa979bbad1be8548055be2f356fce6317bcedf20dea47ff4f8e829dba5
eef2dac7ed636d058af69dd274f31b09ad33ce1671003671ede3eb0fd80afee1ab4abcdcaa7f7f00693fc15af2e928a826f882a09e601b4131412dc1
5a433a9c91cfc1195e4be5d7b04c2e85b9721bcc953e820c69109c945f86635a89a3e198dc0bca31090670171c925b6008c19d7c156470350cc616d2
48c66db4e4c212f881cd6267a47ed28bd2a7e8c53d728e5c2dff85dfc94bf867fc862168f8c2e833169a1215a3f298f29a72457d4c3d11e613fc980f
bda00c2c64c90ed8a24941f64891546adf338c3059fb82252b343953ff46a4d51944522b5c97c0c44674d5f1967ef9963a876836aeab6e00372b857b
a01c2a88ee793003a613f6f9e0873498463ee4872064d29345b5a934c34ffccda0f14a8279500245301b7a53ef289843f3fb52ed6e98458f1f1ebcb9
57a5de2aa1b284d62ca47731cd54ff05acfd6f62cd274c0b0997f6d5620ecdd6e828a235ff3b8cc3a8f638ad9b080b68c6349a5ba4ef56a2af28d239
f2d32e73e85d4173a6d2be33689e9fd69713f6227deca7fb8cd777a9248ae87e0b33a957c35a4973cbf59d82843b0bb26f5bd5bd26fcad1c3a9fd2bf
97fee3cfafdb85f6ad5ccb393d10095114f963e96611de71048c24d98f8571703f4ca09d1f8249f0303cc224780b4e3264b271c19c1923723333bbca
acaeb21f4093b422d47943e07537fe3d803f04f15a3d7e6fc3bf09bc2af03f03f8571bfea51eaf04f0cfcfdccdff2cf0723d7e578f6dd7f14fd7f13f
047e3b18bfc9c5af057e15c43f5e1acfff588f9768e2a5f178f1cb0c7ef13a7e99815f08fc83c0cf83f8ef6efcac1e3f15f87b17fedf27f193e3f83b
811fd3f48f9fc40be7efe5179ec4f3f7e2b9dfc6f273027f1b8b1f09fc50e007027f23b0b51edf3febe3ef0b3cebc35f07f18cc077573bf9bb5e7c27
124f097c5be0af049e14f896c05f0a3c21f04d812d028f0b3ce6c4e6aa006f16d874f4386f1278f48d29fce8713cba427ee3df02fc8d29a14e7c2324
ff5b008f087cbd1e0f0b7c4d60a3c057051e2ac65fd8f0e081003f588c07f6bbf88100ee77e13e227adf75dc2bf015817b04ee76e12e813b5fb6f19d
417cd9862f1563034d69a8c71d02b7bf68a17b04be68c16d2fc4f06dc5f8c256077f2106b73a708b8acf0bdc5c6fe59b05d65b71132dda548fcf6db4
f1e7d270a30d9fbd8e1bea8ef30d02eb6aa7f0bae358b742aefd7980d74ec1da90fcf300ae17b86e6d5fbe4ee0dabef80cb1f9ccdd58b3c6cc6bdcb8
862eb7d4515d8c5524a9aa00ae76e2cf04ae5ae9e4ab04ae74e2d30257085c2e30d4f9d4934ff2a7043ef9243e518ccbf23d7c5900970a5c2270b10d
175970a18a0b04cebf8e95d771de759c7b1d2b04960b9c23705602ce14f8b833973f3e1e67082c7b12a753a3546089c06281d3044e155834180bafe3
a3169c22f061819305164c5279c1759ca4e2439131fca1204e143881304fc8c57c0f8e670e3e3e1a1f74e303a323f80302f3cc78bfc071f739f83881
f73970acc031343246e0e8510e3e3a0247c559f928078eb4e2bd0247d4e3f07a1c26f01ea90fbfe73ae61ec7bbc76048e0508177dde9e277b9f1ce1c
3bbfd3853943ac3c27d469c721561c2c7090c08103dc7ce0751cd0dfc107b8b17fb699f77760b619fbf930cb8ac13bcc3c28f00e3366669879a61533
ccd8b78fc2fb3ab08f82bd83d8ab6780f72ac69ee92ede3380e92e4c4b0df0b4bb3135802901334fb163c08cc902930426da3181f84c70a1bf18e3af
a38f58f015639c15bd2441afc0d8ebd8231763a8112330ba18a348525102236951640c7a04ba05460874d10497a03b5f1feecc45c793682f469b40ab
25925b055a68b62512cd0255072a024d34cd24d0e8464331ca3428930578907a51a0446da90f32078240d6c48a57af67bdfe7ff8c1ff6b02fee92fee
bf003a33d18e
>}
\immediate\pdfobj useobjnum \csname EF3O6\endcsname {[ 32 [ 318 ] 48 [ 636 636 636 636 636 ] 54 [ 636 ] 56 [ 636 ] 68 [ 770 ] 80 [ 603 ] 82 [ 695 ] 97 [ 613 ] 99 [ 550 ] 101 [ 615 352 ] 105 [ 278 ] 107 [ 579 278 ] 114 [ 411 521 392 634 ] 120 [ 592 592 ] ]}
\immediate\pdfobj useobjnum \csname EF3O7\endcsname stream attr{ /Filter [/ASCIIHexDecode /FlateDecode]}{
789c6360a0103013906761606560636067e000b23981980b873a6e1ce23c40cc8bd7063e20e60762010641305f088885194490548832883188334880
d9920c52003fb8015d
>}
\immediate\pdfobj useobjnum \csname EF3O8\endcsname stream attr{ /Filter [/ASCIIHexDecode /FlateDecode]}{
789c5d92cb6e83301045f77c8597ed2222101e898490aa74c3a20f9576156501f61021156319b2e0ef6bfb1a221509aeee9919cf084f78ae5e2bd9cf
2cfcd423af69665d2f85a669bc6b4eaca55b2f832866a2e7b377eecb874605a129ae9769a6a192dd1814050bbf4c709af5c29e5ec4d8d273c0180b3f
b420ddcb1b7bfa39d740f55da95f1a48ce6c1f942513d499e3de1af5de0cc44257bcab8489f7f3b233658f8cef45118b9d8f30121f054daae1a41b79
a3a0d89ba76445679e322029fec5a304656db7e5c7361f72815e2d3e001f12e0d54690187258735092c166be24f3f8087cf4f8089c240e27be41e20f
49d128f5e3a47e9c141dad5ca00e6798c7ca050a8cc1ac5ca0c029b01ff061113dc19e7cf4e4710bcc3ddeac8be6982b4f115d2dfae7c91a442e7e44
ee1b3cecd5ded47a25f6d2ec866d1bc1ef5a9b65706be8b6c0de7f2f69db54352a5b65df3f5a84c077
>}
\immediate\pdfobj useobjnum \csname EF3O9\endcsname {<< /I1 \csname EF3O10\endcsname\space 0 R /I2 \csname EF3O11\endcsname\space 0 R /I3 \csname EF3O12\endcsname\space 0 R /I4 \csname EF3O13\endcsname\space 0 R >>}
\immediate\pdfobj useobjnum \csname EF3O10\endcsname stream attr{/Type /XObject /Subtype /Image /Width 113 /Height 113 /ColorSpace [ /Indexed /DeviceRGB 1 <ffffff000000> ] /BitsPerComponent 1 /Filter [/ASCIIHexDecode /FlateDecode] /DecodeParms [null << /Predictor 10 /Colors 1 /Columns 113 /BitsPerComponent 1 >>]}{
789c63f8c0800a48e5f391c8677c40199f4186443ef301caf80c1624f2d91b28e3331490c8e7a790cff08044be1c857cc603a4f1196c28e4333790c6
67a8a090cf4e229fe46c87ce1f4d03a36960b428a6bc2806001577343f
>}
\immediate\pdfobj useobjnum \csname EF3O11\endcsname stream attr{/Type /XObject /Subtype /Image /Width 113 /Height 113 /ColorSpace [ /Indexed /DeviceRGB 1 <ffffff000000> ] /BitsPerComponent 1 /Filter [/ASCIIHexDecode /FlateDecode] /DecodeParms [null << /Predictor 10 /Colors 1 /Columns 113 /BitsPerComponent 1 >>]}{
789ca5d3310ec2300c85e1541db2351740e51a0c95b816633646ae1489816b04710177eb50d53c512f9db01d2fd6373ffd81c2f1acde8c9e6b9b79b5
99fad2e4ad5f4c9e3fafdc6266268b81d462ec93aac1d8238d0dc63e898bded4971426bfb14faad7ac36f6788cd3cd6dde4f6d3a3cbba52ff65afaa2
aab4f4c59bd3d497dfeb8aced21777abcfd217bd9f5965e90b6f7199f6bef0a2cad217f688e4b1f4857de2a0b1f4853d22578769ef0bfbc4705258fa
c21e91cec56ee90bfbdc87d3e5bff970d9ea2f65a5bd67
>}
\immediate\pdfobj useobjnum \csname EF3O12\endcsname stream attr{/Type /XObject /Subtype /Image /Width 113 /Height 113 /ColorSpace [ /Indexed /DeviceRGB 1 <ffffff000000> ] /BitsPerComponent 1 /Filter [/ASCIIHexDecode /FlateDecode] /DecodeParms [null << /Predictor 10 /Colors 1 /Columns 113 /BitsPerComponent 1 >>]}{
789ca5d1bb0dc0300c03518f9ed1328a4750a94290f24102d89d8863f75adeb0b14f758aae091da2eb8476d175409be8a29ea2137abdb8e5805e2eee
d9a14bb54117f5149dd07fe3b603fa6bdcb743976a834ed16f53e210fd344576d17753d1179eac5834
>}
\immediate\pdfobj useobjnum \csname EF3O13\endcsname stream attr{/Type /XObject /Subtype /Image /Width 113 /Height 113 /ColorSpace [ /Indexed /DeviceRGB 1 <ffffff000000> ] /BitsPerComponent 1 /Filter [/ASCIIHexDecode /FlateDecode] /DecodeParms [null << /Predictor 10 /Colors 1 /Columns 113 /BitsPerComponent 1 >>]}{
789ca5d4bd1184201086612b381b60acc3c0a20c0dcd0caf25662eb00d1c1a80069803045d22f687ec09be37d9190637b48fea40b437320745b3b75a
e4a06692bd3d7789636fa538eebf12a79e2138ed278153efd278e7fdc277eefd76b4effdc6f6dd3bd02e7bc775e98d5cd79e41baee27a66bafdeb8e7
67bff0fcf4ca8d7b7ef71bcb6fef4019ec1dc7a0f7c118ecf3d74a35ec2984e1de6aba9bdedc77b38f5f2bd17f4d4ceaf9
>}
\immediate\pdfobj useobjnum \csname EF3O14\endcsname {<< /A1 << /Type /ExtGState /CA 0 /ca 1 >> /A2 << /Type /ExtGState /CA 1 /ca 1 >> >>}
\immediate\pdfobj useobjnum \csname EF3O15\endcsname {<<  >>}
\immediate\pdfobj useobjnum \csname EF3O16\endcsname {<<  >>}
\immediate\pdfobj useobjnum \csname EF3O17\endcsname stream attr{/Type /XObject /Subtype /Form /FormType 1 /BBox [0 0 464.8787779 148.7861494] /Resources << /Font \csname EF3O1\endcsname\space 0 R /XObject \csname EF3O9\endcsname\space 0 R /ExtGState \csname EF3O14\endcsname\space 0 R /Pattern \csname EF3O15\endcsname\space 0 R /Shading \csname EF3O16\endcsname\space 0 R /ProcSet [ /PDF /Text /ImageB /ImageC /ImageI ] >> /Filter [/ASCIIHexDecode /FlateDecode]}{
78dab5974d73dc360c8675e6afe0b13e0426007e1ee3baf534a73ad9991c3a3d64dcb59bd46ec71faddb7f5f90f2465056ab5defda23cbb3c2502008
be7804de5ab44e2eb46f5cfdbbb831c7a7cb7f3e5f2cdf9f9dd8ef3fe8a78b7b83f68bdc57f2c217b91fe5b533b9af4c7571637cf490534e29cbe3b5
7e449f21e588be88dd8d1f7f37e6d21cbf1537f7f2da99315c804be1142c47e040415c233940c2fac6ca78ad8d88095cf294bcd8070723739be8d6ae
bbcf0e12614adebb92bf79ba5bda8ff64f7bfc966a7cb27eb91f5b9ce3ccdcda8cc0b125b3ffb53ecfc58d3dfe09ede95ff6dc9cdbdb95532709ad8e
1de4de75b518ef20461703ab2c281b41ee93604e64471ecdada91bf9a6ee2427f0d99518880259424821a73abb395998e31fd1a2b38bcbb6678bdfcc
2ff6bbce1dd95fede29df96161ce4d0bc3206699a1e4447a17947126004482e09c2f31669f768a803bd431286729424894229138b409989046aef01b
577fac2f460b2240462a25c972587eaf99a71644119c6fcb721845716d59cfcba7524382885cab4247305827335a20661765e1cc12682c8e52dddce7
64b41579f588bc0a25790821c9b42db733397dd0aec8beeb01d0f43a2e83a9dadd508ee6c3ee952e635d9bf54a95e095ce500ce072e6dc2f0525b19c
e3aa88675676da2dbbebeea1fbb4417e2140f018b0971f123f2d609bdff747d67b702ce3ea8562fa5b265a76b6a3ce755e4ff7043f0d56f7045683c1
41e14258548e88b1824052a772a48ca33c2b0feb38ac18fa3a201fc2445343dec6c4c9a91a1669472c6220a0c616ad19659ce3922fe0c865974bda9f
8cc4aeaa366118edc8609c89804834e1ab2a3c1f4c46cc198449d18512f646e3481d8a8de8e52bb2669fccaaa02ba4927d5fe10730526b5d415285b2
85925201b54e6353d79ea09ca3db64316eaa2f61d6eec5bb1d7028b076898308273c8b703faf93e85377df5dc8ffeb4dca4a0e4296297a65bd08f5e2
78c7e7984791c0a7e2a2ee7dd873adf1cc49a54d1947b9571ea699f77500bd36f2a6666ac4e31d894791215566680d0db639da44d9c6da2ff635b027
efd86315b16861b41b8371ae15159586ec307b3cbc13140281acc5fb92f66f0547ca50bca390a1acd927b32a5f11c4d2207310eeb4cc15ee5424f3b8
234f90a263a4d7a1dd64196eaa2c01d8ee65bb9d7694657b8b8ffcdc7e4e6857c31246b62b88e9eea80a8f5dccd82e312dbbcbee73f7ef26a1a50c51
5ad0d20bed859a3eda1d805c598f5952a90f7fc1d79a9755a84c2ae3f8083c789806e030e0d50938395543a0df1181d25900358c686529e3ec699800
5b6fe80f68fa7ce0aa6d41ee684706e34c045efa1739d5486d96c39b3e8f0858884adf06ec791ed6ead007e2fadd5fb34f66553e2d51ba83069fc30e
c64aebfa643c84324f41e988a4b61d86f84a189c2cc64df5550fb53b17ef760c72c9907868c40e6bfaee84790fdd7f1b94c585a52f45b76a305f867a
59cf766eccf9ff5faa1a9f
>}
\expandafter\gdef\csname EFWidth3\endcsname{464.8787779}
\expandafter\gdef\csname EFHeight3\endcsname{148.7861494}
\pdfobj reserveobjnum
\expandafter\xdef\csname EF4O1\endcsname{\the\pdflastobj}
\pdfobj reserveobjnum
\expandafter\xdef\csname EF4O2\endcsname{\the\pdflastobj}
\pdfobj reserveobjnum
\expandafter\xdef\csname EF4O3\endcsname{\the\pdflastobj}
\pdfobj reserveobjnum
\expandafter\xdef\csname EF4O4\endcsname{\the\pdflastobj}
\pdfobj reserveobjnum
\expandafter\xdef\csname EF4O5\endcsname{\the\pdflastobj}
\pdfobj reserveobjnum
\expandafter\xdef\csname EF4O6\endcsname{\the\pdflastobj}
\pdfobj reserveobjnum
\expandafter\xdef\csname EF4O7\endcsname{\the\pdflastobj}
\pdfobj reserveobjnum
\expandafter\xdef\csname EF4O8\endcsname{\the\pdflastobj}
\pdfobj reserveobjnum
\expandafter\xdef\csname EF4O9\endcsname{\the\pdflastobj}
\pdfobj reserveobjnum
\expandafter\xdef\csname EF4O10\endcsname{\the\pdflastobj}
\pdfobj reserveobjnum
\expandafter\xdef\csname EF4O11\endcsname{\the\pdflastobj}
\pdfobj reserveobjnum
\expandafter\xdef\csname EF4O12\endcsname{\the\pdflastobj}
\pdfobj reserveobjnum
\expandafter\xdef\csname EF4O13\endcsname{\the\pdflastobj}
\pdfobj reserveobjnum
\expandafter\xdef\csname EF4O14\endcsname{\the\pdflastobj}
\immediate\pdfobj useobjnum \csname EF4O1\endcsname {<< /F1 \csname EF4O2\endcsname\space 0 R >>}
\immediate\pdfobj useobjnum \csname EF4O2\endcsname {<< /Type /Font /Subtype /Type0 /BaseFont /BMQQDV+DejaVuSans /Encoding /Identity-H /DescendantFonts [ \csname EF4O3\endcsname\space 0 R ] /ToUnicode \csname EF4O8\endcsname\space 0 R >>}
\immediate\pdfobj useobjnum \csname EF4O3\endcsname {<< /Type /Font /Subtype /CIDFontType2 /BaseFont /BMQQDV+DejaVuSans /CIDSystemInfo << /Registry <41646f6265> /Ordering <4964656e74697479> /Supplement 0 >> /FontDescriptor \csname EF4O4\endcsname\space 0 R /W \csname EF4O6\endcsname\space 0 R /CIDToGIDMap \csname EF4O7\endcsname\space 0 R >>}
\immediate\pdfobj useobjnum \csname EF4O4\endcsname {<< /Type /FontDescriptor /FontName /BMQQDV+DejaVuSans /Flags 32 /FontBBox [ -1021 -463 1794 1233 ] /Ascent 929 /Descent -236 /CapHeight 0 /XHeight 0 /ItalicAngle 0 /StemV 0 /FontFile2 \csname EF4O5\endcsname\space 0 R /MaxWidth 974 >>}
\immediate\pdfobj useobjnum \csname EF4O5\endcsname stream attr{/Length1 9548 /Filter [/ASCIIHexDecode /FlateDecode]}{
789cd579795c545796f0b9efbc57fb4e158b1455c55220ae0444458d94c45d635089ada64d530a887101c1258a06dc40a3366e60628c568c1aa3c626
c436e0168db44b08dde95633c94c32692389498718d363e234c265ce7b8062f77c33f3fdf1fdbedfbc57e7bdbbdfb39f735f0103001b3d44f08c1a3e
622404810780f5a2d6e051e94f4d765f889b40f5e1049f8c9afc745adcfed40b00e8a4feca2787658cd6a626ada5fa71aafffcd4e4be8973fe965742
8b1da5fa9459f3fdf9c20746358024d2fc7eb3962cf2c09c881400d540aaf39cfcd9f317f65bf21c8086ea7064b6bf301fd47483a691ea86d9f396e5
1c290afd27aadf01085a9e9bedcfd2cc78bf14c0398ffafbe7528371af7a33d5f7533d2677fea2e7b70e315750fd32d54be6e5cdf28fcb18497d11a3
a99e36dfff7cbeb859b590ea9ba8ee59e09f9f1df7fde3a7a94ef8b26bf979858b84fd3806c04d34c0ccfc82ecfcc1ea1fa9e8267ca45c90796580f6
4ba01a825b699341077d600808c3478ecf00d33cffa205104a3ca5abad0de04149193d37bb6001683ae631ea1394b7060496208f64c3590ed56cd49a
a2ccf3137c26df1d7b437ba9ebf3bfbedade529e73dbae2aeffd0fdaf73f7ccbbbc8abf1b738d564683b2fb7f0fdfc2aa7f9fcca7fbfcfff9babed2a
d17eb593dace12b4f3cc0466b080151cd00dc2c1496d666a63c453ec18253e584802156917a37e5bc708b1a34d035a92a09e9e06303e58b35daedb60
07d815b92ef717f867c25a7fc1fc05b07666817f0eac9de55f5048cfdcec027a2e2b98076b6767e751797641f65c589beb5f406372b36752cb5cff02
3fac9de7cff3c84fd28fb5f3fd8b7261ed82b9724bde6cff7c585bb078018d5c94b360363d73e5f5bbe8d0c3ab1daf3fc337900e430ca0fe5c6e1456
13a99944e4f50eb2331fe562679de576949ffeef39cf4cb4ee92ff8188a6b78f95df0fca5dd678a47d7a977ac1c3b2403ee07e3995873c98da9da8d4
3da0989ea2896d2189819424bd4c5557fb1bff097204b21641af42d48882203e98d171a5e78cc8021f79b6c52a3bb7b35deaf9ece62363b00364bf26
efb8836a4ca98bf03cbdbb913e20f578c82287c170980419e0876c980d73200f162bf2f1404297be594adf3c28686b6bbbd976adada1adbeed6cdb99
b6b7dbaada7ed376aced68dbe14771fcbbabddcf04bad43c1d20e39900edbe41a6745807c8bc1dde016682491d6021c8e8002b81bf0364ff328b209b
20886076070413cce9801002d9d7e6118412c8f25a2cf38525430dd4d37d1e0ec36e76906a39d4be905a024235aca351357081d5b30d426f6a3b0877
e02a8d2c837a3c2c021b0b49d40af0a924c05d9601c7698d146667296a9508e204f1b83849ac116f890d30402c141bc44cb19025e13e698a74902005
7f47f2be42deb7867d01857012bfc5243c2d0e174df00536e061f88a769179530fe5b01f8a08173bcb8362a14898442d97a406d845771ef537b03dec
2a617792ad81ebf0128ac268d8c3ae135df5f033acc10ca198589e24e410fe9768ad069abf0b0ac9715c673ae0424f6a23ec69af99ca33027b4bd795
fb0e14d3ce19b05f55a3b2aba36917996307d905d6a4da4e92bd8abfc485f82f6c9d182d1e124743793b073013ca69ed5df21c550e5b46b4cb7791bc
bab054cc6487e15b31533d93d6fe9d4c11ed795c984414e5c06982a52a0bd13498adc30d84a9dc1b010deab1625f9a4f2ba8572ab2ccc364788e4a45
700caaa137564239ada4d0ab1a20fd4c33778b3788e672b659f8191a7038c4438e789b780d768afc00efaa55928802835e1e4b95e01d9355e59b38d5
73795a64ef5e7f57f558d49e2a48af322ef3d4b4b5a54f15c3a5695592b30abd9a2ad11b7de3ffd479a377af71e9533d55ad238677ac3a227338b54d
9e4a45b946cdd43e62b8d2276f5a2579e93726b3ca332bd7f3a2e5c5e8412f5ab207f596ed4790a32b592952a9aceda6584ed2d1937647fb8254011b
040c5b6d9b42b54eb30b9d8ef0504b4bd3dda6c7c0d278b7c9723b814509568b2d29d166b508718960b5407494fc1436ee7ef555fabdfaea7da6e5f7
eedfe7f798564ae70dfc4382069644773f9614e085bc9497f142b6992d63cbd966d9efdc20d39d4e9e5a073e9f230d03a2109056a921a0d5b8554e4a
2b98de726d5c9539636a2d0df60d9cd654d74208f56d4abcdb74ad298178302d8a1d37a35914660c88b44ac9de246ba42392b3b1fc6596fd011bdbb2
ffb05838ba6674f375d9cb90bcc4b144b113f6f8e2c2ba8563a8d34a29995592c434cb6bd61dc6807dab48760b169dc074ce100baa222c2de3aa1c19
e3aa82339e195765cf788630c1b67303a7d55d6b3a77ce6a4be9c0e6ae828dda227daf96be67554e8b35248570f3253e2d4e91a6a8978bcba525e165
616ab2ea30b11b89d7b90896a816772b0c5fe45c0da561abbbad0e5fed3c0487c2ad336086978848ee0f0386b2e47eb1d1512a75f2509694283aec2a
b50ac8959c6f194f6c4cf23ff946e9afae3ebffcdad46f987dc43361fceee1c38797b2ad83e6ef1cb3b432ed890f1f4bfce6fd5f1ec88fe0df11f5bb
49de85447d77c8f7f5014790ae54eb2ef504051cc68076bbca19f06c8fdeaadae4783d3ed81904680f73c67a2c4eb4bbb5aa789909c1199df46b15fa
8901779b884ae280a5a9f16e6393e5ebdb16e526ae24309f36cbe577fb3d599122cc602ee6b08b9151b171c92e22a43f51d59325b7171e210f53b7be
ce3fe2df3c7be9b98ccbf3cf5eaa3d70ec44c59ed75f9a7cb6a0f0cab4af99e1d7e875d76df9fcaf5eef85c7122bcbd7561c5c9a5f5814137bdce3f9
63f58a23b286679194f7934e0994c1acf24530231a01d19806a857072486abb4cca003a74a231a4c96cfc655e98930a342984126ecda90baa644ab2c
d7c66b439a12891645b0e21512ee1559a43df4d00346c3340a0c4be1455007b39e10cb7a627f36813d6578ca3885e5b0c56c39ae634612a596456292
35c9116d8db64626a38a0b8c27f3ebd7afb43e2b795b6e62434bd2211e60991748427b4842598479043ceb8b16bba9ada596886e01b53d60d9601402
b0cab849bddf15e2643a3a77e82c2a97a58575958b4546bfc35a2cb2b590882c75b76503962d98c4c3ebdaa51344fa6595790e0e3b3c2216591a9f63
586ba0d7d45ecd2c865fe33f3c7b2177fab9b96f7df0c15b135fcb90ae1fe6dbcc667efb2f3ff29f3c9efac7124eecde7d222696b85d4ed8572afe24
06a6fa628254602c354020581570061fb0040c1ba2b63a37790d515a67982bc88991ee702f391852a246c5c534b6343e541f9f9d622a6b101ab041ac
97ea554477b54b98c166b02895c31edc8e2b73f461d15102761212ed91dd516462b0b07ffddebdeb099876fc2be32f5f350fae9e7b8349fcce97bc95
df66e92c7cfc2b38f8e4bed74e9d7a6ddf4961594d4c2cff2bffe11733f80fdf7dcdffa238a899ec804bf65087489b7249262a98e50b95ac020a6815
c95f48240f9490890c546a4bcb877556d912fa76f10304a42eb280a69ea1d4da478bd1a18e64641d30709acf3655602aec26a548a3a5d95805552a35
690b098645b3c84378aef5cbab8cb72649d7a734af927acad9ce46e2ef4685bfd1d0179ef0794389bb71aa80ab77c0b6d5b529eef58450434c0fa723
c669d692f726176e8e0c4fb0b4d435ddad6b5218db69ab4a2d858cb40b33bd7dc8d7c4242506cb4e4631d7e8a898e47efd833a079066081bb71c38b0
65cbc103fcc0eaadd0f6af5ff0adabb6bdceefddbbc7efed1fbd75cdeaeddb57afd92afc6e5759d9ae574acb764df15497bcf3d147ef94547ba22e96
7ffacd379f965f64fe45ab572f22208d594514951145a18ac644abdd61ac14c202ba0362003604bb0396adc19bbc6aa73332c80551514ea3a230847e
674cfa9affd4a92fc17561ef773b177ece792ee27d579d5b7dd876daf6ad0d49630628ba6d0b3291ae40723f486ad792a858d64916f1e0c6f8dde348
4f0655cffb33bfcf2c5f326456fe36ff6afc6e36b44397dca425ccc86c537ec9ccdf7dcd829570b6973fe31276766a92ec7d2cf20959f13e3a78df37
0cac74b09544418e31ccaa431d58e9bc853a35051e95dca8b5a24e2377907f5257c8de49d26a28fb93730aada4a3f85717222bd690c66b4d9dbe4809
320f5e9aef3bd54d2e4f8baaf618189be14b3333b360569b3566980a4b201f368156cd34820ab562300b13a6b0a942ba6136cb159e674b841558202e
553faf2963eb8512c34bc2cb582986b4bb2f0adfd11889d1c2697e5bf0f2a2af84943fad6ffdd5faeb92a9350c8f35f764c57c15d14be762a98a28d7
5076ddc7e7800aed2a5661d108161d4861c644706a459b9251901e2ab62207cdeacc204296593ba4e28d54def18c6dbfcb92999bdfe0f53c8ded65d5
ac92e7f274ee97fade5fca42591fd68b851ce43b79097f815776ec4ede89327827bced4b261ea24e65451125ab28621aa5d30e141d155a7b8571955e
945468d58233d824e9c2c2446baa5de7348814eb9b12c95c28a859db311c22a719b614f9ee6ad64a6caff6b914362f0f6212484c22c6aa450738985d
08c610d10b5ee61562314e15ab8ed5c46a3daefeacbf30928d1472a5c5e2626969d07ad57af54baa97d4ee194a800c098ac63eac2793e38447b64262
4abb86e2e66145431b3e7d6fecc6e73ffb805d66d0b2a67503df5651b14d381dbce5059ecb8a2b67b66e90ae7ffcc9e693c253adb7cbd6ac5947beab
ad997cd7b7c413358cf599544205ac12998f62874fd258ae91bf55fc556202f9289deca3348a8fd280a6d3470581d60d166611dc6a8bd6a7cdd7eed5
6a67a01c3b084795f843ebedfad6db14119aafcb1e4a8022b2e7de94a1ebc00ba729c372eb43b426783344556bb27a4add279db5d135d64d210608c1
50a356a377a3c63e2296b8fee1b5a6c4c476075ad778b7854cfba2e2afac32e37d0b1222125c09ee044f426442546a9c2fc2e7f2b97d1e5fa42f2a3d
22dd95ee4ef7a447a647a5c7e5c7ad8b287395b9cb3c6591eba2b6c405e2eec4b93aa7764eea9c90e9ca74677a3223f35df9ee7c4f7e6489abc45de2
29890ced1a651e6703acd1c9b2eb88255f9814d9355f0916ce7e717455decbb53535a9a7d71fad6fbdcf843776669ec8c83e3bfddfee08493945330b
3f3d1e3fbe75d5e11cfff97d67ced98a37f6e973382eae458e2b0bdb6ee22de25518a4fac2a194ad174da5c6f5ba5aab581b424ceaa6b61961b47d44
374b4b6362a7bfe3776f5b7eba9de0d39bc32de125e15bc203e11221ab84c10e840738143fd71e07f1d68457d3dfb978f19df457273c7960462bff98
f566aaa7f789c9477bf6bcd9d070b367cfc331316c2833311b1b144d1224acc4e92a3b7931270cf48575ab0593bd56d26c32d5b09da4d4a01146596d
fa11118a1d2726ca46dc284797badb0927325d25ae800b156bee601a1d188018c5bac410dc575333e8ed15f56dd056bfe2edd64b6f6cdb76e8d0b66d
6fe009e1d9bf351dcaf2b3e14c43f7703f77d4dfba554fd081573171cb0ee194bdc69095694b35eb25c79b4caa35b053a1b5b61ac32667b843d03834
304eb09947381514eb940c5d665e7b7271b73d5ac4a746e44704223e8ab81321a5422a4b15521da9e1522f755f4d5f6d2f5d1ee4b13c21cf9117ae9d
b1506670a4921e29bc55fc13e51b6a85e96ab1b8a5dad0f0ee739766cefa682ebfcb2fb1f8962f99ba4638b07e57ad497876fad94bfdfa1debd18b0d
643a16c49ee09fd7ed3c7e6c8f6c2d0bf914713ad1a4a78c7cac2f3acc10a1b5950605d79ab13636ba26eeb4b6d67ca65b446c18680ca354369b6744
bc12cbdbd95ed7d8ce787e5da62885b8dfa3a447a087ccfdce98461887588487d9dde3ac432474860b0e494ec27d072a761c38b0a3e2400de7cdfea3
1327ee99f4dbe329d52b7edfd2f2fb15d52935c2e3973ffbecf2a5cf3efb8e7fc9bf8d70bdd3abc799f79e9935930d6272de3368e6acc3321de7499d
9791ce20f9999ee467ce8a6fc36941621a11466a2c2d9443cbb9744b13e9adec45d2b599e449486f83484be4c4f87c0d5d62e6fd80cafead6c190fd7
8b7a17760a4c0323454bfbb12fc167b4483e295dca94f2a53b92aa7d115a4065ff5b933c9742b03a82781a05d37db12a9b36d40caa08b5c35016e1c1
9af0d361163558cd1a8d2addaa31a73b4335dd46462b4ebfa5a5a9fd2c376448e35de500600b213d094a88498fc98fd91213a0fbbd982f62da62b4c4
61c5fa1d847d3babbb16921c4aa7183fe2dceadf9cad2d585c7eb0b660e9e683b5b5a955cb961fc10d2b96fcf465eb2f853dafed3ebbbfb54cd8b3ef
95f75e6f2d13338fcd9eb9a28302318b280882febe30d4029a98aacc64ad319cd63141031364bf39d22e232d1fcd292d202d50903d9ee9f8834390ad
ef3f4127ab66c58a8aa3b5b569ef2c3e7f51d82f23b0778f8c006d9c9df54387852d56b43184b43148556b835a438dfc35c0669e8836c788bffb1ae0
8b4e0d2b822255b1ba5853ac2dd615eb8b0cc5c66253b1b9d8526c2db205c2ee84591fcdd71ff96850b8e3e8918aed478f6ebfc36cfcf69d1ff90fcc
8a5fdcba72e5d637972f7dbb9b5fe64dfc7b32a714b21a3b1bd8ce191c4b185a818e072a3d49528f65a61aed69b54e45416ba44d360f45d7c8eeaf7d
281bfaf1f4a0bd41324fda3de4438684e058f7985ebbdfa8ad1d74725d501f271eb759ebcfb656133b72664912ed9647fef912ed1607b77c438c06c1
a49fec7669b4825a37d9ed76a5e9f42e3745fb52b641b4973a3684ca7edb4b7ebbbb4ba77787ab6152b8c6a4d6d8a3467497b1bad6d4285b694a4aa7
23ff4976e4b2d494acc22467146a53475e0171725e31dfa973ea9d863ee48e7ae97b19066b07eb06eb071bf41ef0b018a1bbaebbbe47505f7b5f478f
e0eeaeeeee784f7c644c5ca9ae545f6a2835ca5f349920a8742a3d1ad0882634a305c3b01b86a3538cd0c6f58d4f8dff557c717c49fc96f840fc9df8
504a41163e8c236ee5bcae8aee7a30eccbe403427fe21d6e9c7068fa860d3377a4d61db8f7c9f40bf3722efa576fca3ee23bf2d29f7f9f735c4c3dd6
bd7b46866f4ca4a9c7cb1b769f888e3e9b9c3c6de2b874af39a662f59ea3caa96a00398ebf4a7b48cb29ca98248d19dfa44cf9b4a64ca7271e539cb1
d84cb2960fa9a35f6247b6d87e304f49a8fe8d43c919e544c41e3c9839e4e84c1127c9ca96b222be6e5ce19933d7f79595497bf8fbe5ad810d1376ed
fd939059ce86cadeea18e9f954c5beec30d8e77c68619b74ecb4bdc640f665d74f204b1be990153ea55da31a131f98599ee39c6c664114e4da15fb41
b48b65c764337baba6e689b7179fbfccfec04e0a075bfd7bf79edd2f14dd0f1ccd9975070f757ec517a6bd5ca01be2f99579c84fe0d6281fb7fff882
a9b1f37defe396f1a6695af91f1ecd836fe1344f3d9f470098f8bd8f9b279aa6fdc397738fd8a07c7f064176cd1b29270e8132821b049504bb09b208
f61094131c22d848b04a5a0616553c5c91cae18a78b3ad59bc0545620e2c94ecb0506c6a7f0b29705e06753d9c946c7052fc0a16e2587af7843c8c86
01d476ac0b2e69b00cfe9dcd635784eec276e1131c8647c4eee26cf147295e7a8f74738aea9c3a463d5efd9cba5ea3d64cd5946b876a8bb4673af8e3
c10ce809b96050ce542fcb5c101d4230bde5efd96a982eff33216a697082f25f805c66104cb5f6b2001a36b2a38c5ddac52e650942d9848eb20aec2c
079e803c3a292d83029803b369f745e0a1e83c0be2e99d080974275169268df0107d73a8bf90a000b2c10ff3a117b58e810534be0f9586c13cba3d30
e9c15a854a2d9bded93467093db368a4ee7fb06bff07bbcaffa92ca1bde4afd60b68b48c879fe6fcdfed389c4acfd1bc29b09846cca2b17e65b56c65
865fa1c843ab2ca0673e8d9949ebcea1711e9a9f47bbfb95bebf5f67b2b24a2161441914cca55679d742e5bfa1050a2d7d887fc98fccea9c23b42b4d
db0bcaff60ff787914bd90ffab942dd701c114a742297b6effb73386fc74327995e1301246c168180be3e12948878944fdd3b4cf2fe80c0c354289af
ed3ec7663bfecd8bff9e88f72af16713fec4f12ec77ff3e25f4df86325def1e20f2f0e937ee078bb12bfafc4a666fcae19ffc2f1db41f84d1adee2f8
75227ed53859faaa121b6960e364bcf9655fe966337ed9176f70fc33c72f12f15fedf879257ec6f15f6cf8cf2bf1d353f809c78f69f8c72bf1fab551
d2f595786d145efd53b87495e39fc2f18f1c3fe2f8078ebfe7d850891fd6bba40f39d6bbf08344bcc2f1e23aab74d189bf0bc63a8e1738becff13cc7
731cdfe37896e3198ea7399ee278d28ab5a55ea99663cdbba7a41a8eef9e9821bd7b0adf2d114ffcd62b9d98e16bc3133ef1b75e3ccef19d4aace6f8
36c72a8ebfe1782c0bdf32e1d1235ee968161e396c938e78f1b00ddf24a4df6cc6431cdfe07890e3011beee7f8fa3e93f47a22ee33e16b5918a02181
4adccb71cfab06f2baf8aa0177bf1226edcec2577659a457c27097055fd6e14b1c77561aa59d1c2b8d5841932a2a71c77693b4a33b6e37e1b666dcba
e594b495e396f219d29653b8a5442cffb5572a9f81e53ef1d75edccc71d3c63ed2268e1bfbe08b44e68bc370c37abdb4c18eeb2915a086b22c2c254e
957a719d15d7725cb3da2aade1b8da8aab3896702ce6e86b7b61e54ae9058e2b57e28a2c2cca7048455e5cce7119c7e74db8d4804b74b898e3a2662c
6cc682665cd88cf91cf3382ee0382f12e7727cce9a263d3719e770cc5d89b3a992c3319b6316c7591c6772f40fc2cc667cd68033383ec3713ac76953
75d2b4669caac35f048749bf48c4291c9fa69d9f4ec30c074e66166972284eb2e3c4b141d2448ee97a7c8ae384272dd2048e4f5a703cc771d4338ee3
d83116696c108e89304a632c38da88a3388eacc41195389ce313426fe989664c3b85c3c6a18f632ac7a18fdba4a1767c7c88597adc8643061ba521be
36330e36e2208e291c070eb04b039b71407f8b34c08efd93f5527f0b26ebb19f0b938c98f8985e4ae4f8981e13faeaa50423f6d5639fde5aa98f057b
6bb15722f6ece1957a66618f789bd4c38bf136ec1ee795ba0fc3382fc67af552ac19bd7a8ce118cd31ca8c914467a40d3d59e86e461791e0cac20823
3a89834e8ee1cdd82d0dc3a812c631340b438853211c83695270183a38da390671b4d1001b472bd16a4d43cb4a3467a189a3d1102c19391a68b42118
f51c7516d472d4d0300d47b51d55592852a7481ae0406a458e02d585dec82c081c590dcb5ab799f5fcdf70c1ff6f04fecb2be23f006e0195af
>}
\immediate\pdfobj useobjnum \csname EF4O6\endcsname {[ 32 [ 318 ] 48 [ 636 636 636 636 636 636 636 636 636 636 ] 65 [ 684 ] 68 [ 770 ] 82 [ 695 ] 84 [ 611 ] 97 [ 613 ] 99 [ 550 635 615 ] 103 [ 635 ] 105 [ 278 ] 108 [ 278 974 634 612 ] 114 [ 411 521 392 634 ] ]}
\immediate\pdfobj useobjnum \csname EF4O7\endcsname stream attr{ /Filter [/ASCIIHexDecode /FlateDecode]}{
789ca58cc90a80300c441fb855ad4b5debfeff7f69c851d01e0cbc999019023f270ae431092919869c82124bf568d4aacdc78f5670af6927f40c8ce2
9330ebd5b3b0b2e9be737072dd59d301ee
>}
\immediate\pdfobj useobjnum \csname EF4O8\endcsname stream attr{ /Filter [/ASCIIHexDecode /FlateDecode]}{
789c5d923f6f833010c5773e85c7748848089846424851ba30f48f4a3ba10cc43e22a4622c4306be7d6d3f934a45829feeddbbb38d2f3e572f95ea67
167f9851d434b3ae57d2d034de8d2076a55bafa27dc2642fe610f9af185a1dc5b6b85ea699864a756354142cfeb4c969360bdb9ce478a5a7883116bf
1b49a65737b6f93ed790eabbd63f34909ad92e2a4b26a9b3ed5e5bfdd60ec4625fbcada4cdf7f3b2b5657f8eaf45134b7cbcc796c42869d2ad20d3aa
1b45c5ce3e252b3afb941129f92f6f8fe3cbaeddc39f383fd08017271f201f8e90d7700f24c00148810ce0400e3caf6d7cd714f50e0d08191d1c1ad0
cb1996716840c87067c19d0537476f1e7af3d09b63933c0b7208d3558509bb756840c847c8e127f0700a2e2077410ea104684d7a6f8e33e461f935c4
2ef2744d5edc65adb7e2eecd0dd96328c4dd183b0f7e12fd20b811e8153d86558fda55b9f717009dc0f0
>}
\immediate\pdfobj useobjnum \csname EF4O9\endcsname {<< /I1 \csname EF4O10\endcsname\space 0 R >>}
\immediate\pdfobj useobjnum \csname EF4O10\endcsname stream attr{/Type /XObject /Subtype /Image /Width 507 /Height 489 /ColorSpace [ /Indexed /DeviceRGB 1 <ffffff000000> ] /BitsPerComponent 1 /Filter [/ASCIIHexDecode /FlateDecode] /DecodeParms [null << /Predictor 10 /Colors 1 /Columns 507 /BitsPerComponent 1 >>]}{
789cedda416ec320108561aa2cbcf41172ff5be5082cbdb0421dcf8c2105c7d8c16916ff93a84a25bee51343ed9ce427dca7f5c8dd5dc238fde532fd
3eba2e0cbac2f4b30f7e59d7709b17000000407b20e871c945f783ee3bdd7bddf7ba07000000680f48a9c9cf47a5c55a934a13c4cf47fd7cfca9d000
0000001a0272c9ba2ffb88383d3c2820e9010000004e038cb1b14f863e3b3838bb66f9a5d2b24b16000000c0b700f69215e7c651ef68127bcdf2ba97
4e04000000380348a15869e9dc68c7d34acb8e030000003400e27ffb24cf5363ac348b2f3d84010000003401522abeacc7ef0facd27a5d9222020000
00f026907f11e566e2efd817000000004e07f2c46fa29e2b2d28f3f230000000c01b405e696b639f551a000000c0594029e5b1cf2acd6d2300000000
878052a53db23ef65de71d000000c01702e5e4df1ff4212853794f03000000d80d942b2d4e8ee54a03000000380358cbabb9b18a00000000d80dac55
5aa7078c7173a1f57a74f3250b000000e000b01ebb5e0d7ad8295375c5020000003800ac575a2c3300000080cf005bb1523b70c902000000d8096c55
9a3d8ba7b9a6b506000000d00ca8893d8ecbda7d1c000000e053404d270ec9dc9805000000a019509f5869b73d83270000004035505f690000000067
038410f2fff905af775ad6
>}
\immediate\pdfobj useobjnum \csname EF4O11\endcsname {<< /A1 << /Type /ExtGState /CA 0 /ca 1 >> /A2 << /Type /ExtGState /CA 1 /ca 1 >> >>}
\immediate\pdfobj useobjnum \csname EF4O12\endcsname {<<  >>}
\immediate\pdfobj useobjnum \csname EF4O13\endcsname {<<  >>}
\immediate\pdfobj useobjnum \csname EF4O14\endcsname stream attr{/Type /XObject /Subtype /Form /FormType 1 /BBox [0 0 457.676383 392.87952] /Resources << /Font \csname EF4O1\endcsname\space 0 R /XObject \csname EF4O9\endcsname\space 0 R /ExtGState \csname EF4O11\endcsname\space 0 R /Pattern \csname EF4O12\endcsname\space 0 R /Shading \csname EF4O13\endcsname\space 0 R /ProcSet [ /PDF /Text /ImageB /ImageC /ImageI ] >> /Filter [/ASCIIHexDecode /FlateDecode]}{
78dab5584d6fdc3610d599bf82c7f8107a66f831e4d1ae5ba3b93959a0872287c0769c047602c769fdf7fb24a31699dda516bb8e7675d083387c7c22
df0c796fd9127e6c5fd3f8bfbc33c767d7ff7ebebc7e7b7e6a7f7b573f5d3e18b65f70dfa0c117dc8f68768efbc68c21ee4c88ea92269f3d1e6feb47
5fc4652d510053f3f4c9988fe6f804411ed0e8dc98ac0e2d345a9f9c8f12a7b8ec8a7a95f40cde3620da64ff146f6e5e815327f7762db44fc1c55c0a
482615eb11326425b1dfafed5ff6ab3d3e919117468dfb71e2d7ea718f10d151983444484779bd97cb3b7bfc27dbb36ff6c25cd8fbffa312741c2393
cb4fb147c4148f181452aec65f6110ee69f8e6141fe2d1dc9bf1fbbd1e3f602eae1069a69c49adb0d39875ecdd9caeccf11f6c99eceae3f4a95657e6
6ffb6aa023fbdeaede98df57e6c24c340ce7e044b424a9faafc10e01ce0c29b268288a57766190d719884607f530792a0635d86120a938e24499630a
7e27063ca4750e3e2517454ae48a430d7638f8448e938f1a534cbc130719c23a8710921355a27a26d4608743f0980b8c592d25a7ddbe841fb8e6508b
8a69ed451903ca3e59cc6e616962f14fb15647568acb902a4c3f409f87bbe17ab0c38ff581562b969df7148931d2ec5dfe19dd28b738c6c41b17b3aa
8bde27498cab3756f0632cdb2423b30fc3773caae351a057c30d68fe00d1b34d73b35adad97996907c43754637ae0f0ca494a717e13b1c239ae1ea51
3da9b95de16787cbe1db703bfc0341bfe229f458c6e8a484008bac59cee8464133baa3519bc90e8bc64831c6d463f9f60853ce4d1caf87077cec2bf0
fb30dcf6b805c14c4672c80db7195d5010340352086335683e48c1d86389792c31fbd8907c0617f493e842e2d05f78fb68c7708140986a0dad195dd0
8ec585e0037227ae83b4d31e4bf26eec4043c37246fbea49292e4804e780eb851594424e2888d79a5b85f615949c5c20f589534a7a90823d97c1dc46
d3c4a5342c67744141f50854bc6039e5f2d20ac2c238654e8d5957e88282c8925e19422e99f5a282a5c71216c6481dd2f874852e281810287a098b3e
bd8f82633f5e2837ee5ca19bb84dc5a6f880941885178d79513c1e7c8f221c8c39a0a86b28cee8827c283c3ca3725834e97de4838b51496db13083db
c5a38cf457d8c70385eb655e2c4e47aaa8836b7615da178e0b5274662a8b0ebd8770f00e4711cbb0f1e50add2a1d6764e7845a6cd1927790af977619
ce45817369ecb94217e4c33640422459b4e77de483731156696c4cb942b7cb1791a0b145c98b7ebc837ca94711b64514b3346ba34217e48324582525
2c7af33ef279ef4a861d348e5ca1dbe5f3d858152a7478d1ccddc285855c49454363cc15ba20dfd8b74ad65f513633a943d5abd43872856e978f909a
2332d9e115b36c3a3978a658822bd886354bf719eb4b579097f15679e982195bc882d5e71b379ec1ad9a691e37971a0f2f9465900e3fec1be16ab0d5
9adf0cf645c34e9689947f459d8cade27850179b8f39835b859b4ec2642c040fad8fa59b28b05d440519b9f1e2195cd89f211ba790c28b94c762df3c
9d8a4ea779ed29e18623cd8dc794e6ddcea79f78f5c2988bff00b18d4246
>}
\expandafter\gdef\csname EFWidth4\endcsname{457.676383}
\expandafter\gdef\csname EFHeight4\endcsname{392.87952}
\pdfobj reserveobjnum
\expandafter\xdef\csname EF5O1\endcsname{\the\pdflastobj}
\pdfobj reserveobjnum
\expandafter\xdef\csname EF5O2\endcsname{\the\pdflastobj}
\pdfobj reserveobjnum
\expandafter\xdef\csname EF5O3\endcsname{\the\pdflastobj}
\pdfobj reserveobjnum
\expandafter\xdef\csname EF5O4\endcsname{\the\pdflastobj}
\pdfobj reserveobjnum
\expandafter\xdef\csname EF5O5\endcsname{\the\pdflastobj}
\pdfobj reserveobjnum
\expandafter\xdef\csname EF5O6\endcsname{\the\pdflastobj}
\pdfobj reserveobjnum
\expandafter\xdef\csname EF5O7\endcsname{\the\pdflastobj}
\pdfobj reserveobjnum
\expandafter\xdef\csname EF5O8\endcsname{\the\pdflastobj}
\pdfobj reserveobjnum
\expandafter\xdef\csname EF5O9\endcsname{\the\pdflastobj}
\pdfobj reserveobjnum
\expandafter\xdef\csname EF5O10\endcsname{\the\pdflastobj}
\pdfobj reserveobjnum
\expandafter\xdef\csname EF5O11\endcsname{\the\pdflastobj}
\pdfobj reserveobjnum
\expandafter\xdef\csname EF5O12\endcsname{\the\pdflastobj}
\pdfobj reserveobjnum
\expandafter\xdef\csname EF5O13\endcsname{\the\pdflastobj}
\immediate\pdfobj useobjnum \csname EF5O1\endcsname {<< /F1 \csname EF5O2\endcsname\space 0 R >>}
\immediate\pdfobj useobjnum \csname EF5O2\endcsname {<< /Type /Font /Subtype /Type0 /BaseFont /BMQQDV+DejaVuSans /Encoding /Identity-H /DescendantFonts [ \csname EF5O3\endcsname\space 0 R ] /ToUnicode \csname EF5O8\endcsname\space 0 R >>}
\immediate\pdfobj useobjnum \csname EF5O3\endcsname {<< /Type /Font /Subtype /CIDFontType2 /BaseFont /BMQQDV+DejaVuSans /CIDSystemInfo << /Registry <41646f6265> /Ordering <4964656e74697479> /Supplement 0 >> /FontDescriptor \csname EF5O4\endcsname\space 0 R /W \csname EF5O6\endcsname\space 0 R /CIDToGIDMap \csname EF5O7\endcsname\space 0 R >>}
\immediate\pdfobj useobjnum \csname EF5O4\endcsname {<< /Type /FontDescriptor /FontName /BMQQDV+DejaVuSans /Flags 32 /FontBBox [ -1021 -463 1794 1233 ] /Ascent 929 /Descent -236 /CapHeight 0 /XHeight 0 /ItalicAngle 0 /StemV 0 /FontFile2 \csname EF5O5\endcsname\space 0 R /MaxWidth 989 >>}
\immediate\pdfobj useobjnum \csname EF5O5\endcsname stream attr{/Length1 10716 /Filter [/ASCIIHexDecode /FlateDecode]}{
789cd57a0b741455b6e83eb5eb54557fd2e9ee74877c3a49279d2684f089090102086d0c7fc4280101c5494312108104c247884c002101c109080445
3e1101111003329800226a1419c88c33c2bcf13a5e4551748cc878719c09c9c9db559d20cccc7b7766adbbd65bafaa77d5f99ffd3b7bef73aa810180
931e327887e70e1d0611900cc07a5069e4f0bc7bc7f9e59e6994cf25f872f8b8f13929bb87bc0d807a7dcd3d77e58fb01debff2ae51ba9de73efb8de
198ffcad64150df629d54f98363b588a0bc39701f0fe54ffd0b485f3bdf0485c3680b289f2a2b874faecb97d16ce04d0280f07a607cb4a41a51b4c01
ca5ba7cf5a5c6c4fa87b9ef26300bafc624651b0509bf2562580f78f54df77061584ed54370024ba289f3c63f6fcc70ef48d48a23ccd07eb67954c0b
16ac7c90884b7c8ff2a366071f2b953728730192f4f6de39c1d94529dfde7992f27d089f0ba52565f37114fe8e86d2a87e66e9bca2d281ea9ff5a1e7
130d3340e7951542974439049f51a683197ac1209072878dc907dbace0fc3910453ca5abbd1de066ca68fd68d1bc39a075f4635427196f0d2496aeb7
64b9ac9472ced044edafc1ffd0d5bed478ae26f820946adffd3f3bc3bf8dd107ed7fd4c148df4cfdcb97ceb370b013a72221163c100f495466271dd6
392b75b4c29bed65e0f4543a722ab50a07173d25a3955eab50a9062692a6061643d6616033e67040a79c9f864de032e4bc24382f38155606e7cd9e03
2ba7ce0b3e022ba705e794d17346d13c7a2e9e370b564e2f2aa1f4f479458fc2ca19c139d46646d1542a79343827082b67054bbcfa93f465e5ece0fc
19b072cea37a49c9f4e06c58396fc11c6a39bf78ce747aced0c7bf45a76ee5838ed7a7f015e4c1202ba81feb85d20a22b880f4f66207f105b733af33
cf6674a4c7ff0b0cb7d1b80bfffb763039d4567fdf4cdf32c66de5936fc9cffb292dd11abe514de94137bb762389d80d7a233bd60e23e9c81db512e8
a375e3ba346dd416b8be46996c63eb75b9f34cfe2c75890fbdf17f41b1442b4cb228889a2c4972071f6f5e79c5430b21005e58a2b8848b6d5567b3cf
6f6b831de0019dfb00bb29c78cbc0cb4c6481bed5462a311b2e14ec8815c180b0f40108ae011288505b008961872f442fa6df5d36006cc8279a1faf6
cfdb2fb4ffb6bda9fd7cfbdbed27da8fb41f6e7fa5fde5f683ed07da5fba1ddf7f72856c557d47ce66cc15021def7402b2c506ffeeec00bd4d4e0784
13e47680cef5b11da0af84073a40b751c10e882098465044a0db57d22ba2940c37c12c82528268025dc60b0cfe00d11802b2d9b0a4037c2c8b703e47
f79bb01fb6b1bd942ba6f2b954522b1d8155d4bb1ede66e7d81aa92795ed856bf001b5ac8273b85f06360a32a914e04392ff75960f47698c6ce662d9
aa22833c563e2adf2fd7cb57e426e82797c94d72815cc63271179fc0f71264e33ba417672101ead9275006c7f16bccc49372ae6c834fb009f7c31734
8bceb373504d322f275c5cac042aa472e97e2a39c39b602bdd2554dfc476b00f08bbe3ec09b808cfa02c8d801dec22d1750efe024f60be5441a2c894
8a09ff33345613f5df0a6564882e323308298dca087b9a6baaf18cc39efca2715f830a9a391f762bf58a4bf5d12c3ac7f6b2b759b3b2116ae1037c08
e7e2476c95ec93f7c923a03ac4012c806a1a7babde4729668b8976fd2ed7479716c9056c3f7c2d17a85369ec77748a68cea3d2fd4451319c2458a4d8
89a6816c15ae214cf5da38685247c9bda93f8da02e25aa014a300b6652aa1c0ec111e88935504d2319f42afdf85fa8e736f912d15ccd9e92fe024d98
0ba9502c5f255eeb6a5303f09aaa701925063dbcf63ac93fb2b02e70df44ef7b93127bf6f8bbacd7ae7aeb20af2e6cb1b7bebd3d6fa21ccb27d5714f
1dfab53ad9efbbf47faabcd4b3c7e8bc89debab6a1b91da30e2dc8a5b2711329a9e7a898ca87e61a75faa475dc4fbf910575de6933bc4fda9ff40d78
d25e34a067c8ee90e7a6958c94aa6aff5cae26e95848eb7d8108a5d609b5d60dce7551264f783c7adcb151f6d6e6ebcd7780fdf2f566fbd574962439
eccecc0ca7c32ea56480c30ebe24fd29adddb67d3bfdb66fbfc14ce2c71b37c48fccc4f34493384fd0c432e9eec3326b4599a81455a28c3dc516b325
ec29dd3e5da2253d99acbe190201770ed6ca522d5fae42ad494b503c0809cc62bf30ba2e3c7f6203350ef49fd4dcd84a08f56eceb8de7ca1399d7830
29891d0dc770599ad22fd1c1b3fc998e4477a260a3c4b3ace8576c54ebeefd72d988fa112d17f7d300242f791451ec811d8194e898588cf238b80c0e
cee51cfbf38e4d61b5ae0d32ad5bb09b2566f674b1a312676f1d5de7ce1f5d1799ffe0e83a57fe838409b69fee3fa9f142f3e9d30e67760736d70d6c
543bff56e5dfb23a8fddd1259b700b648c9727f009ea1279095f185b15add2aa8e966348bc9ef9b050591053163bdfb3022aa357c4ac885de1d907fb
621d53608a9f88c8ea0bfd06b3ac3e5d7d498a9a35986566c86e97a22a40a6e4cdd631c4c6cce03d2f56feec83c7965c98f815730d7d305a5cdfbf7f
ff22b661c0ec2d2317d5e4dc7dfe8e8cafde7a684f699cf886a8df46f22e23eabb4169a017b823cc95a6844a6f44ad3bacd6b451f1d47a37fa3628eb
dc2fa4467a22005dd19eae5ebb075d092625556742647e27fd26837e62c0f566a29238606fbe7cfd72b3fdcbab76e326aea4b380a9303e9810f41626
ca3085c533b74b4e4cea9a92154f84f425aad2585628711b793864c30be27df1d5c36766e6bf37fbd499863d878e6ddef1c233e34ecd2b3b3be94b66
fd05fa131ad77ffcbddffff61d1935d52b37ef5d545a569edcf5a8d7fbdb238f1fd035bc90a4bc9b744aa2786879208e8561182086e5005ad45ace70
b98959cde05134d96ab3ff71749d85080b3308b3ea845d18d4d89ce1d0e57af9c2a0e60ca2c510ac7c96847b561769770b74871130891cc6227812d4
4896065d591af66563d9bdd67bc326b062b6802dc1552c8c4469628998e9c874fb1c3e4762162a4262224b5cbc78b6ed61ee6ffd1c9b5a33f7895a56
f03649680749a890308f8387033e39467554dae3626a5557ad7d4d98540bcbc3d6a9bbe3bb7898193d60b62bf1f65676ab5cec3afa1dabc5aeaf1612
91bdf1aabe80f5154ce2118d21e944907e39749e83db05b7894597c6c718dd56db63628f16962c2e88ef1e7e7bc6e4d38fbefcab5fbd7cdff3f9fce2
7ef17478b8b8faa73f8b1fbcde7377a41fdbb6ed587257e27635615f63d893649818488e5020acd20ab5914aad27728fbdd6ba266983679ddf9a64f2
44c74778303121d64f068694e8b261622eb75efe497d022ef2a9ac496ac226f91c3fa710dd47e2a5296c0a4b52dcaec810aeccdd8bf99224ec24c4e7
d5cd516246a4b47bf5ce9dab099869cc7363defb207ce091472f312eae7d26dac45596c762c73c87038fef7afec489e7771d9716d7277715df8bef1e
9822befbe64bf127c3404d657be2750bb58fb46906c94481698128ee905042874cf682933c9023931928aabdf57ca3435f09bd6fb10304a42eba8026
be4ea17a8006a30d23c9c8d1afffa48073a2c4148ce1d97c049f8e7550a7a8a42d2418e66389fbf074db671f30d196c92f4e6859ced3f428682df177
adc15f1ff486bb03fe28e26e8a521bdfb3d6b9217e5dca0be951d6e4ee1e77b227dc44d69b4c7878626cbabdb5b1f97a63b3c1d8ceb56ae4b26991de
c24c7f2fb235c9991991ba913196ab2f2939ab4fdf88ce06a419d2daf57bf6ac5fbf778fd8b36203b4ffe72762c3f2a75f103ffef8a3f871f7880d4f
acd8b871c5131ba477b656556d7daeb26aeb04ef9165afbefffeabcb8e7893deadfef0abaf3eac7e9705e7af58319f803466395154451445191ae353
13a2592544d79af7c8b5b02632a1d6be21729d5ff578122368b794e409331486d0eff4495f8a1f3af525b231faad98d3b1a73da7e3de8a6f4c50f73b
4f3abf7622694c3f43b79d1136d215c8ea0399212d49eaca3ac9221e5c1ab36d34e9c98023b33e153798fd3386cc210e8b2fc66c63833b742981b484
8531e7848758f8375fb248c39ded140fc64b5b3a3549b73e1fb202fe11ee2279a9e00984c13649d9266b9cc9e0d3ecad831a3308f70bbaf3d2d7a07e
53e8d72274c05dfbaeeda77e6749e39a49e3348a5ebd24675f346c369b363b97b3cde697131c164d8a884ee060f344f2684f2f13789c7222f125a395
46d5edb1e1929a752167a71f094f62b4681c89215a6f26fc89b70a3a916d64b92f6cdffe8238c9d2366dd8b0495824f94acbb2c737ef11d76eb47d25
9d6dfbb86aedba5552b1185c326f6ee9ded387d7ec7279cf3df3de7f10c5842fafebc0b757c00d9b4d84a95d93ec66e0d16119e031c94e239220fd0b
21478be448418481588734fc89c63b95b18dd759164b1097c43991c376b223ac46cc107922c87bdf58c4a2582fd68375d92bb68865e2e7a28696d459
63b77491f6c18b03f1aa83494c729097cc5125845734ae3055f2c87d55dae8980989d6e6d03225c39edd69d8bfd56152d291040b611448ef2bf55747
48c3d547a4627599a42acca4b8598c328c8d541e60139522f688b25859c59e5436536cbad36237562d11e2a087ef2cb34b358de25adbcc467ef14682
7ca9254dbe742381f0a49d89ba97f0b4b2f2c0481ea37093d924c7984d1863b698a51826592c66c5a16a2a27f3a269aae440c94aad690763cd317309
15a2c7a2592d6693160a392d2a84d92f9cefd2414f4697ec9f4c8fd601fca777284974824da7f37b4552b864a6c8cbec3477e3c966af79b03498f731
a79bc748f7f01c73c03c499a293dcaa79b0bcce55285f438afe0cbcc35d2661ea78249d2106485ebc7104c9535d054139864b3d90ab61874cb6e2dda
6ab779e544ee55bcaa57f39992cd7e8bd7e6b50d920660969cc9d3b5bea66ccb106bba6d180c63a3a4803c9407780e092e470b680153aef91e6bc016
b04d94266893ac79b662693a06e5a9bc4029500bb44253a1b9d0b20816b272e9315c24cfe78b95c5ea22ad547bcc5a61adb0554a55b85a5ec357999e
b454dbb6c83b6dafd81ed425956962fa8ff94ccc977b9e0d60d99feb8f26b1468877c45b82a4e694afead092c6ed2dd748bfcbc94af5a47d8719fc70
92e2c6044b17930d5eeaa234d81cdeca84e39e065fbd635d172b74c1a83093664940cd35b42b29dbf90b249590be355ebede4a2bf25dc30a3bb2758b
35273d2e3d3e3d21dd9b9e989e3424251017880f2404bc81c440525e5c5e7c5e429e372f312f292fa53465555c557c554295b72a7155d2fa94da946b
29f19d5d3b3b757628882f4828f0162496c69726947a4b1397c52f4b58e65d961875abefbc93f573f8b27483d8952c7c66e2ad5158a474ea9383cb4b
9e6da8af1f7272f5c1736d3798f4e2968263f945a726ffd73529b3b87c6ad9874753c7b42ddf5f1c7c73d7eba79d156b7bf5da9f92d2aa7bcbb9ed9f
e315e255340c09c442255b2ddb2ac3569b1b1c72431762528cea0c8311aea131f6d6cb199d565c5cbf6affe16a7ac0121e6b8f5d16bb3eb6369613b2
8673ef40b89fdbb0de21ef8e57c66ecf7bf5dd775fcddb3ef69e3d53dac4ef594fa68cdf25671d4c4bfbbca9e9f3b4b4fdc9c96c30b331271be02309
1256f264c545419107fa07a2631ac0e66ae0da3a5b3ddb825d64d0a4e10ea765689c61a5323274137559f7998d57d38f15c42f8baf8d47c35675308d
b641408c62b77846dc555f3fe0f0e3e7daa1fddce387dbcebcf8f4d3fbf63dfdf48b784c7af86fcdfb0a832c976974e70685fbdc952be7083af0aa20
6eb9209662f264703353a5b69abb5f62bcc1ca4e443538ebadeb3cb16e49736b305a72860ff51828361afb0e9d79a190e97ac807a60e892b8dab8d7b
3fee5a1c1f0243d81069887b482cefa1f6d67a9b7a984ba084954825ee9258d394b93a83138da0cfe0ad617d298a520da6ab7245eb116bd36b33cf4c
9df6fea3e2ba38c3525b3f636abdb467f5d6069bf4f0e45367faf439d4bd07ebcfcc2c82dd2d3e6edc72f4d00e5d03c809ca85c4eb08e813884613a0
8d29553647bd758b99491a8cd557c83097eeaff4ad25192e62b4c349b6eb4881dbf0083e47c8295322d3885123e5c2fac71fdf7cb0a121e7d5056fbe
2bed6e7b48dab173c7a9dd6d558aab6d4751e177badf7d93265f4cf3ea7e372d60534ec987e1a4c49926c330ddf136eb62bddcda4c7a6637054c79a6
0253a989f42c22539fc8ed7bb39e2eb9e046ade2fa5aa7e3a7f1925e832d12d360986c0f6d3ed3036176b25579bc8097f26be4628c416800c5f5b7e6
0e1ea87124d724981ce8aa384d51e1a0c4a96e6b559c17eb634f46db5570846b9a92e7d0c2f33c515acc309fe1c05bc93d193bca41832e5f37b6213a
630211e9c979c9a5c9eb936be97e23f993e4f6641371cae08dfb567efdc438778871a9434faf78e554c3bc05d57b1be62d7a6a6f43c390bac54b0ee0
9ac717fef099cec6e7b7e96c9476ec7aee8d17daaae48243d3a73e7e538a444104f4bd5d8a27ffb9142f774af16881fb376ee9efe5e8fe6fe44813eb
620cad880534af7e26312a10a13438a1c15aaf9f4938c3ef43a77be8df9d49047c43a2cba15ca9502bb40a5385b9c2526ead08abb0558457d82b1ce5
cedae86bd18edb770db71d5d946d3a7860f3c68307375e634e71f5da9fc577cc819f5c397bf6ca57ef9df97a9b784f348b6f49fdb349cb5dac3f6178
5c4c90771386ba2d191c88edb425f5b675ec753c19477664b8615186e9d624232384ebe54e73123085ecc9a7f1329be2bfc91ac2452263cc3ac351c2
949535340c385c7e1edadbcf971f96fa93457951877d6d8714f3fec2a03829fe4af7c920fba6d3a084e486a3083b07d0164aa108c161c12a5bbde9a4
6a56c8490f73eadb0163259015b9705e371b47f3227646e8120bd9db9fc4d50547258cecb1ed45c2e3f8aa885e1e3cea749c3bd5768484553c8d739a
ad84acfd199a2d05ae04068559259b655c42bc669254f3b88484f81cb3253e417693175823bb2add6ba2742fe0272fd02dde6c498855e1fe58cda66a
aea4a1dd74ac2e345f26748c0831e4167ed0dd82b333aab17d4bb1ad6a3cf51826458f61667bcc1e8bc7da8b8c5b0f4b0feb40d340f340cb40abc50b
5e962c753377b3748fe8edeaedee1ed92dbe5b42aa37353139a5d25c69a9b45686e9a7c14c9214b362412b86a10dc3d18ed11883b1e891e34c29bd53
87a4fe2cb5227559eafad4dad46ba951143eccfdc92b2518671a8aefd6cd736fa66fa2fa12ef70edd87d93d7ac99ba6948e39e1fff30f9ed59c5ef06
57ac2b3a1038f0cca7bf2e3e2a0f39d4ad5b7e7e6064a2adfbb36bb61df3f94e65654dba6f749e3f3c79f38a1d078d9d673f326bdff31db406c967d9
b8168e2f81839dd4aacc16e231e998dd69d3d7e0a0467d5bd11159870e2fc89abe12b2a6fa46df153990b9755f4ffe2bd3c116b172b16a74d9ebaf5f
dc5555c57788b7aadb6ad78cddbaf3775241351bacdbd243b40a271aabdf0503039e9fd6ff3a333be9aab7d2ea7759c6921d18e6d697637648a32e67
dc340225eed3ba1188200b1e5a76377d6757764837022fd7d7df7d78c19befb1dfb0e3d2deb6e0ce9da7764be5376a0f164fbb86fb74ea2700285972
0145c0df50bc3c3e142f8fa77879bc1e2f8fff97e2e537fe49bc3cbace9c3fbacea19fe639f54784feb018e727a09f99d086dd1a3a101a5d670f15eb
c72aff7698cd02ed5c8a9422799239cb3c521ac9875144fda0f4201f6fce33cf91e6f062f3628aaa1753545d253d2b3dc337994f4a27f9afa533f81b
1ec725132ab2859b358b895e56b7148d91720c8fd5624d2e8bdbea073ff349299828fb799292a4fab5148ab0132d3e6b36f695fb6ad97a5c2d8dc061
7240cee10125a006b45c8aa9732d7a4c3d0126b009529e7c1fbf5fb95fcdd3c699f2cde32dd3a090154933b1489ec9672a33d539a6a065bab5c4b600
16b0c5d2527c4c5eca972815ca12b5427d4c5b6caa30959b175a965aaba4d59c626cd8c236491b719bfc1c7f4679467d560bf4aeb1eeb4ed85bd6cb7
b41b0fc807f84bca4bea016db7f515db2fa5c3f8ba7c82d79bdeb0354a6fe379f9577cb1119fc732fdc77c16e69b50ffe5171f7ef945bdf8e8c33f7f
ffa15cd05a833375b8518b35ad3349472692b5fb8874c404db03315ae86c863422477b094ee24b5c43063253cc9d877c56439c4ae8942c423f25538c
53b28c8eb39be67f38bc09e4e832ec2a0d9746aadca2855ba230564bd3bc96be98ada55b022c200dc5801ce0776b0fe024ed67960256201563815cc0
a76a15966596572cb11da73afa892e4b9c8b33dbc648475b974a47db8ae4827dad1f6ddc877ee8f856274d7a3623af46fa59f8a01f2041333e84fdf6
e7b6cb9def1f7fdf3ac636c9a47f77d66e7e37a37eea6c110760133ffebee53edba47ff8d2962c3719dfa440da4fb096f6cb5da08ae012410dc13682
42821d04d504fb08d6122c578ec2877c299c5552e12c6f85b3ead7902b5f8172b918e67217cc959be9ed84e35236bca9837a0e8eeb79f90ba3fe388e
a2741a94a00ffa51f921652d4c903375a9ddbc726031fc95cd6267a56ed246e90f78171e205d4896a7c9aff14d4a776581f2a69aa64e57abd597359b
364a7b4efb9329d1f4be39cc7caf256839d3c1b764cc8734980156b25d767856e78eec9622e9ad7ffb5261b2feb5533651e374e37ba29e661049b950
5a028d0deb48e32de5f22d690e516c6c475a01172b86bba1044a09ff79f0084ca7d9e78317bac13448a57706a4d39d49a9a9d4c24b743e42f56504f3
a00882301b7a50e9489843ed7b51ea2e9845b717eebf395699912ba27711f55948cf426a69fe1766ed7b73d67c9a6921cda57fe19a43ad753c82d4e7
df9b31975233a9df0458402df42fbe4163b422a347d0a0c84ba3cca1a7fead782a8dfb08b5f352ff129a3d68d4fdfd38e38c51ca0823da97c0a354aa
cf5a466d4b8c913268ee4cc8baad57679f8e7f4bb4ffdcf8fefe8f57b2a11712494bf7606e88a468328af6a431c63f2fbcb48f4f813b60008d3d0c86
c308180d63e05ec883fb88fe71309ee67a80b474324c01a8979605da6f086c71e1dffcf8d70cfcb106ff62c31f045e17f85f7efcde867faec16b7efc
eec9bbf87702afd6e0b735d8dc82dfb4e09f047e3d00bfcac12b02bfccc02f2e8fe35fd4e0656a78791c7efe596ffe790b7ed61b2f09fc54e02719f8
9f2efcb806ff28f02327fec752fcf004fe41e0efa9f9ef97e2c50bc3f9c5a57861387ef0bb58fe81c0dfc5e26f05be2ff037027f2db0a906cf9f8be7
e7059e8bc75f65e05981efae72f0773df84e24360a7c5be05b02df14785ae01b024f097c5de0498127041e776043a59f3708ac7f8d2cb2c0d78e4de1
af9dc0d796c9c77ee9e7c7a604daf15840fea51f8f0a7cb5068f083c2cb04ee02b020f15e2cb363c78c0cf0f16e281fd4e7ec08ffb9df81221fd520b
ee13f8a2c0bd02f73871b7c01776d9f80b19b8cb86cf17622d35a9adc19d02776cb7521482dbadb8edb968bead109fdb6ae7cf45e3563b3e6bc66704
6ea909e35b04d684e166eab4b906376db4f14ddd70a30d9f6ec10deb4ff00d02d7574fe1eb4fe0fa6572f52ffcbc7a0a5607e45ff8f12981ebd6f6e2
eb04aeed854f12994fde856b565bf81a17aea6d0980aaa0ab1923855e9c7550e5c29f089150efe84c0150e5c2e7099c00a8181f69f2f5dca7f2e70e9
527cbc10cbf3ddbcdc8f4b042e16f8980d175971a11917089cdf82652d38af05e7b660a9c012817304ce4ac44705ce74e4f099e3f011813396e274ca
140b2c125828709ac0a9028303b0a0051fb6e214810f0a9c2c70d244339fd48213cdf84064347f200327081c4f338fcfc17c378e63763e2e0aef77e1
7da322f87d02f32c78afc0b1f7d8f95881f7d8718cc0d154335ae0a891763e2a0247c685f191761c1186c3050eabc1a135982bf06ea927bfbb05734e
e05da331207088c0c1773af96017de39289cdfe9c44103c3f8a0407b380e0cc30102b305f6efe7e2fd5bb05f5f3befe7c2be5916ded78e5916ec138f
9961987187856708bcc382e9bd2d3c3d0c7b5bb0574f13ef65c79e26ec918169ddfd3cad10bba73a79773fa63ab15b8a9f77bb0b53fcd8d56fe15dc3
d16fc164813e8149e1984874263ad15b88092d184f24c417625c187a88831e81b12d189383d19489161855885d88535d044652a7c868740b74098c10
e8a4064e810ea2d59183f6a5185e88368161d6481e26d04aadad91681168b6a349a046cd3481aa0b954294a952260d702395a2a030d2cea59ec8ec08
02593d2b5cf5144bfbffe182ffd708fc5fafb8ff0d696f7e02
>}
\immediate\pdfobj useobjnum \csname EF5O6\endcsname {[ 32 [ 318 ] 48 [ 636 636 636 636 636 636 636 636 636 636 ] 61 [ 838 ] 66 [ 686 ] 68 [ 770 ] 78 [ 748 ] 87 [ 989 ] 97 [ 613 ] 99 [ 550 635 615 ] 104 [ 634 278 ] 108 [ 278 974 634 612 635 ] 114 [ 411 521 392 634 ] 119 [ 818 ] 122 [ 525 ] ]}
\immediate\pdfobj useobjnum \csname EF5O7\endcsname stream attr{ /Filter [/ASCIIHexDecode /FlateDecode]}{
789ca58bc70a80400c441fd8bb6beff5ffffd1b08817173c181832c99b819f637d701b07170f9f80908898e426e993c944b9a1ab5e9fc2902a451535
8dbe5a3abd7b06462666f10b2b1bbbb843b3f3027a750274
>}
\immediate\pdfobj useobjnum \csname EF5O8\endcsname stream attr{ /Filter [/ASCIIHexDecode /FlateDecode]}{
789c5d923d6f83301086777e85c774880884e056424855ba30f443a59d5007c73e22a4622c4386fcfbda7e4d2215091edd7b1fbec3971e9b97460f0b
4b3fec245b5a583f6865699e2e56123bd179d049963335c8255ae12b476192d425b7d779a1b1d1fd9454154b3f9d735eec956d9ed574a287843196be
5b4576d067b6f93eb690da8b31bf34925ed82ea96ba6a877e55e85791323b134246f1be5fcc372ddbab47bc4d7d510cb839da12539299a8d9064853e
5352eddc53b3aa774f9d9056fffc5981b4537f8bcf7d3cd0813f5ede43de3f415ecd0cc8813d500007a00438f0b89641550553c5aa0a72817a1e1d08
19853d3a103241a62813e4034ef4e8c0209768d9a30321a3f7f210e56816ab8a20745fc69f7037e195c1e4f1cf45b3c488259a2cfb3526a4700ccae3
a9ab89c379b13a118b79789c87c779b8802ca22c9cec2e7abd517fe77e416f0b252fd6ba5d0a5b1c96c8afcfa0e9b6e866323ecbbf7fb20cced7
>}
\immediate\pdfobj useobjnum \csname EF5O9\endcsname {<<  >>}
\immediate\pdfobj useobjnum \csname EF5O10\endcsname {<< /A1 << /Type /ExtGState /CA 0 /ca 1 >> /A2 << /Type /ExtGState /CA 0.15 /ca 1 >> /A3 << /Type /ExtGState /CA 1 /ca 1 >> /A4 << /Type /ExtGState /CA 0.8 /ca 0.8 >> >>}
\immediate\pdfobj useobjnum \csname EF5O11\endcsname {<<  >>}
\immediate\pdfobj useobjnum \csname EF5O12\endcsname {<<  >>}
\immediate\pdfobj useobjnum \csname EF5O13\endcsname stream attr{/Type /XObject /Subtype /Form /FormType 1 /BBox [0 0 472.07952 270.47952] /Resources << /Font \csname EF5O1\endcsname\space 0 R /XObject \csname EF5O9\endcsname\space 0 R /ExtGState \csname EF5O10\endcsname\space 0 R /Pattern \csname EF5O11\endcsname\space 0 R /Shading \csname EF5O12\endcsname\space 0 R /ProcSet [ /PDF /Text /ImageB /ImageC /ImageI ] >> /Filter [/ASCIIHexDecode /FlateDecode]}{
78dacd9b4d8f1cb71186e7dcbf628ed1c12db2f87df0c18a13033e04902dc007c3076325cb36242bb29208c8afcf5b9c0fbed5333b3bb39a5dc5f2da
33b5ddd5e4cba78b45b2f47eedd70e7ffcfa0ba7ffdebc9d9e7efdea3fbfddbcfaee9b67ebbf7ecfdf6e3e4c7efd3b7e5ee386dff1f311b77d839fd7
93ba783bc522b32b2d09bebda16f52dc1cfba737b892bffd3a4dbf4c4fbf828b0fb8e79b690a75ced56549eb90e790f0ffb7938430bba037ec6c6fc8
26c9cdc1c7ee6edccdd6fe90f7eb43d7be457ca8b5e17a1fe696c5c9facf57eb1fd67fac9f7e25da28597f8b0ea2b3735d7f9cdc9c9bf32dbb52237a
6eb589654e3e79a196ef4ddc9ce9fbe9f9fa3d1e1036bd7e3d1ea01aa8e5a4af2073dd88303dc3607c9cde4f3a865fe820467428a418a44809e8d45c
522d4947f5d98be9e9dffddabbf58b5ffa70bd7839fdb8fecbca3f59ffb47ef1edf4b717d3f3deae8796c93b3f4b73c135ea1c192f95ea2e7fa7e46a
75f61e7dc54f91b3e4aa8f2e5796b9c49273e1ee0de3c572dde1ef945c3ee18e5aa5797df03970adf263eb25126771c9d908328c97ea7597bf537a89
2f738d29c756f03e9ea598ac222bc6e27b8453e77283833207f1621cf985a3af57af56ff5cfd6bf5eb6abd7afbd88330fce782eb8a8fc5c4f1613d6f
108efb0b69f64beb516cdbfe6ef4cc3929bec652b7ea41b3d9e5d87f0fedc6a860309eac5ffc3ee539a6ead04f553ccc55fa87eda56573a5ef577e06
7d11bee6165288d5084ce64b15361e87c4d6e3098d114be75662f42e20ac90c8ade4f4492287cf26728ab32f98579b1579982f16993d92c8c6e32991
358a69bc0e6879be26c9e9b389dcca1cbd2f36e31bd68b25267fa430fb3b25700b73f6b162564ce19afa96cfa5af445c185b4bde084ce64b15361e87
c4d6e3098d055dce156d0f223e5e3352b421f2fba92f6ef4f15033b41630f6150f4ea1caed53689f43ffb17ab7fa63f55fcca47f3ed1d93bb85c7dff
835fbec3c47a83ffbed35f62a277316dfe2434e0e5ea37dcf933e6df57ab0f987da74f997d038f39f2dad407dd2bc6391697a158e02fb10abcb688b5
c2774b1ec6322266acceaa43f4783b250c5e2bd2c8f866424e584bc3473662ed96b2cf615df5a1a18f71a9b3c4505a63634d7393d4b231ea4be5aa0f
adcf12a5b814f445a46cdd4b7f1812ac6eaf604174aef6020a730f76b07b1dbf983117e115452b231eaf76d1e58f8ed1c28e24cdf9da2aecc9cf9899
44ba9fd0e6e2bc4f05d703a51242edd743a4500356518beba193c3839b5e0f0e4b74bd99237bf6c5eb153506b5eb42575a0d79edabcc21e351dd4dd1
6587209783bd7f8cbecb00cd20380612404091005abb1db2c17d886acfb3c73b93fb1a18b2f954319c3dac2129f712ba1d7d4798aa1bb0b26bbea7bb
02d992484dc9faa76c56fd6329e4026566f7c29550fd71f30143a03e7ec2fbf51290662931b906acf005f958f0a1085ae640424568121d9513f87222
b8c797f3c03dbe6cdce36b8c83df366bd2d2ad036041c80a2d766d07c16ce5e52bccadb9ecbd0558ed984a625f2a10c0a2ab021f818005587cc54485
b5c30260eb6700acd737574b59006cfd0c80d5ee00c986485a4fe24543b353f09660d3dd01b02e717246ff9b0558dd24ac769c5880d90df16bdc10bf
6a4767300c965fe3861677b8c261cccb55f1c52ba4bef7e8a60db8605483ed76e6855e3a116fe3b0bf4ee43583ba675727d824a5160bafb836a34f12
65412fb29654bd5fe0eb34d54cadd505bfec85f9751a80319e71c12fae2fc1a74d04237e9da0b9691bb0895f97e03f8445f86d100f9358098bf08b7e
c01f22de825ea7f12e27a90b7ac18b945c364f257ad19184713a88bfad228f46ebfd025fd37ac657f076f8d85533fc7a55d339bfe457b707a06b5af0
0bf5d57f59c45fed57888806f59a0043f56d0cde335c34c0425868ac2faa4055cd2135fc86edafc2951836b4ee193664ef1936f318335ce782c62e43
b0659518366c530c26370b86871f66989967867b86ebfb28198a058381e9332f62b0698f8dc1a33d44b1690f531ce786d813da2208a38b012debaf38
53ece35c043d4c4b8a91fd3ba45a794131decde435fb5e528cf66750d79614533b0dc579164ca33e3e741671b8e0a13d77cc11c0a1264674184d981d
4622948c23c87234b529c2ee629b22ecac8b1461984d86b037139c6cb6f9c198bf393f3076ca0fd0d8ecf0262ef303b6339bf4580eb064b6d9c13053
764066e612d138b49cc3924b7a1f4c7445b2dd426e47b2831402165dcbec00e3d5304665991e8ce63c5876e0b7d92ba8004a61135a7ba6b08db1ae9e
0aac3d5ce2e9251a6c87f5eef420f839355c1d2db96c2672a5cfc74e16e84a5f3f65bf6497bd30bca2fb2a2d8625bde67ac2d7d8095fdd4f2fb25db6
0d7a4344625791222ee8356e885e2c23911ea45c16f806446e74449691d5f821808d9d08861fd129a02d08663b130c3f3082e105c12c33136cfc10c1
b063d95c0e376b747f606cd4dcb609348e308e1f95c2e7b927b0b8d4e933b79b44be3ff935eff4d4dc5f06eda61e1cf547d0768b2c4e2c7e78b2d645
7ed8fd03936ea2bc5cbd5b7d5cadf169bdfa123f69e57527681c683cfd2a6e5ecaedb13444d81c4a83cf32d7aa9c443f8e6ef2b62b3b9b0a0ca00a84
0f7ba30eccf632657373e2356c6173d5cd346cfa92f4ed003661d55f7a5fe811de05bd7034646bb8a1f66e4dfa8a8189a8bb563b9bd7047a7bd1de3d
1bfdced9ceb46bed1bb6ed7b454fd875fe50b71b3d447f362d3706b75b82bcbb07c709809675f3738c594e9e763d5bfdbcfa8051fe8071fd73f5efd5
9bd52b1ed9bee7d5e3ecfabefb5d5e5394de958a37aac4d0f783bd3fb46a783a66359c7f3cdae7b093abe8a656bfef44a7f3ca7127b71dfca45d92e9
b64ee7847027b6cb3b1b7778d87a77375dfd71d78c535dceba20923b4e8117a79afb2edf73657d6b77939e517501b9c3c3ca5d66eba59dc6041437f7
9eec76b5e502fb6edf7f3db6ec39806b19f364b278d723669da78e9acfeb7c5f7b6c6e3e8b72b9a5efe765f1d43142b81e184da72e8598ba7406c5b2
7267527c710638d91e1b8aeb11b3e9f5e51c53bfcf0259b4ef76c6f566c6dd86c53ee74a91c3fc21e672903f0c1b272074fb410997663b47dc5fb97e
a4fa83d2ab61bbb87ae4a4b793b52348423fb594ebc1c542030f6bafc878a95c77f93b2518d6039f5cccf5f082e1bd3da8be22e3c582dde1efa460cd
7d7235d7830b16533d2cbf22e3c5859577f83b595c99fcf5cab9826e965ea99cebe183e278006ffc8e10ce3bbf67c5c5e30e25b75d1581f178181d53
d8dfafaf7df1ae2243cad72c2270b652e3314546eccf1e643423f2b05e2c323b24918dc7532297088da4d6280559cdc315ce3da6c85ecf6cb28f982c
5865325f2cb371493a5b9fa784d6a5b73897101e3467bc22cef219958e612e2ef86a7341325fae34bb64a58dcf934aebee6214c144e9427cc03ac547
551aab0889b1b66a951ee6cb956697acb4f179526938c3854e82d354e38ad1237e3ea57513b7e886b4374a93f962a58d4b52dafa3ca5745fd9d55642
6ad2da0396853eaad2ba09293548b44a0ff3e54ab34b56daf83ca97440d2964b40d28125da35a3473e5abb287ae6f1ff50bc78f1c85fb57a71ac9d4d
da58db910a04c132e8b08031f4ca0b3d276b053aa6d80baf8297dd399935b7dd3999314bd9d73106dd4189be1f10d14ad5673db2ccb9e95166408eba
ab63c4524347af1fd9052c5ef6250880d5bb7ef2092d4d05c2ce0c97fb02043d904f583da8bd8451801066a95813f556422b2e40d85f0fb576558c0d
0dae2e6908a75563afc268053dd47d1adcbbaf62f4baf3b7294f8f5a0ab3ab62c4ade2dac6dc461123bbd103c2dd31afee6a7ae9d950d4c2e15df981
60704b3fe38b906c5fc3a8a7b6fa45edc9710d233ad093575ac199a76ee3c2fda07da81ac601b149cbf71073563e2066eb80d85807c426e1248ab5c2
a3248cf10263ad5fca887cbd5c8f377078ea218cad1fc258f44815413d5b8ecda030c88c820199c6dc803c10311cf3e5b49d42af0f716c5e07e6985e
07e2d8704918b3f341b171c214f3cb302836caf3c6063df47a145fa994f17e61588bdba4ef7e33c1823cbc86523691921036b16c206c621605e20a25
bcd483406cdc50202e40a679d72b0b89608e2a0c30873e02d82b848aec32126bcd41a8b5172010c09bbfe1d36413ff8960e3874331d1744624e6c772
2866750cc3fc267028a6f79b2836cc532836cda158cc8fbd1ec5d7ab67bc27c806d94132f34d20830af1b16d63f49db1d8124b241bc2ef8ec5d60fa1
6cc8279479416750e654865036ed1928dbf630ca3c8f33cbcc26b1cced3c27adb88565939d7040e6f481b20a936d10cac3cbc36615476a1ae900caac
11f79c9a799a232e4fdf042a9b39e452686550f9724a1a8cfd9ca481ed9cfb1a3b25bff43e71f24be681e922c7bd33f7b576ce1938f471d2c0764e1a
189833b2865b28355433a6340150dac0de396b20ef0f9736dcafc69128e62247a298ab1cefcc1b246bb158ad62215ed81bd790bb125d5d240e229007
f36eff2b264cb1f14314db1497525fbe9e283676c258407dd3bf5ab6e0584b14738d75116ead1fe2185d14cc7b2e2c38e67632c78247e179fd7ae2b8
17824a141f171cb36e6615b729b7a8cb551cb78753078ff623df6ce5cc04989f7b7eade3d1fa835baa0ae0f5ec7285bbcb1d83c4fd89de350b1eb33d
437c3e4dcfff0714908b63
>}
\expandafter\gdef\csname EFWidth5\endcsname{472.07952}
\expandafter\gdef\csname EFHeight5\endcsname{270.47952}
\pdfobj reserveobjnum
\expandafter\xdef\csname EF6O1\endcsname{\the\pdflastobj}
\pdfobj reserveobjnum
\expandafter\xdef\csname EF6O2\endcsname{\the\pdflastobj}
\pdfobj reserveobjnum
\expandafter\xdef\csname EF6O3\endcsname{\the\pdflastobj}
\pdfobj reserveobjnum
\expandafter\xdef\csname EF6O4\endcsname{\the\pdflastobj}
\pdfobj reserveobjnum
\expandafter\xdef\csname EF6O5\endcsname{\the\pdflastobj}
\pdfobj reserveobjnum
\expandafter\xdef\csname EF6O6\endcsname{\the\pdflastobj}
\pdfobj reserveobjnum
\expandafter\xdef\csname EF6O7\endcsname{\the\pdflastobj}
\pdfobj reserveobjnum
\expandafter\xdef\csname EF6O8\endcsname{\the\pdflastobj}
\pdfobj reserveobjnum
\expandafter\xdef\csname EF6O9\endcsname{\the\pdflastobj}
\pdfobj reserveobjnum
\expandafter\xdef\csname EF6O10\endcsname{\the\pdflastobj}
\pdfobj reserveobjnum
\expandafter\xdef\csname EF6O11\endcsname{\the\pdflastobj}
\pdfobj reserveobjnum
\expandafter\xdef\csname EF6O12\endcsname{\the\pdflastobj}
\pdfobj reserveobjnum
\expandafter\xdef\csname EF6O13\endcsname{\the\pdflastobj}
\pdfobj reserveobjnum
\expandafter\xdef\csname EF6O14\endcsname{\the\pdflastobj}
\pdfobj reserveobjnum
\expandafter\xdef\csname EF6O15\endcsname{\the\pdflastobj}
\immediate\pdfobj useobjnum \csname EF6O1\endcsname {<< /F1 \csname EF6O2\endcsname\space 0 R >>}
\immediate\pdfobj useobjnum \csname EF6O2\endcsname {<< /Type /Font /Subtype /Type0 /BaseFont /BMQQDV+DejaVuSans /Encoding /Identity-H /DescendantFonts [ \csname EF6O3\endcsname\space 0 R ] /ToUnicode \csname EF6O8\endcsname\space 0 R >>}
\immediate\pdfobj useobjnum \csname EF6O3\endcsname {<< /Type /Font /Subtype /CIDFontType2 /BaseFont /BMQQDV+DejaVuSans /CIDSystemInfo << /Registry <41646f6265> /Ordering <4964656e74697479> /Supplement 0 >> /FontDescriptor \csname EF6O4\endcsname\space 0 R /W \csname EF6O6\endcsname\space 0 R /CIDToGIDMap \csname EF6O7\endcsname\space 0 R >>}
\immediate\pdfobj useobjnum \csname EF6O4\endcsname {<< /Type /FontDescriptor /FontName /BMQQDV+DejaVuSans /Flags 32 /FontBBox [ -1021 -463 1794 1233 ] /Ascent 929 /Descent -236 /CapHeight 0 /XHeight 0 /ItalicAngle 0 /StemV 0 /FontFile2 \csname EF6O5\endcsname\space 0 R /MaxWidth 974 >>}
\immediate\pdfobj useobjnum \csname EF6O5\endcsname stream attr{/Length1 8492 /Filter [/ASCIIHexDecode /FlateDecode]}{
789cd5597b5c5575b65fbfbdf63eef7338e7700e0f399c73100e882f08c4777122df9a9192a316c551407cf0107ca468a82948e9e023b18c94ccd4d0
1c741c0325d3645243679ad466ba53b74cca9ac86c2ed50cc2efdcb5f701d3e633f7cebdffdccfdd9bb5f7efb97e6b7dd7e3f73b1b600060a58708ee
b1a3468f81607000b0fed41a3236fda1a99ee5fd0e517d1451c6d8a98fa4c5ed493d03203c4ffdd50fde9f314e9b9abc160013a9fec343531392e6fe
bd701531eb43f569b3f37d45c206e349aa531b1b377bc92237cc8d1c06209da33acf2d9a93bf70d09279006aaac38139be922250d30d9a455437cc59
b02c7793e68974aad3fca0b6bc1c5fb626f3ed7280b005d43f388f1a8cbbd49554df43f598bcfc454fae381564a5baccaf6841e16cdfa06b89245ff8
50aa0fcff73d592456aa16525de6ef2ef0e5e7c47d736f13d59f23792e1715962c12a609370122f2a87f6651714ed108f577548c682699f340c6ca00
814ba01a4284d226930e06c24810468d999401a605be4505104698d2e5f703dc2e29a3e7e7141780a67b1ea33e41796be8dda08c3c075f800a4cf02f
5dfe8ffec9fb12d19eeef2cfdefffbcbff91ff5260859eb722bb8a703192c4160821bd19950df414945ebc3d59d64da2760398955e81fa446a5191cd
35f4d4828ca35ee115c0740b3c073605d3e5be62df2c58eb2bce2f80b5b38a7d7361ed6c5f41093df3728ae9b9ac7801ac9d935348e539c539f3616d
9eaf80c6e4e5cca296f9be021fac5de02b74cb4fb2cdda7cdfa23c585b305f6e299ce3cb87b5c58b0b68e4a2dc8239f4cc93f9df61bf9fae805c9fc2
97900e230da0fe586e14d6907259a4dc956e45b3ee46ada7cef2bacb8ffcf74833b2bfb0e45f30c9ccc058f97dbb7c078fbbda67de512ffea92c507c
dcaaa2f2c8db53fb9096badb1ad35334b14d642b9092a517a8ea0cbcf18f902b50c4097a15a2461404f1f68cee2b3d77743678c10d8b55366e633bd4
f9ecda5d63b09b28ef282bd6508d29751156296f33b51888430a0c8161300a26c114f0410ecc857c2882c58a85dc9078576f36e4c102288462bfdf7f
cd7f99bcf57dff39ff697f93ff90ffa0bfcebfdfbfcfbff76e49ffe10ac4fac1ee9a51592540b2bc72ce4b8180570fe92659fe61dda4271ad54d32a7
49dd24f399d24db24d7cdd1444944d944364269273d05c224217e47c974f64232a242a220a21926db898288ca550ee68a1fb34d4410ddb4bb55c6a5f
482db5c2115847a31ae00c6b6195c2006adb0b37e1128dac8016ac13814d80646a05f85012a09d65c051e2318cd9d830b54a0471b278549c223688d7
c58b30442c112f8a5962094bc6ddd234692fd130fc2df9c079704103fb044ae0387e85c9d8248e124df0095ec43af89c5691716a812ad803a5248b8d
154299502a4ca196b3d245d8417721f55f643bd92592ee387b1aaec0f3280ae36027bb427ab5c00ff034660865047fb2904bf29f255e1769fe0e28a1
347285e9800bfda88da4a7b56629cf481c205d51ee9b50462b67c01e5583caa68ea65564c4f6b233ac4db5156ae1123e860bf1cf6c9d182dee17c741
550001cc822ae2bd439ea3ca65cb4877f92e95b90b4bc52c56075f8959ea59c4fbb7b246b4e651610a69940b4d444b5566d269045b879524a9dc1b09
17d513c4049a4f1cd42b157b16620acca352291c82233000aba18a3829faaa86483fd0cc1af12ae95cc5360a3fc0451c05f1902bde20ac6597a80678
43ad92441418f4779beb05cff8ec7aefc3d3dde766440de8ffb3aadbac76d7437abd7199bbc1ef4f9f2e464833ea25473d7a34f5a227faea3febbc3a
a0ffc4f4e9eefaaed1a3bab98ece1a456d53a75351ae5133b58f1ea5f4c98bd64b1efa1b9f55ef9e9de77ec6fc4cf4f067cc39c307c8b124c8bb1dc5
2d52a9c27f4dac22ebe82114a2bdc1aa5a2bd41a365b3784691d414e74d823c2cc9d6ded6df780b9b5bdcd7c2391f5162c666b7292d56216e292c062
86e8def25378b6e6a597e8efa5976e312dfff1d62dfe23d34ae9fc22bf40749125d33d8825d7f2125ece2b7809dbc896b1e56ca39c8bae5218cfa4ec
ad03afd79e86b5a2502bad5643ad56e35239105c4c6fbe3cb13e28637a230df60e9dd1d6dc490225b425b5b75d6e4b240c66f4664783304814328744
59a4144fb225ca1ec5d904fe02cb79974de8dc5327968c6b18d771a58e1890bdc409a4b103767ae3c27b456098c3228960912431cdfcb2e53963ad6d
b348710b669dc0748e5033aa22cd9d13ebed1913eb43321e9d586fcb78942441ffa9a1339a2fb79d3a65b10eeb96a65d91466d96be514bdfb07a87d9
123a8c64f3263d224e93a6a9978bcba5251115e16a8aea70b11799d7b10896a816f72a8958e45803e5e16b7aad8958e3d80ffb232c9990e921255206
c390fb58caa0d8e8de2a75ca7d2c3949b4db546a15502a39dd3989604cf63db8affc894b4f2ebf3cfd4b661bfd68386fafabab5bca360fcfdf3e7e69
75da0317ee49faf2edc75e2d8ae45f93f63564ef12d2be0f147907823d5857ae7595bb836bedc65aed5695a3d6bd357ab36a83fd95f8104730a02ddc
11eb363bd0e6d2aae2651042327af4d72afa1300ed6da42521606e6b6d6f6d337f71c3acdc844a22f36ab39d3e97cf9d1d2542267332bb4d8cea1d1b
97e22445069356fd584aa070977a98baf915fe1efff2f1b3f332cee59f3cdbf8eaa163db76bef2fcd493c525e7677cc10cbf448fab79d3c77ff578ce
dc93545db576dbdea54525a531b147ddee3f1c597140f670cae5e21ef2298132fe6a6f2433a211108d69807a75adc470b5961974e050694483c9fcd1
c47a3d2966541433c88a5d1ed9dc966491edda7a79645b12e9a218563c4fc63d2f9bb4af1efac23898419bc4527806d421ac1fc4b27e38984d660f19
1e324e63b96c315b8eeb98914ca96551986c49b6475ba22d5129a8e202e329fcca95f35d8f4b9ece6b78b133793faf655967c8423bc942d92479243c
ee8d167ba92de5e6c85eb56a5badb9d228d4c26ae306f51e67a883e9d0013ab3ca69ee6477dac52c8bdf1d2d66395ac844e6e61b7200cb114ce6e1cd
01eb04937f5964ccc16e83bbcc225be3630cefaaed3fbd7f078be197f9b78f9fc99b796afeebefbefbfac32f674857eaf896a0207ee32fdff1efddee
967b128fd5d41c8b8925e9f713ee7924bd0a667bc3248b80025a448a2c892447099948c755b5b9f342b345f699843b22868880955599fe261d49bdc4
8c7e8e90369621436778add305a6c25ed230699c3407eba15ea5265c490516cda2f6e3a9aecf2e31de952c5d99d6b15aea277bc06ac2b1823c3d0c62
60ba375aed0a67e5105eab7b55ac85ca1057ad7973c8068fdae1880a7642efde0e638487321d41d493ebbee0df77bbb137a439fced5ea7224e394e45
beed6c76a9ebac4dd6afac98c93287289859834d8cd260ca204876cbc990c0643d40da6d215727d54c3c772968f891059ff25bccfc194366e187f9e7
936ad87deb77ed5a4fe46a88896546669df6180bfafa0b16a2a4c95dfc51a7b0fdf8ee974f9c7879f77159a7f37402ac27743574d61fe8b5c336ed6a
b6cdac11cc3a90c28d49e0d08a56255f13ba0abe724a3a9215cc5826b374cbe68952def18c6d6d6729ccc5aff2169ec676b123ac9ae7f174ee93126e
2d65616c20ebcf42f7f2ed7c157f8a5793396875711ead2e81c76b10b6c16a91397038882a5a524ec8ad6d895e6da23a5dbd0a578922cb0c56f2f0f9
77853f763e215d91b36f4003f21fe597ee616f8a56a3469dca82224a1651c4343af0d851b46fd3dab61957eb454985162d38424c922e3c5cb4a4da74
0e8348d9b82da9b3594e3b96809623e58dc03a4cbeef742725fb1ef13a497b6fdaf2602681c42441856ad10e766613423054f480877984588c53c5aa
6335b15ab773301b2c8c6163843c69b1b8585a1abc5eb55efdbcea79b52b534961a1c1d13890f5637224bb43285b11b0015be3c6fb4befbbf8e15b13
9e7df2a377d939069d4f7755f22ddbb66d119a42363dc5f35859f5acae4ae9ca077fda785c78a8eb46c5d34faf93ad5a4a9e3a80ce343af04013ed49
2e7da8d604af85aa1a4d1677b9ebb8a331bac1b221d400a11866d46af42ed4d846c7120a172eb725250502a9b9b5bd939cf61dd9650916d96b0b1223
139d89ae4477625462efd4386fa4d7e97579ddde286feff4c87467ba2bdd9d1e95de3b3dae286e5d6485b3c255e1ae885ad77b535c6ddccd3867cfd4
9e493d13b29c59ae2c77565491b3c855e42e8a5ae55ce55ae55e151546d1c014772738ee65432cd1297250c4a60c1a9c1c7567860f114e7e727075e1
0b8d0d0da94deb0fb674dd62c2beed59c732724ecefc8f9b42726ee9ac920f8fc64fea5a5d97eb3bbdfbcd53d6b267070eac8b8beb94b15a4858cd54
d928b53960a837bc5723986c8d926683a9816d27738246186bb1ea47472a5190942487406b7bb39c01138f653957396b9da8c442b778749801128929
5207a4c4dd0d0dc30faf68f183bf65c5e1aeb3fbb66cd9bf7fcb967d784c78fcef6dfbb37d6c14d3d03dcac7ed2dd7afb71075cb554636b44104edac
31e45fda72cd7ac9fe1a931a0dec4458a3b5c1b0c111611734760d4c14ac41a31d8a88cdcae9414e38adcabed91ec838f1a9914591b591ef45de8c94
522195a50aa9f6d408a9bf3a4193a0edaf2b8442562814da0b23b4990b491f7b9492ba87d8492725ba29a9ab07321973b1acf388e1e21bf3cece9afd
de7cdececfb2f8cecf98ba417875fd8e4693f0f8cc9367070d3ad4b73f1bca742c983dc03f6ede7ef4d04ef95c462947cc26ac836190371cb58026a6
aa30591a0cdb754cd0c064d917c7d8e468940f880923db08688b357418651cbb9271a22d811448856465a70911b31b56acd876b0b131edd78b4fbf23
ece97a4cd8b96be7c93d5d152a5bd7ce9cec6f651b9fa6c597d1ba481b403faf4975523c0c4d82c434228cd1983b693f96934d27651bbd59ebd5a66b
b3b4455a49c937d1f2267bba812e31eb56adcaf695acc74ffc7abf01db05a68131a23970844cf41acd92574a97b2a422e9a6a40a3021062adbdfdbba
315047925d7bc34c6faccaaa0d0b0255a4da6ea8887463434453b8590d96208d46956ed104a53bc234bdc6442be9a9b3b32d702e1c39b2b55d394cc8
c078831363d2638a6236c5d4d2fd56cc2731fe182d21a56063bf13af9f80b307808b1f7d6acdaf4e36162faedadb58bc74e3dec6c6d4fa65cb0f60e5
8a25df7f26c3f8728d0ca3b073f78b6fbdd25521661d9a336b454fc42c260de413fe043ae1375aa1d1d0209ff0ad410fa3d53efa67277c6f746a7829
94aacad4659a326d99ae4c5f6a28339699ca82cacc6596526b6df8cd70cb1db14e5173d70f8192e70e1ed8b6f5e0c1ad379995dfb8f91dff9659f093
ebe7cf5ffff2dcd9af6af839dec6bf21371c46de66634349c2e37c9ab887249463fa3e6f444f4c379836b037b12992e279ac12d96394cc9f1490b5b5
27acbdda405c7feaa42dc7731b319245a0f4c37a366192949534360e3f5c7a01fcfe0ba58785a114d9fb64dadf7548a5abcbf6f126fe37ba9b7ceceb
9ec00e78004e20e92c90e8b5a9f4646f3d56981ab44d6a9d4a039a31b4e136073c92a2f9f205397c8fa607ef0a1608232504efb062284e708def5fb3
8fe438be2e78a0038f5a2d2d27bb8e90b572674bf2d7b942ff353c4babc5c175ef48a34130e9a7ba9c1aada0d64d75b99c693abdd345bb5739ab146d
e5f6cab0468bd8e8a1cda18f53a77745a8614a84c6a4d6d87a8fee234b75b9ad95c451f648c5b8f289e6fb1bb2232abba449de21d5a6ee7d12e2e47d
32dfa173e81d86819464faebfb1b466847e846e84718f46e70b318a18fae8fbe6f70822dc1de37a48fb38f2bde1d1f151357ae2bd7971bca8df2b712
26082a9d4a8f0634a20983d08ce1d80b23d021466ae312e253e39f882f8b5f15bf29be36fe667c186da90b993d80922dc4a5fc425045df79144d200c
e52d2429049f9dbc7f6665e5ace7529b5ffdf14f33cf2cc87dc7b76643ce01ef81e73ffd5dee5131f5509f3e1919def151a6be2f54d61c8b8e3e9992
32e3e189e99ea0986d6b761e74cab61c4201f15769276534da3b4c9226085f030b6bd254e8f48431f998d96a9233dac866fa4bea3e41057e0a5056fb
5520abc9c7665bc80866977737da47922d6c292be5eb2696bcf9e695dd1515d24efe7655576de5e41dbbde17b2aad87d72141ea2289c4e760da61d62
84d7f15336dda0634db60603e5529b7e3265d53176391c87053caa35a927a51e2db49fb2934b0553260d84dded3d2c961d9293eaeb0d0d0f1c5e7cfa
1cfb3d3b2eecedf2edda75728f507aabf660eeec9bb83ff0cd4f98f1427be9ad822782467e0f2e8df261ed0f4f995a7bde3f7ed039c934432b7f59d6
dcfe0e47f3d4f93c12c0c47ffca0e361d38c7ff86ee7102f2adfbb40a8237a964e78a150417495a89aa886289b6827d17ea2d5aa78382f2e86f35215
bdaf43a9648385621b2c94ac705c1806a76552b7c071f1736ab3c1719c40e57e5088d13084c61cba63ed3458067f630bd879a18f902ffc11437028be
4d48ef914cd204a958da247da74a536d569d532f551f55ff10d0071c9801fd200f0c641933bc206b2bda85107acbdfc9d43053fe0a2a6a6970a2f2bd
512e3308a15aa02c80868de92ee31dede21d6509c2d8e4eeb20a6c2c171e80422822798be947e31c5a7d11b8e977f86c88a7771224d29d4ca55934c2
4d7acda5fe12a262c8011fe4437f6a1d0f05347e2095ee870574bb61ca6d5e254a2d87de393467093db369a4ee5f5875f0ed553368a525b496fc35ac
8046cb72f868ceff6cc551549a47f3a6c1621a319bc6fa146e39ca0c9fa2919bb814d0b388c6cc22be73699c9be617d2ea3ea5efe77ca62a5c4a4822
3afdc07c6a95572da1b1850aa7245a3b1952ee9ad53327f0bf10f03fa57c73ffc7cba1f885fc1f11393eed10427b6518f4a2f6384850788ea51ffb13
61123c04e9a4f3547884f8ff02a6033408abbcfe5b1c3b6cf8770ffe2d097facc61f4cf83dc7768effe1c1bf9af0bb6abce9c16f9fb95ffa96e38d6a
fca61adb3af0eb0efc0bc7af86e39769789de31749f879eb54e9f36a6ca581ad53f1da6709d2b50efc2c01af72fc94e32749f8ef36fcb81a3fe2f867
2bfedb4afcf004fe89e30734fc839578e5f258e9ca4abc3c162fbd1f215de2f87e04fe81e37b1c7fcff1771c2f56e38516a77481638b13df4dc2f31c
df596791de71e06f43b099e3198e6f733ccdf114c7b7389ee4f826c7268e27381eb76063b9476ae4d8f0c609a981e31bc732a5374ee01babc463bff1
48c732bd7e3ce6157fe3c1a31c7f5d8d47381ee658cff1571c0f65e3eb263c78c0231dccc6037556e98007ebacf81a09fd5a07eee7b88fe35e8eaf5a
710fc757769ba4579270b7095fcec65a1a525b8dbb38ee7cc94039155f3260cd8be1524d36beb8c32cbd188e3bccf8820e9fe7b8bdda286de7586dc4
6d34695b353eb7d5243dd707b79a704b076ede7442dacc715355a6b4e9046e5a2556fdd2235565629557fca5073772dcf0ec406903c76707e233a4e6
33f763e57abd5469c3f5b4d153434536961352e51e5c67c1b51c9f5e63919ee6b8c682ab39aee258c6d1eb7f6ae54ae9298e2b57e28a6c2ccdb04ba5
1e5cce7119c7274db8d4804b74b898e3a20e2ce9c0e20e5cd881451c0b3916705c1085f339ceb3a449f3a6e25c8e792b710e557239e670cce6389be3
2c8ebee198d5818f1b3093e3a31c67729c315d27cde8c0e93afc4548b8f48b249cc6f1115af99134ccb0e3546696a686e1141b3e3c21587a9863ba1e
1fe238f941b33499e383669cc47122f54ce43861bc599a108ce3238dd278338e33e2588e63aa7174358ee2f88030407aa003d34ee0fd13d1cb3195e3
7df75aa5fb6c78efc820e95e2b8e1c6194467afd4138c288c3390ee33874884d1ada8143069ba521361c9ca297069b31458f839c986cc4a47bf45212
c77bf49898a097128d98a0c78103b4d240330ed062ff24ecd7d723f5cbc6bef156a9af07e3add827ce23f5b91fe33c18ebd14bb141e8d1630cc7688e
bd83308af48cb2a23b1b5d1de824159cd91869440721e8e018d181bdd2309c2ae11cc3b23194900ae518429342c2d1ced1c63198a3950658395a4857
4b1a9a576250369a381a0d219291a381461b4250cf5167462d470d0dd37054db50958d22758ae40176a456e428505d1880cc8cc09135b0ec751b59bf
ff0f17fc5f0bf05f5e91ff09988e668a
>}
\immediate\pdfobj useobjnum \csname EF6O6\endcsname {[ 32 [ 318 ] 48 [ 636 636 636 636 636 636 ] 55 [ 636 ] 57 [ 636 ] 68 [ 770 ] 76 [ 557 ] 82 [ 695 ] 97 [ 613 ] 100 [ 635 615 ] 104 [ 634 278 ] 108 [ 278 974 ] 111 [ 612 635 ] 114 [ 411 521 392 634 ] ]}
\immediate\pdfobj useobjnum \csname EF6O7\endcsname stream attr{ /Filter [/ASCIIHexDecode /FlateDecode]}{
789c6360a0103013906761606560636067e060e004f2b880981bab3a1e343e2f129b0faf0dfc60528041104c0b310883691106512029c6200e242518
2419a418a4013f190178
>}
\immediate\pdfobj useobjnum \csname EF6O8\endcsname stream attr{ /Filter [/ASCIIHexDecode /FlateDecode]}{
789c5d523d6f833014dcf9151ed3212240128a8490aa7461e8874a3ba10cc47e4448c558860cfcfbda3e03522dc1e9ee7cf683f7c24bf95aca6e62e1
a71e7845136b3b29348dc343736237ba77328862263a3e79e6debc6f54109a70358f13f5a56c8720cf59f865cc71d233dbbd88e1464f01632cfcd082
7427ef6cf773a920550fa57ea92739b14350144c506b8e7b6bd47bd3130b5d785f0ae377d3bc37b16dc7f7ac88c58e4728890f8246d570d28dbc5390
1fcc2a58de9a550424c53f3f4a10bbb5ebfed8ee07d4c0ab9513c8c909f24223400c4800c7652b922968ea93a99733c8999733c847842dd440c81c32
f732877cc2c5166aa093cf28cb420d848cb3cffe1b360af719d417b451b828e02cbcbb52b8ada3a9ff691b756e8a2a537fef42f1b7d2e3625e6d8396
4ed85ed9c15a07813fb43633e0a6cf35dfb6bd93b40ea81a944dd9e70fc12bbf59
>}
\immediate\pdfobj useobjnum \csname EF6O9\endcsname {<< /M0 \csname EF6O10\endcsname\space 0 R /M1 \csname EF6O11\endcsname\space 0 R >>}
\immediate\pdfobj useobjnum \csname EF6O10\endcsname stream attr{/Type /XObject /Subtype /Form /BBox [ -6.5 -6.5 6.5 6.5 ] /Filter [/ASCIIHexDecode /FlateDecode]}{
789c6d90390e03210c457b9f820b7c648b00433b65ae4113459afbb713430011a7b1bcfde745dc9bd83de963203eba8bd887920f8e3d669f73095983
f09012c589e7c42969a63b95660da35da598a4066a7d0a9cc9af68693a6f715765b66e44b595f63930225830cc706cbb616c8cfd10d84bf1e71fb05f
c32f1b7df917d14937a3cb48fa
>}
\immediate\pdfobj useobjnum \csname EF6O11\endcsname stream attr{/Type /XObject /Subtype /Form /BBox [ -6.5 -6.5 6.5 6.5 ] /Filter [/ASCIIHexDecode /FlateDecode]}{
789c6d90390e03210c457b9f820b7c648b00433b65ae4113459afbb713430011a7b1bcfde745dc9bd83de963203eba8bd887920f8e3d669f73095983
f09012c589e7c42969a63b95660da35da598a4066a7d0a9cc9af68693a6f715765b66e44b595f63930225830cc706cbb616c8cfd10d84bf1e71fb05f
c32f1b7df917d14937a3cb48fa
>}
\immediate\pdfobj useobjnum \csname EF6O12\endcsname {<< /A1 << /Type /ExtGState /CA 0 /ca 1 >> /A2 << /Type /ExtGState /CA 1 /ca 1 >> /A3 << /Type /ExtGState /CA 0.2 /ca 1 >> >>}
\immediate\pdfobj useobjnum \csname EF6O13\endcsname {<<  >>}
\immediate\pdfobj useobjnum \csname EF6O14\endcsname {<<  >>}
\immediate\pdfobj useobjnum \csname EF6O15\endcsname stream attr{/Type /XObject /Subtype /Form /FormType 1 /BBox [0 0 457.67952 191.27952] /Resources << /Font \csname EF6O1\endcsname\space 0 R /XObject \csname EF6O9\endcsname\space 0 R /ExtGState \csname EF6O12\endcsname\space 0 R /Pattern \csname EF6O13\endcsname\space 0 R /Shading \csname EF6O14\endcsname\space 0 R /ProcSet [ /PDF /Text /ImageB /ImageC /ImageI ] >> /Filter [/ASCIIHexDecode /FlateDecode]}{
78daed594d6f1b3710ddf3fe0a1ee34368cef06386c7b86e0d04edc18e811c8a1e0a477112d8756db7f5dfefe3da92c88db5d2a280dc16d15a80f944
ce721ee773f7d69071b8c8bc76e5efe2ba3f3c5efcf5f96271767264be7b578f2eee7b325ff0bdc4822ff83e60d909be977d1171dd872836498e8cd1
5535a24c9687ffae30b31e7deafb8ffde11b88b8c79a93bef7d9fa9cbd44e393f591e320966df265c112bbaa3012b2446110b75e5da3c34d6ecdd7a2
0379cbc1619a211fac24766cee16e6bdf9cd1cbee1615336fae5077a3feaffd0c227a359fd6dbf2435898d12d4e7e0494d601b445d481e17a836873f
39737cb39a4d6a457d0a9c23b6e4becdf836638f334efb5333df4b7cf112366f8b6bc00fd43cf4cea6ec2827271ae01c6d04896a49bd8337ac1d7c8d
d56edbbf1b36b474c447ef7bbc45891505d922cdb3d5c770d11f15b77d72ccd7858718ad737057170a0f4c56a24269c4bfa3f3fef00710e5ccf9c721
b09d7fe87f36af3a77607e31e76ffbefcff74415b9682930711d0c2b702e59dbe44dd1458e6c4e1c3c36eaf24e7cc5bdf3859c807d326270a5df1a9c
cdd71679937c85682590461116bf135ff4021696b375c9f9502bb8c266f335296d92adcc569ccb29122133eec6d6deed8b43b2f831843ada54e05cbe
b6c99b628ca16c4ace09ef16bc78ffc6c538d42c1a631d6e2a70365d5be44dd2a5c926149e925553d891b1bd1b98c70d203e25ae34acc0b98c6d9337
c51852234a576545fc62d98931bf7f1bf33857cd49a4567085cde66b52da245b5a5a00419d9f62a61dd9dabb7d059c295c451b052b702e5fdbe44d31
16586c40c0778c0a6cb718165afbaadd9b91949fa2a1583f543395201a093aee16ddefdd1fdda7ce74d77b37d995fc553bd8f4b86b7447a37d569e8f
b0da11fa6c6044b6416e569b9470fc4f8de9e6acfb726ca9b79e602edcb0b546e7b255cb5bb3d5c8dbc096641b492908c1d0265dfce5d8220e9622aa
2f6de8aae0b97c3512d784b512373046eca06e24463a519da22cbf20652958940759424bd91a9e4d592db1a2ac91f84c310c6774a9347809f138b3cb
0ab70cd3c58ad4c44163c85c6f24291843af9bcb3ee46b74d846b56f87a25d1d6b6ce657f07841404d41c289da056bb8527478d858b4c5392e77924b
b3e4793a6cff78800c658786bd5c11d0a2fbb5bb47243788e78beeaefbdcdd741ffea909716d423418100544c714c4251cb6af07411922734024381b
db56f528a28ecd92d04ae3434db4c96818a3f3f0961a453a2ecf167390d6d3c8974912f308aebadd0617d8a8069534c2b5b4118a618bb3530b17a034
0a198cac0d7dd58dacb8ee821a1c0375c5c45a182a464d52b659e3deb1f5e27cd96683a33ca088031bc9a96be30607f1c2e4c7dbf42236ba90c7dbf4
3959562da55b83078a1645059766b8c1ab9a69141e1cec66477bb9bbdc3cb31fcdac9e292f8d4a862778cf3c53a6cdcffb02ccc0a72c11d7f6d96eef
f2fe2d33feaf7afd37679cd6499707075b26dc4dc9bc6ae09e7d27044fddf555d3d2a99f92fdf0fa67f592e72963239d681cda0f428879bc4995c978
94c9ce0e4c0896c5e9d33ba357dd9fdd153298e972e74a5752f2d880dd60d4d6917d7ffa37bc2ba963
>}
\expandafter\gdef\csname EFWidth6\endcsname{457.67952}
\expandafter\gdef\csname EFHeight6\endcsname{191.27952}
\pdfobj reserveobjnum
\expandafter\xdef\csname EF7O1\endcsname{\the\pdflastobj}
\pdfobj reserveobjnum
\expandafter\xdef\csname EF7O2\endcsname{\the\pdflastobj}
\pdfobj reserveobjnum
\expandafter\xdef\csname EF7O3\endcsname{\the\pdflastobj}
\pdfobj reserveobjnum
\expandafter\xdef\csname EF7O4\endcsname{\the\pdflastobj}
\pdfobj reserveobjnum
\expandafter\xdef\csname EF7O5\endcsname{\the\pdflastobj}
\pdfobj reserveobjnum
\expandafter\xdef\csname EF7O6\endcsname{\the\pdflastobj}
\pdfobj reserveobjnum
\expandafter\xdef\csname EF7O7\endcsname{\the\pdflastobj}
\pdfobj reserveobjnum
\expandafter\xdef\csname EF7O8\endcsname{\the\pdflastobj}
\pdfobj reserveobjnum
\expandafter\xdef\csname EF7O9\endcsname{\the\pdflastobj}
\pdfobj reserveobjnum
\expandafter\xdef\csname EF7O10\endcsname{\the\pdflastobj}
\pdfobj reserveobjnum
\expandafter\xdef\csname EF7O11\endcsname{\the\pdflastobj}
\pdfobj reserveobjnum
\expandafter\xdef\csname EF7O12\endcsname{\the\pdflastobj}
\pdfobj reserveobjnum
\expandafter\xdef\csname EF7O13\endcsname{\the\pdflastobj}
\pdfobj reserveobjnum
\expandafter\xdef\csname EF7O14\endcsname{\the\pdflastobj}
\pdfobj reserveobjnum
\expandafter\xdef\csname EF7O15\endcsname{\the\pdflastobj}
\pdfobj reserveobjnum
\expandafter\xdef\csname EF7O16\endcsname{\the\pdflastobj}
\pdfobj reserveobjnum
\expandafter\xdef\csname EF7O17\endcsname{\the\pdflastobj}
\pdfobj reserveobjnum
\expandafter\xdef\csname EF7O18\endcsname{\the\pdflastobj}
\pdfobj reserveobjnum
\expandafter\xdef\csname EF7O19\endcsname{\the\pdflastobj}
\pdfobj reserveobjnum
\expandafter\xdef\csname EF7O20\endcsname{\the\pdflastobj}
\pdfobj reserveobjnum
\expandafter\xdef\csname EF7O21\endcsname{\the\pdflastobj}
\pdfobj reserveobjnum
\expandafter\xdef\csname EF7O22\endcsname{\the\pdflastobj}
\pdfobj reserveobjnum
\expandafter\xdef\csname EF7O23\endcsname{\the\pdflastobj}
\pdfobj reserveobjnum
\expandafter\xdef\csname EF7O24\endcsname{\the\pdflastobj}
\pdfobj reserveobjnum
\expandafter\xdef\csname EF7O25\endcsname{\the\pdflastobj}
\immediate\pdfobj useobjnum \csname EF7O1\endcsname {<< /F2 \csname EF7O2\endcsname\space 0 R /F1 \csname EF7O9\endcsname\space 0 R >>}
\immediate\pdfobj useobjnum \csname EF7O2\endcsname {<< /Type /Font /Subtype /Type0 /BaseFont /GCWXDV+DejaVuSans-Oblique /Encoding /Identity-H /DescendantFonts [ \csname EF7O3\endcsname\space 0 R ] /ToUnicode \csname EF7O8\endcsname\space 0 R >>}
\immediate\pdfobj useobjnum \csname EF7O3\endcsname {<< /Type /Font /Subtype /CIDFontType2 /BaseFont /GCWXDV+DejaVuSans-Oblique /CIDSystemInfo << /Registry <41646f6265> /Ordering <4964656e74697479> /Supplement 0 >> /FontDescriptor \csname EF7O4\endcsname\space 0 R /W \csname EF7O6\endcsname\space 0 R /CIDToGIDMap \csname EF7O7\endcsname\space 0 R >>}
\immediate\pdfobj useobjnum \csname EF7O4\endcsname {<< /Type /FontDescriptor /FontName /GCWXDV+DejaVuSans-Oblique /Flags 96 /FontBBox [ -1016 -351 1660 1068 ] /Ascent 929 /Descent -236 /CapHeight 0 /XHeight 0 /ItalicAngle 0 /StemV 0 /FontFile2 \csname EF7O5\endcsname\space 0 R /MaxWidth 989 >>}
\immediate\pdfobj useobjnum \csname EF7O5\endcsname stream attr{/Length1 3804 /Filter [/ASCIIHexDecode /FlateDecode]}{
789cb5567b7054d519ffcefdeeb9bbb9fbc826d93c37090b9b0dc4640137124dc8c0027941c04413624056b2c96e368164b36c2224c4204a10915a02
d215100cb62922be526b9db4e9cb2aa38e655aab4cc7e938ead469c7199cb1333eda389cf4bb3721225367ea1f9e73bf73bedf77bed7b9e7dc732e30
0048a246065b75456515a48002c0b2499a565d5fd790f774e12f08df443852ddb071f5c2b195af101e231caf6b58e2dd5614f90d80544eb8a9ad3b10
85316992f031c25d6d3bfb9ca747c7f3019020eb6a8f86bbffb9ed5faf52b02e121c0a077aa360a00af2bf099bc35d03edae97433f02e032f9f84747
28103434fd2119c07099c64b3a48607e8a5f0030e611ceebe8eeebcf7c89f512ae219cd1d5d3162044b9195b085bbb03fd51dcae24108e12764602dd
a1c5256b628447c8bf2ddad3db37fd278803a867687c7134168a1e3d222d20fc67cac90cdabb31c34c910821705da6910a8ba11ca48aaaf58d60ed0a
f4452083de2195e96980394ed7de1e8a45c0386bc7684cd27b23f5765d330b72094b7a1c768d86aca3abf18ec231b0ebf17607628156180ec4ba2330
dc1a0b74c2705b20d24b6d472846ed40ac0b86c3a11ee2c3b1d07618ee084448a723d44a92ed81480086bb023d4eada5bc87bb037d1d301cd9ae497a
c2816e188edd1d21cdbef64898da0ecdff3573fbba30d9ca46e89d002fe627a084e56afd7425fe15da255a35c9a420a26c92e4d939cc95faf6ca2049
ee806ec52eececa4a19bfdfd1b3a384bd9facc012a08311dcbb0887a457f3792e641cf49efa79f9d3e3ee76366dd16cdad9ea633435a36ddb344b933
175b0213f026d1cb701e5e935e83337033d5dbe1b7eca0e4a191b37000dee512bc08a75829b3b3521a7d5bb12b037c3f3f47e35bc8b686bcbc0b8f90
4fcdd304dc2b0d4af5d00eaff18b70926a8f2eff147ec5f6c125380e6f4a35f005ecc346384cf524f4cac02f31158454483b982251056825ca410fbf
a4d74fe15e188446185326143b7b5bcffa2c7b855d868f29e7b7710bee208b53708eed975df239b9060ecfe48b2d7058dac74eca2d7a1da417b20b4e
c92decbc62870b5aae24a9a74cdbe1d744bbe0225bcef6e341ca6c50cb805f828b8675f29299ac0c43b88ce60344cfc10be0c138d9eb7351dae194d4
4eb1bea04c2e6205149037fa3eed745a40ea4b0a97516250e4b48d4beeb5c171df6dcdced737cdf7145d079d3683731ceac72d03ce89e9e9fa66d9c1
378df3ec71741bc765b7ebc36f1bfcd053545bdfec1cff7965c5acd7ca960a92353413ab211293bcb2c233b327e80ba4dda3ed2ced4c29a0599a99c9
b796a3c213d40419d50444d5a44ac824934955b8c168e05ce646a341e22899499bf68fb94ce5122a085526a3d9a426186766693280c5f6ce1fd34b6f
8425e51f79bde9a54b290f838d7f62b01967897fddcfb09b16bc603031e6f7dd9fa02e920a78beba5c5ace8bd55a5ea53e24dda7ee5547a5137c9427
2a92010daac19421a51b33cd0b5981942717f042e506c322a327215fbdc15c22dfcc97194bd412d32de65aace66b8c356ab5a9c9d8640a4b9dd8a184
4d61f300f62bc3d2017c507ec0b83f61587d543a818ff1f3788e3f69bcd10ffe94e214a63dcc85cc75e1d51aa9f8c37592e7e2950e5179e52dc12f7d
a5ca5f68345528bff7954bfb7ee91c36aca4fd55c8527dbdf694e4241b4f4de4668bd9ca2d16b3d140ef8627a0768a4a9c31c8e4598eacb474293583
3bb2e6e54ace1ceab3b9c391559691684639c705c85213d6e73b773bf6645a586656764672a2953b2c28b914286066578ec3b2d0c5162abcc0919557
64fbec85d12ce67ff5c5a8a3c521f9dff159a60bce501d29e829a82b50fc972fbff3e2f3ae1697e4bf9c544a35597b88d209680b55fed167e597bd49
e9a5d72c17ff84f06cab2fd90c5d0389a5955beaa1957bc9e769f1443de8f7a54c7bdef7fcce7386ea88a7c753efa9f324f819ba5256b06537e5bb16
2886ebf862c30a56ec4d4bd7db54bb6258b26fe23ea331b8f981bfcc1b9cd8475cf3831af7b70957ed1bbdbb4fdb1adf1b1c396d94c6ae6c911e2f5c
91b639f8c613570e4b8fbb5766de19d258b9e5b9d6f0ded8ce7bce1e9f5f07b3e7abb4e9c48187cf976c4d2cff1ce619f5c3f1ad3dd6f7aff65f7e70
a5d852653ca47f1d73e73cad6ab7c801b0ecfff283fffcd2520541fa4fb8b628b2767668eecf131d22fd8fe182e14d6d3fcc9531eda6d67350f05628
840e3aa325b0814ffb0340ab64a188dad96c80cdda3729d3bdcd96eaf780c633482334c34b606555b33c423e6b9ce5e56b780e196cf72caf401e3b0a
6be8bc8ac200c4a013c214bd0f9c743fb4d119e5042f2ca55a4c5c2b69386135e9f4d1c9d547da2108d05d5144d2b51021fdc5c4ad822eaa4eba23ae
faead55188fa10d9eca436489aeaff11b5642e6a2345da49b1b6914d84b4b53c0264f3dd225610b78dec9ae06ed26823dd80ee2da45b04f41939c94b
84da28e9b492df4ed273927d0f450fe863d7fb69d0bdf442ddacfe0e9286be45c7799d56939e612fe11e3daa97f22c8665dfb0be6aebb9ce966eeee9
cf89f6d0aef85f45d1f794442b5d06f50013d25edff4cf048ebbf1792f3e17c767adf84caf953fe3c5a7059e77e353563ce7c627e378760a7f3a8563
027f52863f16f88417cf8c36f033711cddb08a8f36e0e35e3c6dc753717c4cc593024f24e3f1217c7412e3028f91c6b1217c44e0d123d5fce8101ea9
c691c30e3e22f0b0037f28f061813f107848e0430773f943020fe6e2835e3c2070380df709bc5fe07d02f70abc57e01e8143b56e3e14c47b040e26e1
ee8149be5be040bf9f0f4ce2c05eb97f979bf7fbb1df27ef72e34e8177c7b12f88bd568ced70f35810774493f90e374693b187d2ea99c2886f5a60b7
c02e81dbd3705b6719df16c44e8ad159861db79a78470686dbad3cecc5762b86821824b3601cdb04b606ccbc5560c08c2d5b33794b10b7de65e35b33
f12e1bfa55dc72a7856f1178a7053793c5e6386e6ab6f24d8bb0d98a774c61d3c649de247063a39f6f9cc48d7be5c606376ff463a34f6e70e3ed026f
ab5fcc6f1358bf18eb28893a3bde6ac20d94d58655b89ebaf5026bd725f15a37ae4bc2b5026baa93788dc0ea24ac125829b042e09ad5437c8dc0d543
b84aa06f0a574ee18a292c2f59cdcb052e7f1dcb882b6bc052e18be22d437833c112d9c34b56e3328137092c2e43ef142e35e312811e8145020b69b8
f046bcc1860568e3052e5c948b0bf3ad7c6110f3ade8662a777b31cf9cc1f386d0c5cbb84be002420b26713ee9cf77a0739e893b1391fe447eef3b29
cf33616e02e6fae41c1b66937a761c1d71ccca74f3ac20666624f34c376624637a9a9ba7afc23437a60ab40b4c99c2e4a44c9e2c3089bc2665a24d60
a2402b79b0c6d142012d43683699b939034d6654051a69c818a79f0e1b5704729a052f439990ec41b4d1df07fd9264205391f964c84636c182fb1f66
85df6f81efd9ff772d39ff0526de186f
>}
\immediate\pdfobj useobjnum \csname EF7O6\endcsname {[ 87 [ 989 ] 109 [ 974 ] ]}
\immediate\pdfobj useobjnum \csname EF7O7\endcsname stream attr{ /Filter [/ASCIIHexDecode /FlateDecode]}{
789c63601852808504b5ac000195000a
>}
\immediate\pdfobj useobjnum \csname EF7O8\endcsname stream attr{ /Filter [/ASCIIHexDecode /FlateDecode]}{
789c5d50b16ac33010ddf515372643506268b3184349160f4d439d4e21832c9d8ca096842c0ffefb9ca4d6851e488f7b77ef783c7e6acfad3511f835
38d961046dac0a38b93948841e0763d9a1026564fce9f22f47e1192771b74c11c7d66ac7ea1af8270da71816d8bc29d7e3960100ff080a83b1036cbe
4e5da1bad9fb6f1cd146d8b3a601859acebd0b7f112302cfe25dab686ee2b223d9dfc66df10855ee0fc592740a272f2406610764f59eaa815a53350c
adfa37af8aaad7ebfacb91d60bdc0b3e12fdaa329de05ef091cefd0ad3e514c36a5bce2190e39c55b69a4c1a8b6b9cdef9a44aef09357a78c7
>}
\immediate\pdfobj useobjnum \csname EF7O9\endcsname {<< /Type /Font /Subtype /Type0 /BaseFont /BMQQDV+DejaVuSans /Encoding /Identity-H /DescendantFonts [ \csname EF7O10\endcsname\space 0 R ] /ToUnicode \csname EF7O15\endcsname\space 0 R >>}
\immediate\pdfobj useobjnum \csname EF7O10\endcsname {<< /Type /Font /Subtype /CIDFontType2 /BaseFont /BMQQDV+DejaVuSans /CIDSystemInfo << /Registry <41646f6265> /Ordering <4964656e74697479> /Supplement 0 >> /FontDescriptor \csname EF7O11\endcsname\space 0 R /W \csname EF7O13\endcsname\space 0 R /CIDToGIDMap \csname EF7O14\endcsname\space 0 R >>}
\immediate\pdfobj useobjnum \csname EF7O11\endcsname {<< /Type /FontDescriptor /FontName /BMQQDV+DejaVuSans /Flags 32 /FontBBox [ -1021 -463 1794 1233 ] /Ascent 929 /Descent -236 /CapHeight 0 /XHeight 0 /ItalicAngle 0 /StemV 0 /FontFile2 \csname EF7O12\endcsname\space 0 R /MaxWidth 838 >>}
\immediate\pdfobj useobjnum \csname EF7O12\endcsname stream attr{/Length1 12980 /Filter [/ASCIIHexDecode /FlateDecode]}{
789cd57a797c54e5b9f0f39ee72cb36766329375924c3259085b624280b08e08914d8c8208283681104081040208460d4b4940a10490b08830454040
a4112924102948aa22d26a016fb9a54514c52522edc5aa2179739ff7cc04826dbfdefbc7f7fb7edf9c3ce75dcebb3cdbfb2c2707180038e92683f7de
c143f260284c01605da9d7736ffefda3933eeb82d41e0a20edbb77f443836a92777f0f8037e8f98dfbee1e33d4b7a4e0519afc058d997bffe88cace9
3f962c0550a90fc64e9e59588aebed3ba87d909eff7cf2fcb95e981e970b60e84d6d5e5c3a75e6ec1ef31f0730511b5e9d5a58560a1a5d60ba426dcb
d4190b8b2f9db70da0f67580b8afa64d292c324c7cab12a0d3457ade731a7558b769db01d25dd44e9e3673ee822fbf8b0b505bacf7d18c92c9858d7f
78f35d80ce766acf9a59b8a0547e599d4ded45d4f6ce2a9c3925ed9bfe0dd4a639ec5c6949d9dc3ff9478e04e83a57cc2f9d33a5b4aff657aa76a33e
651a085e5920f893e842c8a4bebb2087eae29909ba433f9006e78d1c03b6198573674114f1957e6d6d00b76a62247b62ca9c596010351d645a419406
90d85fc448f683d419cce0a619172137b865db45fd3a1cc2403c09dddb9677ecf9f7bfb6c2b6c17a19bc57ebf727f4fbe150591d1a79a2ed35bd3cab
df77745863c79de5ff706781ede17f3fe6ffceafed44c7d5dbce1237cfdee6a368df7a76ea360ffe573b5c0ceed271c70e3f46327591542348333c10
0749e083144887cef444f4325d0f30345a06e5d64c559c0b5d63c4cf48bac6c00131bafe80ae8b7268861869d0479868370b58c1464fc2c04ee39d10
4efb0775752bec0397aeab4f15ce299c04d5857366ce82ea49730aa743f5e4c25965749f36650edd17ce9901d553a794507dea9c294f40f5b4c25934
66da9449d4f344e1ac42a89e5158e21577d2f99fcf2c9c3b0daa673d217a4aa616ce84ea39f366d1c8b9c5b3a6d27d9a58ff5f9c0b9d4733a64f2dbc
e36cc8103c1b0c7aeaa54254b9887f49d089ce1e121f13a1875e7a219b7891085974f7d2c994a88fd11517a2f963f802f2a19f05b43feb6c5b42db17
d0d93e1f6278c19dd26c6fb369a1fa43ff5e0318715b9affefc7c184e05851deaa7758e38efe091dda736ed725b273375753bddfada99d48ce51ba3d
8a873c98a56b94d00ac185f020af4168583731d6388fee5d694e27e32a7ade856a60ac01e8302ef3d6b88c5be3baffc33889d6a6718ad0de70f15411
1830d9c6aa850e2bd9ca26c23c3e58e27f40b1e4a4496615d1204b5250be1d7ef9c5438ac04f125ca0bab88b6dd666b24fef188321f080902ac0416a
31bd2dc3462a9368aeaceb403ad1d717fac3dd3098f8310a1e8442980e33a014e6c3025dffbc846f57a2f3ce319369cc2c9823c6b47d4ab6e13fdbfe
a3ed23b2101fb67dd076aaededb63ab294bf6e3bd8f646db813b71ff17bfa0dff820d4b2ebfb060175a991d30af1b36b085402c1fdbe04e2ecf70f81
e0f4dd21107e6f7008c85710f64188201815824882074310455018826882c904d309620866842091601641a9ce4700a16f429f53081684a033cb813a
384dd709d80b5bd82e6a1553ff6cea0948076019cca39e93ec345b2175a3be5d701dced2c82a388d7b6560c3e99c9ea6f117484f6eb03124bf2d2c97
b958aea6d2711f251f941f94ebe4abf219e82597c967e402b98c65e37665acb28b20177f4bfa730a12a08e5d823238825f623636c883651b5cc233b8
173ea35d042f4fc36ad801e5848b8b954085542e3d483def286760335d25f4fc0cdbcace127647d852380f1b519686c256769ee83a0d7f87a53846aa
2011654bc584ff3bb4d6199abf19cae8489d6726e05217ea23ec69af49fa3d0ebb29e7f5eb3a54d0ce6360875aa7ba341fed2238b68b9d644dea3a08
c0597c1467e39fd832d927ef9687c2ea2007b00056d3da9bc51cb5982d24dac5552e56979e940bd85ef8522ed026d1dabf151409ed971e248a8aa181
e049d54e34f565cb7005612a9ec6c1196db89c41f36905ed19a21aa00473e071aa95c37e3800ddb00656d34a3abd6a2fe5ef34738b7c99685ecd5649
7f87333898b4b358be46bc162a4727ffb0a62a324a0cba7aedb552cab0a25aff03e3bcef8e4fecd6f5274daf5df3d6427ead75a1b7aead2d7f9c1cab
8caf553cb59862a895537c97ffd5c3cbddba8ec81fe7ad6d1d3238b4ea9082c1d4377a1c55458bbaa97fc860fd99d8b45649a1bf6105b5dec9d3bccf
d99ff3f579ce3ea58f3074640928b2a2d32eac442dff9b54ae3ac92ff6f287a91b6183cdaa013a550837d9ec1747d4868f19570fa6b6e3bdc78fa80d
d3ebe0ef3dfe4a56932337f72eb0b73465325572bb9c91be5429a787b397545eb964e9b240cdfa1736a8cecff980ab5779dfcfbe666f7f7c893536d1
7e3b68bf127dbf047f9826f6d3280c70cae106a0fdfaddb8bd6e787684d3ed92345f4f674e0f69072db9be26b06ce952d5d9c4fb5dfa98f7f9fa33f6
dbab57d95bb4ead636273b099cac6bb4df825b61a98a328b862895163af77e70bd5ed96ef4855f3fbb63f1837c1f3fcefcc487aab64fe5d5a49b66b2
0b3e7fb81a7042c0b2c6b932cae8098b478f3b368a66de68a205aedc68b25fcb644992c3eeccce723aec525a1638ece04b1277e9f92d2fbd447f2fbd
749319f9f7376ff2ef9951c9e767f8fb046758365d3d58768097f14a5ec5cbd82ab6903dc556092b7e990cdd04f2b9c458bf7b10066429a02cd62060
3424a81e840466b69f0bf19e09de3735b61042194d59379ace356592068c4f6207c3304c9626f64a74283929d98e44772267c3f92636e53d36bc65c7
5eb96c68ddd0e6f37b6901d256793851ec81adfeb4e898588cf23814191c8a220fb2ffd2f18235e05a2393d502bb4962264fa41dd5387bcb885af798
11b511631e1951eb1af308618242231acf351d3fee70e686b0b9a163a3d9956f34e51b56ebb13b227309377fd643f25865acf694fc94323fb62a5a23
9b162dc790727be6c27c755e4c59ec5ccf12a88c5e12b32476896737ec8e754c84892944444e4fe83580e5f448f525a95ace00969d25bb5daaa60219
d2132d23898dd985f7bd52f9b3b30b9e3a37ee0be61af24834bfb177efde27d99a3e33370c7bb266d03defdf95f5c55b8fee2c8de35f13f55b48de65
447d2728f5770777b8a9d29850e90d0fb8ad01e33ad513f0aef3ad5157ba5f4e8ff08403baa23da95ebb075d0946355d3021624c3bfd469d7e620029
6ca4aeb14d576e5c69b27f7ecdae5fc4954ce63716c51726147a8b126598c8e299db252726a5a6e5c413213d89aa2e2c2758b9833c1cb8e665fe01ff
e2b1771e1ff3eecc63efd4efdc7f68fdd697378e3e36a7ecd4f8cf99e5179892d058fde7bfa5a49cbc2bab66f5cfd7ef7ab2b4ac3c39f5a0d7fbe181
a75f15e7bb88a4bc83744aa273b6d81fc7ac680544eb2040b31650182e36328b093caa41b6e8a7dc4c845975c22c82b073fd1a9bb21c42ae57cef56b
ca225a74c1caa748b8a784483b9b29481f0ae3c95d3e09cf8116c1ba402aeb823dd92876bfe57eeb5856cce6b1a77019b392288d2c11b31dd96e9fc3
e748cc41954b8ce7f0f3e74fb53ea6a4b47c8a675ab277f3002b38a99fe44fe522c23c0e1ef3fbe418cd51698f8b0968ae807d85550ac062eb4a6d47
7ca48799d00326bb1a6f6f611de562ef60a9ece2b49088ec8dd7c401162798c4c31b83d20927fd72089e83db0577884548e3cf18dd1ae83aae6b334b
e6e7f8b78f9d9c36e1f813afbdf7de6b0ffc728c727e2f5f1b16c6af7df557fe9dd77bfaaecc435bb61c4a4e256eaf26ec6b747b920ce3fcc9e12a58
2b2d108850039e889df6806545d21acfca144b92d1131d1feec1c484d8143230a44457741373a5e5ca6df5f1bb28a26067a43378463ead9c5689ee03
f1d244369125a96e57441057e6eece7c4912b613e2f30a7394981521ed58be6ddb7202661cf9e2c877cf86f53df0c465a6f0eb9ff0567e8de5b3d891
2f62df23db7f79f4e82fb71f9116d625a7f2bff16f1f9ec8bffdfa73fe956ea026b19df1c242ed266d9a46325161b23f4a714828a143267ba1903c50
41263350357bcbfb8dbaedcee8600708485d8480c6bd49619b9f16d3402319397af51eef778e93988a314aae3254998ab550ab6aa42d2418e66389bb
f178eb2767196fcd56ce8f6d5eac7411b1e1f3c4dfe775fefa2802bfc79f1245dc4d5303f1dd02ce35f12bd35ece8cb22477f6b8933d6146b2de64c2
c3126333ed2d8d4d371a9b74c6b69f55bd954b87b4033353ba93ad49cece8a1046463faebea4e49c1e3dc3db07906648cf57efdc595dbd6b27dfb964
0db4fde5125fb378edcbfcfbefbfe7dfef18ba66e99275eb962c5d23fd767355d5e6172bab368ff51e58f4c6071fbcb1e88037e9edd517bef8e2c2ea
b759e1dc254be61290c62c268aaa88a2285d637c5a4234ab84e88069a71c8015110901fb9a8895299ac793181e0f49491eabae30847ebb4ffa9c7fd7
ae2f118dd16fc51c8f3dee391ef7567c6382b6d7d9e0fcd289a431bd74dd7686db485720a7076407b5242995b593453cb83c72cb08d2933e07667ccc
6f32fb270c9983bfce3f1bb9850d08e952026909b332e7d84759d8d79fb308dd9d6de38fc44b1bda3549589feba43427651fc94b038fdfa62e9577c1
5289cc8e0c51067b0b5914e1316ee82ed927ecc2f5b3f4e35cf6719aadfb727db61152fce1d252d0b6ca4b61974933282c9a9630e94be86efd4a8bf0
cae1c16574ef7e56f8775aa8f542bb8fbfc00a943fe1f6103656d822a95b645a4a069f40a69190b19f13ae545804715118decc05e0f6ddd7f7d20a76
f1be4bb7a62678cb7f3738280d576449f84ce630a1091c12595793468e54159d46079a0ce201d95b6dbdb0b68ad140b1bc88108d8a89fc7963a43828
fdae9c6b6ab7adbad3bc5518be693f3ea23e3ee980d7c2d844ffa030162685696186301847594829ac04a3c60c928a463982454b63d93829df32954d
9316b0f9d2d338477e525b60a862cba545968dd226ac912383e698c2111f26a24f6ae0d7a4145efe9994fb87e5ad3f5b7e5eb1b546e3fee62eac822f
267a4f11e5b544b9011cd0ddef86f5c6c56cbddd20d94da0445bb3c063949d7a8444e74a3ffb42aa070ac20959e60869594aa25ea633b6ee06cb6109
fc323fcd07b16dec00abe1d3783e2f54326e3ec9a25877d69545eee21bf822fe2caf215341bbcb3edadd08e97ea7ba5e96d6c362f935129d46c65f26
35202d680c6e4a5a7020cc4cdbea665d80ef144e682d95f25b6bdf13067be8ded65e105c93f26eb1e6427fbce66012931c14510cd224845f19149569
9247eea985d66f0986a319a46fb9ed4ef01b01249204b19d3fb3a7d45b1b2addab4d978ab54592a632a3ea66316a1e1ba63eccc6a953d87475a1ba8c
3da7aea72c669bd9ae5b38628e836ebe53cc2ed534f2ebad8f372ae76f26c8979bbbc8976f2684384f9e86f25b0fbceecf21fd4193ea401915872ce3
204a0cdd28bbd71b5debad8bcdb2a2a2c3089e089b628a8e961d035d268f458e130c22d347018a23289d7e226474e68aaba389d6e3b403fe785dc59e
0a670a284c21a5d26437b8994b8ac048390552588a948a696aaa966a48357ae37bb29e521ecb93a629f3e479ca93e1cbd5e5da4675a39630510f7622
c37dd89d7561c2e77b8545259a83d60657dd5d3ee0cc85df0c7f7ec1c5f7d8bb0c5a96b6aee06bd7af5f2b3544543fcba7b18a9a49ad2b94f31ffd71
d511e9fed66b554b972e1376a59c6c6537cafd4c947d3750f49a608e34da604fa45a6f73782b138e78ea7d758e95911688c428abd1604e40836b482a
71e1fd734d59594149365eb9d142eaf2b6ee0b1c8211fe59997199f1990999deccc4cca48169fe387fbc3fc1eff527fa93f2e3f2e3f313f2bdf989f9
49f969a569cbe2aae2ab12aabc5589cb92aad30269d7d3e2dba7b64f6a9f50105f9050e02d482c8d2f4d28f596262e8a5f94b0c8bb2831aaa307efcf
7a397c39c22ca7929fc94eec180b4648c72eed5b5cb2a9beae6e60c3f27da75b6f32e9950d0587c64c3936e1bfae4bd9c5e593ca2e1c4c1fd9ba786f
71e189ed6f1e77563cdfbdfbdeb4b416a1e7b3db3ec5abc4ab6818e88f854ab65cb6555a979bea1d727d24312946735a61a86b488cbd8572b9902fe1
37aed9bfbb96e93787c5da6317c556c706621542560f314208f772eb3e241863e0d5512fe5bff1f6db6fe4bf34eabe9d135bf947ac1b531fda2ee7ec
ebd2e5d333673eedd2656f72321bc06cccc9faf8488284953c41759145f5406f7f744c3dd85cf58a61a5ad8e6d2025038374afc3691e12a7db942cdd
4d5c119ebbf15ae6a182f845f18178d42d4b8869948c01318a75f0cfb8bdaeaecfeb4f9f6e83b6d34fbfdeface2b6bd7eedebd76ed2b78487aecc7a6
dd45856c3033d035b890bb4f5fbd7a9a2084570571cb05b194192493d61b2b0dcb15f71ea6d45bd8d1a87a679d65a527d62d19dc06182139c3867874
141bf5ec47302f18b8dd087ae2f48171a57181b80fe2aec7290361201b280d740f8c55ba6a19860c6357530994b012a9c45d126b9c385b3038510f3d
75deeab69262394d67ba2657b41cb09c39fcf83b93267ff004bfc1df61e92d9f30ad4edab97c73bd4d7a6cc2b1777af4d8dfb92bebcd4c2c9cddc3ff
dcb8e1e0fead420332e8c8fc40bc0e87f17e8f626716c31e955551c2ad3698a4700ac38c8ac11a661ee912998d4904cc6611308fa8b5e9751148f76b
24ffd8e8d40fce952ce16b2917a098e990df9def0eb829b4701392714cb7b724936ca1c4d20fb593ef6319fcc3fadadafd6faaae4df9d326af6ec9c0
0f578f3afaaae0351f2b4f205e9b290b1beef7455be28ccecaf088fa30ac4ff5d5a53518ebc3de8c894b8d0683e55ed5e9f40e49d7e3b7a03a345e09
2a043f2f389d4b5ad17951e74067a115ed710c7132d22edd8ee8fbb390aa50de1e1199938ddb77ae7f61e7ce17d6eface3bcb970df030f6c7df0d707
730f3cfdbb9696df3d7d20b74eeaffeec58befbe73f1e2d7fc13fe655cfc1b5d3bbff99b47264f627d988875fb4c9abc57f097421db948e76f0f7f34
1a016d4cadb239ea2c1b4c4c32c0286181f25cc2068b1708e442086f87333297fca35bf78f3e471065aa64eb9948845c54f7f4d3ebf7d5d70f7a63de
89b7a51dad8f4a5bb76d3db6a3b54a75b56e9d52f4adb0822768f385b4af8867ba5074754c7e1d1a28ba32c890772bbaa2c088ceb1dde837e61b0b8c
a5463ac7e1d90e3dd83a51473fb9e06640757d29e8b8bd5ed261d8203103e4c9f6e02b864cbfd5aef8957ca5402955ae9373d417a10554d78f4db778
40b20c879e77f2a0e19ff3e04a3b0f0e16b87fef967eca05f7bfe1825cb05f3021785ee7e93a14493a14aed63ba1de5227dedb38c31e40a77bc84fde
dbf87d03a3cba15cadd02a0c15c60a5385b9dc5261adb0558455d82b1ce5ce40f4f568c79d99d51daf77ca5ed8f7eafa75fbf6adbbce9cfcdaf5bff2
6f99032f5d3d75eaea17efbef3e516fe2e6fe2dfd0e1cca533e862bd09c323a4e53b084361e906f863db2d5d9d6d257b131be2c8caddabdbbb3cdd4b
670571bdd26eecfcc6a0b5fb385e6613536eb1867091c85574547556565fdfe7f5f2f7a1adedfdf2d7a5de64ef5e11b0bb75bf6ada5b54c81bf80f74
3514b2afdbcd5d506e389cb07300a599aa59038719ab6c75c606cda41ac090e714474ed723b271e7de1746ed607ef8b67021b1a037b82dae481c9e30
aceb9657088f23cbc2bb7bf0a0d371fa58eb011256f16445fcafa5847cd13bb45b1a5cf5f7b35a249b797442bcc12869a6d10909f1834ce6f8048a34
2ad90ad955e95e11257c540af9a84ef1267342ac060fc61a6c9ac19534a493c0ea5cd31571f27373db9dd677c269099dd2231a9b8866345b28a68134
11d3ccf4983c668fa53b99deaee6ae96bec6bea6bee6be16b317bc2c59ea64ea64ee1c9ee1ca70778ee814df2921dd9b9e989c5669aa34575a2aad4e
f19f1649524daa192d68451b86a11da3310663d123c719d332d207a6ff2cbd227d517a757a20fd7a7a14853fb36ffbcc04fdbd8feaebf88221838944
b327f10e9f1fb57bc28a15935e18d8b8f3fb3f4e3839a3f8edc2252ba7bcea7f75e3c7bf2b3e280fdcdfa9d39831fe6189b6ce9b566c39e4f31dcbc9
19ffc088fc94b0e4f54bb6eed3b3730a6ba5bf295be90c9247b5298630dc43194a83a1ca64261e938ed99d36710675639e158ad2832f78c816fd2a68
8b84057745f415f63c35475872077b9295f36523cade7cf3fcf6aa2a652b7f6b756b60c5a8cddbfe2015ac66038425da4fa7709c7efa5dd0d7efb97d
fe579a5883abce42a7df651e457620cf2d8e636e50a3ae64dd320225eee3c2088493fd0b1ebb5b9e3d95ed1746e0b5baba7b5e9f77e25df67b7644da
d55ab86ddbb11d52f9cdc0bee2c9d771b7a0be3f59a00ab90054b8e94fd35f4b4814d42ba24049058a3a01d44114dcff4651154ac0141934f13e5577
78107478ae31e2fda6787d04fa6bbdc8e01bcd0e31b2e1564e26de7efd62a8f4b8542e554895d222698db44332888d8c94848978390663e4544865e9
982e7b0d3990c3fa601f39d390079413e030394f19aafa0d63612c1b8fe3e57c433114b3e9385d9eaa4c530b0cf3602e2bc7728aa89f5297c132b602
57c82b944ab5066ad80669336e94372a1bd4ddca2b6aade1b8e192a1cd30406414d94691d2f53fc91e638f9de48f36cb052d6370dfcd0071a82f7168
2171c8cceef1e72922d991298fd04441f9abc4d02149cceca0912687d1c44461a65c5b333a0c066d904993996c20ee49a11a116b093250bc667b6444
ad5ddc1c3afbd4767e8a7af00d71a32332944045fe537efe33fe6e34c9b22946769b524dfde5bb4c0fc90f6be34cc5a6f9ec2979be36d7b44a5e62da
246f9337686b4dd5a65d6c8ffc2b79a7f6b22960f298505614a3c91c836ec56d8c31a763aa9262ec6cf65afbb05ceca5f4d07a1a73cd99d66198a70c
310e37fbade3851ca4f1f8b032561daf8d358c358e37e75b4bac0b5885f545f682f62adba1d55a7f6fbd646db36688e459f21919fd11c3e522fe04db
7b811fe1472eb037f89c0b2c9da5cb05ad975a4fb03a3e541a2e45f0d96cb5d0d2dad616ccd7d692bf02e6cbb1873bc3b3014951cb4e068eec0af096
0915ad2d5fe106f68994c9b0f5bf7865ebb5966f82f35899b642fc575478deda9327b5157f2fd357e42e5a7185bea25b2ce89368c11c5656b1ee74e0
83ea80b6e2ab969dfc61eee253d810768d32b7acaf7e8a498e9d0e44b633dce9704bab2a26f096c0ae23018148b8e4640b251b6f6dfd80c7b74cff8a
2ce02a7dde0afd6b1cb2fd9998afa3a2ff07bc232ed96e5a147d625169d5c940f50781d3eb2a0432675affc49dfc287b9135b13df848284255f652e4
910c43fde1a97a406a498cb2c61b1c9644bb6b648afe7646582d7b3f118752aaef7718ad8e3d4e29a60aa236a809ce06735846bfcfb3b278bf6b5914
9c6665de1190de0e4a75c3a68907c2ee2a7bdb2354aee941eafedac969a9ecc73ba2d5f6887553a74ed3260723d70cf1bf13557cb1120d0ffa1db179
1069880873c906034698d49131b7f1e5fdc89ffb9d0632c5f62a5bd4b188d76d1b8cd0a03081ed35aebfe3cda200ad8d52ac6a4ab5ec7aa2f58f2813
c6c4519f3c3c88e96b6fd40bcc7facaf17f15b3b8e877f25906607be0cf13484632f7f58641e4585268bc160979db6911102bf207a023b8afdf71865
4a081cc606ab2410e33a564c67dd4fc37ad266a8e7876f47f6220454ca7f12dbd3ee6a0bed9e0edd60823f2a232fb28ba1b33dd66d88e96c8404d590
1c6f4c4a1dd9fd36a31ab3c4bd456757646c826f4fb28372946ec73abf6e870d115a7243745c6246bf7e57b2b2c86d34d99bb2e82f28e590347bf5ec
758b55ed32ef908c28c43d919008f13eec491bb594447d9f7440b033247b246e92d883527ec8139e2698d9cedc76d2a476deead4a5c2287f44a73cab
c11e11e532d88de21f2989b1c604dfc8b40e94e984e96a10e5f1ee4974485596d40d6e2db1212c263e48d28d7eff484fcfec9fb0ff279955504f3bca
2244c72d1a5eed28975bb2097e2f228ddfb4e9d00ffe9f85f5fb0e12829f4f7df8aced4a7bf9fd472d236de38de2bbaff66fabf479da4c1e0760e3df
7fd4fc806dfc3f7ce191239fd1bf7700891223e9798a2d13a0966007d5b72a9150457099a086600b411181e85f4db09be07982c5d201b84eb0553d08
1794856057d3e194bc194e292d04aba97e15cae56298adb860b6dc04b3a50f2143d415271c9172e1840051973fd3c71cc1e154ef0225e8835ed4bf5f
6e80fe047dc980d5de09d22a512a4b693d5310b4069275c61d34f686c95041b9d20f14177d2d0d955e41273e8cc7e5cef23079b6bc56095346290b94
f7d5b16a93d64b5baa7d63c835bc64b86064c604e323c67dc69ba62cd36573bcb9c0fc57cb1b96df5b2e5a7eb0baac49d601d61aeb17b6345b8dedbb
908c7ae218e802d3c0a2bf79de242421bba5082a65fddb9709e2fff6b2910667eadfcb883a83086a05eb1218585ea88e1dfae50e7505a2d8a8505d05
172b867ba0044a6121cc81e93095769fab7f0934994eb217b22093ae6caa4da2115e184463e64219c11c9802853013ba52ef309845e3bb53ed6e9841
97171ebcb55699de9a42e5149a339fee4534d2f43fd8b5e7ad5dc7d04ef3692ff1a5c62c1a2df028a439ffbb1d0753ed719a3716e6d188c934b6505f
6d8a3ea350a7c84babcca27b298d9944eb4ea7715e9a5f42bb17eacf7ebace687d9532c2a884ae27a857ec5a46634bf495b268ef6cc8b96356fb1c29
a8606dcfeadfb1fde3afa7ae17e29b46f1fd62981e6d8baf2323218afc500cc4eadf4ec6d32ec99461a5eb5fe08a2fb7ee85a12493e1300246c2fdf0
007163343c443b3f0ce3603c3c22223f263385a94ca39cff389c604666d2e6cd9a9e959d3b2854de132a0787ca21a1322f58de9d29cabc4199ed6576
a8ec0150272df2b7dde4d8ecc21f53f0872cfcbe06ff6ec3ef38dee0f85f29f8371bfeb506afa7e0b7cfddad7ccbf15a0d7e53834dcdf875337ec5f1
cb3ef8c520bccaf1f32cfcecca68e5b31abc4203af8cc64f3fc9503e6dc64f32f032c78f395ecac2bfb8f0cf357891e39f9cf89fcfe085a3f8478e1f
d1f08f9ec1f3e7ee55ce3f83e7eec5b37f8855ce72fc432c7ec8f1038ebfe7f83b8e676af0fdd3f1cafb1c4fc7e37b59788ae3dbcb1ccadb1efc6d04
36723cc9f12d8e27381ee7f81b8ec738bec9b181e3518e471c585f99a2d473ac3b7c54a9e378f8d044e5f0513cbc483ef4eb14e5d0447f1b1ef2cbbf
4ec1831cdfa8c1031c5fe758cbf1571cf717e16b36dcf76a8ab2af085fddeb545e4dc1bd4edc4348ef69c6dd1c5fe1b88be34e27eee0f8f2769bf272
166eb7e12f8b3040430235b88de3d6972c9497e14b16dcf262b4b2a5085fdc6c575e8cc6cd76dc64c28d1c37d458950d1c6bacb89e26adafc117d6d9
94173ae13a1bae6dc635d54795351cab574f54aa8f62f52279f52f5294d51371b55ffe450aaee2b8f2f9eeca4a8ecf77c7e788cce7eec615cbcdca0a
172e376315755415612571aa32059739f0e71c972e71284b392e71e0628e8b385670f4b73dfbcc33cab31c9f79069f2ec2f2316ea53c059fe2b890e3
021b3e69c1f9269cc7716e339635e39c669cdd8ca51c4b38cee23823119fe0f8b86390f2f8689cce71da3338951ac51ca7702ce23899e3248e857db0
a0191fb3e0448e8f709cc071fc389332be19c799f0e18868e5e12c1ccbf121daf9a14138c68da3995d191d850fbaf081e1e1ca031cf3cd783fc751f7
d995511cefb3e3488e23e8c9088ec387d995e1e1382cceaa0cb3e3502bdecb31af0687d4e0608ef748dd947b9a71d051bc7b04fa390ee438a0bf5319
e0c2fefdc294fe4eecd7d7aaf4f3b785615f2bf6e198cbb1772f97d2bb197bf5b42bbd5cd833c7acf4b4638e197bc463b615b3ee322b591cef326366
8659c9b4628619bb77332addedd8cd885db3b04be714a54b11764e772a9d5330dd899dd252944e77635a0aa6a69895d4304c316332471fc7a4304c24
3a139de82dc284668c2712e28b30ce8a1ee2a087636c33c60cc2686a44738c2ac248e25424c7089a14118d6e8e2e8ee11c9d34c0c92943eca63806a1
fd190c2b421b47ab2542b172b4d0684b049a399aec68e468a061068e9a0bd52294e9a14c1ae046ea454e99865d91ba21b323706475ac68d92ad6e5ff
871ffcbf46e0fff88bfb6fa8470353
>}
\immediate\pdfobj useobjnum \csname EF7O13\endcsname {[ 32 [ 318 ] 40 [ 390 390 ] 44 [ 318 ] 48 [ 636 636 636 636 636 636 636 636 636 636 337 337 ] 61 [ 838 ] 65 [ 684 ] 68 [ 770 ] 70 [ 575 ] 78 [ 748 ] 82 [ 695 ] 97 [ 613 ] 99 [ 550 635 615 352 635 634 278 ] 108 [ 278 ] 110 [ 634 612 635 ] 114 [ 411 521 392 634 592 ] 120 [ 592 ] ]}
\immediate\pdfobj useobjnum \csname EF7O14\endcsname stream attr{ /Filter [/ASCIIHexDecode /FlateDecode]}{
789ca58e470e83500c449f442fa1b710123a21b9ff01b1befe8a051b2c8d9fc7b62cc3cd304edec452b4b57770f1f00908791011939092c924d71b85
caa5a83addaa359bcb0f9ea296171d6f3ef40caa3b8a266616e1cac6979d9fd4ff03beee0385
>}
\immediate\pdfobj useobjnum \csname EF7O15\endcsname stream attr{ /Filter [/ASCIIHexDecode /FlateDecode]}{
789c5d933f6f833010c5773e85c7748808909844424855ba30f48f4a3ba10c601f115231c890816f5f9b6753a948e4a77bbef33b3b47782d5e0ad5cd
2cfcd0832869666da7a4a669786841aca17ba7822866b213b38bd65fd1d763109ae2729966ea0bd50e4196b1f0d32c4eb35ed8ee590e0d3d058cb1f0
5d4bd29dbab3ddf7b584543ec6f1877a52333b0479ce24b566bbd77a7cab7b62e15abc2fa459efe6656fcafe32be969158bcc6115a1283a469ac05e9
5add29c80ee6c959d69a270f48c97feb114759d36ef9b1cd072af0b6ca67c817276f21560542e15605e4047b250d641f46400c24c01138011c480198
2517a0f69bc2432294ce43423ec2c4a20221c3c6a20221c3cfa2022113647232413ea1738b0a5c650e4bee2cb9b3e438207737e74334c2715e0e7f8e
f3f2b3afc006b85cee2e97bbcbe5682e757f940b79ebd5352945aba93b980fd1438a1ed293cf4109dc2d2af06627c78f881d223bf1db848a87d66638
d7cf629d4a3b8f9da2edcb1987d156d9f71741cfe005
>}
\immediate\pdfobj useobjnum \csname EF7O16\endcsname {<< /M0 \csname EF7O17\endcsname\space 0 R /M1 \csname EF7O18\endcsname\space 0 R /M2 \csname EF7O19\endcsname\space 0 R /M3 \csname EF7O20\endcsname\space 0 R /M4 \csname EF7O21\endcsname\space 0 R >>}
\immediate\pdfobj useobjnum \csname EF7O17\endcsname stream attr{/Type /XObject /Subtype /Form /BBox [ -6.15 -6.15 6.15 6.15 ] /Filter [/ASCIIHexDecode /FlateDecode]}{
789c6d90310ec4200c047bbf820f2cb24310d05e79dfa08922ddffdb4b085a258106b0f18ebd36b78bbaaf1c07cc5b743f511f742d79ed09f5b1a468
e90c75c97929472a5bb0141cf8aac25f50d1f420efc2b5d2c6659ad8ae23f586bf4159fee4b6abcaab1f4621667c8c63e03927383f5ec630f18ee992
30d926860ee85636918ffc0171d84eaf
>}
\immediate\pdfobj useobjnum \csname EF7O18\endcsname stream attr{/Type /XObject /Subtype /Form /BBox [ -6.15 -6.15 6.15 6.15 ] /Filter [/ASCIIHexDecode /FlateDecode]}{
789c3350c8e23250f0e2d235d43334558090b95c489c1c0807cad645e6647071397101008af70be6
>}
\immediate\pdfobj useobjnum \csname EF7O19\endcsname stream attr{/Type /XObject /Subtype /Form /BBox [ -6.15 -6.15 6.15 6.15 ] /Filter [/ASCIIHexDecode /FlateDecode]}{
789c6d90310ec4200c047bbf820f2cb24310d05e79dfa08922ddffdb4b085a258106b0f18ebd36b78bbaaf1c07cc5b743f511f742d79ed09f5b1a468
e90c75c97929472a5bb0141cf8aac25f50d1f420efc2b5d2c6659ad8ae23f586bf4159fee4b6abcaab1f4621667c8c63e03927383f5ec630f18ee992
30d926860ee85636918ffc0171d84eaf
>}
\immediate\pdfobj useobjnum \csname EF7O20\endcsname stream attr{/Type /XObject /Subtype /Form /BBox [ -6.15 -6.15 6.15 6.15 ] /Filter [/ASCIIHexDecode /FlateDecode]}{
789c3350c8e23250f0e2d235d43334558090b95c489c1c0807cad645e6647071397101008af70be6
>}
\immediate\pdfobj useobjnum \csname EF7O21\endcsname stream attr{/Type /XObject /Subtype /Form /BBox [ -6.15 -6.15 6.15 6.15 ] /Filter [/ASCIIHexDecode /FlateDecode]}{
789c3350c8e23250f00262433d4353855c2e5d300d2173b85038195c5c4e5c00e03408e4
>}
\immediate\pdfobj useobjnum \csname EF7O22\endcsname {<< /A1 << /Type /ExtGState /CA 0 /ca 1 >> /A2 << /Type /ExtGState /CA 0.14 /ca 1 >> /A3 << /Type /ExtGState /CA 1 /ca 1 >> >>}
\immediate\pdfobj useobjnum \csname EF7O23\endcsname {<<  >>}
\immediate\pdfobj useobjnum \csname EF7O24\endcsname {<<  >>}
\immediate\pdfobj useobjnum \csname EF7O25\endcsname stream attr{/Type /XObject /Subtype /Form /FormType 1 /BBox [0 0 306 244.8] /Resources << /Font \csname EF7O1\endcsname\space 0 R /XObject \csname EF7O16\endcsname\space 0 R /ExtGState \csname EF7O22\endcsname\space 0 R /Pattern \csname EF7O23\endcsname\space 0 R /Shading \csname EF7O24\endcsname\space 0 R /ProcSet [ /PDF /Text /ImageB /ImageC /ImageI ] >> /Filter [/ASCIIHexDecode /FlateDecode]}{
78daed5c5b8f5c450eeee7f32bce6346da54ea7e61b50f64592221ed4a81483c000f681242d084904980fdf9fbf9f4e92abb7afa72e6d6c08ec2303d
5fbb7d6c97cbe5aab2fbfd68468d7f667cace9bff3b7c393cf5efdf6e6fcd597cf9e8efffc8aff75fe6130e34ff8798d0ffc849fdff1b167f8793d10
8bb783d311bf2fa6dfd67b95f15ad7573f0ec30fc3934f41fe0154cf86c146654b32268cc66b6542d6c14e5cac2aa5a4c8e10b065b5b54b226ba00b8
3111f0f4b0f7e3558f00658a2a691b7d187351019fc8402f5f8d5f8f3f8f4f3eb524a11dbf8066d012b2ff3e68158b3625ea943d5496467159e5e483
ce9d1e1516920d5f0dcfc7f7788c5b1be2757b0c998590c31c8d8b2a6e2c333c9d38de8aae0384d8a76bcacaa49ca391923578b1ae07399e4c576392
f2d959e3a5680c5faced113c4fa76f00bb9c424c9d6c0d5faeef619ea7d3b744a58b29b613adc2cbb53dc4f164ba5a17942bd126274563f8626d8fe0
793a7d93073b1d6dec646bf8727d0ff33c9dbeec21562593b4e9d7d20a1fa9edd51c2d5e6fc193ae06ebe4fb81d289c7944f98a0b4d6761de734a48e
96b28ba72f86279f9b11ca8f2f7e98f286172f876fc6472b7d367e37bef862f8d78be1f949ec169d3258707296766bf062bb718ecc6e82e31ebb45ad
82b3dee5bd76b3a7b65b762a46eda297766bf062bb718ecc6e82e31ebb658357de6029dd67377f62bb59ed9571f8bf305b43975a8df36b4613fc76db
cc6aa37c0aa9ec35593cb5c9ac57708064a2b45983171b8d736456131cf798cd1a557cf231eeb55be6767b3f4cbb2d6256ea4cb7cabb92bd4b6166a4
c926c88f8805314dca12d347ab8bb3f1c54f83551ae1375b6b635c3fe7d1eaddf48e07a710632e3eb8cd3bafe77722acec6da1c8f3d828e3d70f9c9f
009dfcfa1176260f25b968bdc6986c515779be7d340b946df6c558a3cbe6b1bf6c1eabad0e185a2c14f33baf36a246ed7db188869b772e676e25eb64
6c2ebaaaf7e61a8abfdc29c1b767d35b37f364c73d1941274caeac9b2b6bca83e89533f44740504ad12728f565efe562cfc7d65b0f095c70989926c0
2cc1664a57028d0ec6a6c8d584fc28c06e124d10c3a402dae494b73e58726bb6c1e27046aaea30f056c2c5a982144777b041dc48d1784af578788617
a868adf33dceb63902b748193d5cad0bf3c63985e055428f7b981b89968d3d9e55816b941e67db0d814738c7e4d9a3818a1936d594ce986455d0c5d1
68083cc39618191a0e8967659dd1bec75be6cf61ab1d0d65cfc51aa33252481a1389c39af00dd7b3b1091942f0bac759022e700f6bc23f7c8f07cc6f
a846a3227158133e42a322f018558693841e6789b0c033ac092f293d5e345cc93ada58ca754f63521d39932e29c2cb7937eca39e033905e1cd944b25
7891cf9ebf1d9ffc5b8f9fbdaba449e5a873d29875504263b0439e12b3db21d4f7c749ffa564d10ff67820f813f9c5ad671b7856ce788c2d142335f2
2687840ad22963cb1cf716641bc8e3b088d102c1b20d9683f0bc222045c51a4cab86805b6ec25305ac2578112999e4704b59041a95c3724d0b8cc859
6a26c31771ac3a166b789f0bb00487c3588b909dd83e4360798f800b56406dfabc81a54302c6f8230b707d96c1b224911de8808d8f295b5906cb9e04
6e9c0a481bb6b20f9655091cdb138f7c622b2b61d996c48bc21a1eba74822761124f58b383ebb9b3dc4ce21e1b362407a6c75bce26707860866bf459
0fcfe5041e311b90c5f4d910cbf1789661221c12fe41279f1c6e999f84236dd192efa85b3e28618fb9a823dd1008b86689022e7049d0e6d0c1f51c9a
a15623d3816f84d0c19b8c52a27048388695070b2ccf947050096e517aea9a7d0ad8c2c3a14a97bcb19c54c00eae18a76028e176fc2ce0a29c8da91b
0496bf4a187e98a7d30b01b7ac56c2b4d3cddef7d435d71530e253813b74dab40458c0884f09de503abbb23367016715e10dfd28b46459c270407843
2f094ba10d9c2e6266649e431fb53e5cae6f74d77974fbc4b0ef1347e7d1862f90f052a73d31cfae5f206f8150df1f27fd9792453fd8e381e0c12f1e
081e081e08fe4271e239bf0db25346b4b909da7d81ceeff6af2e2f437ab5ac7a6d938fcdf75153a635bee6974aec49488ae727b55ba5ccee94a64ba5
4fcfb037a61d0c7673f8f325febd597dc4cf6fab57ab7175b9fa757581571f569fe0af578c7878b4faefeafbd53988c7d52f20b9c487dee1e323883f
e29d37740b435bc7f58dcf23907e8fb7c0eb6cfc6ed8dc634d27114a87f1fa771ece22cb9f8d8bbcd362af1ac884053b9c1ebe183c59651bbee3d362
2e0d36422135211d72f05238ec68dfee23e5e8b4692a812a2101666c87729ec082cd3d55420c54151142f17e52c7604b61f2482e93aa8e1e3b89603a
70567c6250610c842dc694917335b0ae89703f2601c0e012ed6ecf07262de0844914fdc8f40298ab5e75442a784e8748098033823663435242c9822b
36bec6156f3b0912e64ac6d646488b3deb965e1c643690f06c2fceb55a9649c0478149cb46ace925c676db0fcea9d0f4e9a0a792d867fd94c60654e7
f92359b9548c98d2db17c55f9e8d9e4e0ee8c8c997695eaf67f1b82ab24046ccbeeb9d014e475b86b6c17cf29588cd7e87d2899483fc5bf03dec3287
1254b0d393f9f4a393b61ebd6068f3c70bc681a375e0a6fa80edc1331a5ee0d68fb8e9e0995590c3379741f3526a3d9752f3250712c4906390e50e15
e545d1868484e560355e142de0a9285a9e3f37667756142db468c7ce5cb06535d15733b441d98d59b64b361c94b39054fbe988222a1f9c8be1ba1569
7766365648cdb46447d94bcd7688e17eb325a7723625874299c071663327311c2fca668af2a3f1a5a63bcc72bff1e8b8dee4e85c2e561f6b3d7b1aeb
b1126fae2a3b6a5f6cbd832c0f582f58e58b9eee479c3fd27aee34d66b27f55cd356fbb1d87607181eb05c41de5582f61e33fc58cbf993588e9fff33
4d7901cb52db1d66b9df7a745381f4d795946371475a2f9cc67aec6281abcaca70165bef20cb03d64b48c335160b9d833dd67a515a8fbb32b6889bca
c34dd1e58102c9cfa6c23ea47c197636519b232a08bfdf54107a285ab445bce94b27b72a0877571dfe5ceb273d72c1a49d39fc9cdd659de3fc8e43ea
1203362be97075e3eefacadd35991f37759c1ec36b5c4cf59d1f77eab34bb6279fdb36206f77555adeddbca88f70b4abf4d9ca22eb862ead17e6fc5a
b9b0e0b7a75a382b780a362b7fbcb4b33d22d0dad1ed3836d8526b355ecd568cd76e4bf98259af6dd0e1baad0ff760a96411083336aac2560d5d6a2d
ceafd94bf0db6db1585434480d0eadf227b55881560e9b09794edad0a516e3fc9ac504bfdd162b5ae9ec1007c3757b1eeec162b485b07861bba6a406
2f6eaee11c59738de0b8a7b9c66815a80629871bb53cd0d19887f5a27de8783875c7c3d17ebc55813839f28dfb1df8a2da0a10bd81003a462b0a1039
daea0fdba2d28a0f230ce72c36dfa2f690a3b5f490c7d85679487538da53f9102b3c1460ad3be421a71dad70b4551d0ab4161d0ab4d61c72b4951c0a
b4551c5274c0beba9cae82be0e6458dfc0d963afeaeca14bbf5be0f4fff148fd6720b84a9b9b05a5b6b07e33fd2647c402a4c797d7be9c6445bbe210
86d5f25eb1b9beb5a87ae3baee162979546d91b2c55411296b4c15688daa3c54b6a8caa3628daa02ac515584ca4d541560abe666118d9d803394d572
73b495725f1513efbf22b20ec51c17ddb133e31608f59ff991f7c86911c11d06ab1bdce54e29038f55ac79e12e4395f19b000b297d8ed44200c165b0
5d9e00f2539516ac38da8295480b6bb012b4355871da1aac04588315675083550ba12d5409c21aaa1a650b540cab618a6175c41ad6123f86d5b4af61
6c65e2606d31e160ed2f61606b2ee160bdeb62606b2be160ed29e1e0a6a184db97f59330d2da4cc2b0d649c2c1da46c2c1da43c28782b590f0e58cdd
4109ead64022a85bff88846bff88805bff88846bff88805bff88846bff88846bff88805bff888437d7451c65fd2312def48f48b4f68f08b8f58f48b8
f68fb49161dd2382b6758f48b8768f08b8758f48b8dee948b8768f08b8758f48b8768f08b8758f48b8768f48b8768f08b8768f48b4768f08b8758f48
b85ebd48b8768f08b8758f48b8768f08b8758f8833de75a274545c178952fbc4b0ef13476f20fda1d5f806045deab09bf0f1d1944713eadb7b2695bd
252c18f120cba3298f26d40f047f18825378e92d487db4af394a8d5398bed4693fcba3298f267cf0b18789f040f04070331fbc462b8a28c1b9baa657
74a21c2e183edc8752f74cd80e05331643c5d331a7b4f3be6fbaf0fbfc6c3a1b2c3e25ba88f680deac2ea78e9271f5ebeae7a9a744ab821c5907dae3
d2a518de7eb7ba983a56a8fde4257eff820ffcb8fa84779aec168ebe64a1b8803c73bf707b4ab3fd2a4c1d3266e5567fc7ef9e7468a46ea567d2b843
3c3646d6cf63c4ef37e9e2afdd28e6f57de27fe682a5f9fbdfba2bbee9768f5a7c44711136215487d3e0cd25295d134e97aadd1d628fef60b3b9a69c
05edca96a68a2526e3faea12bb2abab56c4c36554161ba793d48feeb328dae16f1c38c229510cffc64e6429d1de6b0883f6cae76b1e9d7d737e32ef2
8fb305a8bc8a4b7fb54ebb847c591f9a0f0eddb2279ecfa893beb88bc92e3517faeedf76e879b5fe545b36cf9caf2782a23c7fb7ceab7f4cef6665a3
1152a78d14596aa9f7e2b7d5f346e75cde456aeea36338e7102b13dd2b79fa52871ea7b32a2c1357e077ddf626e4c9ca6426a7f7269ac2714293f386
be3a06a8c3f812011dd5f992a9dd0368c0285038048f80891a3c15274f5d361e328c84da5c35d53652e75887cefaaf79303c3a2c046b1e9573817797
e92593c3d02945a0efa618b8d0801d8c31b56d36050d9d62540537e66828f17074121b8391d44e0557b0504bce0e6b82f3d176626066611c74943253
49c8b67e0265d61078b51de3ccecdcc490a3d284e663d81494e3bded1dfbbbe1e8d415beb7fe4c24714b5b1d17b754b95bef873385e6bed1464e4cfa
6adcd4c374349a69b47bf83e3ae2e808978efe8c9c999612be1ebee07073d00bce84c387bae2ac813bf8f5436e3c86c8bd768ee1f52e952cf2c36483
cb5d74b501abc416be3e8a3557e0f77006ca456a23f0763a4ab66916898f237d1df815f8c1114bd4c93d7feec64396bb82c86178fe3f33a8a7c7
>}
\expandafter\gdef\csname EFWidth7\endcsname{306}
\expandafter\gdef\csname EFHeight7\endcsname{244.8}
\pdfobj reserveobjnum
\expandafter\xdef\csname EF8O1\endcsname{\the\pdflastobj}
\pdfobj reserveobjnum
\expandafter\xdef\csname EF8O2\endcsname{\the\pdflastobj}
\pdfobj reserveobjnum
\expandafter\xdef\csname EF8O3\endcsname{\the\pdflastobj}
\pdfobj reserveobjnum
\expandafter\xdef\csname EF8O4\endcsname{\the\pdflastobj}
\pdfobj reserveobjnum
\expandafter\xdef\csname EF8O5\endcsname{\the\pdflastobj}
\pdfobj reserveobjnum
\expandafter\xdef\csname EF8O6\endcsname{\the\pdflastobj}
\pdfobj reserveobjnum
\expandafter\xdef\csname EF8O7\endcsname{\the\pdflastobj}
\pdfobj reserveobjnum
\expandafter\xdef\csname EF8O8\endcsname{\the\pdflastobj}
\pdfobj reserveobjnum
\expandafter\xdef\csname EF8O9\endcsname{\the\pdflastobj}
\pdfobj reserveobjnum
\expandafter\xdef\csname EF8O10\endcsname{\the\pdflastobj}
\pdfobj reserveobjnum
\expandafter\xdef\csname EF8O11\endcsname{\the\pdflastobj}
\pdfobj reserveobjnum
\expandafter\xdef\csname EF8O12\endcsname{\the\pdflastobj}
\pdfobj reserveobjnum
\expandafter\xdef\csname EF8O13\endcsname{\the\pdflastobj}
\pdfobj reserveobjnum
\expandafter\xdef\csname EF8O14\endcsname{\the\pdflastobj}
\pdfobj reserveobjnum
\expandafter\xdef\csname EF8O15\endcsname{\the\pdflastobj}
\pdfobj reserveobjnum
\expandafter\xdef\csname EF8O16\endcsname{\the\pdflastobj}
\pdfobj reserveobjnum
\expandafter\xdef\csname EF8O17\endcsname{\the\pdflastobj}
\pdfobj reserveobjnum
\expandafter\xdef\csname EF8O18\endcsname{\the\pdflastobj}
\pdfobj reserveobjnum
\expandafter\xdef\csname EF8O19\endcsname{\the\pdflastobj}
\pdfobj reserveobjnum
\expandafter\xdef\csname EF8O20\endcsname{\the\pdflastobj}
\immediate\pdfobj useobjnum \csname EF8O1\endcsname {<< /F2 \csname EF8O2\endcsname\space 0 R /F1 \csname EF8O9\endcsname\space 0 R >>}
\immediate\pdfobj useobjnum \csname EF8O2\endcsname {<< /Type /Font /Subtype /Type0 /BaseFont /GCWXDV+DejaVuSans-Oblique /Encoding /Identity-H /DescendantFonts [ \csname EF8O3\endcsname\space 0 R ] /ToUnicode \csname EF8O8\endcsname\space 0 R >>}
\immediate\pdfobj useobjnum \csname EF8O3\endcsname {<< /Type /Font /Subtype /CIDFontType2 /BaseFont /GCWXDV+DejaVuSans-Oblique /CIDSystemInfo << /Registry <41646f6265> /Ordering <4964656e74697479> /Supplement 0 >> /FontDescriptor \csname EF8O4\endcsname\space 0 R /W \csname EF8O6\endcsname\space 0 R /CIDToGIDMap \csname EF8O7\endcsname\space 0 R >>}
\immediate\pdfobj useobjnum \csname EF8O4\endcsname {<< /Type /FontDescriptor /FontName /GCWXDV+DejaVuSans-Oblique /Flags 96 /FontBBox [ -1016 -351 1660 1068 ] /Ascent 929 /Descent -236 /CapHeight 0 /XHeight 0 /ItalicAngle 0 /StemV 0 /FontFile2 \csname EF8O5\endcsname\space 0 R /MaxWidth 989 >>}
\immediate\pdfobj useobjnum \csname EF8O5\endcsname stream attr{/Length1 4216 /Filter [/ASCIIHexDecode /FlateDecode]}{
789cb556797093c7157ffbbd6f2559972559f8926d8465198c1110190c361e10e00b0cd86063cc21b06cc93260cbc636e0037304132034e508553803
49292104884b49c6944edb3450c224cc9406a693b69924133a9dcc904c3a43c8c409ebbe4f36c4659a99e68feeeaedbedfeebbf6edeadb050600666a
643015e51714420c680058128dc616959596a7bd96f906e149844345e58b678d3e35e36dc2a708874bcb27b8d79c6bb1024879842b6b1b7dcdea3f6a
c2840f123e58bba1cd7efc446f3a008e219d86bae660e33fd7fceb2a3953e6f7047dadcda0a60adc4e581f6ce8a8dbb533b09ef0389229ae0ff8fcea
ca3f5800a2da683ebb9e06f4aff2eb84cf114eab6f6c6b771c953613fe13617b4353ad0f5f904e13fe92b0a5d1d7de8c8daa2800adaccc877c8d81f1
d9b35b08a7518ca6e6a6d6b69a627925b956d697dddc12683eb05f4a25dc40fef5a0e4460f8345228414ab32a69016c6431e48f985f32ac0d8e06b0b
413c285e606000e03117915e1b680929591dd295c996d26ba8b74624132185b014f1c38649c811f4c8df013808d688bf4e5f8baf067a7c2d8d21e8a9
69f1ad869e5a5fa895dafa400bb51d2d0dd0130c34111f6c09ac859e7a5f8864ea033534b2d617f2414f83afc9aeb414774fa3afad1e7a426b9591a6
a0af117a5ad68748b2ad2e14a4b65eb13f6c6ddf17261bd93ee0b45f59fc3064b314a51f28c0bf409d44bb26e9548828eb244577db70c5b2ba023fad
6a095e555985951d5137b24f87d63958708892222b07a820c4225886a99131257b122c814658875723713de607ce0f1c1af8d977b387d97bb48f9e21
2447e40749e11b8788d602eb864845215c1d24e58c32079b007df02ed15b7016ae4bd7e1244ca1ba087ecb764b2e9a390d3be1032ec12538c6729895
e5d0ecfb2aabaa83efe067687e05e91693950fe079b2a958ea832d5297540675709ddf8423549b22e35fc2afd976b80387e05da9181ec076ac80bd54
8f40ab0cfc0ed3829032e194e2892a400d5132baf89d48fd12b64017e5ec94aa4f6565ef47a23ecdde66f7e0338af97d5c81eb48e3189c613b64877c
462e86bd83f16235ec95b6b3237275a476d11e6e846372353babb2c23525561a29a348ebe037441be1269bc676e06e8aac4b8980df819beab9f284c1
a8d4dd3899d6034417e022b8304cfa91b5a8eae0985447be1e502437311f32c85a2bd0df01c230e24d1597516230ce6eea959c73fcbd9e8555f67796
8e728d7b02da4d6a7b2f94f51a3aec7d03036555b28d2fede549bde8d4f4ca4ec7273f34f9896b5c495995bdf75705f943560baaf369acbc8a5805d1
308d17e4bb06fff974d6e8842927f01a80268356a9673acf1c8e2a1ea58d92511b85a8d5692564924ea75571b546cdb9cc351ab5c451d293349d2b7d
ae964ba84228d469f43a6d946670953a35184cb7df8bcb790a26e4dd75bbe37226521c6a13ff5c6dd20c11ffbe1f6497a65e54eb18f37a9e8ed28e91
3278ba769a348d67694b78a1f659699b76abf68474989fe0d12a498d6aad5a172fc56912f4a35986942667f04cd558f5188d2b2a5d3b569f2d4fe193
35d9da6cdd547d0916f1d99a626d91ae5253a90b4aabb15e15d405f51dd8aeea9176e22ef919cd8ea81eed0bd2613ccacfe219fe8ae6292f7863b262
98f2630e648e6b578ba5ac4fe64aae9b0feb45c1c35b82dff9562b3f50a83f53fef05b87f23f2fa6d33183ce57261be169b5c658cc263e229aeb0d7a
233718f41a35e586472157be8a9c3148e089b6c4d83869443cb7258e4c91ecc9d427719b2d31373e5a8f72b203908d889a976eefb46d4e30b084c4a4
784bb491db0c28395490c1f48e649b61b4838d56f10c5b62da38d3fd8b271299f7eaa5665bb54df2def61806324e52dd97d194519aa1f2debb77fbd2
eb8e6a87e4bd67cea16a517e447104948dcabb7b3fef9edb1c97336cbbf8e78487dac8960dd230482ceddc4417eddc9b1e57b5abd9855e4fcc80eb23
d7ef5c27a9ee7335b9ca5ca5ae282f4347cc74367952ba2355a57e82cf524f6759eed8b8483bc2aa524fd8deb74da3f12f7be6cf23bbfab61357b54b
e1fed6e728b9d1da79dc54f161d7bee31ae9d4c315d28b99d36397f96fbcf470aff4a27346c2f280c2cad5176a825b5b366c3a7d6854a972de578a65
f275da9d184885f39e79492cd9c62de6f804292e919b2d668949083410c32dcac659cc06a3a48f566672a321396e1e26cfd377e242cb16d366478cd5
64066b74a2c569b61a55a90ed3fd925e4b45d565faba7aa62ebd77fba2c9ccbc7773722259a6d49aeedebd7bff9ee90b73cee3c40e6675588a4dc332
4df9ac4ea37c5e3a99b6354df23a8dcc919a3e799265ca608ea6d09964d6d82cb765287f468695736f749e7d6b6353f195e0a69369f6cbe2ce6591bb
69fd85075b362e39382577f6199fff1f375e66e3db1754d4d49ef9e63b5c74f0289bffc5bec30b172cfa3b9dc8f37465fc55ae06338033761a533b52
47a74f89193565947a14aaa4976c0d5529a14b9e19638362799a1c9526aad624bd37f15875cd6dc9d2d7f7f0d3fd83f79cb4f4f0ad37ba5e5e159df7
158cd4442ea55b9b8d1f3deabffef86196a150b327f2f5797cdfd2bfa651240318767cfdf137970d85e0a73d1a5e34b2f26d56cc9f25da43f29fc135
f5bb50ccadb052bea5c4feb89c62ad9224a50dddbb1a5c0099504ff7a40426ba27e9c643a36420ef72e4fe5ba67cff647a4bb18991bb59e119c4121a
e42530b2c2211e219d550cf1f2309e433ceb1ce25590c60ec06cba1b9aa1035a603504c97b1bd8610cd4d27d6007374ca49a455c0d49d86116c9b4d1
2dd146d201f0d17d3d8e46e74088e4c71337131aa8dae93e7e64ab358202d407486703b57e92d4fe0f5eb31f7bad204f1bc8d71ad20991b412878f74
7e9cc77ce2d6905e25ac27895a92f545ac05221abec88aec6425446d33c9d490ddd5246727fd26f2ee8bcc3d69a73c62a5154a87e4d7d168e00764ec
4f485546226c25dc14f1eaa638b360f27f683fd2753da14bafd381af8836d3a9f86f65f0ad2bd14ee7421954d0e7a2cd47affa3e69ab67e097027b9d
f8ba1b2f84f1bc11cfb51af93937be26f0ac135f35e21927be12c6d3fdf88b7e3c25f0e7b9f8b2c097dc78f244393f19c613f367f213e5f8a21b8f5b
f158188f6af188c0c3163cd48d2f5cc1b0c0832471b01b9f1778607f113fd08dfb8b70df5e1bdf2770af0d7f2af039813f11b847e0b3bb53f8b30277
a7e02e37ee14d8138bdb053e2d709bc0ad02b708dc2cb0bbc4c9bbfdb8496097193b3baef04e811ded5ede71053bb6caed1b9dbcdd8bed1e79a31337
085c1fc6363fb61ab1659d93b7f8715db385af7362b3059b28aca67e0c790604360a6c10b83616d7accee56bfcb89a7caccec5fa053a5e1f8fc13a23
0fbab1ce88013ffa49cd1fc65a81353e3daf11e8d363f5aa045eedc7552b4d7c5502ae34a1578b2b961bf80a81cb0db88c349685716995912f1d8355
465cd28f958baff04a818b2bbc7cf1155cbc55ae2877f20a2f5678e472272e12b8b06c3c5f28b06c3c965210a5565ca0c3f914d5fc99388fba79024b
e69a798913e79a718ec0e222332f165864c642810502f305ce9ed5cd670b9cd58d33057afa71463f4eefc7bcec593c4fe0b4773097b8dc72cc119e66
9cda8d530866cb2e9e3d0b270b9c24302b17ddfd38518f1304ba048e139849d3994fe1581366a0896738704c0a8e4e37f2d17e4c37a29369b9d38d69
fa789ed68d0e9ecb1d025309a55ec151243fca86f6913a6e8f467afdfdde73441ea9c394284cf1c8c9264c22f1a430dac29898e0e4897e4c88b7f004
27c65b302ed6c9e36662ac134708b40a8ce9478b39815b049ac9aa39014d02a3051ac982318c067268e846bd4ecff5f1a8d3a356a086a634617ae899
b84a20a755f05c9409c92e4413bdf8e819188f4c8bcc234312b23ee6dff11ccbfcff16f83fdbffb125f9dfb9b0811d
>}
\immediate\pdfobj useobjnum \csname EF8O6\endcsname {[ 87 [ 989 ] 109 [ 974 ] 113 [ 635 ] 964 [ 602 ] ]}
\immediate\pdfobj useobjnum \csname EF8O7\endcsname stream attr{ /Filter [/ASCIIHexDecode /FlateDecode]}{
789c63601852808504b5ac509a8d160e1905a360148c8251300aa802d8016c520017
>}
\immediate\pdfobj useobjnum \csname EF8O8\endcsname stream attr{ /Filter [/ASCIIHexDecode /FlateDecode]}{
789c5d503d6fc32010ddf91537a643449ca6cd62598ad2c5433f54b7939501c36121d580301efcefcb41eb4a4582a77b778fbb77fcda3eb5d644e06f
c1c90e23686355c0d92d41220c381acbaa232823e34f945f3909cf781277eb1c716aad76acae81bfa7e41cc30abb8b7203de3100e0af4161307684dd
e7b52b54b778ff8513da0807d634a050a7ef9e857f111302cfe27dab52dec4759f647f151fab4738e6b82a2349a770f64262107644561fd269a0d6e9
340cadfa973f15d5a0b7f287732a2fd017bc11fda8324dd017ccf4b9ca34415f90e87b79223a435ff046cd7fdbd01cb4b4cda45c4248fef266b331b2
642c6ecbf7ce938aee37f82c82ef
>}
\immediate\pdfobj useobjnum \csname EF8O9\endcsname {<< /Type /Font /Subtype /Type0 /BaseFont /BMQQDV+DejaVuSans /Encoding /Identity-H /DescendantFonts [ \csname EF8O10\endcsname\space 0 R ] /ToUnicode \csname EF8O15\endcsname\space 0 R >>}
\immediate\pdfobj useobjnum \csname EF8O10\endcsname {<< /Type /Font /Subtype /CIDFontType2 /BaseFont /BMQQDV+DejaVuSans /CIDSystemInfo << /Registry <41646f6265> /Ordering <4964656e74697479> /Supplement 0 >> /FontDescriptor \csname EF8O11\endcsname\space 0 R /W \csname EF8O13\endcsname\space 0 R /CIDToGIDMap \csname EF8O14\endcsname\space 0 R >>}
\immediate\pdfobj useobjnum \csname EF8O11\endcsname {<< /Type /FontDescriptor /FontName /BMQQDV+DejaVuSans /Flags 32 /FontBBox [ -1021 -463 1794 1233 ] /Ascent 929 /Descent -236 /CapHeight 0 /XHeight 0 /ItalicAngle 0 /StemV 0 /FontFile2 \csname EF8O12\endcsname\space 0 R /MaxWidth 974 >>}
\immediate\pdfobj useobjnum \csname EF8O12\endcsname stream attr{/Length1 12508 /Filter [/ASCIIHexDecode /FlateDecode]}{
789cd57a09741455dae877ebab5b5dbda6bbd31db274f68db02526865d692384558c1010509c0442580402844540262c43028a9320121423440dc826
066434c1882851c4c08c0ee00c4f671c05b7df88cc3c1cc790dcfcdfad4e20e8fc6f66ce79e7bcf3ba72abee7ebf7da914300070d34d8598e1438765
c31d702f00eb45bdbee139778f8f3b94ac507b0480b267f8f809595bafedf91400afd0f8d5bbeec81de138daff37b4f812cdf1dd3d3e357df6b5a2cf
00341a8689d3e7e52f503d4127a83d89c6ef9fbe74710ccc8e1c00603a426d51b860e6bc85b72e9d0360a636ec9f995fbc004c7481258fdab6997397
173e33242584da730122d55933f20bf4a96f9502242fa0f1beb3a8c3bed3f434b5f7503b61d6bcc50f5d68e83699daa7a97d746ed1f4fcc7a73c7517
40f70dd4ce9b97ffd002b54a5b48edafa81d333f7fde8ce46f6f6b0048d1099e730b8a8a177fb7fa6232408f8f69fcc9058b662c1864fa2b557b2e06
e0b340d2ca06819f4217420af5f5803e54976316aa0d066568f6985c70cccd5f3c1f4289aef46b6f07b85e9333d9833316cd075dd68ca2d20ef2a983
c2fe2c67b27f28ddc10a5ef80f7eed1fffdb33dfec783e68dc6bba8cd4fce419ddb9b7a01e517363b5a81167c58bf43c45e5c5ff08ce57ffc5f8c737
eefff77f37f66d3fdbfe71fbd91be7c9f64d332bfeed3dffdd998c78ea8260f01067c3201ca22106e2200992692498fa982155819f0afcfa3a4d6a85
212ff26726496310443b48e9518cb968cc96f37463dc4a926a013b38a837089c74aadb383720a73be000780c395d91bf287f1a54e42f9a371f2aa62d
ca9f0d15d3f3e717d37dd68c45745fbe682e54cc9c5144f5998b663c0815b3f2e7d39c5933a651cf83f9f3f3a1626e7e518cbc93bcff6a5efee25950
31ff41d95334337f1e542c5a329f662e2e9c3f93eeb3e4feff834e18f4993b7b66fe4d7aa142402f18f4359e9cf0f2808fa8d69db40d898691d0db78
faa017d122127ad2dd475aa9501fa32bbc03e7bfc0579003836d60fa933c4a594bc7e7915e9fef2077decdacea6cb3591df509ff067b89dacad27f3d
0fa604e6cae7f57a973d6eea9fd2a5bde8465de90f70ad9cea83af2fed4eb4711af40ae9a01fa3b6da31aa1863ddb9941727cd052ead16531dac42ca
0ecfe04fd192a8c013ff00858a9b165935445d5514b90b42975f4ee1b002f093fcfe41f3080fdb6e9ac72e025ce76460b62c3e90d407a8a11633da2a
487b1c455020c9670cf16d20dc0643e16e1807b95008b361012c8315f0074346e4789a313ed6189f06b3602e2c0a8cb75f6cff63fb79d2dc33ed8ded
6fb5bfd97ea4fde5f643edb5ed07dbc77581e59fff02b6bcaea3e530ce0a140977af8e22314fa332908ad4b1db3a8aa4e6d08ee2a23296cadd54886a
0467a078a8e4761469cba75129a4423e8db000c21420820af937903e2d928ae4f1322a7154567414e216fca1a3a4b04c82b989ae37611f54b1ddd492
7b2ea49e6ae530ac8725d4738235b18d4a6feadb0d57e02ccd2c8326dca7021b0519d40b7081f87f95e5c211da6300f3b001268dd46dac7a441da7d6
a95faa67a09f5aac9e51f3d4629681cff1897c379501f836c9c529b25c75ec132886a3f8356660833a5475c0277806f7c1e7748aa459139413cf5712
2c1e560425ca4a651cf59ce467603b5d45347e86ed606709baa36c1d9c8727515546c00e769ef06a82bfc33acc554a8815194a21c17f92f63a43ebb7
433119c6f3cc0242e9497d04bd4157798fc4defcbc715d81123a39176ab43acd638aa75324c576b313ac59db02d57016efc785f8115bafc6ab7bd411
501ea000e64139edbd5daed10ad972c25d5e2be5eeca32358fed83afd53cd334dafb6d89119d794419471815420395659a93701ac4d6e3468274a5c1
cf33a6516a2aada71d4cab086b8022cc8439545b0907e130f4c64a28a79d0c7cb57efcefb4b24afd94702e678f297f87333894a28c42f532d15a8a53
25c0ab268daba830e815e3ac55124716d4faef9914f3eee4d8debd7ed28c719a626a21a7d6be3ca6aebd3d67921ac127d7725f2d26eab56a62fca7ff
d3e0a7bd7b8dce991453db366c68c7aec3f28652dff84954952deaa6fe61438d3179682d4fa4bf9179b531d367c53ce27c247ee023ce19037b07ec0e
4535a4c952fb77b4bbd90910a447617e1bee80751aaa2c0c4235676bf3b9d3b7003dd2fa6578313ef8cad99a35e3c401719cf96987b2f68b6a3971d5
0add20de1fac55bba1dab6d9bd29d4ec0b8a429f372294565e6da60d2e5d6d765e4e63718acbe9ce4877bb9c4a723ab89c101f27efcaa355cf3c437f
cf3c738d99c50fd7ae891f9899e78833e23495332c83ae5b5946b52816a5a24c14b3c7d872b6823d26ed1ac5bdea14f21616f0fbbd5958ad2ad57c8d
09aacd7ab4e643886656e7b9d1b541b993ea69b2bfffe4e6c6560228b539fd6af3b9e634a2dde4387624088354656abf5817cf4ccc70c57a63051b25
9e6233de63a35a6bf6a9c523ea46b49cdf471b109fd55184b10f76f893c3c22330d4e7e22ab83857b39ccfba9eb0577b36aba4efe0b428cce2ebe644
2dd2d93abad69b3bba3624f7bed1b59edcfb08126c3fde7f72e3b9e6e3c75dee011dd05c35a03139f9b726fe2dabf5395ddd06106cfef409ea443ed1
b4425dc19746948599c81a84a9e12416bec5b0545b125e1cb1d8b7164ac3d686af8d58ebdb037b225c53616a222191d917faddce326f4d8a8fd34c99
b7b38c74d5ebd14c1a90097ab3750c913123ffae174a7f71f6a115e7267dc53cc3ee0b1357f7eddbb78c6d1e386fdbc8659559779ebe25fdabb7eedf
b520527c43d85711bf8b09fbeeb0c0df07bcc196527374694c70b5d75e6ddea2f9aa63b6c46fd636799f4f09f105037ac27c49314e1f7aa2cd5a8a24
42486e27fe66037f22c0d566c29228e06cbe74f552b3f38bcb4ee322aaa431bfb9202a3f3a3fa6205685a92c8a793d6a6c5c5272661421d297b0eac9
3203959bd0c3219b9f17ef8baf1e383927f7dd79c74ed6ef3af8cad61dcf3f39fed8a2e25393bf60b65f63627463c59ffe969878e296f4caf25f6ddd
bd6c41f1ca84a42331311f1c7e78bfd48c02e2720dc9944291da1a7f24b3a31d10ed598056533567b8c6cc6c16f069ba6a73383f1e5d6b25c4ec0662
3689d8b9c18dcde92ec9d74be70637a7132e0663d553c4dc5392a53dac14088d80c9e46896c123600a613d2189f5c4be6c2cbbdb76b77d222b644bd8
0a5ccfecc44a338bc50c578637de15ef8acd444d284c648af3e74fb53dc0135b2fe299d68c3da29ae59d3034f9a25a409047c203fe7835dce42a7546
86579b3cd5ce8d76a51ad6d837996aa2baf998057d60716a51ce56d6952f4e097e87b638a5b6108b9c8d97a5024b0d26f688c600778249be5c92e6e0
f5c04d6c91dcf81386b555f79ad4ab85258873e2bb074ecc9a72fcc117df7befc57b9ecde5e7f789c78382c4e5fffaabf83e26a6e996b457aaaa5e49
48226a9713f495863d498049fe84600deca536a80ed1aa7d21bb9cd5b68d719b7d9b126d71665f5854b00f63a32312c9c090105d324ccca5d64b37c4
c7ef215fccce2867f08cdac49b34c2fb709432954d65719ad71312809579fbb0f838053b11898f91e628363d44a9d9b073e7062acc3ce6e931ef9e0d
1a74f8c14f1917573e136de232cb6111639ec641479f7bf6b5d79e7deea8b2bc2e2149fc4d7c77ef54f1dd375f88ff320cd434b62b4a5a28ca82d559
c4130da6fb43b94b41055d2ad90b4efc408e4c65a0999cada71b5d521352bbd8012a242e9241935ea714c24f9b51224e3c72f5eb3fd9ef9ea4300dc3
f9003e82cfc45aa8d54c242dc41816cf62f7e0f1b6cfce32d196c1cf4f6c59c37bcae8e951a2efa3067de32115eef42786127593b5eaa8ded5eecd51
9b929f4f0bb525f4f079137c4166b2de64c2836223d29cad8dcd571b9b0dc276eaaad11a404ada8598897dc8d62464a487482363a86b7c5c42e6ad7d
833b279064288f56ecda5551b17b97d8b57633b4fff913b179cde3cf8b1f7ef841fc503362f3bab55bb6ac5db759797b7b59d9f6a74bcbb64f8c39bc
fae5f7df7f79f5e198b877ca2f7cf5d585f27758fee2b56b17532189594318951146a186c4c49ba2c3582984555b76a9d5b03124badab9396453a2c9
e78b0d8e82b8389fdd101802bfd3277d21beef949790c6b0b7c28f471cf71d8f7c2baa31dab4cfdde0feda8d2431fd0cd976073b485620f356c80848
495c12eb448b68f0e998aad12427030fcffd8bb8c69c9f31642e71487c3ea68addde214bd12425cccedc13ef6741df7cc1420c77b653dc17a56ceb94
24697dae90d09c50e3895f26f0f91dda3a7537ac53c8eca810aa3b5bc9a2488f71d570c9f1d22e5c394b3f21d47841ab2fb03cfe113ed7b1da0e558a
56a5ea9ca9102f1737d262e739e9faa406cb8b02ce16210b3eb7e7ca3edae11405d4b524af3ac5cf7dfc5ed86a5ec3b63a75c569011e664f079f5975
1b7e9da4c1905809cbe1bc6046bae5eaa04d62acf14c616ccb5596c9a2c5a7a24964b19dec30ab14b3448ec8e7a9d796b150d687f562dd768b6d62b5
f8a5a824013f65e43ce7295b5eee8f32b998c21417f9ac2c9382f092ce3566527c6a5f13a52b1602a2b539a0344494019d66f65b5926c71d8eb61244
feb4be4a7fd30865b869b652685aad983466d6bc2c5ccb6623b57bd9246d069bad2dd7d6b347b4ad1461eeb43a0d1d22445c748b3fc59c4a65a3b8d2
36a7919fbf16ad7edad253fdf45a3451a9b8fd22bfd025f6d9ea86ad81d8272c2803c3bcce5003bc2eb18f24488611f524078864dc31f963d1c6f0e3
8f1913ed1fb381ec21b141bc23de16656c391f23eac4e7e20b51c746b07016c146d488fbc40e2933ac86ec0b59980e7e915585608a4a0ef933cdba09
2d9a0b55e42e55c52c4a1fbca87ab79a3d5bed6bac2ad7d065065f88835bc2c254d7108fc567532946694e27352767ec0af074b00c8fdc03e4d5d51c
1931c9617f94a46cd68a60c68133ae686852bde0651e2504bba98990c81295244cd6924c497a923926aa2febab64b36c65165fa22ee1cb8237681b4c
4f6a4f9aa2a71a8ebd5b703cf6613d99f46f31d27a10f5039a858fddb1f2f63317de18f5e8431fbfc7de65d0baae6da3787cebd6c79586908a5f8a59
aca4725adb467efec33f3e7654b9bbed72d9ba75eb497dda5bc8e67e4d3431c128d220652bac51999f7c9e9febce73e4270c3b9b9e46b6d5226dab6e
d8561df44edb1a0ce6687012f3a34d4eb3dfbcc0bcd36c9e8ad2e7118c9afa5ddbe5a6b6cbe4c95ace4bcbaac051b243359a87a4c107fdfd61586f0b
32d7877a3705d5456c0b03b77b78a84dd3c3b30d3aa75f3502d04bd2a6be7339ed95bca8d551d55148dad3694b087b66b8248a93e3c9b74a35c7cf5f
78fcf1176469fbf5c0432b4f437bfbe9958706d6d72ba94d5f7ed944451957902f1ac43fe86ac82fd8a3c9b74a0bdb2fe297942985c1107f0494b20d
aaa3d4bec152ef52ebbbd5b936859bdc7618e11916ee6cbd94de29abe2ea65e7f797d3fcd6a00867c4ea888a88ea084ec019aeb203ba7e5ec316067c
257e39f6999c97df79e7e59c67c6deb56b6a9bf890f566da84e7d4cc033d7b5e3c73e662cf9efb121208210773b381f1442d824a9d42f03903d40aaf
0787a79eeb9b1c756c1b0910e8ca7097db3a2cd2b032e9e9d7a9d57813b54845e333251c8aa15821ac8b9fc1e7eaea061e7ab8a91dda9b1e3ed47692
e8b6670fd10e5f511ef8b1794f413e1bca74ba86e60b6f07f93ae02a216a7928fb5fe04f20893697ea1bb8772fe3f536f65a68bdbbceb6c917e15574
af0ea31577d0309f0162a311c54be2050290ab018f92322472416475e4fb915722f91018c2862843bc4322782f53aa9e6aee6529822256a414798b22
cc53174a02c71a219441db980e636132886e524b5a0fdbcebc3ae7e4b4e9ef3f28ae8a932ca5f53366aa53766dd85eef501e9872ece4adb71eecd18b
f5671616ccee147f6adc76e4e00e197da49278fe83681d0c93fd3eee64367dafc6ca609b436bb028c1144e98b96e0fb28ef1c808dd22033fab0cfc46
d73a8cba0c08073792df68741b76e1527a2b59338a69c9f7bfe2f7e678abbde422bd0464240ba8477c660639fd0ce51fb5d3ef62a9e283fadada83af
6b9ea772664d2f6f4dc50fcac7beb69fe02257a7161870dd4afa62067430adcce1aab36db3304587b176b36ecdf6487d31d465b00c415cee6e03c8d3
780d4f13ef0aa80b55328c4834442da87bf8e1ad07eaebb35e5ef2e63b4a4ddbfdca8e9d3b8ed5b495699eb61d330abe939afa261165399d2bfd634f
b20dc7d443d040de555721fbba77bd4409afdf2a353fc79c47da4ff21f9ce1329ced9b75f453f3ae556b9eaf257d6fec17f72a6c53980ed9aa339062
a6f9ed4eeee7393c8f2fe057c875199bd0069ae7c76608d0c01449f2160753fc499adb1c1a045aa4c96b2b8b8cc1ba888630a7095c41baaee5b8f4a0
1c5f28199078c350b792db33f2c6c1832f5d35920d49187f705a424ec282848a846abade48f824a13dc14c943268e3ed4aaf1b84f306089732ecf8da
978ed52f5a52bebb7ed1b2c776d7d70fa95dbe623f6e7c78e9f79f49323e5b25c9a8ec78eee9379e6f2b53f30ece9cf630747291300886be3773b1e1
9f73f15227178fe4797fe7557eca47efbfe0231d2cd918d0d42574aef4bea3c8fbd6bba1de5627bdaf3be81e747b87fde4cd833f7e48d84a58a99598
4af4127389a5c4bad256622f71940495384b5c2bddd56157c25c37e70637bda0287ee2c0fead5b0e1cd87285b9c5e52b7f15df31177ef2e5a9535f7e
f5eec9afabc4bba2597c4b6a3980b4cfc3fa4b8f2026aa3504a1b471b7fb233a6d5c9d63137b1d1b22c9be0d372c5d179f404956a799f39b0376ee2f
512a9b9a789d341d0ee12647515c5f7fc31f28fd3bbdc49eb6839a655f178fc0bee9347401bee12882ce059428695692332b9639eacc0d268b466e30
db2d837e4313c8ba9d3b2dcdd9919ce09dc19263013f70835ddd7054f4c85e552f101c47d707f7f1e111b7abe958db616256e174cee9b422f24227e9
b464f8d23fd86e531cd6f1d151ba593159c64747476559ac51d1143f94b28daaa7d4bb31547aa744f24edda32cd6e808138c8bd01d26dd1337acbb84
ea5cf32502c7885202eeea7be9aea44c19718a43c628264747a402c9325299e7b3f8ac3e5b1f32babdacbd6c83cc832c83ac836cd6188861094a774b
776b8fe0544faab74748f7a8eed1293129b109c9a596526ba9add42edf1b3345d12c9a156d68470706a113c3301c23d0a7469a93535386a4fc22a524
65754a454a75ca9594500a6a16def096d1c69b0b2dbe6b8a9cca64aad49768878f8edd3365e3c6694f0c69dcf5c31fa79c985bf84efeda4d33f6fbf7
3ff997df161e51871cecde3d37d73f32d6d1e3a98d55afc4c71fcbcc9c7ccfe89cc4a084ad6b771c30f2cb7e64d6fec677900e922f75703d08f7828b
35e865162bd19864cce976481d34cc787a47c41e784541d6f4a5803595b6db1332485af2a44c69c35d6c195b29d68f2e7efdf5f3cf9595f11de2adf2
b6ea8d63b7effcbd9257ce6e97b6f42069e12443fb3d30c8efbba1ff9b2cacc1536723edf758c7921dc8f64a751c1090a84be9d78d4091f7b83402c1
64c1036a77dda727b183d208bc585777e7a1256fbecb7ec78e2abbdbf277ee3c56a3acbc567da070fa15dc23b1bf8d2c50899a47d9f5357fb291582b
943470f94045038a2501b42c4a1edee01a478571154cf28da0e1ea20e0ea3cb9f20d9d7c0102c68ba96e8177725d225fbda31809b9ffd7239439ca4a
a5442955562b9b951a45970799d16c44c1e118ae2641124bc1143546cf844c361007aa697a3650ce8123d56c3e42f3eb1361229b8c93d51cbd100ad9
6c9cadcee4b3b43c7d092c662b7125c5c92bb4f5b09e6dc48dea465eaa554225dba66cc727d527f9366d0f7f41abd58feb9fe8edfaed3263c930b30c
167fdb09f6007be084b8bf45cd6bcdc503d7aa89421389009944211bfbc63f924fd0b8d962562758cc38c162b528139862b552be60d24d942f705d37
49c2d9683611ce9665e10a6a443cab6eb35acc7ae01537d90b7b8086a36b5df2ada65bde82e5cd1a20a3a42b05d7b6c08bb1d1b5ce1bd43dddad2361
4beff64f097cfd19a84a7ab773254409e171964ccb486524cfb6f82df729f7f109961ccb7c653e2fb42c275e2ce725bc4c794a79923f6169501af86f
9593f83b1ec915336aaa955b74ab991e36af1286216a388fd023cc1eabd726b39578251963d5441ea7c59912f564738225d61a6f1b807dd5befa005b
9a235b1981d9aa5fcde27ecd6ff2eb43cd432d43ad7e87df21b93851c951efe1e3b471a61c7dbc39d732c13a1d0ad80c650ece50e7f039da1cd37c73
be75a6adc8b10496b0e5ca2a7c485d45dc2dd156984a4c0fe9cbcd25e69596a5d655b63265032f776c836dec09650b56a94f7399253da5fb532b6d3b
1dbb6137ab516a70bfba9fefd5f69af6eb35b6971cbf510ee1ebea6bbccefc86a3513981a7d5f7f87243222298fc63f156163fb1ee8bcf2f7cf1799d
f8e8c25fff768164a312e7c872ad1a2b5be7908c4c227ff011c988199ef187eb817754241159fa5e68c0bd5c47062ad32c9d2f3b6d063bb5c0dbc260
a9499af1b630bde31d56f3cf5e62f9b3240f9394e1ca4813b7ea41d6508cd07bea31d6be38404fb3fa995f19867ed5cfefd4efc5c9fa2fac792c4f29
c43c358f4fd34bacabad2f59233ade6ec937db2c7621ce691ba31c695da51c699ba1e6ed69fd68cb1e4c248bc028d7e155eadb141bacf7df65a9b25a
0811731509b65a85aac9a4f12acda4a9555c650ac32a85c4b98a2860a9325b18832c8b4e830a25831cad6629f3261279a62ba06b36e7b9c1a7d3a5f8
76540238768a2dffb9081bf24b7ececc92cdac9f99f12ad18f9d5a263245e652f6aee8bf9435b126750cb5fb2d674d2273393bc59a968b4cd61488dc
f93e8a2c1360843f38c908d46db1a1f628dd658b757ac6241a6f73a44d770e96f179e32de07799edaebd6e25bc0c42b769d1ee066b50eae02fd2d3c5
e0cbe914b453bedb3550bf11ac1b66df2407a457e2fb3a2377613282f783b5d39393d88f3745f19d91fc53ddbbcf9a2e237a45c2ab8ed202df508cf3
bb22b2a19b1e12e451751d432cda98f01bf08ac114edf8dd3a392a679923f458c821c73633347026a1bd2c8c77b8e91480b753ea594129a8d348407f
0e3241cc2824574705207df1e57a09f98ff5f5323eef84f1d59724d0ecf0d71dd950078cfdfc41ddb229eab7d874dda9ba1d6342247c01f024749413
ed35ab9428b9cc0d764502260ca89841ba9fa63bca2805eac5ab37321e19e2f39537e53c0685b4563a3d057a53c41f9a9addada7dec319e1d5c37b98
215ad313a2cc714963fadc205463babcb71ae4ea16111dbf37c145b95bef633d0e39615b8829a1212c32562601e9e9e4549b9dcde9f417e0720737fb
f5ed779d549d3cef92a471a29e4cd4247beff5258f5d47acbe4b392cc9d9c17b246a12db035c9ee00b4e96c4ec246e276a4a276d0dec9260ac3fa47b
b65d7786847a74a759fea32436c21c1d3f26b90b660662861884fa62f6c6ba94325bd236af29b621283c2a80d2d5c13fc7a76fc64fc8ff938c3320a7
5d79d181c7751cf677e54b977c94e299c94fedea9e1bf28ba0c1df4374e0a39e0f7ee9b8d4f9fce1c3d6318ec966f92d52e7173fc67715a6792212c0
217ef8b0e51ec7e49f7dddd0473d637c0700ca3e2a8f1a6507ef0665543ea55249a58a4a0115d95f4e650f9547a9ac510ec315ed085cd052e0146f85
53da1828e6e5704abdd8dec23d70542d8485f45ca836c342e50348e56e38aa0c8037653135c151d9563f37e61cc55154ef0945180ffda8ffa0da00b7
698fc244350326694758185f47eb2d81626a206ea6de84870f1e80ad94db6c609f2b43959794ab38154fe217aa49cd56b7ab9ff0d9fc59feadd65ffb
9576d934d0b4d4b4c3d4a047ea93f57dfa8fe634f327169fe501cb7736669b696bb2c7db97d88f38e21d6b0d6aa5622ef484596033be85794a5257f5
2a21f494df2b98608afc1fb56a2662a719df80c83a83106a05ea649a5976471dbbf4ab5dea1c42d9d88eba061e56087742112c80e5b00866c34c3a7d
31c44077984eba1903e990465706d5a6d18c18c8a2398ba198ca229801f9300f7a51ef48984ff3fb50ed0e984b570c8cbbbe57b1d19a41cf19b46629
dd0b68a6e5df38b5eff55373e9a4a57496fc2a613ecd9670e4d39affecc4a1549b43eb26c2129a319de6e61bbbcd3056e41b18c5d02ef3e9be80e64c
a37d67d3bc185a5f44a7e71b633fdd67bcb14b31415444d783d42b4f2da6b945c64ee974760664deb4aa738d1210a8f65f1adf4cfdfc976ae895fce2
2ea8e39bbc10caf2438d2ff322480e2369ff64b885cec8a4dc671864c3701841dc180d63e06ec8817b880ee361029d792fc53593e994a9708529701c
de64c854d392f9b3b3b3d2d23a9e191dcf5beb94d5fef66b025b3cf86322fe231d7fa8c4bf3bf07b815705feef44fc9b03ff5a895712f1bb47eee0df
09bc5c89df5662730b7ed382ff25f0eb81f855167e29f08b74fcfcd278fe79255ea28997c6e3c5cf52f9c516fc2c153f15f817819fa4e39f3df8a74a
fc58e0476efc5fabf0c26bf847811fd2f40f57e1f973c3f9f955786e389efd7d043f2bf0f711f881c0f705fe4ee06f059ea9c4d34d51fcb4c0a6287c
2f1d4f097c67bd8bbfe3c3b743b051e009816f097c53e071816f083c26f075810d025f1378d485f5a589bc5e60ddab14310a7cf595a9fcd5d7f0d5d5
ea2bbf49e4af4cf5b7e32b7ef537897844e0cb957858e02181b5025f1278b0005f74e081fd89fc4001eedfe7e6fb13719f1bf712d07b5b708fc01704
ee16b8cb8d35029f7fcec19f4fc7e71cf86c0156d394ea4adc2970c73336ca23f1191b563d1dc6ab0af0e9ed4efe74186e77e253167c52e0b64a3bdf
26b0d28e5b69d1d64a7c628b833fd11db738f0f116dc5cf11adf2cb0a27c2aaf780d2b56abe5bf4ee4e553b1dcaffe3a111f13b8e9d13e7c93c047fb
e02384e62377e0c60d56bed1831bac58461d6505584a942a4dc4f52efc95c0756b5d7c9dc0b52e5c2370b5c01281fef65fae5ac57f2970d52a7cb800
57e67af9ca445c2170b9c0871cb8cc864b2db844e0e2162c6ec1452db8b00517082c12385fe0dc587c50e01c57169f331e670b9cb50a6752a350e00c
810502a70b9c26307f20e6b5e003369c2af03e8153044e9e64e1935b709205ef0d09e3f7a6e3448113e8e4095998ebc5f1ccc9c787e2380fde332a98
df2330c78a770b1c7b97938f15789713c7081c4d23a3058e1ae9e4a3827164a49d8f74e2083b0e17985d89c32a71a8c03b95defcce16cc7a0def188d
7e814304de7e9b9bdfeec1db0607f1dbdc3878909d0ff6b707e1203b0e14384060ff7e1edebf05fbf575f27e1eec9b69e57d9d9869c55ba330c38ee9
b75879bac05bac98966ae569764cb5629fde66dec789bdcdd82b1d7bf648e43d0bb0478a9bf748c41437764f4ee4ddefc0e4444c4ab4f2a4204cb462
82c078817141184b78c6ba31a600a35b308a50882ac0483bfa88823e81112d189e8561d40813185a80dd8852dd0486d0a29030f40af4080c16e8a609
6e812ec2d59585ce551854800e81765b08b70bb4d16c5b085a055a9c6816a8d3345da0c9835a01aa34a8920478917a51509aebe44a6f644e0481ac8e
15ac7f8cf5fcffe107ffaf01f83ffe22ff1b4b306f95
>}
\immediate\pdfobj useobjnum \csname EF8O13\endcsname {[ 32 [ 318 ] 44 [ 318 ] 48 [ 636 636 636 636 636 636 636 636 636 636 337 ] 61 [ 838 ] 68 [ 770 ] 78 [ 748 787 ] 82 [ 695 ] 84 [ 611 ] 98 [ 635 550 635 615 352 ] 104 [ 634 278 ] 108 [ 278 974 634 612 635 ] 114 [ 411 521 392 634 592 818 ] 122 [ 525 ] 215 [ 838 ] ]}
\immediate\pdfobj useobjnum \csname EF8O14\endcsname stream attr{ /Filter [/ASCIIHexDecode /FlateDecode]}{
789ccd8fd90a835010430fb8d65dab6dddadaddbffffa0978b2f828ae08b81212119c20c5c84b2e3ab0b6be818983cb0b07170f1f06512acf6c38d8e
88a7e4584c7270c38b371f5232a1730ae99554d4347c856ef9f1a7a36790d978e6b15b629a01f162035b
>}
\immediate\pdfobj useobjnum \csname EF8O15\endcsname stream attr{ /Filter [/ASCIIHexDecode /FlateDecode]}{
789c5d53cb6e833010bcf3153ea687884082dd4808a94a2f1cfa50d39ea21e88bd4448c558861cf8fbda1e43a52225a399dd99accd263dd5cfb5ee26
96bedb419e69626da795a571b85b49ec4ab74e2759ce5427a7c8c2b7ec1b93a4ce7c9ec789fa5ab74352962cfd70c571b233db3ca9e14a0f09632c7d
b38a6ca76f6cf3753a433adf8df9a19ef4c47649553145ad8b7b69cc6bd3134b83795b2b57efa679eb6c7f1d9fb32196079e612439281a4d23c936fa
4649b9734fc5cad63d55425afdab6707d8aeedda9ffb7ec005f81d640959465942dea37bdf405e6806c8017bc00150003840001e01c7250de10a54c5
7005f980200f17206482dc4679a5a15a60160f172064641531ab88591cdd9c435e280ec2e1e1c5d2030b8ec08fd1b25254717d22de6aa41c07e41895
b74b4fb008fca810b02c143308cc203083e04b2b9c0d687c29225ea7c25dab18a87cb75b89e5ddfbedf0abbcae9ebc5bebb62eec7b5837bf689da6f5
2f6106e35dfef30b20ddd8c6
>}
\immediate\pdfobj useobjnum \csname EF8O16\endcsname {<<  >>}
\immediate\pdfobj useobjnum \csname EF8O17\endcsname {<< /A1 << /Type /ExtGState /CA 0 /ca 1 >> /A2 << /Type /ExtGState /CA 1 /ca 1 >> >>}
\immediate\pdfobj useobjnum \csname EF8O18\endcsname {<<  >>}
\immediate\pdfobj useobjnum \csname EF8O19\endcsname {<<  >>}
\immediate\pdfobj useobjnum \csname EF8O20\endcsname stream attr{/Type /XObject /Subtype /Form /FormType 1 /BBox [0 0 306 194.4] /Resources << /Font \csname EF8O1\endcsname\space 0 R /XObject \csname EF8O16\endcsname\space 0 R /ExtGState \csname EF8O17\endcsname\space 0 R /Pattern \csname EF8O18\endcsname\space 0 R /Shading \csname EF8O19\endcsname\space 0 R /ProcSet [ /PDF /Text /ImageB /ImageC /ImageI ] >> /Filter [/ASCIIHexDecode /FlateDecode]}{
78daad5ac98e1cc7112dc0b7fa8a3c928094cc7d31e083695903e86083d4003a08ba9833dc304371b304f8fbfc0ffe1dbdc8aace8ca8eee98534c9e1
74475545454646bc1791991f9455067fadfad6d0bf17f7f393ef6e7f7bf3e2f6f9d553f5b71ff9b7179f66abdee2e7151e788b9fdff1d8157e5ecda4
e27ef626e1f75dfb6d6bd0019f4dfff47a9e5fce4ffe8adb3fe1aeab79f65ea7984255aee89c6249113a6cf4daf9607d64e23b2eb6b96aef6cf224ef
3a84747d955b5ef50a26c25c5d60305e4c92d9576daa2dbe88973369d46ef7eef929c6fcfbfc6126577d4bbef249d7e083b24987e83dee82df9e5ecf
4fbeb72ae9acae5f368f5cdfcc3fab479379ac7e51d73fcc7fbf9e9fcdcd8c397b6d92cb518e9d498fbe3e2dcf94e0733cd30677c80a6bb23678b458
39054c7cd40e6b9cceb80f469f6b47386847b0da98908c977630f1713b301fd9e4ec8a3ddb90240d61da72d6010a4855c4fb8d0bab223cad5381cba1
829466ed48e9a3e9bbc7eafaed1c752c2e059b8c4dcb7b1e4db7ed0ae23f9910aab3a5ecaebc5faf24e34c341996efae7c6e579c2e0143c6d072bff2
7a7d26d650315ee3edee8a5a9ff126c41483ab79b9f2e47b37ccbc6f3731bff7e4f11176546f6a4b626d37c2433e777859acc946e59daed19b580dfe
5e9a08dd84e4e0f5e443e4260ce1c1692f2d5d828a55e76a4c0d39d8e3b31e8fda50aa76c523fb99095d7638f076d15a928629d61b5f9d3b6a829dcc
51232ce2d73a8338e26630e971432c3e068a42449c2b272c39ee0e1b686673884e5832a4272c0906f108a36daa3e9c42a7e396249ae1104a16960ce9
094b52d6d617e7b30fa70c895b5c68dc484a8b062650b4398d9c6df83054ad49b6aa6abafe31fd3abd9bfe33dd4e1f1f2b678150b081fe54e4e2af93
c20db7d3efd3e7769b9a5ee0f7edf4127fdfe0f31b7c7e876b9f60ccdc8d510728d3daaa63ce094962a3d131e28bfa78ab7e52ef94533f00b32af11e
88cdd89a803401def5cb276f03aec4602d9ecf60ede7574a123ea34480ac058e518262669170be8421bc9b11170183f34216b42d39e42484519702dc
73cac36d16e061499a819f40af22a52dfd2cece4528cb31684b882753ec11f545f605e40450ea38025704d6d491b83f619a0ed14d804701a22454b4c
9af8a2160555161363dabd59e74c464869d500165fad8a565b0c860031d147577d56b0bbc6b06014f0aa641b8a17b702b03368a42a0cc6b85c1c5990
1218ad44b0498a084ef01cbd0b815a53b60e010570b3b6b8a6b6ea84694f99602656dc411ab2c5d883b35e554b4f99425e18350406004c2ac1b69b03
e205af4190184fa36ccec94495b5c039b5e86072cd9ea40509563c9c2314575d6242ba35dab7a0a6482a8aa590a36c34147ca1b602aca05a8bc178d4
8106b5520533928e12e02b872f10074a274b2042d819ad055a09930b122ee0f541a18044049840d28a480e05b163ad4140243891a448af9072a2f721
8ebc8bad54a87ea57231920a87075fe022eece9a8871033c54b24ed92fc3835f308e1a2d27021aaaf344c20a4fc16bb995a8301e31670226246b0fa8
a9b68991b3a8464a55157189ac71cbdd11f7442a68614f7620fbdac4bdf682f106e62ce5af41be7be7103aade832b6d466099c103c42d993164b71be
883143884e275f094444f0157809eec0633ef82646e542858bdbe840c2b980a9531519506c88a1898baece4713da8cc3102425c91d8a23875a864209
51022e6c52940910d6443e416e38ef9a38e0c90a2a116e45959781442892b85b1dea31cc332804730f7ca17a90c495822217f93e9043b1dec34f3ca2
2d6a144458c1179a25001a62a5c9036208acbac916cc8aaea654f85e8c1d1804dc84a68d76a09001b665915cac7c1526069a3f206f0b3fe42a46d4c4
a812892a23857ba292705102dca41ad46e2c07ece45a806ce42b97709148f0c7f999fa2a9680592b4d802542a1f22a1a849aa08c8b590250090c45a4
58ce131c03065170170ea610d2ce1442da99824b0753b06a763005ddea7c49be0aaa101a3a5588cc1b5cc167717085b879b08578e1a00bae63f005fa
3c4ab4ea055f70341c7cc14d1e7c215e370883930b230c6b35753055f0058c28253bef255f3023065f80f532e0c40bbae02f636cc1dc33d8a234285a
106590859076b2c02829246d145c414f951a62145cc19c33a882410f230a4f79422b189c28047d74a210f70ea260d632a68089786f0a673185b87910
05c7bf411482993a4f0870193421aceb3481d922f00e45d24401a6d7185d9534c1e392d144a1fe13116a254d2458ed13905fb044690d2e7a62c91222
60194d08181e3c21519b3185e0a01d53c0ab0e6056eb862938190ea610303c8842f895f10467ebc113086d9feb923d8c277889c179828519e709e412
da2867373c0150060a262779021980e60cb8b2a10976f7d7f34436b1a0a14268503701a752e34beb1e14506b3f71364f1c6a1c468fc084a31918c201
e65cd8419b0b3b3833e1c0602eec60cb851d559970a027177698e4c20e884c38908f0b3bc43161c7322eeba8c5851d9e9890e1909076c011d28e2c5c
3a2044483b560869c7042eedb92f843dc585b46732978e9415d29e9a42da53904b47aa09694f2921eda9c3a45f9f3940f28c39afc651e6f4728bae7c
4585f5401f9e30394bd1f5458d385770a2113f9890f060447638515c89e7479e029b519e35c41eb515b780a52fbb7794565c2fcb6a66c3a8ac0ee73a
57db0b2ba16040005b78187515378121030a15e35d6b8c475dc5dfc60083ebed75157703c311766f2facb85a862ecc845157717319e870b5bdaee26e
1858c45e36eaaac3088532d2d5dc6ac3515709051db802414dad59d655dc0286677c598715567c140ce7a01a7a735b46659515f704c33fa9bb9756dc
9b0c17f9b2d328adf810195ce2661391fb459656c2ec01a3dc23acb4e27131e055dc3c4a2ba17ac0ae30645456c2ec01c7628cbdb0120e19282ddc37
0a2be1ec81de626a466125cc1ea82e75f7c24a04c9407be19251588941321660b1ca0aab07c881d9f765ecb0ec4cfedc7e034a9cfa45197533d36653
abb0282e4c9b97d0ca2efa7264c976ac9af7fe5a6cabf2b67b35b8f153df237d409dd8283db8ed0a6d97ece9e276d35ebc6cd0dab63dfb4a6cc5d13e
ebb21567d1b8ac2f1abb716de371dde62acb26d74f6d930b54d176eacc6eaf6ebdfa9776b5d08249b6fd8647935b77f01072cbaeea2a8f0fc8d3d84a
13fba149c7d46cb4d42381abbcdc8f8a9bfd86e78f55806f12f2d964fc0fd1bfa73bda361c5b108717f1d5058bf898a30c68407705b6040b034cefe7
a58f1142a02dd0624f28e6696ca1b381a38b0da63d0692b619982a46ed37a3f6728fe5f0fa93ba60fd690c10c547db811503ec423e4026bc6880b4ea
4cab6dc7065828400e0f70344eea82c689a8a0667a2d9f424271b7912e001cf7a4a7c748989bdad0ce9bc5faf02c8e1a575d50e3b241b26964836453
c60679d144f2419e33937b7ba4eb01167e08c6ac87606694e73a38288d122c8dd3b5910f47bf211518ca74ec1f6921ac3ef89263f4b23b72735ecffe
518eebc1809d374f89d19339d9525f723f3b0b5e758887ccc477420c764654b885e299162eff3f8cffac1dd0bdf13f00b9f3e62931fe480ce3a2a12d
7c474ba5e80e509c0cf11d17639a6d8c3556397c2eeea35f0fcc1c0e342e3e7e6006ed82cb3ee5b3cf0f1dd89d77d1a2b82f80676104171f35c2a10e
2b31a0d14b671a41c9e825b4aea60c670a5398f8b829b42a1a4a3029b555f3338f54c5edb105ae12f16dd0e250d110e9d8500e0f1e17206dff9cfe35
7da2b302d36ff8ff6652d3e7e5e40040d103b1963fb8f1be9d09d89ea01a334f87c34c5b7144402086ea567c302080f8a8d0e8380ffa6183642f2e6e
f0ff8cf1ef1943300d6c8e294b6b98fca039c9ebd82a746a3dd0e978f4f8211c3f5d94b7d67c115aec8a718033c252a31f46353edf1048324eb380c0
b6989a556b44976f7e9fd34e27c9b6823e558f9fcefdadc68b68e974510ed046550edb296f23fab46c304b05ef639579aee2a4dc527dfbb59a2e282e
ca72586fada7ffd7ae50bb5b0df805759bacc069491b59e01d3b27671fbce2f841b7b30a2e42fdf398eef9953a5ba784dcac69beed0a31815a7e1f0c
2deac933496113d7d78f5b378156773974448787e8f8d0a7e9357e7f044ea8e965fba4808e6a7adf50e40d24375ba838ab7d68be388bf59e5fa9b375
1ef105b5e9689bbcadbe141643ade15a236875cbeeac65a0d58182de006f90672dfd023115c543da9ea80cb546141cd1d7d3cfbc58af4413e96c0ec2
f56bb4ddac569b9092275a08bb2b7f5eb5d13a0a6dda94b03dd3695c0aaeb485a77ea6b3bbe3c39a52ae007ea2ebc6089f2d5d2fbacfe20175d695ba
cd2e8342839ad0e2b7d9b57fc53d78e59bd5188f120f0889867fdfe23ffdb7b7d4a61d7b8d17196c586ab7127cafc27774006851bcb766d05702362d
38c2c9b5e580b5fd461605b0db37f87d836fef9163aff1d94e41f0cb3c3ffb03e3576712
>}
\expandafter\gdef\csname EFWidth8\endcsname{306}
\expandafter\gdef\csname EFHeight8\endcsname{194.4}
\pdfobj reserveobjnum
\expandafter\xdef\csname EF9O1\endcsname{\the\pdflastobj}
\pdfobj reserveobjnum
\expandafter\xdef\csname EF9O2\endcsname{\the\pdflastobj}
\pdfobj reserveobjnum
\expandafter\xdef\csname EF9O3\endcsname{\the\pdflastobj}
\pdfobj reserveobjnum
\expandafter\xdef\csname EF9O4\endcsname{\the\pdflastobj}
\pdfobj reserveobjnum
\expandafter\xdef\csname EF9O5\endcsname{\the\pdflastobj}
\pdfobj reserveobjnum
\expandafter\xdef\csname EF9O6\endcsname{\the\pdflastobj}
\pdfobj reserveobjnum
\expandafter\xdef\csname EF9O7\endcsname{\the\pdflastobj}
\pdfobj reserveobjnum
\expandafter\xdef\csname EF9O8\endcsname{\the\pdflastobj}
\pdfobj reserveobjnum
\expandafter\xdef\csname EF9O9\endcsname{\the\pdflastobj}
\pdfobj reserveobjnum
\expandafter\xdef\csname EF9O10\endcsname{\the\pdflastobj}
\pdfobj reserveobjnum
\expandafter\xdef\csname EF9O11\endcsname{\the\pdflastobj}
\pdfobj reserveobjnum
\expandafter\xdef\csname EF9O12\endcsname{\the\pdflastobj}
\pdfobj reserveobjnum
\expandafter\xdef\csname EF9O13\endcsname{\the\pdflastobj}
\pdfobj reserveobjnum
\expandafter\xdef\csname EF9O14\endcsname{\the\pdflastobj}
\pdfobj reserveobjnum
\expandafter\xdef\csname EF9O15\endcsname{\the\pdflastobj}
\pdfobj reserveobjnum
\expandafter\xdef\csname EF9O16\endcsname{\the\pdflastobj}
\pdfobj reserveobjnum
\expandafter\xdef\csname EF9O17\endcsname{\the\pdflastobj}
\pdfobj reserveobjnum
\expandafter\xdef\csname EF9O18\endcsname{\the\pdflastobj}
\pdfobj reserveobjnum
\expandafter\xdef\csname EF9O19\endcsname{\the\pdflastobj}
\pdfobj reserveobjnum
\expandafter\xdef\csname EF9O20\endcsname{\the\pdflastobj}
\pdfobj reserveobjnum
\expandafter\xdef\csname EF9O21\endcsname{\the\pdflastobj}
\pdfobj reserveobjnum
\expandafter\xdef\csname EF9O22\endcsname{\the\pdflastobj}
\immediate\pdfobj useobjnum \csname EF9O1\endcsname {<< /F2 \csname EF9O2\endcsname\space 0 R /F1 \csname EF9O9\endcsname\space 0 R >>}
\immediate\pdfobj useobjnum \csname EF9O2\endcsname {<< /Type /Font /Subtype /Type0 /BaseFont /GCWXDV+DejaVuSans-Oblique /Encoding /Identity-H /DescendantFonts [ \csname EF9O3\endcsname\space 0 R ] /ToUnicode \csname EF9O8\endcsname\space 0 R >>}
\immediate\pdfobj useobjnum \csname EF9O3\endcsname {<< /Type /Font /Subtype /CIDFontType2 /BaseFont /GCWXDV+DejaVuSans-Oblique /CIDSystemInfo << /Registry <41646f6265> /Ordering <4964656e74697479> /Supplement 0 >> /FontDescriptor \csname EF9O4\endcsname\space 0 R /W \csname EF9O6\endcsname\space 0 R /CIDToGIDMap \csname EF9O7\endcsname\space 0 R >>}
\immediate\pdfobj useobjnum \csname EF9O4\endcsname {<< /Type /FontDescriptor /FontName /GCWXDV+DejaVuSans-Oblique /Flags 96 /FontBBox [ -1016 -351 1660 1068 ] /Ascent 929 /Descent -236 /CapHeight 0 /XHeight 0 /ItalicAngle 0 /StemV 0 /FontFile2 \csname EF9O5\endcsname\space 0 R /MaxWidth 695 >>}
\immediate\pdfobj useobjnum \csname EF9O5\endcsname stream attr{/Length1 3940 /Filter [/ASCIIHexDecode /FlateDecode]}{
789cb5567d5054d7153ff79d777761d95d76976559047197dd050c2e9845211027aeca8711141030a0125958be59580151a08829388931891fb1db68
49a42949d5da8452a7a521e9a48d8d99b14e9a8936937632491a6792ced04cda3149cbd44bcf5b8893d8a6d3fe9177dfbdf7fcce3df79c73cffd0406
00462a643014e51714420ca8015822712d4565a515f97d257b08af226c2baaa85a9f3ab1f615c2958443a515999eb6d4ce29c2b384b735047c41382e
8d0148f9843b1afa7a6d4f9e9e4c213c49321d4dc1e6c0876d7fbd48c614fd879b7d3d41b246f6e43384b5cd1dfd4d4bfe96b885f005007ca1a5d1e7
576ffbb509202248edd92dc4d09ee5af1256e49d2d81de7df18fb0bb09ff9670624757838f154bf711fe88b021e0db17c40e5524e1bf2bfe77fa028d
19d91bba0122ade49321d8d5d33b7f117e0f109546ed2b83dd8dc1e3c7a464c2d5e4831694d86861e1930821a8c23c256b20119280e5179654420445
8fbef9f945d985f6e37002cce1f6015fb7af1e467ddd814e18adeff6b5c26883afb387ca96c66e2afbbb3b60b4b9b18be8e6eec676186df175924c4b
633d71da7d9d3e18edf075d994b29734047cbd2d30dad9ae70ba9a7d0118edded34992bd4d9dcd54b628fa6ff368d12f59cf8e0207e059fc2464b324
a59e2fc0b7a049a2284b512a4494a32479710cb7beb2a6023f71b642bdca2cccec943ac03ef88a0c2ee6c4f0c801f209b130964189ad8a6a4631240d
619fb6826f7e7e7e6a7ee2968e8538a7dd8ab622b390156f7c94eb2993efccc132611a2e53fe159c834bd22518871c4a5be197ec90e4a69667e14178
9b4b7001c6582e33b35c6a7d536556f5f383fc0cb5efa4be1b49cbdbf038e954344dc3b034289541135ce257e014a5ae30ff1378818dc03578022e4b
1be13318c14a3842e914f4c8c0af310d08291d26144b94167c5c8a6e7e2d9c3e816118844a98504dabcceccdb0d7cfb25768b7fc997c7e1377e26eea
310667d841d9219f9137c291057fb10e8e4823ec945c174e831490bd3026d7b1732a33fc46f1953865e46913bc48792f5c6177b38378883c1b543ce0
d7e08a7a939cb9e0957a0857d37880f27330056e0c51fff058544d30263591adcfc8932b980fcb495b0f80997637c4fe4cc5659418acb0192625d7bd
fe496f79b5edb51abb7bc56dd06650db26a16c52d76f9b9e9f2fab961378cd244f9c4457c4a4ec72bcff758defbb57149755db267f5a90bfa8b5a02e
9f7815d5442a88d8c42fc8772fac89089068f5205139b42caed0282368853de65dc3b82516b8c5126b262286c75a62f3a2b5ba48b6dd04db2d85669d
d912ab5b12638ad647d070222cda257191169eb8c422c5262c35dcb87e75d6989b6bcc35c6298529f74ec85c7363f6ea6cae2937772539a336f0bfa8
e95fa86a92a7aaa319abf5c65647d7186a8c35a6f6e87643bb71307ad03068d4d432b47bb257af4a71c424b1b8180766b074860e66f790772a496e62
f2c99786f7eef156dd357c6d8fbb6fe0ddc1cbb4a42fbdfece8b3cf3e6c7671f0af4de7c2aba6f52b4b0e113fe9b87f8b599c9b1ebcae82be73f08af
051ba4b344af3fc664344473b39e3388d24a1a1dd58e64c969a77af91d525a3aa7adc5654985c490386390978e6960d7c7d9bbb469252a6b71dca0b9
24614055ae65e9c9f638bd4e52bb1824a1d68a2ea323c9a58f4c5e61b831753a92d5ce5e9d5a9b445538404a98284217afdff8e78d59c3c7868f4d14
27e3572245120aa554ff5652f856c631a8f566acb1ae892fb5eeb0ee886fb3b6c50f5807e20f5b0fc547792c1b2cdeb84a4b599cdf5217d76b09c68d
580ec445d552c4639ed18c478d6bc775e3f6f1e471c733cef1f448d2a4fbc43dee3eea0ebacbdc6bdd2a4633e0c8d1334772caea55cee52ce71ea6cc
46b24a4d54567816d48462cd962c4f768e5c32349d5cfefaf0919fd3ce8b61a091a46323db5ff2db5a2ed5befabae68de93fbddd551f7acf290ddf7c
e09cbfe1f94317afdf191163ae69c8c838979a7a7ec25eb3bf7fb8daffd605658636d10c4dd10e35820bce7a374769246d24773a64ae52cb18415432
e5653629c9ce1d4e8739468a3511c3c59d4e479edda4c592586d495279ec80736772b2c305917687c9097aabd3e570a62853e160b5d7af5e30585fb6
4a3419caf254624f93317b5d998aff3209b7cd405d2a2d6076613cf540aa54cb16a2a3678b21b987b1c54049ea701853d50e96859f1ddc53f59d9cbc
0d233df79db82bb7f0e6c5fc2bc1f397f7766e7cb13930eeb14fb3cc69e9fcb1898ad28a3f061e7eb2bc74eb1f586067eb47af3e2d7ed7b7a5d2dfa0
1c5d8b7786547372fbe0a78fec8a5ef3292c8b081ff86fecd7bffb45fdf97b37b374851187c33bfed6dd45e758402c05d01dfcfcbd7ffc4257087e7a
ab7cf953cbca79a8a83f47f9305d16472047fe102ab919367d496e8015b377c27ea8710ba4430bdd3d1218c0abbc44502fe916ef73356c576e3199de
0f6c65f87e53680616420bb4047a56b84823a4d0ab688196bf4473b0b281455a054e761c36d0391c847ee886566826ebbdb4a3d3a081ce5e1b786025
a52ca2ea49c206eb49a6974ee45e926ea45b30002b887b2f74927c0651eba083928deebe2f74f5845123d58dd4a78f4a3f496afe07abd9b7ac5692a5
3eb2d5467d3a495af1c3477dfe3f8bf944b551bf6db087241a48d617d6d618eee10b8fc8465a3aa90c924c3de96d25391bf5ef22ebbe70dbed7a2ac2
5a7aa074517e37711bbf46c6769bd4b6b0873d84bbc2563de46716acfe4aef2ffaba6feb4b2f92f94f29efa755f19f3e75788d4a34d38abe8269e980
77fe2702275df8bc079f0be18ff578be47cfcf7bf04702cfb9f0ac1ecfb8f087217c760e9f99c309813fc8c3a7057edf83e3a72bf878084f6f5ec74f
57e0531e7cd28c6321fc9e064f093c69c22786f0bb3318127882244e0ce1e3028f1f2be2c787f058111e3d92c08f0a3c92808f097c54e023020f0b7c
f850127f58e0a1247cc8830f0a1cb5e088c06f0b7c40e00181c302f70b1c2a76f1213f7e4be0a01107fa67f880c0fe7db5bc7f06fb0fc8fbf6baf8be
5adce795f7bab04fe09e10f6fab1478fddbb5dbcdb8fbb8326bedb85411376915b5d73d8e99d171810d821b0dd826dad79bccd8fad64a3350f5bb644
f1162b3637e979b3079bf4d8e8473f75f387b04160bd4fcbeb05fab458b72b9ed7f971d7fd06be2b1eef3760ad0677eed0f19d0277e8703bf5d81ec2
9a6a3daf49c36a3dde3787dbaa66f836815595b5bc6a06ab0ec895152e5e598b955eb9c2855b05969765f072816519584a4e949a714b146e26af36af
c312aa4a04166f32f262176e32e2bd02371619f9468145462c145820305fe086f5437c83c0f543b84ea0770ed7cee13d73b8267b3d5f23f0eed7308f
a8bc0acc15de20de35843904b36537cf5e8fab05ae129895879e395ca9c54c816e812b04a65373fa9d78870197a3812f77605a12a6a6e879aa1f53f4
e8621aeef2a0536be5ce2174f03cee10984c287906ed246f4f40dbb2286e8b467a61bdec3d252f8bc2a4484cf2ca4b0d9848e289214c08e19278175f
e2c778ab89c7bbd06ac2388b8bc7ad438b0b63059a05c6cca1c918cf4d028da4d5188f0681d102f5a4411f421d19d40da1364acbe98910a5458da03b
2d8f47845045e22a819c46c1f3502624bb110df476d470c98a4c83cc2b4322b269e63ff8284bff663ff886f5ffbfdfd27f013b975f3e
>}
\immediate\pdfobj useobjnum \csname EF9O6\endcsname {[ 82 [ 695 ] 97 [ 613 635 ] ]}
\immediate\pdfobj useobjnum \csname EF9O7\endcsname stream attr{ /Filter [/ASCIIHexDecode /FlateDecode]}{
789c636018028005af2c2b031b00015f0010
>}
\immediate\pdfobj useobjnum \csname EF9O8\endcsname stream attr{ /Filter [/ASCIIHexDecode /FlateDecode]}{
789c5d504d6bc3300cbdfb57e8d81d8a9bc07a0a81d15d72683796f6347a706c391816d938ce21ff7efee85298c0167a7a4f7a889fbaf78e4c00fee9
adec318036a43cce76f11261c0d110ab6a50468647957f3909c77814f7eb1c70ea485bd634c0bf62730e7e85dd9bb203be3000e01f5ea13734c2ee76
ea0bd42fcefde08414e0c0da1614ea38ee2cdc454c083c8bf79d8a7d13d67d943d19d7d521d4b9ae8a256915ce4e48f4824664cd21460b8d8ed13224
f5af5f17d5a037fa6b1de9257d977c4ff0b1caf0f1013fcb7b9afaa74f0bd23536f772f13e1acf27cb8e935743b85dd5599754e9fd02f7d77982
>}
\immediate\pdfobj useobjnum \csname EF9O9\endcsname {<< /Type /Font /Subtype /Type0 /BaseFont /BMQQDV+DejaVuSans /Encoding /Identity-H /DescendantFonts [ \csname EF9O10\endcsname\space 0 R ] /ToUnicode \csname EF9O15\endcsname\space 0 R >>}
\immediate\pdfobj useobjnum \csname EF9O10\endcsname {<< /Type /Font /Subtype /CIDFontType2 /BaseFont /BMQQDV+DejaVuSans /CIDSystemInfo << /Registry <41646f6265> /Ordering <4964656e74697479> /Supplement 0 >> /FontDescriptor \csname EF9O11\endcsname\space 0 R /W \csname EF9O13\endcsname\space 0 R /CIDToGIDMap \csname EF9O14\endcsname\space 0 R >>}
\immediate\pdfobj useobjnum \csname EF9O11\endcsname {<< /Type /FontDescriptor /FontName /BMQQDV+DejaVuSans /Flags 32 /FontBBox [ -1021 -463 1794 1233 ] /Ascent 929 /Descent -236 /CapHeight 0 /XHeight 0 /ItalicAngle 0 /StemV 0 /FontFile2 \csname EF9O12\endcsname\space 0 R /MaxWidth 974 >>}
\immediate\pdfobj useobjnum \csname EF9O12\endcsname stream attr{/Length1 11180 /Filter [/ASCIIHexDecode /FlateDecode]}{
789cd57a797c5445b6f0a97beed27d7b4977a73b6b27e9a4d309618d0901c2222db2a318641170701248022290405884c08435018109080404841603
22204664300144910822a0ce08bee1a90f15109788e8a08c31a97ce7de4e203833efcdfbebfbbd7b736e55ddaa3a75f63ad537c000c0410f113c03fa
f6eb0fa1900ac0dad3dbb001590f0df7bdd3fe28b5fb121c1d307c649fe4cade270090baa1e2c1fb460cec7a61e6656a6fa2feea8786774a9b3cf7e9
3000298bfa474d989a5338f0b17b06507b37f5ff6ec2ec991e783c26134011a9cdf30b274e9dde79f6640023b561efc49ca24250e8066315b5cd13a7
cccd87fdabbfa036d110e19f9497936b18f75629407c03f57799442f2cdb95f500099da99d3869eacc275f5a1af501b547537bff948209399b1f593f
08c06ba0f6f8a9394f168acfcad3a93d85da9e693953f392bfeb45b8bdab889ef3850545337f4d7afd0d005f36f57f513823afb087f2035593887f69
1268b23243f012a885d0497fa7810a1da127087dfb3f3002ac5372664e8308d0b882a62680db357df4137933a681a1791ea33e412f0d54726d24b3b1
3492821dfed995f44fdffe8babe9937f6bd87f8bb3e938c102820f9b3ed12088b7a5a6b74eff6fe8d1f0ddc1dcea62c4b30c61100d6e888504f0d29b
30b0ea520a5e826e7862734ba2d18c20b6d588bb2f4dae48e3b5910ac9570623e909c0445ab410e610b035ebee69580f4e5d77f37266e48c87a53933
a64e83a5e367e43c0e4b27e44c2ba2e7a4bc19f49c3b630a2c9d985740f58933f29e80a59372a6d1984979e3e9cd1339d37260e9949c028ff6241b58
3a3567e624583aed09ed4dc1c49ca9b074c6ac69347266feb489f49ca4e16f6527ada5a1d1f5197c0559d0d30ccaa73a438b89fd6c62fd427090566f
7db5b4d9a4e6fac8ff5927cc4a7867ffcfe3606c70ac56deaeb7c271d7fbb1adda33eed4856e00bf9653bde7eda96dc0091e9200232b1c0ab3490f07
749fd0b496083eaa33d2528bce91c6d31ce32c7a2653bd8d71b53e93de1a2b005a8d6b777b5cdbdbe352fe619c001d6e8f6baff7aed6576fe9eda4f7
e6d3b3a3defb87bb7a6d5aafa459994deb95b498c0442b5b43b483942e3d43dcc7064bfc0fc8171c34c924231a4441d0b0683a5ed82288acfc7eb9e0
078fd7243bb9936d56a6b2cbcd635a780a821b40b7e24a6a31bd2dc2722a63890a24abf6109f1da00b74835ef020e4c144781c0a61163ce935e936e6
215a3b1067c1fe1cbd7f0acc80395eb5a9a9e932f9f6c5a6bf36fd0779786dd32b4d554d2f37bdd4b4af69efa751ffd2cf5aae606cac6e6e59f5b582
80baa681280bcabe4333c8049a8cbb1050dc279a821042d0ab1934293fd80c5a4ccc21c823206912ed41082778bc192208b4f85e481049a0d9df2c5d
3e0073089e2448a0984f32f49aa89ec83288e633741f873db095eda296a6f1e9f426201c806534bb1a4eb0336c85d081deed821bf0218d2c8333b887
0c7730a4d35b808ba4ff9b6c041c241c99ccc9321599cc78a878507c58ac16af89e7a0ab58249e13b3c522968e3ba451d22e824c7c9bece234c44135
bb04457018bfc6743c2af615ad7009cfe11eb84aab68323b03e5a4f362a2c5c90aa04428161ea637a7a473b099ee02ea3fc7b6b10f89bac36c095c80
4d280a03611bbb407c9d819f61098e104a4815e9423ed17f8a709da3f99ba1885ced0253810bede81d514f6b8dd79f31d841baa0df37a084561e0195
72b5ec54bcb48a26b15dec04ab93d741003ec4dfe174fc982d13bde26e712094072580d9504eb8376b73e47c369778d7ee620dbb3047cc667be06b31
5b194fb8dfd638a2350f0a0f1347f97094608e6c239e7ab065b88228d57a63e09c3258ec44f30983b280b80628c00c984cb562d84fb1a30356403961
d2f995bb4a3fd3ccade2e7c473395b2dfc0ce7b02f5961be789d644da107281abca6c892880283f61e5b95e01b945be51f36daf3ce98f80eed7fd3f4
d8144f15645559e67aaa9b9ab2468bd1d2982ac95d853e4395e8f37efeaf3a3fefd07e48d6684f5563bfbecd58fb65f7a577c34753556bd16b7adfaf
afdea72d5a25f9e86f50769567c224cf53b6a7bcdd9fb2e575ef108c3b942990276bde5fc57f148a6507ed665dfd21f226d868b528800e194255abed
932155a12346d780daf466b73143aa42f43af8bb8db9925667cfccbc076c0d75a94c165c4e47b83749c8e8ece82a14972e5eb22c50b161fd46d9f125
bff7da35dee3eab7ece46797586d1dad5749eb15e8ebc5f943146d3d8581c921861a80d6eb79f30eded0f43087cb2928de2e8e8cce4225a1dc501158
b66489eca8e33d2f7dc6bb7f7b95bd7ded1a7b8bb0e6b24b4289b08438b21f82ad824831d6f6c959c2739ef0c4bbe27385e8c6abc2924a3d1a36f511
0ed0fe87d0c1ef84282630210a01fb08db61914801183bd5ea24dcac4b25f12a36e93b0dc624b074e6150c7b1affbe47baf0cb544d8a654d97c572b2
6c13450faf3f540e3820605eeb5815617487c4a2db151d417cdcac235c576ed6d9aea7b204c16e73a4a739ec3621390dec36f026684f61e5d6679fa5
bf679ffd9519f9ad5f7fe5b79851cae2e7f8598273b4703aebccd203bc8897f2325ec456b3b96c1e5bad71f33985c3b1c40da9c5efea83015108488b
1408180d71b21b218e996ce79b35c734cdd5d53610419dead26ed69dd71924c60e866088288ceb1a6f97327ce97692176783f9332cef5d36b8a1728f
5834b07a60fd853d84806c5d1c4c1cbb619b3f39322a1a23dc764904bb24897d6ccfd9d75b02ceb522c53cb0a90253dde13694636c0d43aa5c238654
858d7874489573c4a344096af6547bbeeecd37ed8ecc666a5a895bfa8e55b96df6f04ca2cd9f36521c258d52e689f3a4d9d1659194f88b916214b986
7b26cc9667451545cf742f86d2c8c5518ba317bb77c3ee68fb3818e7232632ba40d77b5946e7246f82ac64dccbd2d34497535628e75b211c6f7880c4
989ef3e00ba5bffff0c979e7477fc59cfd1e8de437f7ecd93387aded3e75e3a039157dee3f7b4fda576ffd6e67610cff96b8df4afa2e22eedb40a1bf
23b842d552635ca92734e0b2048ceb6477c0b3cebb565ee57a3e25cc1d0ae88c7427796c6e74c619e5144d0861235af837eafc9300c8dcc37563abbb
72f34a9dedcbeb36fd26a9a432bf313736272ec7931b2fc23816cb5c4e313e212939239618e9425cb56319c1ca5dec61efb5cff30ff8578f9d9a3ce2
9da9c74ed5ecdc7f68c3b6e7370d3f36a3e8f4982f99f98fe88bab5df3e98f3edf897bd22aca976ed835a7b0a83831e9a0c7f3e703f3f76a769d4b5a
ae249b12c84b17f96398052d8068e90368520212c345466656c12d1b44b31e234cc4984567ccac3176be676d5d9a5dd3eb95f33debd288175db1e269
52ee694da56d4d94520d8431b4d9ce81a7400963ed2089b5c32e6c287bc8fc906514cb67b3d83c5cc62ca44a238bc7747bbacb6bf7dae33350e602e3
19fcc285d38d8f49be86cb78ae217d370fb0ec13a4a16da4a15ca23c061ef37bc528c55e6a8b890a28ce806d854508c022cb2aa53236dccd5474836a
93636d0dacb55e6cade29c4df3165291adf6bae6c09a07937a786d503b5a50b16b32079713ee528ba68d4f31b231d07e74fb7a96c8cff3ef1f3b3169
ec9b4fbcf4eebb2f0d7b6e8474610f7f3a24845fffe607fe93c773e69ed4435bb71e4a4c22699713f5157a3c4984d1fec450192ca5660884c90177d8
4e5bc0bc2261ad7b95cf9c607447c686ba313e2eda4701868ce88a1e62ae345cb9633e7e27e523ec9c700ecf8967a43332f17d205618c7c6b104d9e5
0c0bd2ca5c1d993741c01646bc1e2d1cc5a7850995cbb76f5f4ec08c0f6c79e09d0f437a1c78e27326f11b5ff0467e9d65b1e807b6608fc33b9e3b72
e4b91d8785b9d58949fc47fefd23e3f8f7df7ec9bfd103d478b63356cbdb5612572b75aebc94a9ddeff745104fc97220b643c0b1367655f2f3a911e6
c4b66e57a23bc44831930267487c74aaada1b6ee666d9dce4e8b87e8ad4c728d562cf83a928727a6a78569aead3b89372131a37397d09601a40f61e5
9a9d3bd7acd9b593ef5cbc169afeeb125fbbe8e9e7f9ad5bb7f8adca816b972c5eb76ef192b5c2db9bcbca366f292ddb3cca7360e1ab1f7cf0eac203
9e8493e517bffaea62f949963373f1e29904a4a78b2c5bfa1877106f0ab8fd16da70e4ada2416222780db6869eb569fac6d3bcf568372576f55c03dc
b1fb861647e9bc2bc7919d9ae10d7f67b42b0645b033c1a01528185523b3abaab18faa08684078d960928c064a2f2455768bf7aa94ae5b48ef0d5aec
2027d3bd2d3cf3ce666568066dd33a5068656c9c7f142a218610a3a0ba04a712aa2609498a4749523d6a6725437d5c982f142b73d585c26265b1ba46
081399094359347a597b4c36b43176663d7194618c31cf30d938db30d7b888adc60d6c0b3ac93b43e3c93fc9339997d9bda75907b68095b00e6ff392
33bca456bad060c0bfd7b793e21a2869afff9c64574cd6d081f237954e6647690f8933851badf062b85c63b57b4ae30ebb6bbcd5f655e16608c7088b
d1608a4383b35f12b17bf67c5d5a9a6e0a9d6aafdc6c206338a9db863d53b3f669a931a9b1a971a99ed4f8d484dec9fe187fac3fceeff1c7fb13b262
b262b3e2b23c59f1590959c985c9cb62ca62cbe2ca3c65f1cb12d62407926f24c7b64c6d99d432213b363b2edb931d5f185b1857e8298c5f18bb306e
a167617c446b3feac5badabd195672a424b2bbf4f8d611394c387669dfa282676aaaab7b1f5dbeef4ce3af4c786163f6a11179c7c6feed86909e5f3c
bee8e2c194071a17edc9cf39bee3f5371d252b3b76dc939cdca0c5e1c324ab4ad9499ee3866efe48ac3187186b225cab42aaa3374682c33120c22c1b
a2fad30e4ba2b9a96feb57349f39793df55076ecc2d8402c129d3a39415299eee8947d10adc91450d3f1ea0b4f3ffd82068d7fecfe4af159686a3a5b
fc4af79a1aa1d3996bd7ce10080fe7e6f0a3fcef741fcdc9dd4dd43098de7419af910e23a1b73f1a4ad972d15a6a59aed6d8c59a70525e94e2b0c040
67bf285b03e589cdd90fbf79ddf6d3f554bf2924da16bd307a4d74205a22e2f400d44c5d579726c4e60884d7863e9bf5eac993af663d3bf4c19de31a
f947645bf2c81d62c6be76ed2e9f3b77b95dbb3d8989c490953958772f498ba812c7127db6a0b4a26ac0eaac910cabacd56c23868b601006d81da67e
317a4e9696765b5ab57749cb9edeac4c4ad58014c85ac511dc515dddfd95f9679aa0e9ccfc571a4f91dc76ef26d9e121e1b15fea76e7e6b0becc4077
df1cee6a165f335d25242d274453de90082e662c352c975c2f32a9c6cc8e44d438aacdabdcd12ec1e032c010c111d2cfad9358abe7469af08261fd66
30aea7f48e298c09c47c10732346ea0dbd596fa1b7ab77b4d45ee964e8646caf1640012b100a5c05d1c671d33501c7eb1b932e5b3dba930128bad015
b1a4e180f9dc6b934f8d9ff0c113fc263fc5521abe604ab5b073f9e61aabf0d8d863a73a77dedfb63debc65416caeee79fd66e3cb87f9b6699d3f928
712cf164a25c68b0df1b698e313a4a43c36a42b026c95b9d7cd45813f27a544c522418cc036487c3d32f458fe741b1d75e090a9e5fd038ca24e9b75d
d836d0f637b61a6e13eeecabbd58b34a287b0e0bcfa0c3efce0deb77ee5cbf616735e7f539fb860ddbf6f09f0e661e98ff5e43c37bf30f64560bbdde
f9e493774e7df2c9b7fc0bfe754cecabeddbbefec6a313c6b3ee0c99c8ba8f9fb047e3e338a5d073c966b448dece6f958f89afc05141620611fa6ba1
bc4ea3f70a1d46fc269bd16fcc32661b0b8d64b7a164255a4a72bc9a2e31fbd780ecfc5acbc6efe04b780d360acc00fd455b30e14ef55b6c925fca92
b2a542e9862407911002d9f94b9d36f7309d476348a60930d69f243b8c112120c7282e73598c07aba38f46da14b087180c7296dd1092e58e20c7f76a
8edfd0401b819e45f7ec79e5a6be19386837f087a62666251626ae490cd0fd46e2a5c4a6442349588f4a2ea23e28ead6957497de29a6f47b73f1cbc7
6a66cc2adf553363ceea5d3535bdabe6cedb8b2be6cffee98bc6df09db9edb7aacb2b14cd8b663cb1bcf379689d9fb278e9fdfcc81984b1c8442178a
5646402b93cbacf66af35195363718aac5f3fe4e8d683d58f5d4ac4027f660b6eb7d97a079df3f2127b77afefc0dfb6a6afabc3aebf849a1522360fb
368d005a382ff7fb660f9ba55b63385963a85ce3801a73b5760e73840c4387abdf6fce617e6fefc86228964b94124389b1442d31159b4b2c25d69290
125b89bdd81188bc1169bf3b53baebb856b47edfde0debf6ed5b778339f8f51b3ff0ef991d2f5d3b7dfada57ef9cfa7a2b7f87d7f1efc89d32c96b9c
ac9b16c9c95f2a89422d36ddeb8f6e894dd5d655ec753c1a437169801ea15ac5724a395bc293df188c4f9fc58a6c9cefb6689a03f95d01bea8a6e64e
1c17bab544f7dd8dfb65754fab48cebe6d095041bde160a2ce0e9436ca26b233139659ab8d4715553680a1bf43735edd13282a9d3fab85a18359a1db
43358d05e3f71d7585e3e0b841edb7be40741c5e16dad18d071df633c71a0f90b2f2274812ad5640bbc7295a2d19aef97b5acc82d5343c2ed6601414
75785c5c6c1fd5141b27ba685759213a4b5d2b22b45dc547bb4a9b58d51417adc0c3d106ab627026f46ba35175beee8a164332335bb6999fb46dc6d1
920e59bfa3cc48d19f940b41b2960b4d75ab6e93dbdc9182657b537b730f630fb587a987d9e4010f4b14daa86d4c6d433b393bb9da86b5896d1397e2
49894f4c2e554b4da5e6528bf68b21130459954d68460b5a31046d188951188d6e31c698dc29a577caef534a5216a6ac4909a4dc4889a01469fa9d5d
2e4e3fc7c9ded607864e4c4b61bb90ec70e5d0dd6357ac18bfbe77edce5b7f1d7b624afec99cc5abf2f6faf76efaecbdfc8362effd6dda8c18e11f14
6f6dfbcc8aad87bcde63191963860dc9f285246e58bc6d5faca6cbae14d67e94b6910fd21e68950c21f822d8d95143996a2219938dd91c56cd077bd6
6a89aa9e4ed505a30645e2032fbb98e685dae1c619d683b9b49c86f6c3743b9bc38af9b22145afbf7e61475999b48dbf55de18583174f3f6bf08d9e5
ec5e2d96ee272f1cad7bbf137af8dd77fc7f95ca8e3aabcde4fd4ed3508a03fd5d9a3b66062dea4adaed2050e07a530b02a1b40507ddeef65e9cc4f6
6b41e0a5eaeafb5f9975fc1df63e3b2cec6accd9befd58a550fc6b605ffe841ba87d20845194576788d96066dffa0749236589726971a46ac491aa49
154632c16452652ddd96eca2d49c6e9b69b41dc0dc4795049411de3019cc26d56808fea447de60d17e4151e99068d77ec170688f50ed61d2cf8ca09d
13e914690e1e828754d982afb5a3e4d960824e1efdcff3f3db65b0aa9d8d9b24214c089312d40c75903048eaaffad5478547a5916a963a4d9826e5ab
73851261ae54229509cf089ba4f5ea51e1a8f49e700adf976224c188b268925483c94885d9254462981825451ba28d4e93cbec031ff30ac9182ffaa4
043941f119928d896abcc96bcec42e621743a639d5da5f1888fd45bfd847f2cb7ec56fe86becabf635f9ad7eeb2818c5460959e230e961f96125cb30
dc38421d699a00b92c4f988c79e26469b23c599966cc314d34175867c12c365758804f8a0ba47972893c4f29519ea4b34389b1589d6d5a602e13964b
e5d68db091ad17d6e156718bb449dea43c63f077aa306fb7ee825dac52a8c4bde25ee945f94565afa1d2fcb2f54fc22bf8ba7844aa36be61ad154ee0
59f15d69ae7606498f66da1ff39a987754f597572f7e79b59a7f7cf1871f2f8ad90d153859835f0358d130996ca407ed5273c9464cec7e7f7fc92e2b
b2684751d10a49640243bb406ab7d348d56e54995698543219a39d0c86ce6422130df0862434d704836c6e31909066fdeba64246206bb6a1fdb8a0d5
83bf79d5da5b6ce25f99c46f8f70ccbf491545354a74a9496a2ff11e75a4f888325acd5767b379e26c65a6ba5a5cac3e236e17372a4fab6bd45dec45
f16571a7f2bc1a50dd2a8a12f980290a5d92cb18654ac124c9676c6bf258bab34cec2a7556ba18334da99641d85fea671c6cf25bc69096c70863f011
69943c46196518651c63cab214589e6425962d6cbdb297552a5596f72d972c4d964edaaf3382d7c8e82fddc8c45cfe04db73911fe6872fb257f98c8b
2c85a588d98d971a8fb36a3e50182c84f1e9ac1c9a7f751e4d799204a0b97bbaabaafa06ff5149f8656a39f59ed27f23d67b593c059f74a1808fbe51
cd7f2c97d7e9b31b1b304b799af67b60de0c5ba823341d900ed045270287770578c3d892c6866f7023fb424865d8f8375eda78bde1bbe03c56a4acd0
be2569b957d58913ca8a9f8b82f43809e30a1da34b43e8150861062b2a597726f0c19a80b2e29b869dfc11eee479ac1fbb2e2461da37bfa524c346e1
22dd11eab0bb84d52563794360d7e1804648a8e02057b0f2c6c60f786cc3e3dfd00eb25a9fb742ff564f7b672a66e9a4e8df0e5bd392ee22a4e8d590
0aab4f04d67c1038b3ae4423e65ce3c7dcc18fb02dac8ebd888fc2bfc59d264f3db6b322f66963b6b2e2d6d013da87cb167ae6697319cda5b54e28f3
7e5e7c7bee3c7d6eba3d3e239efa1ab3d9a7d4bdf984fe3d4f18f38c5af6d15f7f1fd2f3278833e81feefefc07eb9596f2d6470d0f58c718b56ffd86
dbdff9689e3295c70058f9ad8fea8759c7fcc397c164f19cfe0d0d042d515f49194a1c5411540a972157380020854319c1e70415045b097209b61194
13ac940fc245f91a9c16af41b1e484c3623e4ca772ba58172c854c38ae8172060e4b0eeabfaabf3f8c83a9de0e0ad00b5de9fd7e79258c128f428fe6
f54f91caaa82b7b0ba55d9faea0613602e6c83cf5838ab105461a27012efc1cd7844ec26be27454979d22b7234c5c417e4b7952465bcf2aa4130641a
be32fa8c13d5b9a61ea6c1a6b1a6d9a625a6eda623a6d3a6aba61fcd68b6993d41b9412a8e80763009ccfaf7e36734a98a2e218c4a51ff063a56fbae
231a6970aafedd54ab3308a356b02e8081f56fae63abf762abba04116c68735d0627cb87fba1000a89ab19f0384ca4d567ea5f8127400a9569904a77
3ad5c6d3080ff4a13133a1886006e4410e4c85f6f476104ca3f11da9761f4ca1db030fdfc655a4b7f2a8cca339b3e9994b23d57f63d52eb7571d412b
cda6b5b42f79d368b446470ecdf9dfadd8976a9369de2898452326d0d81c1d5b9e3e2347e7c84358a6d153fb263e9ef03e4ee33c34bf8056cfd1fb7e
8b67b88ea58828a2f3393c416fb5558b686c818e298dd64e878cbb66b5cc118246d5f407fd7f20fef14ad5ed42fbcf16edbf526ce0d033301784d169
280222219af0a712e67ed01f06c040180c0fc0439005c388ffe13092d67a044613ee47b5bd8fcec6129399c20ccc48a77e1333338b326bdae369e9a9
f735977d8265664b797f73d9b7b9ecd75cf60f96f7a5369769cd657a7399d15c7601a81616fa9b7ee558efc45f7cf8f734bc55813f5bf1278e3739fe
cd873f5af1870abce1c3ef9fba4ffa9ee3f50afcae02ebeaf1db7afc86e3d7ddf1ab3e788de3976978f5ca70e96a055ea1815786e3e52f3a4997ebf1
8b4ef839c7cf385e4ac3ff72e2a715f809c78f1df89f0bf0e211fc2bc78f68f8470bf0c2f901d28505787e007ef89768e9438e7f89c63f73fc80e3fb
1cdfe378ae02cf9e8995ce723c138befa6e1698e2797d9a5936e7c3b0c6b399ee0f816c7e31cdfe4f806c7631c5fe77894e3118e87ed5853ea936a38
56bf46190cc7d70e8d935e3b82af2d140ffdc9271d1ae76fc2437ef14f3e3cc8f1d50a3cc0f1158e551c5fe6b83f175fb2e2bebd3e695f2eeedde390
f6fa708f035f24a25facc7dd1c5fe0b88be34e0756727c7e87557a3e0d7758f1b95c0cd09040056ee7b8ed593365edf8ac19b76e8994b6e6e296cd36
694b246eb6e1332a6ee2b8b1c2226de45861c10d34694305ae5f6795d6b7c175567cba1ed7ae3922ade5b8a67c9cb4e608ae592896ffd127958fc372
bff8471faee6b86a65476915c7951df12962f3a9fb70c57293b4c289cbe928492fca72b1942455eac365765cca71c962bbb484e3623b2ee2b8906309
477fd31f162c90fec071c1029c9f8bc5235c52b10fe7719ccbf1492bce31e36c1567719c598f45f538a31ea7d76321c7028ed3384e89c727384eb6f7
91260fc7c7394e5a8013a991cf318f632ec7091cc773cce98ed9f5f89819c7717c94e3588e6346abd2987a1cade2236191d22369388ae3485a79641f
1ce1c2e1cc260d8fc0879d386c70a8348c6396091fe238f4419b3494e383367c80e310ea19c271f0209b34381407c558a441361c68c1011cfb5760bf
0aeccbf17ea183747f3df63982f70d413fc7de1cefede590ee7562af9e21522f07f6ec61917afa9b42b08705bb73cce4d8adab53ea568f5dbbd8a4ae
4eec926192bad830c3849d6331dd8269f798a4348ef79830b593494ab560271376ec60943adab08311dba761bbb63ea95d2eb64d71486d7d98e2c036
c93ea9cd7d98ecc3249f494a0a419f0913397a392684603cf119ef404f2ec6d5632cb1109b8b3116749304dd1ca3eb31aa0f4652239263442e8693a4
c23986d1a4b048747174720ce5e8a0010e4ed97307c9de076d0b302417ad1c2de630c9c2d14ca3cd6168e2a8dad0c8d140c30c1c1527cab92852a748
16e0427a8b9cf2289b2474406643e0c8aa59eeb2d5acddff850bfe7f13f0df5e31ff0f06cb228d
>}
\immediate\pdfobj useobjnum \csname EF9O13\endcsname {[ 32 [ 318 ] 40 [ 390 390 ] 45 [ 361 ] 47 [ 337 636 636 636 636 636 636 636 ] 56 [ 636 ] 61 [ 838 ] 77 [ 863 ] 97 [ 613 635 550 635 615 ] 103 [ 635 ] 105 [ 278 ] 108 [ 278 974 634 612 635 ] 114 [ 411 521 392 634 ] 119 [ 818 592 ] 8970 [ 390 390 ] ]}
\immediate\pdfobj useobjnum \csname EF9O14\endcsname stream attr{ /Filter [/ASCIIHexDecode /FlateDecode]}{
789cedce470a02500c05c007f6de7b6ff7bfa21f11176e145cb89981c04b0821c98f2a6f7d35b567aa976aa49956dae9a49b5efa6532786d0e3f5c1e
7df9c138934c33cbbce445a9e563baca3a9b6cb32b799f438e399574cee5cbab0000000000000000000000000000000000ff70cded0e217e02bd
>}
\immediate\pdfobj useobjnum \csname EF9O15\endcsname stream attr{ /Filter [/ASCIIHexDecode /FlateDecode]}{
789c5d52cb6e833010bcf3153ea68788f088dd4808a94a2f1cfa50694f550e602f115231962107febeb6772155916034b33b9e95d9f85c3d57ba9f59
fc6e4759c3ccba5e2b0bd378b312580bd75e4749ca542f6762e12b87c644b133d7cb34c350e96e8c8a82c51fae38cd7661bb2735b6f01031c6e237ab
c0f6faca765fe71aa5fa66cc0f0ca0677688ca9229e8dc712f8d796d06607130ef2be5eafdbcec9deddef1b9186069e0098e244705936924d8465f21
2a0eee2959d1b9a78c40ab7ff524475bdb6dfda9ef47f846bc04f911e513c91bc5aa42aaa8aa48ee829c71928962429620a40819428e705c8de19c0c
d33c7c23a28ca919a566949aa39c939c93cc318f1f515e29c6738ce7f9da831681549045907c4299ee82d32570196441374794e3341c10bab5275804
a60b9a69a5388cc8d722f6e230822ee14e5d35cd0e8da30e5a5ffd4b2ffeb7afffd76f805fd76dbde4cd5ab75961a7c34af965ea356c6b6f46e35dfe
fd052ce1d2c4
>}
\immediate\pdfobj useobjnum \csname EF9O16\endcsname {<< /I1 \csname EF9O17\endcsname\space 0 R /I2 \csname EF9O18\endcsname\space 0 R >>}
\immediate\pdfobj useobjnum \csname EF9O17\endcsname stream attr{/Type /XObject /Subtype /Image /Width 148 /Height 148 /ColorSpace [ /Indexed /DeviceRGB 3 <e3e7eab44f2117658a88bac7> ] /BitsPerComponent 2 /Filter [/ASCIIHexDecode /FlateDecode] /DecodeParms [null << /Predictor 10 /Colors 1 /Columns 148 /BitsPerComponent 2 >>]}{
789cedd5b111c2300c05501d1c0d0d154bb0441881021feb300a25270a46c04bb0041535c745c45608c124471448f77f27e7b99172321165722512f1
3493bb566311d22a1cc70acaac6ecc1293339f54311ff5e8ec3650bd547bef9d5b769810d48faa4aa5f2e7e0a07aabb4f713b7eeb073a05a55dc265e
bb9d6e9357efa16c2aedfd88b9c3fb083584aa6f930b1fa07aa805ef6a6aee56d414289b229aea13572c15e6ad1e7171358e200b57a1cc2af65ef1be
5625811a4295c9de261492cc11caa262b7c3e76a473705caa6ca6de22976dbeb117ffed85016f56533400da890bfe701b5610d49
>}
\immediate\pdfobj useobjnum \csname EF9O18\endcsname stream attr{/Type /XObject /Subtype /Image /Width 148 /Height 148 /ColorSpace [ /Indexed /DeviceRGB 4 <e3e7eab44f2117658abf7e9a88bac7> ] /BitsPerComponent 4 /Filter [/ASCIIHexDecode /FlateDecode] /DecodeParms [null << /Predictor 10 /Colors 1 /Columns 148 /BitsPerComponent 4 >>]}{
789ced98518ac240104447c703ecde40740f30267b8005bdff9936d530159b9901853418a8fa2aa75f9ef9e84034a525a7c71294342f29ae6434573a
d02a904aaa0895e582cccf997072ae73c309fde213caed7bc95772914aaa8d551967d6508a2b9c75a0ce0d4925d51e5496470d4f4a0b21937bcc8a54
52c5abacf19c85b315e2927b8b5452c5aafceb0736feeed0e3a5666e52a4922a5e658de72c9c996a0025a9a4da956a90d242705e1d64bf0ea5922a4c
b5861763f68772c0cbc60062b1af914aaa3055ae3becff309b70f25317d9439d48255598eadd67a2a17b9054527da42acfcff10fced5a93cc4dca592
2a54d5d95f96dc529c311d482aa93654597be9cdc220ee7707924aaa08d511c060a34fb8d06f34672c592aa9f6a222300af1414a7b1f5249b58dea1f
b0fae7ba
>}
\immediate\pdfobj useobjnum \csname EF9O19\endcsname {<< /A1 << /Type /ExtGState /CA 0 /ca 1 >> /A2 << /Type /ExtGState /CA 1 /ca 1 >> >>}
\immediate\pdfobj useobjnum \csname EF9O20\endcsname {<<  >>}
\immediate\pdfobj useobjnum \csname EF9O21\endcsname {<<  >>}
\immediate\pdfobj useobjnum \csname EF9O22\endcsname stream attr{/Type /XObject /Subtype /Form /FormType 1 /BBox [0 0 306 216] /Resources << /Font \csname EF9O1\endcsname\space 0 R /XObject \csname EF9O16\endcsname\space 0 R /ExtGState \csname EF9O19\endcsname\space 0 R /Pattern \csname EF9O20\endcsname\space 0 R /Shading \csname EF9O21\endcsname\space 0 R /ProcSet [ /PDF /Text /ImageB /ImageC /ImageI ] >> /Filter [/ASCIIHexDecode /FlateDecode]}{
78daed5a4d6f1c37129d33f74f34908b7430c5e2370f7b88378991007b5022600feb1c1249f1c690ecc84ee2fdf9fbc8e96e16d9ad1e8d2640f610db
c268debc6ad614591f7cf0c34083c25f1a5ea8fceffa5e5c7c71fbfbcfd7b7dfbe7a39fce33bfeeefaa3a0e12d7edec0e02d7e3ec1ec157ede88fc88
7b6194c7eb5d79d5e4f19b1a5fff23c44fe2e273503f82f14a0893648843d43205edbc8631592753b2317806df7198629214833531af511ed0406591
87a17d34a568dc40ca4b4df929ca18eadf7eb81dfe35bc1b2e3ed7d93f7c47fc7c82efaf86f6db3f1443e7f711dbffbab6daf5fd70f1350d5fbc1f2e
c5e5f0303d582170f9e14a46842e3f1e88b046aa9434b9261a150d49da2918e2a5c89e3d88bc692ff2ae59259353210483000525bd576af642bcbc12
175fd1e06518ae7e2a9b747523fe3d9cedd4f9f0fd70f58df8f24a5c8ae28ff0497aa793b18d1f15ddf6c307e994b5c99b68cd118ed8a523097ba3bc
77a671a4a2db8e242b958b089c49fa9888c4a523a4b5cc2724e8f6a05678db15a228ad37084a522a1ee10bedf48a37564b1f9cf7d4a5cd0c1ff0c644
69420cd1930b4739e3b833f93c9f9a617af86620911341bbb188b489b6ac0e2be92ebe13d64b4bf9c035a933813ddb21a79c8fbe61cfe0828d939f9c
8ea9a5cf68cff75e1a6302b5fc8af6fc5058645b7e45177cac4cca858e3fa33d3f62036c39789c5fd19e9fb48cfbb74ddacde8829fa4d32efa8e3fa3
3d3f9f08f256a7aeec57786141585ce9a0436b51e1a50552c0523ee0adc50c2f2cb4972a24173b8b0a2f2c0cd29fa2a2cea2c24b8b248df3d6face62
8617166bbdf0fed15e088b953643acc9f04aeda429868abc1bf279eeeb00925ffa185cae00175fa1ca499d6bc2d9eec7f3e1eaad4064e088f6c6a302
aa3d892ae9ef851424dca4a83d55127bd2b723c9246c0a3693d69ef4faacb0b0db3ada449a54dab3ce76f7a3bdf2da8490740ad327efcb27565a84d6
c764d130c64f6e46ef134a9633984bfcf4098d361edb1e544cd14e9ff8473f797d5e3e62057a2a514ee2f86b6d72d542b635d05a61d6e812d1104a21
4583a2492918b2eae8063eae8f1d85399a715dbf4207d6771edd3e6a67719af5b17d7b5cdf94ee18a2abeb57e8c0fa464952d83247da1edbadc7e573
df510193425dbe42abcb5b2492cfab2b239d53917ca9144737e8b16961b14418c8e6f52bb2b93cf6de1b87bc8a2efdbf7664567536abd1c4d6d2ed9b
4dc79e71ce0e20ec5b59cbae38676330d5fb46d9b22bdeb03124eddb70c79e71ce7699509a7ccbae3867db52c6a3efc833dc704b93b0a6e74e30e7e2
f42c22b7c7384b97d6e642479ce1869b1b272a41cf9d60cea5d296035147ae3867abdcf6bd8e3dbbe20dbb0c15d12dd833ced8a94c2c5677e40a336e
2cd310ee4ddd7d728639f7f0ddf3297db55c9d733663d80f665fc6f02b76e94033fda1f40f7406e4bcb65e257b44333dfbecf5bb479ae34a9f7512b7
151b9c43695959e4627c92b5d82a4cb7864e688c9ffdadf6c55c55d47ce17dde887fcc66a9b21edfa566f20956a2b5602f5c9e7c28b971a13208d582
1bbb72fbcfddfbddcdeeb7dd1d5e87b6f08f8a065744d4a888084ad878553a19f75ee3e046dd783e42cd3767d64b8d231fca996092f7e999a55f54b9
63789adcb1b96cd13df413750f4a46266d8d692f340c3e709f4d2443704991394a7158999cb45658292ad71e34066fbba291f006bba8559e5c4e533f
b4c5546bcab8d0f852e103be588b228a3b6732798e384900d101135144ab6eaf210c3ee04b6ec611d38d8a3e9c2c80689c0ca7ca18d926d40c1ff006
79841dc22c91b3ea2401e40f4bbe32770d1b73d72355e4b1f290ef8e09851db71fd75bcc706f81a0e0c864e7dbc0567861410a7d5de19ad05a547869
1124b9d86b350c5e5868230388a1b3a8f0c202d7077cc3a05abd8ec14b0b6455c8b79dce6286171616cb53e1b5e939c30b0ba7300896e9a3b1a8f0d2
0297d458669bd6628617161e97075d26a7c6a2c20b0b9c7fedcb64d6a6f60c2f2db0bc2a935f6b31c30b0b5c699ccd83656330a30b3e0e35c55ebca9
e892dfb5d2fbd5567a943aa22967eb9ccb7fc9237f983cd2942ca691609f70e9eff1d58e8ffa6e54d43696aa7b9a5ac2dde1920977675b37e9dd3949
3ce1ee700585bbb32da3f4ee9ca2a5706fb8a0c2bdd95655f2ad5d1199b0ef90a76a2bbc0d328185f9b3adb274ee9ca4b5fc29cdbf515e1ea9720d9d
492f33bd955d389d6b2f13bdd35d389d8b2f13bd135e1a3a535f667aabbc34e9c8e49789de492f9ccef49789dd6a2f0db90a3033b9115f9a441c1598
89c9d4174e6312ccc46ce597865c359899dce82f9ccc459889dd09304daa321566a2770a4c436732cc4c6f251846e73accc8ee34189ea44c8899da74
2bc23cfd7a7c9c0a8363e85cad827f493187a598e7dc318ed8bc03aa8cd64e1ac43c4b3ae6d9b28c599165f05d93a2945b2642822915170d6d9c4ef9
0d05245146860fad78c3adc4969510314a0a45f2d3e8fdc69a505ab622e97af88ec386a4d1235c9fc1d1f1ffcb3c6562cd2d26996eb9ad96768d80bd
db7ddcfdbafb01afbfae06ce99e90fde049d2c956f9edfa054e5ff9782e45a04ae5a892d2b1e38543347b18fdb8c3661d3848923765163e0f38236af
b515b35f762f763f226ebf216237bbdb3c2caf444df1e366f6bfe5ba8ec88ca1507e1935c58fdbe3563c6a5e5a557e69a236a34dd488b0c208b3b833
f479719b573b14b737bb0fe78346f2c418954b41113e78bffbb4fb19c17cb31ac860a798e08d9dce11e5fa3d9ea3b812c86a25b6ac84c8a33a26c19c
284de2e2f24b3d7cc7619ea2ec21cfcd5c4c583675eb6d45f31e21fbeff980c9de50c01080913fc7f2b1e31878e9322aed7310d73a49353797519cad
c4965513459ec5355a3c8b59686bc2f27d785e1ad7083e298df767eea69cbe17bbdbf36c859d43ab8d1ec3da598e6e0f5d23beb7cba4bf14e2f27f68
3b3a3a
>}
\expandafter\gdef\csname EFWidth9\endcsname{306}
\expandafter\gdef\csname EFHeight9\endcsname{216}
\endgroup
\edef\EFSetResources{\noexpand\pdfpageresources{ /XObject << /EF1 \csname EF1O26\endcsname\space 0 R /EF2 \csname EF2O30\endcsname\space 0 R /EF3 \csname EF3O17\endcsname\space 0 R /EF4 \csname EF4O14\endcsname\space 0 R /EF5 \csname EF5O13\endcsname\space 0 R /EF6 \csname EF6O15\endcsname\space 0 R /EF7 \csname EF7O25\endcsname\space 0 R /EF8 \csname EF8O20\endcsname\space 0 R /EF9 \csname EF9O22\endcsname\space 0 R >> }}\EFSetResources

\begin{document}

\title{\texorpdfstring{Diagonal Bases and Diagonal Periods\\ of Elementary Cellular Automata}
{Diagonal Bases and Diagonal Periods of Elementary Cellular Automata}}

\author{\authname{Tigran Nersissian}\\[2pt]
\authadd{Independent researcher}\\
\authadd{Nice, France}\\
\authadd{ORCID: \texttt{0009-0001-8755-7412}}\\
\authadd{Email: \texttt{tigran.nersissian\symbol{64}hotmail.com}}
}

\markboth{Complex Systems}
{Diagonal Bases and Diagonal Periods}

\maketitle

\begin{abstract}
Which cellular-automaton diagonal families form bases in every finite
window? For canonical polynomial lifts of elementary rules, two
truth-table bits determine triangularity, and units on the matrix
diagonal determine invertibility. Exactly 24 rules give universal
binary bases; all remain universal over every modulus. Among triangular
binary coordinate maps, the Pascal transform is uniquely characterized
by converting OR convolution into pointwise multiplication, while
increment becomes strict prefix summation. Explicit inverses and coordinate
comparisons distinguish sparsity from evaluation cost. A Rule 30
polynomial construction gives Fibonacci bounds on interpolation order
and prime-modulus periods. Exact additive periods anchor a finite
census modulo two and three. These results separate all-window basis
classification from optimization of a representation and from period
patterns observed in finite windows.
\end{abstract}


\section{Introduction}

Parallel diagonals of a cellular automaton can be used as columns of
a coordinate matrix. The central question is which rules make that
matrix a basis at every finite length. For the canonical single-seed
elementary lifts considered here, the binary answer is an explicit
list of 24 rules. This is an all-window statement: one local truth
table must generate invertible matrices at every size. Once that
classification is established, a second question becomes possible:
which of the available coordinates simplify the operations needed
to represent and evaluate another diagonal?

The modulus matters at the first step. Over a field, a triangular matrix
is invertible when its diagonal entries are nonzero. Over a composite
modulus, the required condition is that those entries be units. A
greedy elimination procedure must therefore include solvability and
uniqueness conditions at each pivot. We formulate that procedure with
its exact domain of completeness before applying it to cellular-automaton
diagonals. This prevents a successful elimination in one finite window
from being mistaken for a basis theorem valid at every length.

The classification follows two steps. Theorem~\ref{thm:tri} shows
that two rule bits decide triangularity for all window lengths.
The unit condition in Corollary~\ref{thm:criterion} then reduces
universality to the values of the center column. Proposition~\ref{prop:congruence}
gives exactly 24 universal binary rules and proves that each remains
universal over every modulus. For a general composite modulus the
question reduces to its prime divisors, with eight candidate rules
left in the unresolved odd-prime part. Universality throughout this
classification means that every finite target vector has unique
coordinates in the diagonal family. Computational universality
concerns the full dynamics and enters only in the later comparison
of those two meanings.

Invertibility leaves open the quality of the representation. A target
can have fewer nonzero coordinates in one basis and still require
more work to transform or evaluate. We examine this issue through
four explicit bases and through the operators needed by the Newton
support calculus. Pascal coordinates have a precise algebraic
characterization: among triangular binary maps, they uniquely turn
OR convolution into pointwise multiplication. Under the same change
of coordinates, increment becomes strict prefix summation. The latter has
a dense matrix but an inexpensive direct implementation, illustrating
why matrix density and transform cost have to be assessed separately.

The diagonal arrays also connect this basis question to recurrence
theory. For Rule 30, the integer polynomial lift gives a system of
successive discrete integrations. Tracking its degrees produces a
Fibonacci bound on interpolation order. Reduction modulo a prime
then gives period bounds through the periods of binomial sequences.
The additive rules supply exact comparison cases in which the
coefficient structure is explicit.

The finite comparisons test two possible extensions of these
structural conclusions. First, does the distinguished role of
Pascal coordinates make them weight-optimal? Explicit examples
give a negative answer within the tested family, even when the
complete relevant Pascal support lies inside the window. Second,
do period patterns detected in one window persist when that window
is enlarged? The census records how depth, width and test order
affect the answer. These comparisons explain which stronger claims
are unavailable from invertibility or transform identities alone,
and lead to the optimization and period questions at the end.

\subsection{Relation to earlier work}

The algebraic theory of linear and additive cellular automata is
established in \cite{mow,dennunzio}; Willson \cite{willson} develops the
fractal viewpoint. Rowland and Yassawi connect automatic sequences with
rational diagonals and with columns of linear cellular automata
\cite{rowlanddiag,rowlandcolumns}. The arithmetic tools used here are
Lucas' digit criterion \cite{lucas,fine}, binomial periodicity
\cite{zabek2}, and the Frobenius map \cite{lidlnied}. Mattarei
\cite{mattarei} studies a related converse question for signed binomial
coefficients. Fiorot, Gilblas and Tonolo \cite{fgt} study iterated
antidifferences of modular periodic sequences. The subset zeta transform and covering product belong to
the standard transform calculus \cite{bhkk}.

The specific object classified here is the diagonal matrix family
generated by a rule. Pascal inversion, triangular elimination,
Frobenius and the covering-product transform supply established
tools for studying it. The rule-bit criterion and the resulting
24-rule classification are the principal conclusions about that
family; the uniqueness theorem explains Pascal's role in the
associated support operations. The Fibonacci interpolation argument
is included with its attribution in Section~\ref{sec:recurrent}.
The period census then compares finite profiles for which the
preceding theory supplies exact additive reference cases.

\section{Preliminaries}\label{sec:bg}

We use $\FF=\mathbb F_2$ and $\ZZ_N=\mathbb Z/N\mathbb Z$. A binary
support $S$ represents the sequence
$t\mapsto\bigoplus_{r\in S}\binom tr\bmod2$. For a finite window,
only indices $r<t+1$ contribute. On coefficient vectors the product is
OR convolution,
\[
 (a\conv b)_r=\bigoplus_{i\vee j=r}a_i b_j,
\]
where $\vee$ is bitwise OR. On a window $0\le r<n$, terms with
$i\vee j\ge n$ are discarded. Increment means
$(\Inc a)_r=a_{r-1}$, with $a_{-1}=0$.

The following facts fix the connection with \cite{U1,U2}. The matrix
criterion is proved below, and the carry identity is also proved here
so that the linear-algebra arguments are self-contained.

\Needspace{5\baselineskip}
\begin{background}[Support as a binomial transform]\label{bg:I-1-18}
$\mathbf 1_{S_m}$ is the binomial (zeta) transform of the diagonal $b(m,\cdot)$, the
transform matrix being $\big(\binom nr\bmod2\big)_{n,r}$, unitriangular hence invertible over
$\FF$.
\end{background}

\Needspace{5\baselineskip}
\begin{background}[Increment is not multiplicative]\label{bg:V-4-3}
In general $\Inc(A\conv B)\ne\Inc(A)\conv\Inc(B)$.
For $A=\{1\}$ and $B=\{2\}$, the two sides are $\{4\}$ and
$\{3\}$ respectively, on any window containing $0,\ldots,4$.
\end{background}

\Needspace{5\baselineskip}
\begin{background}[Triangular carry matrix]\label{bg:V-4-4}
In the zeta basis $\conv$ is a pointwise product while the carry is a full triangle:
$\Zet\Inc\Zet^{-1}=L$ with $L[i,j]=[i>j]$, nilpotent, and $\Zet(I+\Inc)\Zet^{-1}=I+L$,
unipotent.
\end{background}

\Needspace{5\baselineskip}
\begin{background}[A repeated-root column]\label{bg:III-3-20}
For $n\ge1$, the sequence with generating function
$z^{n-1}/(1-\alpha z)^n$ is zero at $k<n-1$ and equals
$\binom{k}{n-1}\alpha^{k-n+1}$ at $k\ge n-1$.
The convention $0^0=1$ includes $\alpha=0$. Equivalently this is the
solution of the recurrence with characteristic polynomial
$(x-\alpha)^n$ and initial data $0,\ldots,0,1$.
\end{background}

\Needspace{5\baselineskip}
\begin{background}[Universality criterion for a diagonal basis]\label{bg:XI-4-1}
The diagonal basis of a rule over $\ZZ_N$ from $\delta_0$ is lower triangular for every
window exactly when $p(0,0,v)\equiv0$ for every $v$ in the value set of the centre column
together with $0$; it is universal when in addition every centre--column entry is a unit.
\end{background}

\begin{figure}[!htbp]
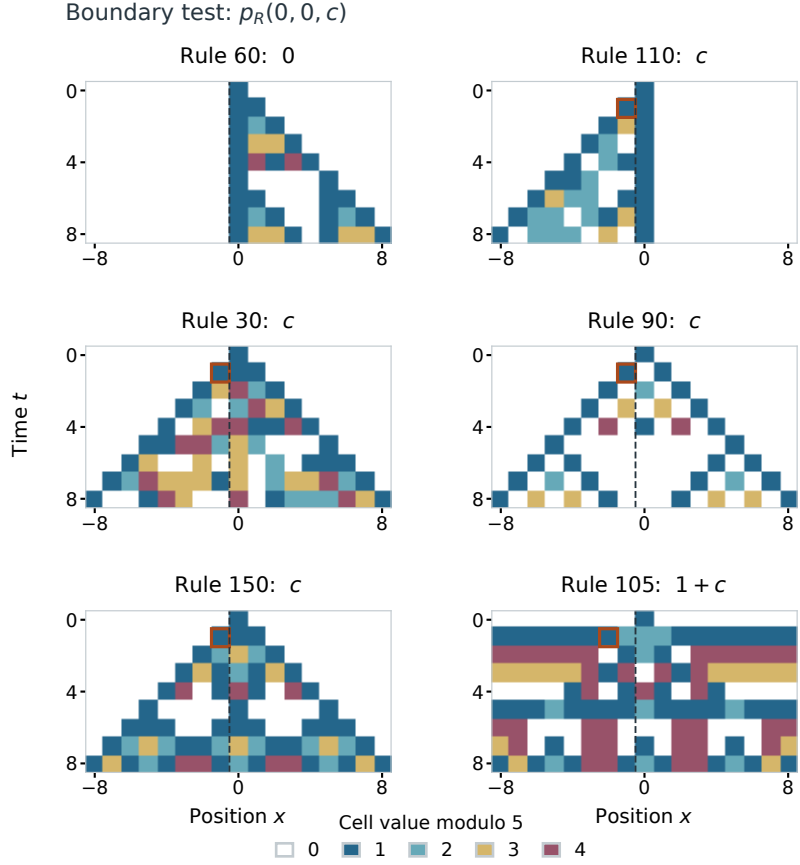
\centering
\EmbeddedFigure{\textwidth}{1}
\caption{The left-boundary test for six canonical
polynomial lifts modulo five, from a single seed, on
$0\le t\le8$ and $-8\le x\le8$. Each title gives
$p_R(0,0,c)$; the dashed line separates $x<0$ from $x\ge0$.
For Rule 60 this polynomial is zero, so the left half-plane
remains zero by Theorem~\ref{thm:tri}. For Rules 30, 90, 110 and
150, it equals $c$, and the outlined cell $A(1,-1)=1$ already
violates that condition. For Rule 105 it is $1+c$; the outlined
cell $A(1,-2)=1$ exhibits the evolving background, which is
retained in every update. The discrete key identifies all five
residues, including zero as white.}
\label{fig:rules5}
\end{figure}
\section{The mod-$N$ lift and the diagonal family}

The first task is to specify which modular dynamics defines the
columns. A Boolean rule admits many integer polynomials with the
same reduction modulo two, and their reductions at other moduli
can differ. We therefore fix the canonical lift before forming
the diagonal matrix. The definitions below tie its entries to a
particular orbit and initial configuration, so that all later
invertibility statements refer to the same family.

Throughout, $N\ge2$ is an arbitrary modulus, not necessarily prime, and $\Zm$ is the ring
$\mathbb Z/N\mathbb Z$.

\Needspace{5\baselineskip}
\begin{definition}[Canonical polynomial lift]\label{def:lift}
Let $R\in\{0,\ldots,255\}$ have truth-table value equal to bit
$4a+2b+c$ of $R$, for $a,b,c\in\{0,1\}$. Write its unique
multilinear algebraic normal form over $\FF$ and lift its coefficients
to $\{0,1\}\subset\ZZ$:
\[
 p_R(a,b,c)=\sum_{S\subseteq\{a,b,c\}}\varepsilon_S
                  \prod_{v\in S}v\pmod N.
\]
The variables $a,b,c$ are the left, center and right neighbors.
\end{definition}

The canonical multilinear form is expanded over $\FF$ before lifting.
For example, $a\vee b$ has lift $a+b+ab$, whereas $a\xor a$ has
lift zero. Translating an unreduced Boolean expression term by term
could instead give $2a$ and would define a different map at odd
moduli. The convention above removes this ambiguity.

For an initial condition $c_0=(c_0(0),c_0(1),\dots)$ placed at $x=0,1,\dots$ on a zero
background, write $A^{R,c_0}(t,x)$, $t\ge0$, $x\in\ZZ$, for the orbit under $p_R$.

\Needspace{5\baselineskip}
\begin{definition}[Diagonals]\label{def:diag}
The \emph{$m$-th diagonal} of $(R,c_0)$ is the sequence
$D_m^{R,c_0}(t)=A^{R,c_0}(t,\,t-m)$, $t=0,1,2,\dots$ Truncated to a window of length $n$ it
is a vector of $\Zm^n$.
\end{definition}

This is the cut used in the diagonal--period census of Section~\ref{sec:scan} below
(there written $\mathrm{Diagonal}[\,\cdot\,,\;\text{size}-m]$), and it is the cut along
which the support--set calculus of \cite{U1} is defined: $D_m$ runs parallel to the right boundary of the light cone
at depth $m$.

\Needspace{5\baselineskip}
\begin{definition}[Diagonal basis matrix]\label{def:basis}
For a \emph{basis rule} $R'$ with initial condition $c'$ and window length $n\ge1$, let
\[
 B_{R',c'}[t,k]\;=\;D_k^{R',c'}(t)\;=\;A^{R',c'}(t,\,t-k),\qquad 0\le t,k\le n-1 .
\]
Column $k$ is the $k$-th diagonal of the basis automaton; the matrix is $n\times n$ over
$\Zm$.
\end{definition}

The decomposition problem is then: given a target vector $y\in\Zm^n$ --- in practice a
diagonal $D_m^{R,c_0}$ of some other automaton --- find $x\in\Zm^n$ with $B_{R',c'}x=y$, and
find it with as few nonzero coordinates as possible. For a fixed
invertible matrix there is only one such $x$; optimization then
means choosing the basis family. For a rank-deficient or degenerate
array, several solutions can exist and minimizing within a fiber is
a separate problem. The staircase solver below addresses exact
coordinates under its stated pivot hypotheses.

\section{The greedy staircase solver over $\Zm$}

A triangular arrangement suggests solving for one coefficient at a
time. Over a ring, each step is a linear congruence whose solvability
and number of solutions depend on the pivot. The solver records that
dependence explicitly. Its role is both constructive and diagnostic:
it produces coordinates when the hypotheses apply and identifies
the obstruction when a target cannot be reached or its coordinates
are not uniquely determined.

Since $N$ need not be prime, $\Zm$ is not a field and column operations must be done with
gcd arithmetic. Write $\lead(k)$ for the smallest $t$ with $B[t,k]\neq0$ (undefined for a
zero column), and $\Lambda=\{\lead(k)\}_k$ for the set of \emph{leading positions}.

\Needspace{5\baselineskip}
\begin{algo}[Greedy leading--position elimination]\label{alg:greedy}
Set $r\leftarrow y$, $P\leftarrow\varnothing$, and pool the columns by leading position. For
$t=0,1,\dots,n-1$: if $r_t=0$, continue. Otherwise, among the unused columns $k$ with
$\lead(k)=t$ keep those with $\gcd(B[t,k],N)\mid r_t$, and choose one, $k^\ast$, minimising
$\gcd(B[t,k],N)$. Put $u=B[t,k^\ast]$, $g=\gcd(u,N)$ and
\begin{equation}\label{eq:coef}
 c\;=\;\Big(\tfrac{r_t}{g}\Big)\cdot\Big(\tfrac ug\Big)^{-1} \bmod \tfrac Ng ,
\end{equation}
represented by the integer $0\le c<N/g$. Break ties by the smallest
column index. Set $r\leftarrow r-c\,B[\cdot,k^\ast] \bmod N$, append $(k^\ast,c,t)$ to
$P$, and mark $k^\ast$ used. If no admissible column exists, leave $t$ unmatched and
continue. Return $P$, the residual $r$, and the flag $r=0$.
\end{algo}

\begin{lemma}[Local solvability]\label{lem:local}
With $u,g$ as above, the congruence $cu\equiv r_t \pmod N$ has a solution iff $g\mid r_t$,
and then \eqref{eq:coef} is one; the solution set is a coset of $\frac Ng\Zm$ of size $g$.
\end{lemma}

\begin{proof}
$cu\equiv r_t$ is solvable iff $\gcd(u,N)\mid r_t$, the standard criterion. Dividing by $g$
gives $c\,(u/g)\equiv r_t/g \pmod{N/g}$ with $\gcd(u/g,N/g)=1$, so $u/g$ is invertible mod
$N/g$ and $c\equiv (r_t/g)(u/g)^{-1}$; lifting the residue mod $N/g$ to $\Zm$ produces
exactly $g$ solutions.
\end{proof}

\Needspace{5\baselineskip}
\begin{theorem}[Monotone matching]\label{thm:monotone}
After the step that treats position $t$, the residual satisfies $r_s=0$ for every $s\le t$
that was matched, and no later step changes $r_s$ for $s<t$. In particular the algorithm
terminates after at most $n$ steps and uses $O(n^2)$ ring operations,
in addition to the gcd and inverse calculations.
\end{theorem}

\begin{proof}
Every column selected at position $t$ has $\lead(k^\ast)=t$, hence vanishes identically above
row $t$; subtracting a multiple of it leaves rows $0,\dots,t-1$ untouched. By
Lemma~\ref{lem:local} the subtraction sets $r_t=0$ when an admissible column exists. Positions
are treated in increasing order and each column is used at most once, giving at most $n$
subtractions of cost $O(n)$ each.
\end{proof}

There are at most $n$ gcd tests and $n$ inversions when the $n$ columns
are pooled once by leading position; the pools together contain at most
$n$ entries. With $b=\lceil\log_2 N\rceil$, schoolbook arithmetic gives
the conservative bit bound $O(n^2b^2+nb^3)$. The stored matrix takes
$O(n^2b)$ bits. The residual, coefficients and column indices use
$O(n(b+\log(n+1)))$ additional bits. These costs are for a supplied
matrix; generating it takes a separate $O(n^2)$ local updates.

\Needspace{5\baselineskip}
\begin{remark}[What an unmatched position means]\label{rem:residual}
If $t\notin\Lambda$, the procedure makes no new pivot at row $t$.
This is a restriction of the procedure, not a claim about the column
span. Over $\FF$, the matrix
\[
 B=\begin{pmatrix}1&1\\0&1\end{pmatrix},\qquad
 y=\begin{pmatrix}0\\1\end{pmatrix}
\]
has both leads at zero. The algorithm returns $y$ as residual, although
$B(1,1)^{\mathsf T}=y$. Cancellation between two earlier-leading
columns creates a later-leading vector.
\end{remark}

\section{The trichotomy}

The elimination procedure can now be related to the geometry of the
single-seed orbit. The relevant cases arise from the position of the
first possible nonzero entry and the arithmetic of the resulting
diagonal entries. Stating them together distinguishes a failure of
triangular shape from a failure of invertibility within that shape.
This distinction is what allows the subsequent rule classification
to be read directly from local data and a center-column condition.

\Needspace{14\baselineskip}
\begin{definition}\label{def:tri}
The basis $B=B_{R',c'}$ (window $n$, modulus $N$) is
\begin{enumerate}
\renewcommand{\labelenumi}{(\alph{enumi})}
\item \emph{universal} if no column is zero, $\lead(k)=k$ for all $k$, and every $B[k,k]$ is
a unit of $\Zm$;
\item a \emph{proper staircase} if it is not universal, $\lead$ is injective on the nonzero columns and every
leading entry is a unit;
\item \emph{degenerate} otherwise.
\end{enumerate}
\end{definition}

\Needspace{5\baselineskip}
\begin{theorem}[Universal case]\label{thm:universal}
If $B$ is universal then it is invertible over $\Zm$, every $y\in\Zm^n$ has a
\emph{unique} preimage $x$, and Algorithm~\ref{alg:greedy} returns it. Consequently the
number of terms
\[
 w_{R',c'}(y)\;=\;\#\{k: x_k\neq0\}
\]
is a well--defined function of $y$ and of the basis, not an artefact of the greedy order.
\end{theorem}

\begin{proof}
$B$ is lower triangular with unit diagonal entries, hence $\det B$ is a unit and $B$ is
invertible over $\Zm$; uniqueness is immediate. At position $t$ the pool contains exactly the
single column $k=t$, with $g=\gcd(B[t,t],N)=1$, so Lemma~\ref{lem:local} gives the unique
$c=r_tB[t,t]^{-1}$ and no choice is made anywhere. By Theorem~\ref{thm:monotone} the residual
is annihilated position by position, so the output is the solution, and it is the only one.
\end{proof}

\Needspace{5\baselineskip}
\begin{theorem}[Staircase coordinates and residual]\label{thm:staircase}
Suppose the $r$ nonzero columns have distinct leading positions and
unit leading entries. Let $C$ be these columns, ordered by their leads,
and let $E$ be the coefficient map returned by the algorithm. Then
\[
 EC=I_r,\qquad P=CE,\qquad P^2=P.
\]
Every $y$ has the unique decomposition
$y=CEy+\rho$, where $\rho$ vanishes at all pivot rows. Thus
\[
 \ZZ_N^n=\operatorname{im}C\oplus
       \{\rho:\rho_t=0\text{ for }t\in\Lambda\}.
\]
In particular the algorithm is exact if and only if $y\in\operatorname{im}B$.
\end{theorem}
\begin{proof}
Write the pivot rows as $t_1<\cdots<t_r$. At $t_j$ the coefficient is
\[
 (Ey)_j=C[t_j,j]^{-1}
       \left(y_{t_j}-\sum_{i<j}C[t_j,i](Ey)_i\right).
\]
This formula is linear in $y$, and later columns vanish at $t_j$.
Consequently $y-CEy$ vanishes at every pivot row. Applied to $y=Cz$,
the same induction returns $z_j$ at each step, so $EC=I_r$ and
$P^2=CECE=CE$. A vector in $\operatorname{im}C$ that vanishes at
the pivots has first coefficient zero, then second coefficient zero,
and so on; hence the intersection is zero. This proves uniqueness of
the decomposition and the membership criterion. Zero columns contribute
arbitrary unused coordinates and do not affect it.
\end{proof}

\Needspace{5\baselineskip}
\begin{remark}[Degenerate arrays]\label{rem:degenerate}
An exact output always certifies membership, but a nonzero residual
does not certify nonmembership outside the staircase hypotheses.
Repeated leads are one obstruction. Nonunit pivots are another:
over $\ZZ_6$, the one-row matrix $(2\;3)$ spans the target $1$,
since $2\cdot2+3\cdot1=1$, although neither column can match it alone.
Over $\ZZ_4$, the column $(2,1)^{\mathsf T}$ spans $(0,2)^{\mathsf T}$
with coefficient $2$; skipping the zero first residual misses that
solution. A different lift of a nonunit-pivot coefficient may therefore
matter at later rows. The reported residual remains an explicit check
of the proposed expansion.
\end{remark}

\Needspace{5\baselineskip}
\begin{remark}[An example with missing columns]\label{ver:staircase}
Rule 90 supplies a staircase with half its columns zero. Its image and
the resulting residual criterion are described explicitly in
Theorem~\ref{thm:90}.
\end{remark}

\section{Which automata are universal bases}

We now require the basis property at every finite length. This
strengthens the finite-matrix question and makes a uniform criterion
necessary. The rule-bit test supplies the triangular part of that
criterion, while the unit condition controls every diagonal pivot.
The classification below first settles the binary case and then
uses reduction at prime divisors to organize the questions for
other moduli.

\Needspace{5\baselineskip}
\begin{theorem}[Triangularity from two rule bits]\label{thm:tri}
For the canonical lift from $\delta_0$, the following conditions are
equivalent, for every fixed $N\ge2$:
\begin{enumerate}
\item $B_R$ is lower triangular for every window length;
\item with $V=\{A^{R,\delta_0}(t,0):t\ge0\}$,
\begin{equation}\label{eq:leftcone}
 p_R(0,0,v)=0\quad(v\in V\cup\{0\});
\end{equation}
\item $p_R(0,0,c)=0$ for every $c\in\ZZ_N$;
\item $R\equiv0\pmod4$.
\end{enumerate}
Thus exactly 64 elementary rules have triangular diagonal arrays, at
every modulus.
\end{theorem}
\begin{proof}
The entries above the diagonal are $A(t,t-k)$ with $t<k$.
Allowing all windows includes every cell at $x<0$, so (1) says that
the left half-plane stays zero. It implies (2): at $x=-2$ the update
tests $p_R(0,0,0)$, and at $x=-1$ it tests $p_R(0,0,A(t,0))$.

Now $p_R(0,0,c)=\varepsilon_\varnothing+\varepsilon_c c$.
The test at zero forces $\varepsilon_\varnothing=0$, since its
representative is zero or one. The seed gives $1\in V$, so the test at
one forces $\varepsilon_c=0$. These are precisely truth-table bits
zero and one. Hence (2) implies (4), which implies (3). Finally,
under (3), induction preserves zero on $x<0$: the neighborhoods there
are $(0,0,0)$ or $(0,0,A(t,0))$. This proves (1).
\end{proof}

\Needspace{5\baselineskip}
\begin{corollary}[Universality criterion]\label{thm:criterion}
$B_{R',\delta_0}$ is universal for every window length $n$ if and only if
\eqref{eq:leftcone} holds and every entry of the centre column is a unit of $\Zm$:
$\gcd\big(A^{R',\delta_0}(t,0),N\big)=1$ for all $t\ge0$.
\end{corollary}

\begin{proof}
Triangularity is Theorem~\ref{thm:tri}; the diagonal entries are $B[k,k]=A(k,0)$, so
Definition~\ref{def:tri}(a) is precisely the unit condition on the centre column.
\end{proof}

\Needspace{5\baselineskip}
\begin{remark}[No orbit is needed for triangularity]\label{rem:uniform}
For this canonical elementary lift, the uniform test and the orbit-value
test are equivalent. There is no exceptional modulus and no rule
separating them. This simplification uses both the binary coefficient
representatives and the seed value one; the orbit-value formulation
remains useful for more general local polynomials.
\end{remark}

\Needspace{5\baselineskip}
\begin{remark}[Finite center-column tests]\label{ver:criterion}
At window length 64, among all 256 rules and moduli $2,\ldots,13$,
the accepted rules are the 24 in Proposition~\ref{prop:congruence},
with Rule 84 added at moduli 3 and 9. Restricting this test to moduli
$2,\ldots,8$ therefore has 169 accepted rule--modulus pairs on this
window. Acceptance of Rule 84 here is a finite observation. A failed
unit test supplies a counterexample to all-window universality;
survival of a finite test does not prove it.
\end{remark}

\Needspace{5\baselineskip}
\begin{proposition}[The 24 binary rules]\label{prop:congruence}
The rules universal over $\FF$ for every window are exactly
\[
 \mathcal U=\{R:R\equiv12\pmod{16}\}
 \cup\{R:R\equiv4\pmod{16},\ R\mathbin{\&}16=0\}.
\]
There are 24, and every member is universal over every $\ZZ_N$.
The only possible further universal rules at any modulus are
\[
 \mathcal E=\{20,52,84,116,148,180,212,244\}.
\]
None of these is universal at an even modulus.
\end{proposition}
\begin{proof}
Triangularity forces bits zero and one to vanish. The center value at
time one is bit two, so that bit must be one. On the left boundary the
canonical polynomial is now
$p_R(0,b,c)=b+\varepsilon_{bc}bc$.
If bit three is one, then $\varepsilon_{bc}=0$ and the center stays
one at every modulus. This gives $R\equiv12\pmod{16}$.

If bit three is zero and bit four is zero, the single seed is stationary:
the only neighborhoods are $000,001,010,100$, with outputs $0,0,1,0$.
This remains true for the lifted polynomial over every modulus and
gives the other eight members of $\mathcal U$.

The remaining possibility is that bit three is zero and bit four is
one. These are exactly the eight rules in $\mathcal E$. At time one,
the center and its right neighbor both equal one. At time two the
center equals $1+1=2$, which is not a unit at an even modulus.
Over $\FF$ this exhausts all possibilities. Counting the free bits
gives 16 rules in the first class and eight in the second.
\end{proof}

\Needspace{5\baselineskip}
\begin{corollary}[Reduction to prime divisors]\label{cor:crt}
A canonical single-seed rule is universal over $\ZZ_N$ for all windows
if and only if it is universal over $\mathbb F_p$ for every prime
$p\mid N$. In particular, universality modulo $p^a$ is equivalent
to universality modulo $p$.
\end{corollary}
\begin{proof}
The integer polynomial update commutes with reduction modulo any
divisor of $N$. Its center value is a unit modulo $N$ exactly when
its reduction is nonzero at every prime divisor. Triangularity is the
same two-bit condition at all moduli. Apply
Corollary~\ref{thm:criterion} at each time. The same reasoning holds
when the times are restricted to a finite window.
\end{proof}

\Needspace{5\baselineskip}
\subsection{Finite seeds and composite periods}\label{sec:seedextensions}

The universal family contains a subfamily whose proof extends to
arbitrary data to the right of a fixed left endpoint. This gives
an all-window seed result directly from the local polynomial.

\begin{proposition}[An all-window family for finite seeds]\label{prop:finiteseed}
Let $R\equiv12\pmod{16}$ and use its positive ANF lift modulo
$N\ge2$. Let a finite initial configuration $c$ satisfy $c(x)=0$
for $x<0$ and let $c(0)$ be a unit modulo $N$. Then every
associated diagonal matrix $B_n$ is lower triangular with constant
diagonal $c(0)$, and hence invertible over $\ZZ_N$.
\end{proposition}
\begin{proof}
The four truth-table bits with left input zero are $0,0,1,1$.
Their ANF restriction is the polynomial identity $p_R(0,b,c)=b$,
which holds over the integers, not merely on binary inputs.
Inductively the configuration stays zero at all $x<0$ and its
value at $x=0$ remains $c(0)$: the update at $-1$ is
$p_R(0,0,c(0))=0$, and that at $0$ is $p_R(0,c(0),c(1))=c(0)$.
For the matrix entry $B_n(t,k)=A_R(t,t-k)$, $k>t$ therefore
gives zero and $k=t$ gives $c(0)$. The determinant is $c(0)^n$,
a unit. This proves the assertion for every window and for any
values of the remaining seed entries.
\end{proof}

The hypothesis on the endpoint is necessary for this particular
triangular construction: already the first diagonal entry must
be a unit. The proposition does not classify every finite seed
for the other members of the single-seed family.

\begin{proposition}[Least periods under the Chinese remainder map]\label{prop:crtperiod}
Let $N=\prod_i p_i^{e_i}$. A sequence modulo $N$ is eventually
periodic if and only if all its reductions modulo $p_i^{e_i}$
are eventually periodic. If those components have least eventual
periods $P_i$ and admit transients $T_i$, the original sequence
has least eventual period
\[
 P=\operatorname{lcm}_i P_i
\]
after the valid transient bound $T=\max_i T_i$. If every component
is purely periodic, the same formula gives the least pure period.
\end{proposition}
\begin{proof}
Equality modulo $N$ is equivalent to equality in every prime-power
component. Thus $P$ is a period after $T$. Conversely, an eventual
period of the original sequence is an eventual period of each
component and must be divisible by its least eventual period.
It is consequently divisible by their least common multiple.
The same argument from time zero proves the pure-period assertion.
\end{proof}

Prime reduction suffices for the unit criterion of
Corollary~\ref{cor:crt}; exact period reconstruction requires
prime-power components. A period observed modulo a prime alone
need not determine the period at its higher powers.

\begin{remark}[The outstanding Rule 84 case]\label{rem:84}
The polynomial of Rule 84 is $(a+b+ab)(1+c)$. Modulo three, its
center column begins $1,1,2,2,1,2,2,\ldots$ and follows the repeating
block $(1,2,2)$ after the first entry for the tested times
$0\le t<2048$. All center values are also units modulo 9 and 27
on that range, as Corollary~\ref{cor:crt} predicts.
An all-time proof of the observed pattern is not supplied here.
If Rule 84 is universal modulo three, it is universal at every power
of three; three could not be the sole exceptional modulus.
\end{remark}

\Needspace{5\baselineskip}
\begin{remark}[The 32-term census]\label{ver:census}
Mod~2, window $n=32$, seed $\delta_0$: of the $256$ elementary bases, $24$ are universal,
$72$ are proper staircases and $160$ are degenerate (Fig.~\ref{fig:coverage}). Staircase
coverage $|\Lambda|$ takes the values $1,2,3,4$ for the collapsed rules, $16$--$20$ for the
half--rank rules (Rules~90 and~170 among them) and $32$ for the universal ones.
\end{remark}

\begin{figure}[!htbp]\centering
\EmbeddedFigure{\textwidth}{2}
\caption{The binary basis trichotomy of
Definition~\ref{def:tri}. In (a), the $32$-term census has the
same class for the eight rules $R=32q+r$, $0\le q<8$, represented
by each numbered residue cell. Thus the full $256$-rule census is
shown with $24$ universal, $72$ proper staircase and $160$
degenerate cases. The universal residues $4,12,28$ modulo $32$
are exactly the all-window classification in
Proposition~\ref{prop:congruence}; the other displayed classes
refer to the stated finite census. In (b), the $8\times8$ matrices
for Rules 60, 90 and 30 explain the distinction. Dark cells are
ones and circles mark column leads: every diagonal pivot is present
for Rule 60, the nonzero columns have distinct leads for Rule 90,
and several columns share a lead for Rule 30.}
\label{fig:coverage}
\end{figure}

\section{Four named bases}

Explicit examples make it possible to compare more than invertibility.
The four bases considered here have recognizable coordinate maps,
so their inverses and the resulting weights can be written without
performing a general elimination. They provide reference cases for
the later comparisons: a coordinate change can simplify one
operation, alter support size, or expose a cumulative structure
that is hidden in the original diagonal matrix.

\Needspace{5\baselineskip}
\begin{theorem}[Pascal basis]\label{thm:pascal}
For every modulus $N\ge2$, $B_{60,\delta_0}[t,k]=\binom tk \bmod N$.
\end{theorem}

\begin{proof}
$\mathrm{ANF}_{60}=a\xor b$, so $p_{60}(a,b,c)=a+b$ and the orbit from $\delta_0$ obeys
$A(t+1,x)=A(t,x-1)+A(t,x)$, Pascal's recurrence with $A(0,x)=[x=0]$; hence
$A(t,x)=\binom tx\bmod N$. Therefore $B[t,k]=A(t,t-k)=\binom t{t-k}=\binom tk$.
\end{proof}

\Needspace{5\baselineskip}
\begin{corollary}[The Pascal basis is the zeta transform]\label{cor:zeta}
Over $\FF$, $B_{60,\delta_0}[t,k]=[\,k\subseteq t\,]$, where $k\subseteq t$ means containment
of binary supports. That is, $B_{60}$ is the \emph{zeta matrix} of the Boolean lattice on
$b$ bits when $n=2^b$, and
\[
 B_{60}\;=\;\begin{pmatrix}1&0\\1&1\end{pmatrix}^{\!\otimes\,\log_2 n},
 \qquad B_{60}^{-1}=B_{60} .
\]
\end{corollary}

\begin{proof}
Lucas' theorem: $\binom tk$ is odd iff the addition $k+(t-k)$ carries nowhere in base~2,
i.e.\ iff the binary support of $k$ is contained in that of $t$. The containment predicate
factorises over the bit positions, giving the Kronecker form; over $\FF$ the zeta transform is
its own inverse.
\end{proof}

\Needspace{5\baselineskip}
\begin{corollary}[Cost of the Pascal change of basis]\label{cor:fft}
For $n=2^b$ the Kronecker form is a butterfly: both $x\mapsto B_{60}x$ and $y\mapsto B_{60}^{-1}y$ run in
$O(n\log n)$ ring operations and $O(n)$ space, despite $B_{60}$ having
$3^{\log_2 n}=n^{\log_23}\approx n^{1.585}$ nonzero entries. Individual entries are available
in $O(\log n)$ without forming the matrix.
\end{corollary}
\begin{proof}
Apply the $2\times2$ zeta factor successively along each of the $\log_2n$ bit
coordinates. Every stage performs $O(n)$ additions and may be done in place, giving
$O(n\log n)$ time and $O(n)$ storage; the same algorithm inverts the transform over
$\FF$. The Kronecker product has $3^{\log_2n}$ nonzero entries because its factor has
three. Lucas' containment test reads an individual entry in $O(\log n)$.
For arbitrary $n$, zero-pad to $N=2^{\lceil\log_2\max(1,n)\rceil}$, apply the
transform, and truncate to the first $n$ entries. Lower triangularity makes
these entries independent of the padding; the leading principal matrix is
still its own inverse over $\FF$. This costs $O(n\log(n+1))$ Boolean operations.
\end{proof}

\Needspace{5\baselineskip}
\begin{remark}[Sparsity is not cost]\label{rem:sparsity}
Corollary~\ref{cor:fft} shows that the number of nonzero entries is the wrong measure of the
price of a change of basis; what matters is whether the matrix admits a factorisation into
$O(\log n)$ sparse factors, or a short linear recurrence. Section~\ref{sec:usable} classifies
the $24$ universal bases on that criterion, and the classification does not follow the density
ordering.
\end{remark}

\Needspace{5\baselineskip}
\begin{proposition}[Mirror]\label{prop:102}
$\mathrm{ANF}_{102}=b\xor c$ and the mirrored cut $\widetilde D_k(t)=A^{102,\delta_0}(t,k-t)$
satisfies $\widetilde D_k(t)=\binom tk \bmod N$. Rules 60 and 102 are reflections of one
another and carry the same basis on opposite cuts.
\end{proposition}

\begin{proof}
$p_{102}(a,b,c)=b+c$ gives $A(t+1,x)=A(t,x)+A(t,x+1)$, so $A(t,x)=\binom t{-x}$, and
$A(t,k-t)=\binom t{t-k}=\binom tk$.
\end{proof}

\Needspace{5\baselineskip}
\begin{remark}[Small Pascal matrices]\label{ver:pascal}
At length three, Rule 60 gives
$\left(\begin{smallmatrix}1&0&0\\1&1&0\\1&2&1\end{smallmatrix}\right)$
modulo $N$. Its determinant is one even when $N$ is composite.
Reflecting Rule 102 and its diagonal cut gives the same entries.
\end{remark}

\Needspace{5\baselineskip}
\begin{proposition}[Delta, prefix and parity bases]\label{prop:named}
Over $\FF$, with seed $\delta_0$,
\[
 \begin{split}
 B_{204}&=I,\qquad B_{220}[t,k]=[k\le t],\\
 B_{28}[t,k]&=[k=0]+[1\le k\le t,\ t\equiv k\pmod2].
 \end{split}
\]
Thus $w_{204}$ is Hamming weight, and the Rule 220 coordinates are
$x_k=y_k+y_{k-1}$, with $y_{-1}=0$.
\end{proposition}
\begin{proof}
Rule 204 fixes the seed. Rules 220 and 252 give
$A(t,x)=[0\le x\le t]$: interior $111$, left boundary $011$,
right boundary $110$ and advancing edge $100$ all update to one,
while $000$ and $001$ update to zero. The inverse is first difference.

For Rule 28 the physical orbit is
\[
 A(t,x)=[x=t]+[0\le x<t,\ x\text{ even}].
\]
Its seed row is correct. In the alternating interior, $010$ updates
to one and $101$ to zero. At $x=0$, both $010$ and $011$ give one.
The old right edge updates to one at even $t$ and zero at odd $t$,
as required for the next interior; $100$ creates the new edge.
All exterior cells remain zero. This proves the formula by induction,
and substituting $x=t-k$ gives the matrix. The same orbit contains
no $111$ neighborhood, so Rule 156, which differs only there, gives
the same matrix.
\end{proof}

\begin{figure}[t]
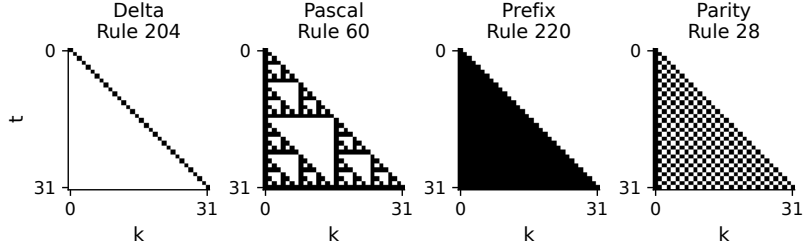
\centering
\EmbeddedFigure{\textwidth}{3}
\caption{The four named universal bases as $32\times32$ matrices over $\FF$
(black $=1$). Left to right: delta (Rule~204), Pascal (Rule~60), prefix--sum (Rule~220),
parity staircase (Rule~28).}
\label{fig:bases}
\end{figure}

\section{The support calculus inside the construction}

So far the columns have been treated as vectors. Their origin as
cellular-automaton sequences supplies additional structure through
Newton supports. Products and integrations of sequences correspond
to specific operators on these supports. Expressing those operators
inside the basis construction allows us to ask a sharper question
than sparsity alone: which coordinates simplify the actual operations
used to propagate the rule?

\Needspace{5\baselineskip}
\begin{theorem}[Support coordinates on every prefix]\label{thm:S}
For the Rule 30 diagonal $D_m$ over $\FF$, the unique solution of
$B_{60}x=D_m$ on $0\le t<n$ is the indicator of
$S_{m+1}\cap[0,n)$. In particular,
\[
 D_m(t)=\bigoplus_{r\in S_{m+1}}\binom tr\pmod2.
\]
The entire support is present when $n>\max S_{m+1}$.
\end{theorem}
\begin{proof}
The support expansion in Proposition~\ref{bg:I-1-18} uses the
columns of Theorem~\ref{thm:pascal}. For $t<n$, every omitted column
with $r\ge n$ is zero at row $t$. Restricting the expansion therefore
gives the indicated solution on any prefix. It is unique by
Theorem~\ref{thm:universal}.
\end{proof}

\Needspace{5\baselineskip}
\begin{corollary}\label{cor:weight}
On a length-$n$ window,
$w_{60}(D_m)=|S_{m+1}\cap[0,n)|$. This equals $|S_{m+1}|$ when
the whole support is included.
\end{corollary}
\begin{proof}
The coordinate vector in Theorem~\ref{thm:S} is an indicator vector.
Counting its nonzero entries gives the result.
\end{proof}

\Needspace{5\baselineskip}
\begin{remark}[The index shift in the support notation]\label{ver:S}
Our diagonal index starts at zero: $D_m=b(m+1,\cdot)$ in
\cite{U1}. Thus the first fourteen weights, for $m=1,\ldots,14$ at
$n=64$, are
$1,1,1,3,3,3,7,13,11,27,27,29,23,29$.
They are $|S_{m+1}|$, not $|S_m|$. At a shorter window the intersection
in Corollary~\ref{cor:weight} must be retained.
\end{remark}

\begin{figure}[t]
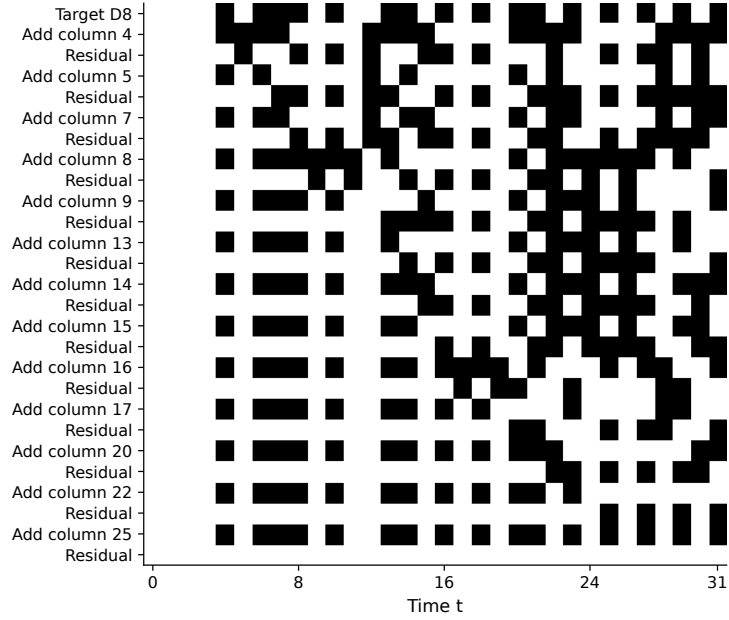
\centering
 \EmbeddedFigure{0.90\textwidth}{4}
 \caption{Algorithm~\ref{alg:greedy} on a Rule~30 diagonal in the Pascal basis: each row adds one
column, the row below it shows the residual. No step disturbs a cell already fixed, which is
Theorem~\ref{thm:monotone}, and the final residual is empty, which is
Theorem~\ref{thm:S}.}
 \label{fig:recon2}
\end{figure}

\section{The weight family, and the non--optimality of the binomial basis}

Corollary~\ref{cor:weight} makes $|S_m|$ one coordinate of a vector of invariants
$\big(w_R(D_m)\big)_{R\ \mathrm{universal}}$. Since a change of universal basis is itself a
triangular unimodular transformation, these weights are not comparable a priori, and no basis
can minimize the weight of every target. Indeed, a non-Pascal basis
column has weight one in its own basis; Pascal can improve other targets.

\begin{table}[t]\centering\small
\caption{Number of terms $w_R(D_m)$ for the Rule~30 diagonals, mod~2, at window $n=512>\max S_{m+1}$ for $1\le m\le24$. The Pascal column reproduces $|S_{m+1}|$ exactly, as
Theorem~\ref{thm:S} requires, and attains the minimum over the $24$ universal bases for every
$m$ except $m=4$.}
\label{tab:weights}
{\fontsize{9}{10.5}\selectfont\setlength{\tabcolsep}{3pt}
\begin{tabular}{|r|r|r|r|r|r|l|}
\hline
$m$ & $|S_{m+1}|$ & 60 & 204 & 220 & 28 & Minimizers\\
\hline
 1 & 1 & 1 & 256 & 511 & 1 & 28, 60, 156, 188\\
 2 & 1 & 1 & 256 & 511 & 1 & 28, 60, 156, 188\\
 3 & 1 & 1 & 256 & 255 & 510 & 60\\
 4 & 3 & 3 & 256 & 383 & 255 & 124\\
 5 & 3 & 3 & 192 & 383 & 127 & 60\\
 6 & 3 & 3 & 256 & 319 & 254 & 60\\
 7 & 7 & 7 & 256 & 351 & 191 & 60\\
 8 & 13 & 13 & 256 & 383 & 159 & 60\\
 9 & 11 & 11 & 256 & 335 & 254 & 60\\
10 & 27 & 27 & 256 & 383 & 191 & 60\\
11 & 27 & 27 & 248 & 383 & 143 & 60\\
12 & 29 & 29 & 248 & 399 & 142 & 60\\
13 & 23 & 23 & 232 & 415 & 111 & 60\\
14 & 29 & 29 & 232 & 383 & 159 & 60\\
15 & 25 & 25 & 256 & 327 & 254 & 60\\
16 & 55 & 55 & 256 & 371 & 191 & 60\\
17 & 79 & 79 & 256 & 383 & 163 & 60\\
18 & 95 & 95 & 256 & 383 & 182 & 60\\
19 & 105 & 105 & 252 & 367 & 223 & 60\\
20 & 123 & 123 & 254 & 367 & 211 & 60\\
21 & 127 & 127 & 236 & 387 & 166 & 60\\
22 & 117 & 117 & 244 & 351 & 219 & 60\\
23 & 105 & 105 & 250 & 367 & 191 & 60\\
24 & 125 & 125 & 256 & 373 & 170 & 60\\
\hline
\end{tabular}%
}
\end{table}

\Needspace{5\baselineskip}
\begin{proposition}[Finite sparsity comparisons]\label{prop:opt}
At $n=512$, Pascal minimizes $w_R(D_m)$ over $R\in\mathcal U$
for $1\le m\le24$, except at $m=4$: Rule 124 uses two terms and
Pascal uses three. At $n=4096$ the same conclusion holds through
$m=29$. It fails at $m=30$, where
\[
 w_{188}(D_{30})=1452<1925=w_{60}(D_{30}),
 \qquad \max S_{31}=3843<4096.
\]
Thus including the entire Pascal support does not make Pascal a
minimizer on every such window.
\end{proposition}
\begin{proof}
For each of the 24 rules, form the triangular matrix $A_R$ and the
target vector $v=(D_m(t))_{0\le t<n}$. Since every diagonal entry is
one, its unique coordinate vector is given over $\FF$ by
\[
 c_0=v_0,\qquad
 c_i=v_i+\sum_{j=0}^{i-1}(A_R)_{ij}c_j\quad(1\le i<n).
\]
Counting the nonzero $c_i$ gives the stated finite comparisons.
At depths 27 and 28,
the $n=4096$ weights for Pascal are $499,903$; for Rule 188,
$1511,1490$; for Rule 28, $1614,1535$; for delta, $2048,2044$; and
for prefix, $3027,3079$. The previously missing Rule 124 comparison
gives $2095,2010$, and Rule 92 gives $1614,1535$.
At depth 30 the same recurrence gives 1452 nonzero coordinates for
Rule 188 and 1925 for Pascal. Reconstruction follows from the
triangular equations at each row.
\end{proof}

\begin{figure}[t]
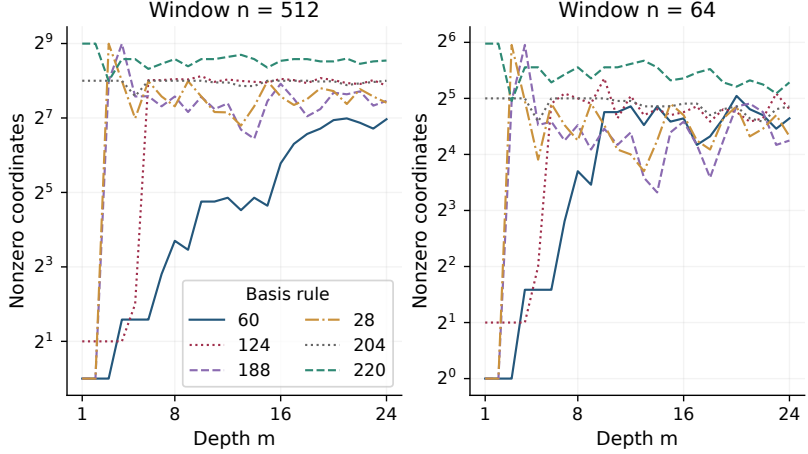
\centering
\EmbeddedFigure{\textwidth}{5}
\caption{Left: the weight family at window $n=512$ (log scale); the Pascal
curve attains the minimum except at $m=4$, where Rule 124 uses two terms
against Pascal's three. Right: the same computation at $n=64$, where the window is too
small from $m=15$ on and the ordering inverts --- the minimum is a property of that finite target. Even a window containing
the full Pascal support can favor another basis, as the depth-30
comparison in Proposition~\ref{prop:opt} shows.}
\label{fig:weights}
\end{figure}

\Needspace{5\baselineskip}
\begin{remark}[Basis search as preprocessing]
Computing $w_R(D_m)$ for all $24$ universal bases requires $24$ solves.
The dense triangular implementation costs $O(n^2)$ ring operations per solve.
In the tabulated range, Pascal is a minimizer except at $m=4$, where Rule 124
uses two terms instead of three. This finite comparison does not establish
optimality at untested depths or among block decompositions.
\end{remark}

\section{What makes a basis usable}\label{sec:usable}

The preceding comparisons show that Pascal coordinates need not
minimize weight. Their role in the support calculus therefore
requires another explanation. Theorem~\ref{thm:uniqueconv} supplies
it by characterizing the triangular transform that converts OR
convolution into multiplication. Using such a representation also
requires computing the change of basis, applying its operators
and recovering values. The following definitions keep these
costs separate, and the inverse and increment formulas show
why sparsity of a coordinate vector or matrix is only part
of the calculation.

We distinguish four properties of a compatible family of triangular
matrices, each referring to a specific part of this calculation.

\Needspace{15\baselineskip}
\begin{definition}\label{def:usable}
A universal binary basis family $B$ has
\begin{enumerate}
\renewcommand{\labelenumi}{(\alph{enumi})}
\item a \emph{banded inverse} if $B^{-1}$ has bandwidth $h=O(1)$ as $n$ varies;
\item a \emph{butterfly factorisation}, on $n=2^b$ entries, if $B=A^{\otimes b}$ for a $2\times2$ matrix $A$;
\item \emph{dyadic self--similarity} if, in the window of length $2n$,
$B_{2n}=\begin{psmallmatrix}B_n&0\\ X&B_n\end{psmallmatrix}$ for some $X$;
\item the \emph{convolution property} if $B(a\conv b)=(Ba)\cdot(Bb)$ pointwise, for the
OR--convolution $\conv$ of the support calculus.
\end{enumerate}
\end{definition}

\Needspace{5\baselineskip}
\begin{proposition}[Banded inverses]\label{prop:banded}
If a supplied binary inverse $B^{-1}$ has bandwidth $h$ then both $x\mapsto Bx$ and $y\mapsto B^{-1}y$ cost $O((h+1)n)$ ring operations: the
second directly, the first by solving the $(h+1)$-term linear recurrence $B^{-1}y=x$ for $y$ in
increasing order.
\end{proposition}
\begin{proof}
Multiplying by a bandwidth-$h$ inverse uses at most $O(h+1)$ operations per row. Conversely,
$B^{-1}y=x$ is a triangular recurrence involving at most $h$ previously computed
coordinates of $y$; forward substitution therefore also costs $O((h+1)n)$ and stores only
the vectors and a constant-width working band.
\end{proof}

\Needspace{5\baselineskip}
\begin{proposition}[Explicit short inverses]\label{ver:cost}
Among the 24 rules, sixteen give $I$, Rules 220 and 252 have inverse
bandwidth one, Rules 28, 92 and 156 have bandwidth two, and Rule 188
has bandwidth four. These statements hold for every window.
\end{proposition}
\begin{proof}
The sixteen stationary-seed rules have bits zero, one and four zero
and bit two one. The prefix and Rule 28 formulas follow from
Proposition~\ref{prop:named}. For Rule 28 (and 156),
$x_0=y_0$, $x_1=y_1+y_0$, and $x_t=y_t+y_{t-2}$ for $t\ge2$.

Rule 92 has physical row
\[
 A(t,x)=[x=t]+[t\ge1,\ x=t-1]
       +[0\le x<t-1,\ x\text{ even}].
\]
For $t=0,1$ this gives the seed and $11$. The alternating interior
uses $010\mapsto1$ and $101\mapsto0$. At the last three occupied
sites, the neighborhoods $011,111,110$ have outputs $1,0,1$;
$100$ creates the new right edge. These give the displayed row at
time $t+1$, according to the parity of $t$. The left boundary stays
one because $p_{92}(0,1,c)=1$ in $\FF$.
Its columns are constant for $k=0,1$ once $t\ge k$, and are
$[t\ge k,\ t\equiv k\pmod2]$ for $k\ge2$. The inverse differs
from Rule 28 only in $x_2=y_2+y_1$.

For Rule 188, whose local polynomial is $a+b+abc$,
\begin{equation}\label{eq:188orbit}
 A(t,x)=[0\le x\le t]\,[x=0\text{ or }x+t\not\equiv3\pmod4].
\end{equation}
In the interior, the four residues $x+t=0,1,2,3\pmod4$ give
neighborhoods $011,111,110,101$, with outputs $1,1,0,1$.
These are exactly the next-row residues. At $x=1$, replacing the
periodic left value by the actual value $A(t,0)=1$ can only change
$011$ to $111$, leaving the output one. At $x=0$ the output is
always one. The old right edge is $010$ at even $t$ and $110$ at
odd $t$, giving respectively one and zero; $100$ creates the next
edge. The seed handles the coinciding boundaries at $t=0$.
This proves \eqref{eq:188orbit} by induction.

Substitute $x=t-k$ in that formula. Cancellation of the four residue
classes gives $x_0=y_0$, $x_1=y_1+y_0$, $x_2=y_2+y_0$, and, for
$t\ge3$,
\[
 x_t=\begin{cases}
 y_t+y_{t-1}+y_{t-3}+y_{t-4},&t\equiv0\pmod4,\\
 y_t+y_{t-1}+y_{t-2}+y_{t-3},&t\equiv1,3\pmod4,\\
 y_t+y_{t-2},&t\equiv2\pmod4.
 \end{cases}
\]
To verify the cancellation, a column $k<t-4$ has the periodic
entries $1+[k\equiv2t-3\pmod4]$ in all terms; each displayed sum
contains each surviving residue an even number of times. For
$t-4\le k<t$, direct substitution cancels the boundary entries,
and column $k=t$ contributes one. This proves that the displayed
operator is the inverse, including the first three rows.
\end{proof}

Table~\ref{tab:transformfamilies} collects the resulting transform
choices. Rules generating the same binary orbit give the same matrix
family even if their local maps differ on neighborhoods absent from
that orbit. The costs refer to both directions of the coordinate
change, so they can be compared with the weights in
Table~\ref{tab:weights} without forming a dense matrix first.

\begin{table}[htbp]\centering\fontsize{9}{10.5}\selectfont
\caption{The seven universal binary matrix families and available transform costs at length $n$.}
\label{tab:transformfamilies}
\begin{tabular}{|lrr|}\hline
Family or rules & Number of rules & Boolean operations\\\hline
Identity & 16 & $O(n)$\\
28, 156 & 2 & $O(n)$\\
92 & 1 & $O(n)$\\
188 & 1 & $O(n)$\\
220, 252 & 2 & $O(n)$\\
60 (Pascal) & 1 & $O(n\log(n+1))$\\
124 (dense substitution) & 1 & $O(n^2)$\\\hline
\end{tabular}
\end{table}

The inverse nonzero counts at $n=32,64,128,256$ are
$63,127,255,511$ for the bandwidth-one and bandwidth-two examples;
$107,219,443,891$ for Rule 188; $243,729,2187,6561$ for Pascal;
and $269,1069,4125,16355$ for Rule 124. Rule 188 has
$7n/2-5$ entries when $4\mid n$. Rule 124 has density close to
$1/4$ on these four windows. Dense substitution costs $O(n^2)$
ring operations for it, but these counts establish no lower bound
on the best possible algorithm.

\Needspace{5\baselineskip}
\begin{theorem}[Uniqueness of the triangular convolution transform]\label{thm:uniqueconv}
For any $n\ge1$, the Pascal matrix $Z=B_{60}$ is the unique
invertible lower triangular binary matrix such that
\[
 B(a\conv b)=(Ba)(Bb)
\]
for all length-$n$ vectors, with OR convolution truncated to that
window. On a dyadic window it also has the butterfly factorization
of Corollary~\ref{cor:zeta}.
\end{theorem}
\begin{proof}
For every $t<n$, $i\vee j\subseteq t$ holds exactly when both
$i\subseteq t$ and $j\subseteq t$. Summing coefficients proves
$Z(a\conv b)=(Za)(Zb)$. Terms with $i\vee j\ge n$ cannot enter
such a row, so this also proves the assertion for arbitrary prefixes.

If $B$ has the same property, $C=BZ^{-1}$ is an invertible linear
map preserving pointwise products. Let $e_0,\ldots,e_{n-1}$ be
the coordinate vectors. Their images are nonzero idempotents and
have disjoint supports, because $e_i e_j=0$ for $i\ne j$.
There are $n$ nonempty disjoint supports among $n$ positions;
each is therefore a singleton. Thus $C$ is a permutation matrix.
It is lower triangular, since both $B$ and $Z^{-1}$ are. If the
one in column $j$ is in row $\pi(j)$, then $\pi(j)\ge j$.
Summing the indices forces equality throughout. Hence $C=I$ and
$B=Z$.
\end{proof}

\Needspace{5\baselineskip}
\begin{corollary}[Uniqueness in the rule family]\label{ver:uniqueconv}
Among the 24 universal rule families, only Rule 60 has the
convolution property at every window. It is also the only nonidentity
family of the form $A^{\otimes b}$ on every $2^b$ window.
\end{corollary}
\begin{proof}
The lower triangular invertible binary $2\times2$ matrices are
$I$ and $\left(\begin{smallmatrix}1&0\\1&1\end{smallmatrix}\right)$.
Their tensor powers are the identity and Pascal families. At length
five, the other six matrix types already differ from Pascal:
the nonidentity types are represented by 28, 92, 124, 188 and 220,
in addition to 60. Substituting the local rule in the first five
rows distinguishes each of the five from Pascal. The sixteen
identity rules are accounted for in Proposition~\ref{ver:cost}.
Theorem~\ref{thm:uniqueconv} then gives the convolution claim.
\end{proof}

\Needspace{5\baselineskip}
\begin{remark}[Finite witnesses]\label{ver:unique}
At length 64 the failure of the convolution identity for each of the
23 other rules can be tested on the $64^2$ pairs $e_a,e_b$:
compare $B_R(e_a\conv e_b)$ with $(B_Re_a)(B_Re_b)$ row by row.
Each rule has a differing row, in agreement with the structural argument above. By bilinearity these pairs determine the identity
on all vectors of that window; they are not used to infer the
all-window theorem.
\end{remark}

\Needspace{5\baselineskip}
\begin{proposition}[Dyadic blocks and prefix restriction]\label{prop:dyadic}
Exactly twenty universal elementary rule families satisfy
Definition~\ref{def:usable}(c) at every dyadic window. The exceptions
are 28, 92, 124 and 156. Independently of this property, every
compatible lower triangular invertible family preserves coordinate
restriction to every prefix.
\end{proposition}
\begin{proof}
The identity, prefix and Pascal formulas give their repeated
diagonal blocks directly. For Rule 188, replacing $(t,k)$ by
$(t+n,k+n)$ preserves $k=t$ and the residue $2t-k\pmod4$
when $4\mid n$; the cases $n=1,2$ follow from the first four
rows. These are twenty rules in all. For each of 28, 92, 124 and
156, $B_4[3,2]=0$ whereas $B_2[1,0]=1$, so the required family
identity fails already at $n=2$.

For the final assertion, the first $n$ rows of $Bx$ depend only
on $x_0,\ldots,x_{n-1}$, and their coefficient matrix is $B_n$.
Uniqueness of the triangular solution identifies the restricted
coordinates. No dyadic hypothesis is needed. The six-pattern block
classification of \cite{U1} may therefore be applied to any such
binary coordinate set, although only Pascal transports OR
convolution to pointwise multiplication.
\end{proof}

These distinctions give different answers to different questions.
A banded inverse supplies a linear-time transform; dyadic repeated
blocks describe a matrix symmetry; and the convolution identity
identifies how multiplication changes coordinates.

\Needspace{5\baselineskip}
\begin{corollary}[Reading the weight comparison]\label{cor:reading}
The prefix and parity examples, including Rule 188, have linear-time
transforms. Pascal has an $O(n\log(n+1))$ binary transform and
uniquely turns OR convolution into pointwise multiplication among
triangular coordinate maps. Neither fact implies that Pascal
minimizes coordinate weight on every window.
\end{corollary}
\begin{proof}
Combine Propositions~\ref{prop:banded} and~\ref{ver:cost} with
Theorem~\ref{thm:uniqueconv} and Corollary~\ref{cor:fft}.
Proposition~\ref{prop:opt} supplies the two finite exceptions to the
suggested weight preference. A recurrence can be transported to
any invertible basis; the convolution theorem specifies its form
in Pascal coordinates, rather than prohibiting other algorithms.
\end{proof}

\subsection{The conjugated carry, and what the round trip really is}

The product operation does not exhaust the support recurrence.
Increment is equally essential because it represents discrete
integration, and a coordinate change has to transport it as well.
We therefore compute its conjugate under Pascal. The resulting
prefix operation explains both the form of the transformed
recurrence and the work involved in returning to the original
coordinates after applying it.

The next theorem also appears, with the same proof, in \cite{U2}, where it is used to
identify the operator acting at each generation. It is restated here
because the conjugating matrix is $B_{60}$, one of the $24$ bases classified above.

\Needspace{5\baselineskip}
\begin{theorem}[The conjugated increment]\label{thm:carry}
On any binary window of length $n$, let $J=\Inc$ be truncated
increment and $Z=B_{60}$. Then
\[
 ZJZ^{-1}=L,\qquad L[i,j]=[i>j].
\]
Thus increment is strict prefix XOR in the cell coordinates and
costs $O(n)$ Boolean operations. Also $L=B_{220}-I$.
\end{theorem}
\begin{proof}
The hockey-stick identity gives, for every basis column $r$,
\[
 (LZ)[t,r]=\sum_{s=0}^{t-1}\binom sr
          =\binom t{r+1}=(ZJ)[t,r]\pmod2.
\]
At the truncated last column both expressions vanish. Hence
$LZ=ZJ$, and multiplying by $Z^{-1}$ proves the identity.
The prefix matrix formula is Proposition~\ref{prop:named}.
\end{proof}

\Needspace{5\baselineskip}
\begin{remark}[Nilpotent and unipotent operators]\label{ver:carry}
The strict prefix matrix has $n(n-1)/2$ nonzero entries and
bandwidth $n-1$, although its application takes linear time.
It is nilpotent, like $J$. The full triangle $I+L=B_{220}$ is
instead the conjugate of $I+J$. Confusing these two operators
would change the recurrence by the identity map.
\end{remark}

Theorem~\ref{thm:carry} turns the support recurrence into a computation that never leaves the
cell domain: by Theorem~\ref{thm:uniqueconv} the convolution is pointwise there, and by
Theorem~\ref{thm:carry} the carry is a prefix XOR, so
\begin{equation}\label{eq:cellstep}
 y_m \;=\; L\big(y_{m-1}\oplus y_{m-2}\oplus y_{m-1}y_{m-2}\big),
\end{equation}
with one $\zeta$ at the start and one at the end.

\Needspace{5\baselineskip}
\begin{remark}[Relation to the rotated local update]\label{rem:notnew}
Write $y_m=Z\mathbf1_{S_m}$, so $y_m=D_{m-1}$, with
$y_1(t)=1$ and $y_2(t)=t\bmod2$. Equation~\eqref{eq:cellstep}
applies for $m\ge3$. The local Rule 30 update on diagonals is
\[
 D_m(t+1)=D_m(t)\xor D_{m-1}(t)\xor D_{m-2}(t)
                 \xor D_{m-1}(t)D_{m-2}(t).
\]
For $m\ge1$, its initial value is zero, so summing in $t$
gives the strict prefix expression. The carry theorem identifies
the support recurrence with this rotated computation.
\end{remark}

\Needspace{5\baselineskip}
\begin{remark}[Cost of a complete recurrence calculation]\label{ver:crossover}
Directly forming an OR convolution of supports $A,B$ enumerates
$|A||B|$ pairs. On a supplied length-$n$ window, the cell form
uses $O(n)$ Boolean operations per depth, so $M$ depths and a final
Pascal inversion cost $O(Mn+n\log(n+1))$. Two predecessor arrays
and one output array require $O(n)$ working bits when earlier
depths are discarded. If full supports are required, the window
must also contain all output indices. The algebra gives
this operation count, not a speedup over the same rotated local
update implemented directly.
\end{remark}

\Needspace{5\baselineskip}
\begin{theorem}[No frame diagonalises both]\label{thm:nodiag}
On a coefficient window of length $n\ge2$, there is no invertible $\FF$-linear change of coordinates $B$ --- of any
radius, arising from a cellular automaton or not --- for which $\conv$ becomes pointwise
multiplication, $B(a\conv b)=(Ba)(Bb)$, and $\Inc$ becomes a diagonal linear map,
$B(\Inc a)=d\cdot(Ba)$ for a fixed vector $d$.
\end{theorem}

\begin{proof}
In the standard coefficient basis, truncated $\Inc$ is the shift $J$ with
$Je_i=e_{i+1}$ for $i<n-1$ and $Je_{n-1}=0$.
Thus $J^n=0$ and $J\ne0$ for $n\ge2$. Similarity preserves both properties.
A diagonal nilpotent matrix over a field has every diagonal entry zero, so is
zero, a contradiction. At $n=1$ the shift is zero and the conclusion fails.
\end{proof}

\Needspace{5\baselineskip}
\begin{remark}[A small multiplicativity witness]\label{ver:comm}
On a window of length at least five,
$\Inc(\{1\}\conv\{2\})=\{4\}$ whereas
$\Inc(\{1\})\conv\Inc(\{2\})=\{3\}$.
This illustrates the interaction of increment with binary OR.
The nilpotence argument above gives a stronger, basis-independent
obstruction to diagonalizing increment alone.
\end{remark}

Similarity preserves the obstruction regardless of whether the matrix
comes from an elementary automaton, a larger radius, or another
construction. It only excludes a diagonal form for nonzero truncated
increment. Theorem~\ref{thm:carry} shows that this operator nevertheless
has a linear-time implementation.

\subsection{Relation to the Rule 30 prize problems}\label{sec:prize}

Wolfram's third prize problem \cite{prize} asks whether computing the $n$-th cell of the Rule~30 centre
column requires at least $\Omega(n)$ effort; running the rule directly performs $\Theta(n^2)$ cell
updates, so any improvement below $n^2$ is already of interest. The present construction describes query cost after a representation
has been supplied; producing that representation is a separate problem.

The centre column is $c(n)=A(n,0)=D_n(n)$. The indexing of \cite{U1}
is shifted: $D_m(t)=b(m+1,t)$ and its support is $S_{m+1}$.
If $K(m+1)$ is the number of blocks in a supplied representation,
its evaluation costs $O(K(m+1)(b+\ell))$ bit operations
at query time $n$, where $b\ge1$ bounds its anchor and finite-mask widths and any explicit
modulus thresholds and
$\ell=\max(1,\lceil\log_2(n+1)\rceil)$. Construction and minimization
remain separate costs; sequential evaluation uses $O(b+\ell)$ working bits.
For fixed $m$ this is logarithmic in $n$ after the representation is supplied.

The measured full supports in the indexing of \cite{U1} have sizes
$|S_{20}|=105$, $|S_{28}|=499$, $|S_{36}|=3979$, $|S_{44}|=64\,939$.
These sizes are not lower bounds on the cost of one centre query: in
$c(n)=\Xor_{r\in S_{n+1}}\binom nr\bmod2$, every $r>n$ contributes zero.
At most $n+1$ support indices can contribute, before block compression.
Obtaining those coefficients still requires work, but the full support size
does not prove an exponential evaluation cost or a lower bound for this query.

Equation~\eqref{eq:cellstep} supplies the usual triangular computation if
all diagonals through time $n$ are constructed. Proposition~\ref{prop:opt}
compares finite-window coordinate weights in the specified basis family;
it includes exceptions at $m=4$ and $m=30$ on the stated windows.
It does not exclude other algorithms or certify a minimum at all depths.

The Zeta--Floor factorization in \cite{U2} separates a pointwise
product from prefix summation. Theorem~\ref{thm:nodiag} excludes
replacing the latter by a diagonal map on the stated finite window.
This is an obstruction to one proposed form of simplification.
It supplies no lower bound for computing a specified center cell.

\subsection{Linear universality and computational universality}

The word ``universal'' carries two unrelated meanings here --- spanning $\Zm^n$, and simulating
a Turing machine --- and they meet once.

\Needspace{5\baselineskip}
\begin{proposition}\label{prop:110}
$\mathrm{mirror}(110)=124$, and Rule~124 is one of the $24$ universal bases for every modulus
$N\ge2$. Rule~110 itself is not: $110\equiv14\pmod{16}$, its bit~1 is set, and its left cone is
non--empty, so \eqref{eq:leftcone} fails. The diagonals of Rule~110 read along the mirrored cut
are therefore a universal basis of $\Zm^n$, exactly as Rules~60 and~102 carry the Pascal system
on opposite cuts.
\end{proposition}

\begin{proof}
Reflection acts on the rule number by swapping bits $1\leftrightarrow4$ and
$3\leftrightarrow6$. On $110=01101110_2$ this gives $01111100_2=124$. Membership of $124$ is
Proposition~\ref{prop:congruence}; failure of $110$ is bit~1.
\end{proof}

Cook's computational-universality result for Rule 110 uses structured
backgrounds \cite{cook}. Reflecting the construction gives the same
computational capability for Rule 124. The one-sided single-seed
orbit used for its basis is a different initial-value problem; its
matrix universality does not prove computational universality from
that seed. Among our 24 rules, Rule 124 is the one for which no fast
transform is supplied. Its measured inverse densities do not derive
any algorithmic lower bound from Cook's theorem.

\section{Rank--deficient bases: the Rule 90 example}

A diagonal family that fails to span the full window can still
represent a well-defined class of targets. Rule 90 gives an explicit
example in which the rank and the representable subspace can be
identified. This case completes the invertibility discussion by
describing what remains available after full universality fails,
and it supplies an additive comparison for the nonlinear
constructions that follow.

Not every useful basis is universal. Rule~90 has $p_{90}=a+c$;
only its even-depth diagonals are nonzero. The time indices within
a nonzero diagonal are governed by the binomial formula below.

\Needspace{5\baselineskip}
\begin{theorem}[Rule 90 span]\label{thm:90}
Over $\FF$ with seed $\delta_0$, $B_{90}[t,2j]=\binom tj$ and $B_{90}[t,2j+1]=0$. Hence
$B_{90}$ is a proper staircase for $n\ge2$, with $\Lambda=\{0,1,\dots,\lceil n/2\rceil-1\}$, its image is
the span of the first $\lceil n/2\rceil$ binomial columns, and for a target $y$ with Pascal
coordinate set $T$,
\[
 y\in\operatorname{im}B_{90}\iff T\subseteq[0,\lceil n/2\rceil) .
\]
\end{theorem}

\begin{proof}
$A(t+1,x)=A(t,x-1)+A(t,x+1)$ from $\delta_0$ gives $A(t,x)=\binom t{(t+x)/2}$ when $t+x$ is
even and $0$ otherwise; substituting $x=t-k$ gives $A(t,t-k)=\binom t{k/2}$ for even $k$ and
$0$ for odd $k$. So the nonzero columns are the binomial columns $j=0,\dots,\lceil n/2\rceil-1$
in order, each with leading position $j$. For $n\ge2$ there is a zero column,
so this is a proper staircase. At $n=1$ it is the universal matrix $(1)$.
In both cases its image is the stated span.
Theorem~\ref{thm:staircase} then makes greedy exactness equivalent to membership, and by the
uniqueness of the Pascal expansion (Theorem~\ref{thm:universal}) membership is the stated
condition on $T$, including $T=\varnothing$.
\end{proof}

\Needspace{5\baselineskip}
\begin{corollary}\label{cor:90}
For the Rule~30 diagonal truncated to $0\le t<n$,
$D_m\in\operatorname{im}B_{90}$ iff
$S_{m+1}\cap[\lceil n/2\rceil,n)=\varnothing$.
Only the Newton coefficients below $n$ are coordinates of this finite vector.
\end{corollary}
\begin{proof}
Theorem~\ref{thm:S} identifies the Pascal coordinate set of $D_m$ with $S_{m+1}$.
On the finite window, retain only $T=S_{m+1}\cap[0,n)$ and substitute
it into Theorem~\ref{thm:90}. For example $n=2,m=4$ gives the zero vector,
although $S_5=\{2,3,4\}$ has maximum $4$.
\end{proof}

\Needspace{5\baselineskip}
\begin{remark}[Window-dependent membership]\label{ver:90}
At $n=48$, the half-window threshold separates $m=7$, with
$\max S_8=16$, from $m=8$, with $\max S_9=25$.
For arbitrary windows it is the intersection in
Corollary~\ref{cor:90}, not the full support maximum, that decides
membership. In particular $D_4$ on $n=2$ is the zero vector even
though $S_5=\{2,3,4\}$.
\end{remark}

A small example makes the role of the window explicit. The Rule 30
diagonal $D_3(t)=\binom t2\bmod2$ is not in the Rule 90 image at
$n=4$: the only nonzero basis columns are $\binom t0$ and
$\binom t1$. At $n=5$, column four supplies $\binom t2$, so the
longer target is represented by a single Rule 90 coordinate.
Changing the window changes both the target vector and the available
column family. The support criterion tracks these changes exactly.

This is the cleanest available consistency check on the whole construction: an independently
computed truncated support $S_{m+1}\cap[0,n)$ predicts when a decomposition against a
\emph{different} automaton's diagonals succeeds.

\section{Diagonals as recurrent sequences}\label{sec:recurrent}

The basis viewpoint emphasizes the relation between different
diagonal vectors. We now follow a single diagonal through time.
When it has an integer-valued polynomial lift, its Newton expansion
terminates and its interpolation order can be controlled by the
rule recurrence. This provides the link between the matrix
construction and period bounds after reduction modulo a prime.

The common coordinates are the Pascal basis of Theorem~\ref{thm:pascal}:
the binomial spectrum of Section~\ref{sec:spec} is the coordinate vector of
Theorem~\ref{thm:universal}, whose support measures coordinate sparsity.

Every rule $f$ induces on the diagonals the triangular system
\[
 D_m(t)=f\big(D_m(t-1),\,D_{m-1}(t-1),\,D_{m-2}(t-1)\big),
\]
the first argument being the left neighbour, since each cell depends on three upper
neighbours sitting on diagonals $m$, $m-1$, $m-2$ at time $t-1$. Diagonals of negative depth lie beyond the right signal cone and
equal the uniform background. That background follows
$v_{t+1}=f(v_t,v_t,v_t)$, starting at zero. On a finite state set it
has a transient followed by a cycle of length $c$; for quiescent
rules it is identically zero. For example, the canonical Rule 165
lift modulo three has background $0,1,0,1,\ldots$.

The triangular diagonal update can be computed without an artificial
spatial boundary. Initialize $D_0(0)=1$ and $D_m(0)=0$ for $m>0$,
retain two background entries for negative depths, and update the
depth array from its previous value. Computing $W$ time samples
through depth $M$ uses $O(W(M+1))$ local updates and $O(M+1)$
working residues, excluding stored output. The first terms at every
depth are retained. Restricting the physical lattice to $x\ge0$
would delete those terms and alter the period experiment.

\Needspace{5\baselineskip}
\begin{definition}[Left--linear rule]
A rule is left--linear mod~$p$ when its \ANF{} polynomial has the form
$f(l,c,r)=l+g(c,r)\bmod p$: the left variable appears as a lone linear term and in no
product.
\end{definition}

\Needspace{5\baselineskip}
\begin{theorem}[Telescoping]\label{thm:tel}
For a left--linear rule, subtracting consecutive terms eliminates $f$:
$D_m(t)-D_m(t-1)=g\big(D_{m-1}(t-1),D_{m-2}(t-1)\big)$, so each diagonal is the running
sum of a function of the two diagonals above it.
\end{theorem}

\begin{proof}
Substitute $f(l,c,r)=l+g(c,r)$ into the rotated update:
$D_m(t)=D_m(t-1)+g(D_{m-1}(t-1),D_{m-2}(t-1))$.
Subtract $D_m(t-1)$. Summing this identity includes the initial value
$D_m(0)$ as the integration constant.
\end{proof}

The identity is the diagonal form of the one-sided recursion in
\cite{U1}. Summation of a periodic input over $\mathbb F_p$ can
add a factor of $p$ to its period. The next result also gives a
different bound for Rule 30, using polynomial degree.

\subsection{The Rule 30 column system and its Fibonacci dependency}

The right-column notation of \cite{U1} satisfies
$\binom nk_R=D_k(n)$. Of the three frames $C,L,R$ used there,
it is $R$ that agrees with the slope-one diagonals here. In this
frame the same-column term is linear and the nonlinear term couples
only the preceding two columns.

\Needspace{5\baselineskip}
\begin{theorem}[Rule 30 telescoped]
Because the same--column symbol is linear, Theorem~\ref{thm:tel} collapses the whole
$n$-recurrence of the right columns into a single sum over the two previous columns:
\[
 \bin nk_R=\sum_{j=0}^{n-1}\Big(\bin j{k-1}_R+\bin j{k-2}_R+\bin j{k-1}_R\bin j{k-2}_R\Big)\bmod N,
\]
for $k\ge1$, with $\binom0k_R=0$. At $k=0$ the diagonal is one;
its initial value must be added to the telescoped sum.
\end{theorem}
\begin{proof}
Rule~30 is left--linear, so Theorem~\ref{thm:tel} applies with $g(c,r)=c+r+cr$; summing the
resulting first--difference identity from $0$ to $n-1$ telescopes the left side to
$\binom nk_R$, the base value being zero.
\end{proof}

The degree calculation appears in the public discussion
\cite{degreepost}, and the corresponding binary ceiling is used
in \cite{U2}. The same argument settles the interpolation question. Positivity gives
the exact integer degree, while the Newton form avoids division
at composite moduli.

\Needspace{5\baselineskip}
\begin{theorem}[Fibonacci interpolation order]\label{thm:fib}
Let $F_0=0,F_1=1$ and $F_{k+2}=F_{k+1}+F_k$. For the integer
polynomial lift of Rule 30, $D_k(t)$ is an integer-valued polynomial
in $t$ of degree exactly $F_{k+2}-1$. Consequently, modulo any
$N\ge2$, its first $F=F_{k+2}$ values determine all subsequent
values by
\[
 D_k(t)=\sum_{j=0}^{F-1}c_j\binom tj\pmod N,
 \qquad
 c_j=\sum_{i=0}^{j}(-1)^{j-i}\binom ji D_k(i).
\]
This is an upper bound on the required interpolation order modulo
$N$; cancellation can make the actual order smaller.
\end{theorem}
\begin{proof}
Over the integers, $D_0(t)=1$ and $D_1(t)=t$. A product of
integer-valued polynomials of degrees $r,s$ has degree at most
$r+s$, and a sum from $0$ to $t-1$ increases degree by at most one.
These statements hold within the integral Newton basis: for example,
\[
 \binom tr\binom ts
 =\sum_{j=\max(r,s)}^{r+s}
 \frac{j!}{(j-r)!(j-s)!(r+s-j)!}\binom tj,
\]
where the coefficient counts two subsets of sizes $r,s$ whose union
has size $j$. Also $\sum_{u<t}\binom uj=\binom t{j+1}$.
The telescoped Rule 30 formula therefore constructs an integer-valued
polynomial at each depth. Its leading coefficient is positive:
the initial two have that property, and sums and products in the
recurrence introduce no negative leading terms. For $k\ge2$,
the product has degree $d_{k-1}+d_{k-2}$, dominates the linear
terms or ties them with positive coefficient, and summation gives
\[
 d_k=d_{k-1}+d_{k-2}+1,
 \qquad d_0=0,\quad d_1=1.
\]
Thus $d_k+1=F_{k+2}$. A polynomial of degree less than $F$ has
the Newton expansion with coefficients $\Delta^jD_k(0)$ displayed
above; all coefficients of order $j\ge F$ vanish. The formula
uses only integer additions and multiplications, so reduction
modulo $N$ is valid without dividing by a nonunit. It includes
every initial and previously described singular index.
\end{proof}

\Needspace{5\baselineskip}
\begin{corollary}[A Fibonacci period bound for Rule 30]\label{cor:fibperiod}
Modulo a prime $p$, every depth-$k$ Rule 30 diagonal is purely
periodic, with period dividing
$p^{\lceil\log_p F_{k+2}\rceil}$. In particular its least period
is less than $pF_{k+2}$ unless $k=0$, when it is one.
\end{corollary}
\begin{proof}
Every Newton index in Theorem~\ref{thm:fib} is less than $F_{k+2}$.
For a power $q=p^L\ge F_{k+2}$, the identity
$(1+z)^q=1+z^q$ over $\mathbb F_p$ implies
$\binom{t+q}{j}=\binom tj$ for $j<q$. Applying it to every term
gives the period. The smallest such power is less than $pF_{k+2}$
when $F_{k+2}>1$.
\end{proof}

\Needspace{5\baselineskip}
\begin{proposition}[Period multiplication]\label{prop:mult}
For a left-linear rule modulo a prime $p$, suppose both driving
diagonals $D_{m-1},D_{m-2}$ have period $\pi$ for $t\ge T$.
Then $D_m$ has period dividing $p\pi$ from $T$ onward. More
precisely, if $\sigma$ is the sum of its increments over one such
period, $\pi$ is a period when $\sigma=0$; otherwise $p\pi$
is a period. An additional background factor is unnecessary once
the common input period $\pi$ is specified.
\end{proposition}
\begin{proof}
The increment given by Theorem~\ref{thm:tel} is $\pi$-periodic
for $t\ge T$. Hence
$D_m(t+\pi)-D_m(t)=\sigma$, independently of $t\ge T$.
Repeating this equality $p$ times gives zero change, because
$p\sigma=0$ in $\mathbb F_p$. The least eventual period divides
each period of the tail, giving the assertion.
\end{proof}

This proposition bounds a multiplication step; it does not determine
how often the least period increases. The Fibonacci argument gives
a smaller ceiling for Rule 30. Given $F$ initial values, a forward-
difference table computes their Newton coefficients using
$F(F-1)/2$ modular subtractions and $O(F)$ working residues.
The initial values can be generated in $O(kF)$ local updates.
In the binary case, evaluating the resulting expansion at time $t$
uses $O(F(\log(F+1)+\log(t+1)))$ bit operations by Lucas tests.
This is an algorithmic upper bound, not a minimality claim. Additive rules give another scale:
their binomial periods are comparable with depth, and their period
logarithms grow logarithmically. A finite plot of the logarithms
must therefore be distinguished from a theorem about growth rates.

\section{Measuring periods: the strict criterion}

An exact period theorem and a period observed on a finite array
answer different questions. For the finite comparisons, the
criterion must specify which repetitions are tested and how a
candidate is rejected. We give that convention before presenting
the census. It ensures that counts for different rules and moduli
refer to the same observation procedure, while leaving the
infinite-sequence claims to the preceding proofs.

For a sample of length $W$, a candidate period $q$ is accepted only
if the matching suffix starts at $\tau\le W/2$ and contains at
least three full repeats. The precise convention is given in
Definition~\ref{def:crit}. We call the result a \emph{detected
period}: it is a property of the sample, unless a separate proof
establishes the infinite sequence's period.

Rule 45 illustrates the need for a transient restriction. At depth
14 and $W=700$, a threefold repetition of a block of length three
starts at index 692. The suffix has eight entries, so it does not
even contain three complete blocks; a detector admitting a partial
last block can report it. At $W=2500$, the strict criterion finds
period 512 with transient zero. The $W=700$ strict result is
unresolved. This change concerns the observation window; it is
not evidence that an eventual period changed during the evolution.

The detected periods begin
\[
\begin{array}{ll}
\text{Rule 30:}&1,2,2,4,8,8,16,32,32,64,\ldots,\\
\text{Rule 45:}&2,4,8,8,16,32,32,64,64,128,\ldots.
\end{array}
\]
Both later have long plateaus, so these prefixes do not establish
a fixed number of depths between doublings. Rule 150 has the exact
period $2^{\lceil\log_2(m+1)\rceil}$ by
Theorem~\ref{thm:colper}; Rule 90 gives the same binomial scale at
even depths and zero at odd depths. Rule 86 has slower detected
growth on this window. Its eventually periodic tails and unbounded
periods across depths are treated in \cite{U2}; no constant growth
rate is inferred here.

\begin{figure}[t]
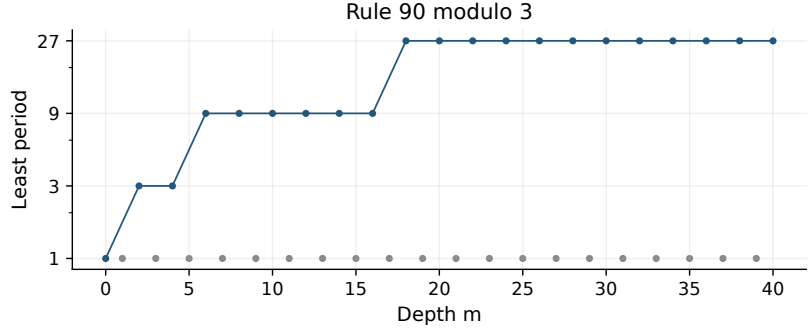
\centering
 \EmbeddedFigure{\textwidth}{6}
 \caption{Rule~90 mod~$3$:
 $1,1,3,1,3,1,9,1,9,\dots$, tripling exactly when $m/2$ crosses a power of $3$
 (Corollary~\ref{cor:90p}). Blue points show the even-depth columns;
gray points show the zero odd-depth columns, whose period is one.}
 \label{fig:stairs}
\end{figure}

\begin{figure}[htbp]
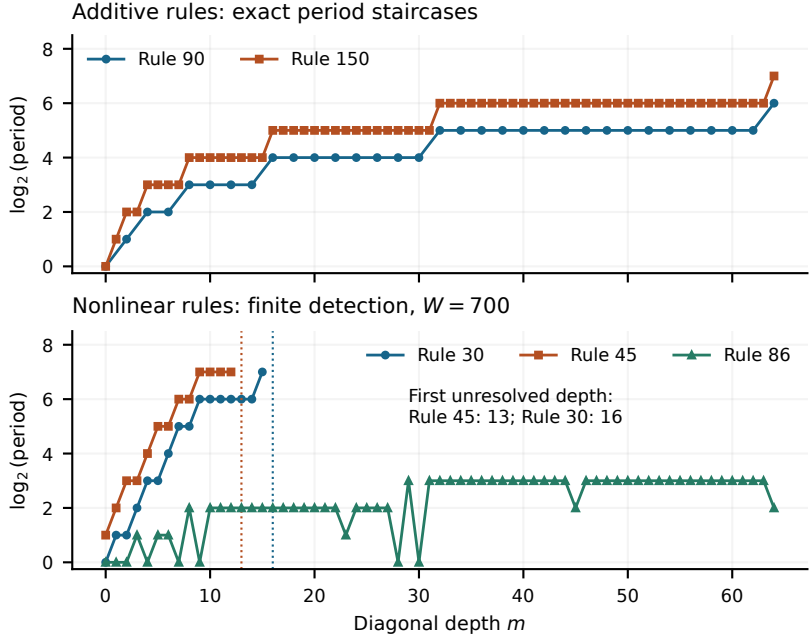
\centering
\EmbeddedFigure{\textwidth}{7}
\caption{Diagonal periods modulo two. Above: the exact
staircases for the additive Rules 90 and 150. Only even depths are
drawn for Rule 90; its odd-depth diagonals are identically zero and
have least period one. Below: finite detections for Rules 30, 45 and
86 on $W=700$ terms, using Definition~\ref{def:crit}.
The first unresolved depths are $13$ for Rule 45 and $16$ for
Rule 30, marked by dotted lines; these cutoffs are limits of this
observation window. Period one is drawn at height zero on the
logarithmic axis.}
\label{fig:stairs2}
\end{figure}

\section{The additive core}

\Needspace{5\baselineskip}
\begin{theorem}[Additive diagonals as polynomial coefficients]\label{thm:multcol}
Let $f(a,b,c)=\alpha a+\beta b+\gamma c$ over $\mathbb F_p$.
From $\delta_0$, write
$P(y)=\alpha y+\beta+\gamma y^{-1}$ and
$Q(z)=\alpha+\beta z+\gamma z^2$. Then
\[
 A(t,x)=[y^x]P(y)^t,\qquad D_m(t)=[z^m]Q(z)^t.
\]
In characteristic two the additive elementary rules are
$0,60,90,102,150,170,204,240$.
\end{theorem}
\begin{proof}
Encoding a row by $\sum_x A(t,x)y^x$, the left, center and right
terms multiply it by $\alpha y$, $\beta$ and $\gamma y^{-1}$.
The seed polynomial is one, so induction gives $P(y)^t$.
Since $P(y)=yQ(y^{-1})$, the coefficient of $y^{t-m}$ is the
coefficient of $z^m$ in $Q(z)^t$. For Rule 90, $Q=1+z^2$;
for Rule 150, $Q=1+z+z^2$.
\end{proof}

\Needspace{5\baselineskip}
\begin{corollary}[Background factors for left-linear rules]\label{cor:ppow}
Suppose the background has period $c$ after time $T$. For every
left-linear rule modulo $p$ and every $m\ge0$, $D_m$ is periodic
after $T$, with period dividing $c p^{m+1}$. In characteristic
two the least eventual period is therefore a power of two.
\end{corollary}
\begin{proof}
Both negative-depth inputs have period $c$ after $T$.
Inductively the preceding two diagonals have a common period
dividing $cp^m$. Proposition~\ref{prop:mult} gives period
$cp^{m+1}$ for the next diagonal, with the same permissible
starting time. A binary background cycle has length one or two,
so the bound is a power of two. The statement concerns actual
tail periods; a shorter accidental match in a finite sample need
not have that form.
\end{proof}

\Needspace{5\baselineskip}
\begin{theorem}[Periods of polynomial-power columns]\label{thm:colper}
Let $Q\in\mathbb F_p[z]$ have constant term one, and let
$T(t,k)=[z^k]Q(z)^t$. If $p^L>k$, then $T(t+p^L,k)=T(t,k)$
for all $t\ge0$. If in addition $[z]Q\ne0$, the least period
is exactly $p^{\lceil\log_p(k+1)\rceil}$.
\end{theorem}
\begin{proof}
Frobenius gives $Q(z)^{p^L}=Q(z^{p^L})$, which has no positive
degree below $p^L$. Multiplying by $Q(z)^t$ proves the first claim.

For the second, put $H=Q-1$. Expanding $(1+H)^t$ gives
\[
 T(t,k)=\sum_{j=0}^{k}a_j\binom tj,
 \qquad a_j=[z^k]H^j,\quad a_k=([z]Q)^k\ne0.
\]
For $k\ge1$ put $h=p^{L-1}\le k<p^L$.
Vandermonde's identity in characteristic $p$ gives
\[
 T(t+h,k)-T(t,k)=\sum_{j=h}^k a_j\binom t{j-h}.
\]
Its highest Newton coefficient is $a_k\ne0$, so it is a
nonzero sequence by triangular Pascal inversion. Thus $h$ is
not a period. Since the least period divides $p^L$, it must
equal $p^L$. At $k=0$ the column is constant one.
\end{proof}

\Needspace{5\baselineskip}
\begin{theorem}[The binomial jump is exact]\label{thm:jump}
For $k\ge1$ the period of $t\mapsto\binom tk\bmod p$ is \emph{exactly} $p^L$,
$L=\lfloor\log_pk\rfloor+1$.
\end{theorem}
\begin{proof}
``Divides'' is Theorem~\ref{thm:colper}. To rule out $p^{L-1}$: write
$k=d\,p^{L-1}+r$, $1\le d\le p-1$, $0\le r<p^{L-1}$; Lucas gives $\binom kk=1$. If
$d\le p-2$, then $k+p^{L-1}$ has digit $d+1$ at position $L-1$ and Lucas gives the factor
$\binom{d+1}d=d+1\not\equiv1$. If $d=p-1$ the addition carries, the digit becomes $0$, and
Lucas gives $\binom0{p-1}=0$. Either way the value moves.
\end{proof}

\Needspace{5\baselineskip}
\begin{corollary}[Rule 90 mod $p$]\label{cor:90p}
For $k\ge1$, the depth--$2k$ diagonal has exact period $p^{\lfloor\log_pk\rfloor+1}$: the period
multiplies by $p$ exactly when $m/2$ crosses a power of $p$. Mod~$2$, the doubling
staircase; mod~$3$, the measured profile $1,1,3,1,3,1,9,1,9,\dots$ with $9$ for
$k=3..8$ and $27$ from $k=9$ --- exactly $p^{\lfloor\log_3k\rfloor+1}$
(Fig.~\ref{fig:stairs}).

The Frobenius step is standard in $\mathbb F_p[y]$; see \cite{lidlnied}.
\end{corollary}

\begin{proof}
After shifting to the moving cone, Rule 90 has $Q(y)=1+y^2$.
The depth-$2k$ diagonal is $[y^{2k}](1+y^2)^t=\binom tk$.
Theorem~\ref{thm:jump} therefore gives the stated least period for $k\ge1$.
At $k=0$ the diagonal is constant one, with least period one; odd-depth
diagonals are zero, also with least period one.
\end{proof}

\subsection{The binomial spectrum of an arbitrary diagonal}\label{sec:spec}

Theorem~\ref{thm:multcol} covers the additive core. To probe every other rule with the
same lens, invert the construction: instead of ``which column is this diagonal?'', ask
``which \emph{combination} of columns is this diagonal?''. The Pascal matrix
$B[m,k]=\binom mk$ is lower triangular with unit diagonal, so $\det B=1$ and $B$ is
invertible over $\ZZ_M$ for every modulus.

\Needspace{5\baselineskip}
\begin{theorem}[Newton spectrum]\label{thm:spectrum}
For any modulus $N\ge2$ and sequence $s(0),s(1),\ldots$, there
is a unique locally finite expansion
\[
 s(t)=\sum_{j=0}^{t}c_j\binom tj\pmod N,
 \qquad c_j=\sum_{i=0}^j(-1)^{j-i}\binom ji s(i).
\]
Each finite prefix has the same coefficients below its length.
\end{theorem}
\begin{proof}
The Pascal matrix is unitriangular over $\ZZ_N$. Its inverse
has entries $(-1)^{j-i}\binom ji$: their product with Pascal
reduces by $\binom ji\binom ir=\binom jr\binom{j-r}{i-r}$
to $(1-1)^{j-r}$, giving the identity. Each coefficient depends
only on $s(0),\ldots,s(j)$, so the finite inversions agree on
overlaps and define the asserted sequence.
\end{proof}

The finite Newton spectrum also determines the exact period, once
its largest surviving coefficient is known. This refines a degree
ceiling such as Corollary~\ref{cor:fibperiod} without requiring an
explicit value for every earlier coefficient.

\Needspace{5\baselineskip}
\begin{proposition}[Exact order and period of a finite Newton spectrum]
\label{prop:exactnewtonperiod}
Let $p$ be prime and let
$s(t)=\sum_{j=0}^{d}c_j\binom tj$ over $\mathbb F_p$, with
$c_d\ne0$. Its minimal monic constant-coefficient recurrence
polynomial is $(X-1)^{d+1}$, and its least period is
\[
 p^{\lceil\log_p(d+1)\rceil}.
\]
The zero sequence is assigned recurrence order zero and least
period one.
\end{proposition}
\begin{proof}
With $E$ denoting the time shift, Pascal's identity gives
$(E-I)\binom tj=\binom t{j-1}$. Hence $(E-I)^{d+1}s=0$,
whereas $(E-I)^d s=c_d\ne0$. The minimal monic annihilator
divides every annihilating polynomial, by polynomial division;
as a divisor of $(X-1)^{d+1}$ it is a power of $X-1$.
The nonvanishing $d$th difference therefore makes its order
exactly $d+1$.

Put $q=p^{\lceil\log_p(d+1)\rceil}$. Frobenius and
Vandermonde's identity show that $q$ is a period. For $d\ge1$
let $h=q/p\le d$. The same identities give
\[
 s(t+h)-s(t)=\sum_{j=h}^{d}c_j\binom t{j-h}.
\]
Its highest Newton coefficient is $c_d$, so triangular inversion
makes it a nonzero sequence. Thus $h$ is not a period. The least
period divides $q$, a prime power, and must consequently equal
$q$. For $d=0$ the sequence is a nonzero constant.
\end{proof}

For example, the integer Rule 30 diagonal at depth two is
\[
 D_2(t)=t^2=\binom t1+2\binom t2.
\]
Modulo two, the second coefficient vanishes: the exact recurrence
order is two and the period is two. Modulo three it survives:
the order and period are both three. The integer degree remains
two in either calculation. This illustrates precisely where a
modular period can improve on a bound obtained from the integer
lift alone. For binary Rule 30 diagonals, the proposition recovers
the support-period identity of \cite{U2} with $d=\max S_{m+1}$.

\Needspace{5\baselineskip}
\begin{remark}[Other column families]
For $Q(y)=1+y+\cdots+y^{q-1}$ and $q\ge3$, specify the finite matrix
$B[t,k]=[y^k]Q(y)^t$, with $0\le t<n$ and $0\le k\le(q-1)(n-1)$.
Its rows have distinct degrees $(q-1)t$ and leading coefficient one, so it
has row rank $n$: every length-$n$ sequence has a decomposition. For $n>1$
there are more columns than rows, so the decomposition is not unique.
These conclusions concern this rectangular matrix, not every square
truncation of a $q$-nomial column family. Nothing below uses this case.
\end{remark}

Spectrum weight measures sparsity in this particular basis. It does
not measure sparsity of the sampled values or distinguish additive
from nonlinear dynamics. For example, Rule 102 is
additive, but its depth-zero diagonal is $s(0)=1$ and $s(t)=0$
for $t>0$. Its binary Newton coefficients are all one, by
Theorem~\ref{thm:spectrum}. Thus an additive rule can have full
spectrum density, and a sequence with an eventually constant tail
can have dense Newton coordinates.

At $W=256$, each nonzero Rule 90 diagonal has one spectrum
coefficient. Rule 150 generally needs more: already its depth-two
column is $\binom t1+\binom t2$ over $\FF$.
The plotted Rule 30 and Rule 86 spectra have many nonzero entries
at the larger measured depths. For Rule 30 they are exactly
$S_{m+1}\cap[0,256)$. Periods and spectrum weights describe
different aspects of these specified finite sequences.

\begin{figure}[htbp]
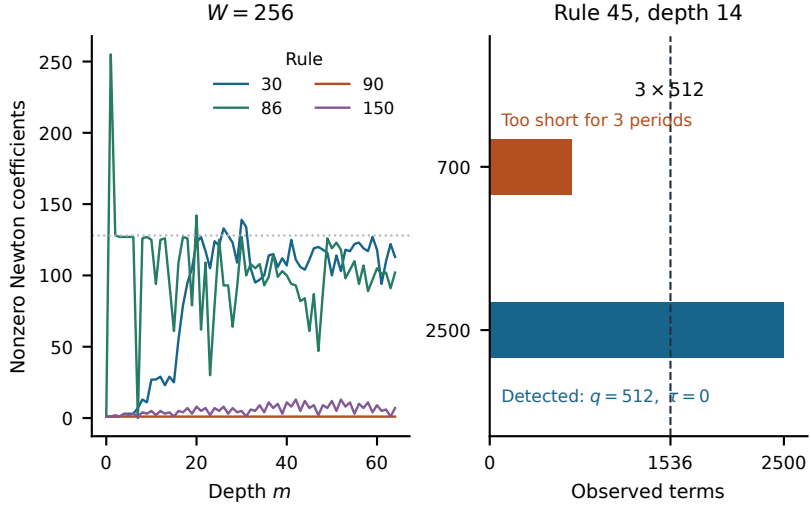
\centering
\EmbeddedFigure{\textwidth}{8}
\caption{Left: the number of nonzero Newton
coefficients in binary windows of length $256$. Rule 90 has one at
even depth and zero at odd depth; only its even depths are drawn.
The other curves show Rules 30, 86 and 150, with $128$ marked as a
reference. Right: the finite detection criterion needs at least
$3\cdot512=1536$ terms for a period-$512$ tail with zero transient.
For Rule 45 at depth $14$, $W=700$ yields no detected period,
whereas $W=2500$ gives period $512$ and transient zero. The bars
show the available observation lengths, not an infinite-orbit proof.}
\label{fig:spec}
\end{figure}

\section{Scanning the full rule space}\label{sec:scan}

The exact additive formulas provide a reference for the full
elementary-rule comparison. Applying the same finite criterion to
all 256 rules makes the differences in orientation and modulus
visible within a common window. The resulting categories summarize
what that window contains. Their interpretation therefore includes
the observation bounds, particularly for diagonals whose detected
periods are still growing or exceed the available range.

The scan computes the strict period profile of every rule $0..255$, seed $\{1\}$, window
$700$, diagonals $m=0..64$, and sorts into five classes: \emph{constant} (every detected
period $1$); \emph{$p$-power bounded} (all detected periods powers of $p$, none reaching
$p^3$); \emph{$p$-power growing} (reaching $p^3$ or beyond); \emph{mixed} (some detected
period not a power of $p$); \emph{window exceeded} (at least half the diagonals are unresolved).

\Needspace{5\baselineskip}
\begin{definition}[Finite-window period criterion]\label{def:crit}
For a length-$W$ sample $s$, test $q=1,\ldots,\lfloor W/3\rfloor$.
Put $E_q=\{t:0\le t<W-q,\ s(t)\ne s(t+q)\}$ and
\[
 \tau_q=\begin{cases}
 0,&E_q=\varnothing,\\
 1+\max E_q,&E_q\ne\varnothing.
 \end{cases}
\]
The detected period is the smallest $q$ with
$2\tau_q\le W$ and $3q\le W-\tau_q$; otherwise the sample is
unresolved. Apply this to depths $0,\ldots,M$.
A rule is \emph{window-exceeded} if at least half the samples
are unresolved. Otherwise it is \emph{constant} if all detected
periods equal one; \emph{mixed} if one is not a power of $p$;
\emph{$p$-power growing} if all are powers of $p$ and their
maximum is at least $p^3$; and \emph{$p$-power bounded} otherwise.
The words ``bounded'' and ``growing'' name these thresholds on
this window, not asymptotic properties.
\end{definition}

\Needspace{5\baselineskip}
\begin{remark}[Classification mod 2]
$146$ constant, $81$ bounded, $25$ growing, $4$ window--exceeded, and \textbf{the mixed
class is empty}: every detected period in this scan is a power of two.
For the 16 left-linear rules, Corollary~\ref{cor:ppow} separately
proves a power-of-two restriction on their actual eventual periods.
The census describes the detected periods for all 256 rules; those
two statements have different quantifiers. The growing
class is
\begin{multline*}
 \{18,22,26,60,82,86,90,105,124,126,129,146,149,150,154,\\
 161,165,167,181,182,193,195,210,218,225\},
\end{multline*}
containing the right--active additive rules $60,90,150$, their affine complements
($105,165,195,\dots$), and the nonlinear rules $18,22,26,82,126,146,154,182$. The window--exceeded class is exactly
$\{30,45,75,135\}$ --- all in the left-linear family. In binary arithmetic, 135 is
complement-conjugate to 30, and 45 and 75 are mirrors. Remark~\ref{war:census} specifies the data behind these lists
(Fig.~\ref{fig:scan}).
\end{remark}

\Needspace{5\baselineskip}
\begin{remark}[Classification mod 3]
The growing class is
\begin{multline*}
\{18,22,26,60,82,89,90,105,123,\\
146,150,154,182,210,218,225\}.
\end{multline*}
The window--exceeded class is again exactly $\{30,45,75,135\}$. The intersection of the two measured growing lists consists of
\begin{multline*}
\{18,22,26,60,82,90,105,\\
146,150,154,182,210,218,225\}.
\end{multline*}
These fourteen have the growing label at both tested primes; rules $89$ and $123$ join only mod~$3$. A large mixed class appears mod~$3$
($115$ rules in the rerun), led by the affine rules: Rule~165 measures
$\{2,2,6,2,6,2,18,2,18,\dots\}=2\cdot3^{\,j}$ with zero transients. Modulo three, the uniform background of Rule 165 alternates
$0\to1\to0$, since $f(0,0,0)=1$ and $f(1,1,1)=0$. For this rule the affine solution is the Rule 90 orbit plus
the alternating uniform background. Let the binomial term have least period $P$, a power of three.
If $q$ is a period of the affine sum, then $2q$ is a period of
the binomial term, so $P\mid q$. The binomial term is therefore
unchanged by the shift $q$; the background forces $q$ to be even.
Conversely $2P$ is a period. The least period is thus $2P$,
including $P=1$ for a constant or zero binomial term.
This argument is particular to Rule 165; a background factor need
not survive in every diagonal of an arbitrary rule. (Rerun: Rule~165 mod~$3$
gives exactly $2,2,6,2,6,2,18,2,18,\dots$.)
\end{remark}

\begin{figure}[htbp]
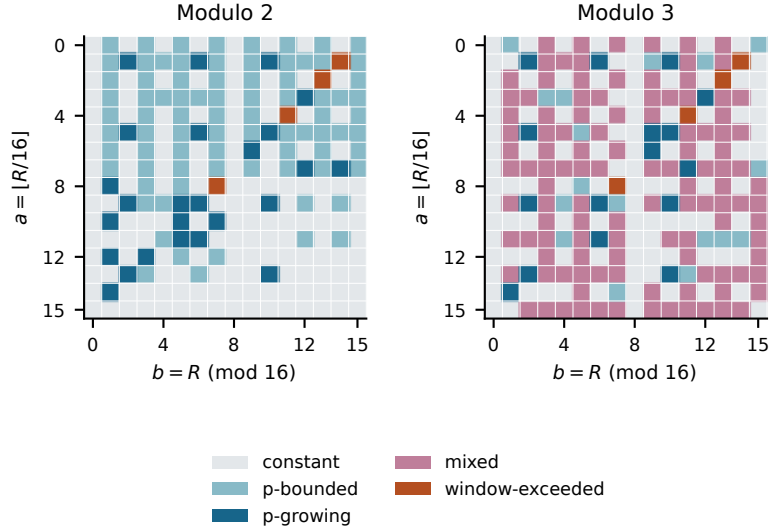
\centering
\EmbeddedFigure{\textwidth}{9}
\caption{The finite period census for all $256$ elementary
rules, with $W=700$ and depths $0\le m\le64$. In each panel the cell
at row $a$ and column $b$ represents $R=16a+b$. Colors are the five
finite classes defined in Section~\ref{sec:scan}; they do not assert
asymptotic growth. The four window-exceeded rules are
$30,45,75,135$ in both panels. The mixed class occurs modulo three.
The grid makes individual rules locatable within the full census.}
\label{fig:scan}
\end{figure}

\noindent\textbf{Finite scope.} The census labels refer to the explicit period
detection criterion on the stated finite depth and width, not to asymptotic
classes of infinite orbits. The census covers all 256 rules modulo
2 and 3, at width 700 through depth 64.
Failure to detect a non-dyadic period on that window is not an all-depth theorem.

\section{Scope of the finite census}

\Needspace{5\baselineskip}
\begin{remark}[Finite scope]\label{war:census}
The comparison uses the single seed, moduli two and three, $W=700$
and $M=64$, giving 512 period profiles. The local polynomial is
Definition~\ref{def:lift}; the rotated update and evolving background
are specified in Section~\ref{sec:recurrent}. For each sampled word,
test candidate periods $1\le q\le\lfloor W/3\rfloor$ in increasing
order and set
\[
 \tau_q=1+\max\bigl(\{-1\}\cup
 \{t:0\le t<W-q,\ s_t\ne s_{t+q}\}\bigr).
\]
Accept the first $q$ with $2\tau_q\le W$ and $3q\le W-\tau_q$,
as in Definition~\ref{def:crit}. If no candidate qualifies, the
window does not determine a period by this criterion.
One profile takes $O((M+1)W^2)$ residue comparisons after
$O(W(M+1))$ local updates. The samples occupy
$O(W(M+1)\log p)$ bits; scanning a candidate needs only a constant
number of additional indices. Factoring a composite modulus is a
separate cost.
\end{remark}

\Needspace{5\baselineskip}
\begin{remark}[Physical boundaries]\label{rem:boundary}
The rotated recurrence avoids artificial spatial edges. A direct
physical calculation can instead start on a sufficiently wide
interval and shrink it by one cell at each end after each update.
Every retained cell then has its complete causal past. Freezing a
narrow boundary can change both transients and period labels,
especially for rules with nonzero background.
\end{remark}

\section{Findings}

\subsection{Two growth regimes, one skeleton}
The exact additive examples have period proportional to depth up
to the jumps between successive prime powers. Their period
logarithms therefore grow logarithmically. The Rule 30 and Rule 45
profiles rise faster over the displayed range and contain plateaus.
Theorem~\ref{thm:fib} proves the Fibonacci interpolation
bound, and Corollary~\ref{cor:fibperiod} supplies an exponential
period ceiling for Rule 30. Neither is an asymptotic equality.
The unboundedness results in \cite{U2} likewise do not determine
the spacing of successive period jumps.

\subsection{Orientation matters: mirror pairs split}
Slope--$1$ diagonals probe the right--moving side of the cone, so the classification is
not mirror--invariant, and the splits are informative. Rule~30 (window exceeded) against
its mirror Rule~86 ($2$-power growing): the right diagonals of one are the left diagonals
of the other, so the split records faster detected growth on the right edge
than on the left over this window --- the slow side being the well--known regularity of Rule~30's
left border, described by the left--column system of the companion--function paper. Rule~60 (growing) against
its mirror $102$ (constant): Rule~102's activity leans entirely left, so every
right--moving diagonal exits its cone and freezes. Rules~45 and~75 mirror each other and
both exceed the window: unresolved at sufficiently many depths on both cuts in this experiment.
Reflection preserves the automaton dynamics while exchanging the
diagonal families being compared.

\subsection{Finite--window census}
For clarity, the five categories give the following complete counts:
\begin{table}[t]\centering
\caption{The finite period census, $W=700,M=64$.}\label{tab:census}
\fontsize{9}{10.5}\selectfont
\begin{tabular}{|l|r|r|}
\hline
Detected-period category&mod 2&mod 3\\
\hline
Constant&146&104\\
$p$-power bounded&81&17\\
$p$-power growing&25&16\\
Mixed&0&115\\
Window-exceeded&4&4\\
\hline
\end{tabular}
\end{table}
These are counts at $W=700,M=64$ with the order of tests in
Definition~\ref{def:crit}. The common growing list has fourteen
rules at these two primes. The same four rules exceed the window.
Changing the depth or width can change a label.

\section{Concluding remarks}
The all-window basis problem has a concrete binary answer: exactly
24 elementary rules in the stated single-seed construction are
universal, and each also works over every modulus. The proof turns
an infinite family of matrix tests into local triangularity and
unit conditions. Two truth-table bits decide the first condition;
the center values decide the second. The modular solver then
constructs the coordinates under those conditions, while the
remaining odd-prime classification is reduced to eight candidates.
This reduction is the main structural outcome of the paper.

The comparison of bases then requires more than a count of nonzero
coefficients. Pascal uniquely converts OR convolution to pointwise
multiplication among triangular binary coordinate maps. The same
transform sends increment to strict prefix summation, giving an explicit
description of both operations needed by the support recurrence.
The examples of inverse transforms show why a dense matrix can
still have a simple evaluation procedure. Conversely, a shorter
coordinate vector does not by itself bound the work needed to
construct or use it.

The recurrence interpretation adds a second consequence of Newton
coordinates. The Rule 30 polynomial system yields a Fibonacci
bound on interpolation order, and binomial periodicity converts
that bound into prime-modulus period control. Exact additive
periods provide cases in which the relation between support and
period can be followed completely. More generally,
Proposition~\ref{prop:exactnewtonperiod} identifies both the exact
recurrence order and the prime-modulus period from the highest
surviving Newton coefficient. It separates integer degree from the
order remaining after modular cancellation. Together, these results explain
which aspects of the diagonal arrays come from a general
coordinate identity and which depend on the chosen rule.

The finite comparisons also delimit possible optimization claims.
Pascal weight need not be minimal among the tested bases, even
when the complete Pascal support is contained in the observation
window. A proposed optimality statement must therefore specify
the target, the admissible bases and the cost being minimized.
Likewise, the census modulo two and three describes period
profiles at its stated depth and width; changes beyond those
windows remain possible.

The classification and the comparisons point to different next
steps. Resolving the eight remaining odd-prime candidates would
extend the all-window theorem. An optimization theorem would need
to control the cost of coordinate conversion and evaluation as
well as coordinate weight. Extending a finite period pattern
requires a recurrence argument at all depths. The results here
provide a defined family of bases and exact operator identities
on which to formulate those questions; the counterexamples and
window comparisons specify where additional hypotheses are needed.

\section{Further work}

The classification reduces the remaining single-seed universality
question to the eight rules in $\mathcal E$ at odd primes. The observed
Rule 84 pattern modulo three requires an invariant or an all-time
recurrence proof. A finite nonzero prefix is insufficient. Once the
prime cases are settled, Corollary~\ref{cor:crt} transfers the unit
criterion to composite moduli; Proposition~\ref{prop:crtperiod}
separately describes how exact periods combine at prime powers.

Proposition~\ref{prop:finiteseed} extends sixteen rules to arbitrary
finite right-hand seeds with a unit left endpoint. The next seed
classification should address the remaining rules and identify how
neighboring seed values can destroy a unit diagonal. For larger
neighborhoods, a comparable local polynomial identity would give a
useful rule-level criterion for staircase matrices. The linear-algebra
theorem already applies once those hypotheses are supplied.

Optimization requires information about both the target and the chosen
window. One aim is to bound $\min_{R\in\mathcal U}w_R(D_m)$ for an
explicit relation between $n$ and $m$. Almost-all counting arguments
cannot establish such bounds for the particular Rule 30 targets.
The counterexample at $m=30,n=4096$ rules out universal weight
minimality of Pascal, while the Rule 124 example at $m=4$ and the
Rule 188 example at $m=30$ exhibit different cancellations. Explaining
those patterns may identify the windows on which Pascal is a minimizer
and whether any eventual comparison holds for fixed depth as $n$ grows.
Finite support in one basis does not itself imply stabilization of
weights in another.

Finally, coordinate weight should be accompanied by transform cost.
For Rule 124, the measured inverse density and failure of one dyadic
factorization leave other fast transforms possible. A new factorization,
or a lower bound in a specified computational model, would resolve
more than another density measurement. Neither the nilpotence
obstruction nor computational universality of a cellular automaton
supplies that finite-matrix lower bound.

\section*{Declarations}

\noindent\textbf{Funding.}
This research received no funding. The author is an independent researcher.

\noindent\textbf{Competing interests.}
The author declares no competing interests.

\end{document}